\documentclass[a4paper,12pt]{report}
\pdfoutput=1

\newcommand{\auth}{Omar Valdivia}   
\newcommand{\thesistitle}{Transgression forms as source for Topological gravity and Chern--Simons--Higgs theories}
\newcommand{\degree}{Doctor of Philosophy} 
\newcommand{\supdate}{August 2014}            

\usepackage[inner=1.5in,top=1in,bottom=0.75in,outer=0.75in,includefoot]{geometry}

\usepackage{setspace}
\usepackage{fancyhdr}

\usepackage{amsmath,amsfonts,amssymb,amsthm,amstext,amscd,array}
\usepackage{mathrsfs}
\usepackage{bbm}
\usepackage{bm}
\usepackage{tikz}
\usepackage{tikz-cd}
\usetikzlibrary{matrix,arrows,decorations.pathmorphing}

\usepackage[citecolor=, urlcolor=,
  linkcolor=, colorlinks=, bookmarksopen=false, bookmarks=true]{hyperref}
\usepackage[nonumberlist]{glossaries}
\makeglossaries

\newacronym{}{}{}

\newacronym{sisd}{SISD}{Single Instruction, Single Data} 

\newacronym{smp}{SMP}{shared memory multiprocessor}

\newacronym{ghc}{GHC}{Glasgow Haskell Compiler}

\newacronym{ghceden}{GHC-Eden}{The Parallel Haskell Compilation System Eden}

\newacronym{gph}{GpH}{Glasgow Parallel Haskell}

\newacronym{gdh}{GdH}{Glasgow Distributed Haskell}

\newacronym{ghcsmp}{GHC-SMP}{Haskell on a Shared Memory Multiprocessor}

\newacronym{hdph}{HdpH}{Haskell Distributed Parallel Haskell}

\newacronym{ghcpps}{GHC-PPS}{GHC Parallel Profiling System}

\newacronym{pe}{PE}{Processing Element}

\newacronym{uma}{UMA}{Uniform Memory Access}

\newacronym{numa}{NUMA}{Non-Uniform Memory Access}

\newacronym{mpp}{MPP}{Massively Parallel Processors}

\newacronym{mpi}{MPI}{Message Passing Interface}

\newacronym{pvm}{PVM}{Parallel Virtual Machine}

\newacronym{hpc}{HPC}{High-Performance Computing}

\newacronym{mcnoc}{McNoC}{Multicores Network-on-Chip}

\newacronym{api}{API}{Application Programming Interface}

\usepackage{graphicx}
\usepackage{subfigure}
\usepackage{xspace}
\usepackage{listings}
\usepackage{pdfpages}
\usepackage{alltt}
\usepackage{float} 
\usepackage[notindex, nottoc, notlof, notlot ]{tocbibind} 
\usepackage[arrow,matrix,curve]{xy}

\newtheorem{theorem}{Theorem}[section]

\newtheorem{definition}{Definition}[section]

\theoremstyle{remark}

  \newcommand{\C}{\complex}
  
  \newcommand{\mbf}[1]{{\boldsymbol {#1} }}
  \def\ii{{\,{\rm i}\,}}
  \def\dd{{\rm d}}

  \def\a{{\sf a}}

  \def\sff{{\sf f}}

  \newcommand{\Hom}{{\rm Hom}}

  \newcommand{\unit}{\mathbbm{1}} 			

  \newcommand{\frg}{\mathfrak{g}}				
  \newcommand{\frh}{\mathfrak{h}}				

  \newcommand{\CCR}{\mathscr{R}}
  \newcommand{\CH}{\mathcal{H}}

  \newcommand{\CG}{\mathcal{G}}
  \newcommand{\CM}{\mathcal{M}}

  \newcommand{\CCM}{\mathscr{M}}

  \newcommand{\Wcal}{\mathcal{W}}

  \newcommand{\eq}{\begin{equation}}
  \newcommand{\eqend}{\end{equation}}
  \newcommand{\eqa}{\begin{eqnarray}}
  \newcommand{\nonueqa}{\begin{eqnarray*}}
  \newcommand{\eqaend}{\end{eqnarray}}
  \newcommand{\nonueqaend}{\end{eqnarray*}}
  
  \newcommand{\bma}[1]{\begin{array}{#1}}
  \newcommand{\ema}{\end{array}}
  \newcommand{\bc}{\begin{center}}
  \newcommand{\ec}{\end{center}}

  \newcommand{\newsection}{\setcounter{equation}{0}\section}

  \newcommand{\complex}{{\mathbb C}} 
  
  \def\alg{{\mathcal A}}

  \newif\ifold             \oldtrue

  \def\e{{\,\rm e}\,}

  \def\be{\begin{equation}}
  \def\ee{\end{equation}}
  \def\bea{\begin{eqnarray}}
  \def\eea{\end{eqnarray}}
  \def\bd{\begin{displaymath}}
  \def\ed{\end{displaymath}}

  \newcommand{\beq}{\begin{eqnarray}}
  \newcommand{\eeq}{\end{eqnarray}}

  \makeatletter
  \newdimen\normalarrayskip              
  \newdimen\minarrayskip                 
  \normalarrayskip\baselineskip
  \minarrayskip\jot
  \newif\ifold             \oldtrue            
  \def\arraymode{\ifold\relax\else\displaystyle\fi} 
  \def\@arrayskip{\ifold\baselineskip\z@\lineskip\z@
       \else
       \baselineskip\minarrayskip\lineskip2\minarrayskip\fi}
  \def\@arrayclassz{\ifcase \@lastchclass \@acolampacol \or
  \@ampacol \or \or \or \@addamp \or
     \@acolampacol \or \@firstampfalse \@acol \fi
  \edef\@preamble{\@preamble
    \ifcase \@chnum
       \hfil$\relax\arraymode\@sharp$\hfil
       \or $\relax\arraymode\@sharp$\hfil
       \or \hfil$\relax\arraymode\@sharp$\fi}}
  \def\@array[#1]#2{\setbox\@arstrutbox=\hbox{\vrule
       height\arraystretch \ht\strutbox
       depth\arraystretch \dp\strutbox
       width\z@}\@mkpream{#2}\edef\@preamble{\halign \noexpand\@halignto
  \bgroup \tabskip\z@ \@arstrut \@preamble \tabskip\z@ \cr}%
  \let\@startpbox\@@startpbox \let\@endpbox\@@endpbox
    \if #1t\vtop \else \if#1b\vbox \else \vcenter \fi\fi
    \bgroup \let\par\relax
    \let\@sharp##\let\protect\relax
    \@arrayskip\@preamble}
  \makeatother

\begin{document}
\doublespacing

\pagestyle{empty}
\begin{center}
\begin{spacing}{2}
{\large{\ \\  \vspace{1.5cm}\textbf{\MakeUppercase{\thesistitle}}}}\\
\end{spacing}
\vfill
{\Large\textit{by}}\\\vspace{0.2cm}
{\Large\upshape{\auth}}\\\vspace{1.0cm}
\includegraphics[width=3.5cm]{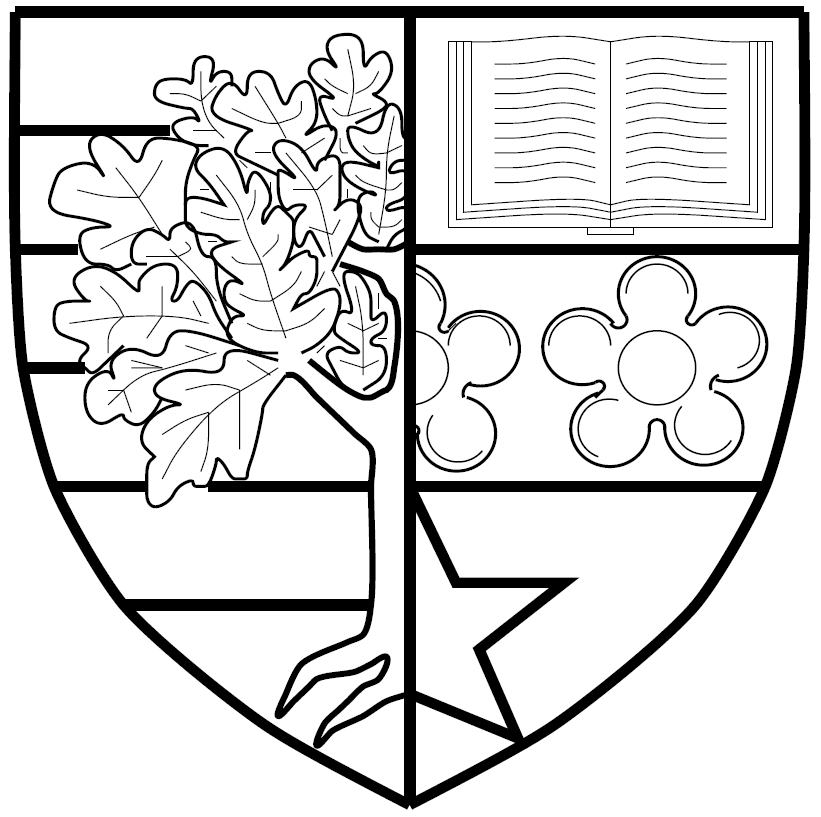}\\
\vspace{1cm}
{\large Submitted for the degree of \\ \degree}\\
\vspace{1cm}
{\large\textsc{Department of Mathematics}\\
\textsc{School of Mathematical and Computer Sciences}\\
\textsc{Heriot-Watt University}}\vfill
{\large{\supdate}}
\end{center}
{\small The copyright in this thesis is owned by the author. Any quotation from the report or use of any of the information contained in it must acknowledge this report as the source of the quotation or information.}
\clearpage
\begin{center}
\LARGE\textbf {Abstract}
\end{center}
\vspace{1cm}

\begin{spacing}{1} 
\noindent

Two main gauge invariant ``off--shell'' models are studied in this Thesis. Both of them constructed by considering different configurations of transgressions forms as Lagrangians.

\begin{description}

\item[i)] Poincar\'e-invariant topological gravity in even
dimensions is formulated as a transgression field theory in one higher dimension whose gauge connections are
associated to linear and nonlinear realizations of the Poincar\'e
group $ISO(d-1,1)$. The resulting theory is a gauged Wess--Zumino--Witten model
whereby the transition functions relating gauge fields belong to the
coset $\frac{ISO(d-1,1)}{SO(d-1,1)}$. The coordinate 
parametrizing the coset space is identified with the scalar field in
the fundamental representation of the gauge group of the even-dimensional
topological gravity theory. The supersymmetric extension 
leads to topological supergravity in two dimensions starting from a transgression field theory
which is
invariant under the supersymmetric extension of the Poincar\'e group
in three dimensions. The construction is extended to a three-dimensional
Chern--Simons theory of gravity invariant under the Maxwell
algebra, where the corresponding Maxwell gauged Wess--Zumino--Witten model is obtained.

\item[ii)] Dimensional reduction of Chern--Simons theories with arbitrary
gauge group in a formalism based on equivariant principal bundles is considered. For
the classical gauge groups the relations between
equivariant principal bundles and quiver bundles is clarified, and show that the
reduced quiver gauge theories are all generically built on the same
universal symmetry breaking pattern. The reduced model is a novel Chern--Simons--Higgs
theory consisting of a Chern--Simons term valued in the residual gauge
group plus a higher order gauge and diffeomorphism invariant coupling of
Higgs fields with the gauge fields. The moduli spaces of solutions provide in some instances geometric representations of certain quiver varieties as moduli spaces of flat invariant connections. 
In the context of dimensional reductions involving non-compact gauge groups, the reduction of five-dimensional supergravity induce novel couplings between gravity
and matter. The resulting model is regarded as to a quiver gauge theory of AdS$_{3} \times \mathrm{U}(1)$ gravity
involving a non-minimal coupling to scalar Higgs fermion fields.

\end{description}

\end{spacing}
\newpage
\bigskip
\begin{flushright}
\textit{To the memory of my father ...}
\end{flushright}
\newpage

\clearpage
\pagestyle{plain}
\clearpage\pagenumbering{roman}
\noindent
{\LARGE\textbf{Acknowledgements}}
\vspace{1cm}

\begin{spacing}{1} 
\noindent
I would like to thank in first place to my family. In particular to my mother and sister for their tremendous love and support. Also to my wife Cecilia, since she arrived to my life my days are full of happiness and without her light, the accomplishment of this Thesis would not have been possible. To my professors Patricio Salgado, for his influence and trust in me not only as a student but most importantly as a person. To Richard Szabo for all his invaluable advices and for showing me what a professional researcher must be. To Jaime Araneda, Ximena Garc\'ia, Luis Roa Robert Weston and Domenico Semianra for their patience and support. To my good friends Leonardo Baez, Ricardo Carocca, Jorge Clarke, Arturo Fern\'andez, Arturo G\'omez, Sergio Inglima, Fernando Izaurieta, Patricio Mella, Nelson Merino, Yazmina Olmos, Alfredo P\'erez, Eduardo Rodr\'iguez and Patricio Salgado Rebolledo. For all of them, thank you for your friendship and enlightening discussions which have been crucial for the development of this Thesis. To the staff people in Concepci\'on as well as at Heriot--Watt University;  Yarela Daroh, June Maxwell, Claire Porter, Marcela Sanhueza, Heraldo Manr\'iquez and Victor Mora; thank you very much for all your help and kindness.

To finalize, I would like to thank to different institutions which have made this investigation possible. To Universidad de Concepci\'on for its dependencies and financial support, to Heriot--Watt University and Maxwell Institute for Mathematical Sciences for its facilities as well as to the Comosi\'on Nacional de Ciencia y Tecnolog\'ia CONICYT for the PhD scholarship grants and co-tutela award.

\end{spacing}

\tableofcontents
\listoftables
\listoffigures
\clearpage
\pagestyle{fancy}
\pagenumbering{arabic}
\fancyhead{}
\lhead{\slshape \leftmark} 
\cfoot{\thepage}
\renewcommand{\headrulewidth}{0.4pt}
\renewcommand{\footrulewidth}{0.0pt}
\renewcommand{\chaptermark}[1]{\markboth{\chaptername\ \thechapter:\ #1}{}}

\chapter{Introduction}
\label{ch:intro}
\begin{flushright}
\textit{``...es una historia que tiene que ver con el curso de la V\'ia Láctea...''  \\ Silvio Rodriguez, Canci\'on del elegido.} \footnote{\scriptsize ``...this is a story which deals with the course of the Milky Way...''. Silvio Rodriguez, Song of the chosen.}
\end{flushright}

\newsection{Motivation}

Physics is a extremely powerful science. It has the property of describing the majority of phenomena observed in nature by using only a set of mathematical equations. This equations encode the information about the fundamental laws which govern processes at different scales in the universe.

Although this is a highly nontrivial fact, it has been the source of curiosity of many people in the course of human kind. The development of physics has revealed another and even more astonishing fact; the idea of that all phenomena in nature seems to share a common origin, even when at first light they exhibit a completely different provenance. This suggest that if all the processes are governed by the same set of fundamental rules, it may be possible to create a single physical framework in which all the interactions are
described in a consistent manner. This seductive idea has become into a very important concept for the scientific community and it is often referred to as \textit{unification}. One of the first lights in this direction was given
by Sir Isaac Newton in XVII century. Newton realized how to unify the motion of planets around the Sun with the laws of movement of bodies on Earth. This is known as Newton's Universal Gravitational Law. Later on, almost
three centuries after, this ordinary equation allowed human to put feet on the surface of the Moon. Of course, much more work has been done since Newton's contribution. In
particular, the works of Faraday and Amp\`ere in the field of electricity and magnetism unified by James Maxwell in what is known today as Electrodynamics, as well as the unification of Newton's gravitational laws with the concept of
space-time made by Albert Einstein in the framework of Special and subsequently General Relativity
in 1915. Today, with the development of Quantum Mechanics, the Standard model of High Energy Physics is the most successful, accurate and predictive model for the interaction of particles. In this framework, three of the four known
interactions of nature are unified (Electromagnetic, Strong Nuclear and Weak Nuclear) and experimentally corroborated with enormous precision. The dynamical content of the theory is written in terms of a Yang--Mills action functional which is built on the assumption that interactions of nature
should be unchanged by a specific group of transformations acting on each point of space-time; a local \textit{gauge symmetry}. This symmetry principle is of vital importance because it sets by one hand, the fundamental
constituents of matter and on the other hand the carriers of interactions. The success of the Standard Model lies on the fact that it is renormalizable and anomaly free; these are highly restrictive conditions so any theory which avoids inconsistencies like those is a believable tool and, at the same time, a prime
criterion for its construction.

The advantage of gauge invariance in quantum systems is that the gauge symmetry does not depend on the field configurations. Since this symmetry relates the divergences appearing in the scattering amplitudes in such a way that they can be absorbed in the coupling constants at some order of the loop expansion, it should remain at all orders in perturbation theory so that the gauge symmetry it is not spoiled by quantum corrections. This is in contrast with some other theories in which gauge symmetry is only realized by using the equations of motion i.e., an \textit{``on-shell''} symmetry. This kind of symmetries are usually broken by quantum mechanics.

Even when there is no clear
understanding if all the gauge invariant theories are renormalizable, the only renormalizable theories which describe our universe are gauge theories. This is really unexpected since gauge symmetry was initially introduced with the motivation of providing a systematic way to describe interactions which respect some symmetry principles more than curing renormalization. Thus, gauge invariance seems to be it a key
ingredient in the construction of experimentally testable theories since symmetry principles are then not only useful in the construction of classical actions, they are also sufficient conditions to ensure the viability of the quantization procedure of a classical action.

The underlying structure of the gauge invariance is mathematically captured
through the concept of \textit{fibre bundle}, which is a systematic way to implement
a group acting on a set of fields that carry a particular representation of the
group~\cite{Nakahara,Eguchi:1980jx,deAz95}. 

The gravitational interaction, in contrast, has stubbornly
resisted quantization. This does not mean that General Relativity is
constructed over weak postulates. In fact, this theory predicts in a good way how the
Universe it is behaved at large scales as well as the dynamic of bodies around
super-massive objects like black holes and neutron stars. However, General Relativity cannot be understood as a quantum field theory due to its perturbative expansion is not renormalizable~\cite{'tHooft:1974bx}. The situation with quantum gravity is particularly annoying because one have been led to think that gravitational attraction is a fundamental interaction. In fact, Einstein equations can be derived form an action principle so one could expect that a path integral for the gravitational field can be defined. But it may also be possible  that General Relativity is only an effective theory for gravity in four
dimensions~\cite{Burgess:2003jk}. Despite these arguments, General Relativity seems to be the consistent
framework compatible with the idea that physics should be insensitive
to the choice of coordinates or the state of motion of any observer;
this is expressed mathematically as invariance under
reparametrizations or local diffeomorphisms. Although this reparametrization invariance constitutes
a local symmetry, it does not qualify as a gauge symmetry. The reason
is that gauge transformations act on the fields while diffeomorphisms
act on their arguments, i.e., on the coordinates. A systematic way to
circumvent this obstruction is by using the tangent space
representation; in this framework gauge transformations constitute
changes of frames which leave the coordinates unchanged. However,
general relativity is not invariant under local
translations, except by a special accident in three spacetime
dimensions where the Einstein--Hilbert action is purely topological. This Thesis is intended to shed some light in the
question of how to construct a gauge invariant version of gravity. The
motivations are quite obvious; in complete correspondence to the Standard
Model, if the truly theory of gravity is such that it can be written as a gauge theory, then in principle it could accept a quantization procedure. This would promote gravity in to the
same footage as the Standard Model, providing in this way for a candidate for the unified
theory of the fundamental interactions. 
 
The classification of topological gauge theories of gravity was introduced by Ali Chamseddine in late eighties~\cite{Chamseddine:1989yz,Cha89,Cha90}. The natural gauge groups $G$ considered are the anti-de~Sitter
group $SO(d-1,2)$, the de~Sitter group $SO(d,1)$, and the Poincar\'e
group $ISO(d-1,1)$ in $d$ spacetime dimensions depending on the sign of the
cosmological constant: $-1,+1,0$ respectively. In odd dimensions
$d=2n+1$, the gravitational theories are constructed in terms of
secondary characteristic classes called
Chern--Simons forms~\cite{Chern:1974ft}. Chern--Simons forms are useful objects because they
lead to gauge invariant theories (modulo boundary terms). They also
have a rich mathematical structure similar to those of the (primary)
characteristic classes that arise in Yang--Mills theories: they are constructed in terms of a gauge potential which descends from a connection on a principal bundle.
 In even dimensions, there is no
natural candidate such as the Chern--Simons forms; hence in order to construct an invariant $2n$-form, the product of $n$ field strengths is not
sufficient and requires the insertion of a scalar multiplet $\phi^{a}$ in the fundamental
representation of the gauge group $G.$ This requirement
ensures gauge invariance but it threatens the topological origin of
the theory. 

However, in this Thesis it will be shown that even-dimensional topological gravity actually encodes a very elegant topological interpretation; it can be
formulated in terms of a generalization of Chern--Simons forms, namely, a \textit{transgression} form~\cite{Merino:2010zz,Salgado:2013pva,PhysRevD.89.084077}.
The construction of gauge theories using transgression forms as Lagrangians is relatively new~\cite{Mora:2003wy,Mora:2004kb,Mora:2006ka} and provides of many advantages form the physical point of view. It will be shown that even-dimensional topological gravity can be written as a transgression field theory which is deeply connected with topological quantities called gauged Wess--Zumino--Witten terms~\cite{Wess:1967jq,Witten:1983tw}. This construction will also be extended to supersymmetry. 

Another interesting area for the study of gauge theories is in the context of dimensional reduction. Dimensional reduction provides a means of unifying gauge and Higgs sectors into a pure Yang--Mills theory in higher dimensions. The reductions are particularly rich if the extra spacetime
dimensions admit isometries, which can then be implemented on gauge orbits of fields~\cite{ForgaCS:1979zs}. The natural setting for spacetime isometries are coset
spaces $G/H$ of compact Lie groups in which Yang--Mills theory on the product space $ M\times
G/H$ is reduced to a Yang--Mills--Higgs theory on the manifold $M$; the construction can be 
extended supersymmetrically and also embedded in string theory \cite{Kapetanakis:1992hf}.
Equivariant dimensional reduction is an alternative approach
which naturally incorporates background fluxes coming from the topology of the canonical connections on the principal $H$-bundle $G\to G/H$~\cite{AlGar12,Lechtenfeld:2007st,Dolan:2010ur}; the reduced Yang--Mills--Higgs model is then succinctly described by a quiver gauge theory on $M$ whose underlying quiver is canonically associated to the representation theory of the Lie groups $H\subset G$. Such reductions have been used to
describe vortices as generalized instantons in higher-dimensional
Yang--Mills theory \cite{AlvarezConsul:2001mb,Lechtenfeld:2006wu,Popov:2007ms,Lechtenfeld:2008nh,Popov:2010rf}, as well as to construct explicit
$SU(2)$-equivariant monopole and dyon solutions of pure Yang--Mills theory in
four dimensions~\cite{Popov:2008wh}. 
A related approach is described in \cite{Manton:2010mj} which systematically translates the inverse relations of restriction and induction of vector bundles \cite{AlGar12} into the framework of
principal bundles. In this formulation there is no restriction on the
structure group and it permits, for instance, the application of equivariant dimensional
reduction techniques to gauge theories involving arbitrary gauge groups
$\mathcal{G}$. In the following we adapt such an approach to the simplest case where
$G=SU(2)  $ and $H=U(1)$, so that the internal coset
space $G/H$ is the two-sphere $S^2$ or the complex projective line $%
\mathbb{C}
P^{1}$. This example turns out to be rich enough to capture many of the general
features that one would encounter on generic cosets $G/H$.

The second result presented in this Thesis is the equivariant dimensional
reduction of topological gauge theories~\cite{Szabo:2014zua}. We calculate the reduction of an arbitrary odd-dimensional Chern--Simons
form over $%
\mathbb{C}
P^{1}$; although Chern--Simons Lagrangians are not gauge-invariant, we circumvent this problem by regarding them in the framework of transgression forms.
The reduced theory is a novel diffeomorphism-invariant
Chern--Simons--Higgs model, composed by a lower dimensional Chern--Simons form coupled to residual magnetic monopole charges plus a non-minimal coupling between the curvature 2-form and the Higgs fields. 

 As physical application, we consider the case of non-compact gauge supergroups. In order to make contact with topological gauge theories of gravity, we perform the dimensional reduction of five-dimensional
Chern--Simons supergravity over $\mathbb{C}P^1$. It will be shown that if the
Higgs fields are bifundamental fields in the fermionic sector of the gauge algebra, then the reduced action contains the standard
Einstein--Hilbert term plus a non-minimal coupling of the Higgs fermions to the curvature.
This reduction scheme thus constitutes a novel systematic way to couple scalar fermionic fields to
gravitational Lagrangians, in a manner whereby non-vacuum solutions of three-dimensional anti-de Sitter gravity can be lifted to give new solutions of five-dimensional supergravity on product spacetimes $M\times S^2$. 

\newsection{Plan of the Thesis}

This Thesis is organized as follows. In Chapter \ref{ch:background}, we review the mathematical background mainly used in the context of gauge theories, namely, the concept of fibre bundle, the notion of a connection over a principal bundle and the Chern--Weil Theorem via the homotopy formula.
 
 In Chapter~\ref{ch:Transgression}, we analyse the general aspects of the construction of gauge invariant theories using transgression forms as Lagrangians. Also, the relations between transgression and Chern--Simons field theories are discussed as well as its relation with the gauged Wess--Zumino--Witten models.
 
 In Chapter~\ref{ch:top_grav},
we review some aspects of topological gravity and how it can be connected with Lanczos--Lovelock theories of gravity. We review the formalism of nonlinear realizations of Lie groups and its application to the case of Chern--Simons gravity invariant under the Poincar\'e group. Then, in Section \ref{firstresult} one of the main results of this Thesis is presented: the construction of even-dimensional topological gravity as a transgression field theory. As a representative example of how to incorporate
fermions into our construction, we derive the supersymmetric extension of even-dimensional
topological gravity starting form a Chern--Simons supergravity action in three dimensions.

Chapter~\ref{ch:Max_alg}, contains an application in which we construct the gauged Wess--Zumino--Witten
model associated to the Maxwell algebra, and it is shown that the Maxwell algebra can be obtained as an $S-$expansion procedure. 

 In Chapter~\ref{ch:covquiv} we discuss general
aspects of $SU(2)$-equivariant dimensional reduction and
revisit the example of pure Yang--Mills theory as illustration. Also, as the second main result of this Thesis, the
symmetry breaking patterns are analysed for the classical gauge groups
and the geometric structure of general principal quiver bundles is described. Then, in Section \ref{secondresult},
we derive the $SU(2)$-equivariant dimensional reduction of Chern--Simons
gauge theories in arbitrary odd dimensionality and discuss some explicit examples.
Finally, we carry out the dimensional reduction of
five-dimensional Chern--Simons supergravity and point out some possible implications.

Chapter~\ref{ch:conclusions} contains a summary and discussion of the main results presented in this Thesis as well as some future possible developments.

Four appendices conclude this Thesis: In Appendix~\ref{ch:app2} we resume the general properties of nonlinear realizations of Lie group theory. Appendix~\ref{ch:app1} contains
some technical details about the construction of Chern--Simons supergravity
actions and a brief summary and conventions for spinors in three and five dimensions. In Appendix~\ref{ch:app3} we summarise the construction of the $S$-expansion method for Lie algebras. Finally, in Appendix~\ref{ch:app4} the group theory data for the decomposition of the classical gauge groups in terms of the Cartan--Weyl basis is summarized.

\chapter{Mathematical Background}
\label{ch:background}
\begin{flushright}
\textit{``...labourers of science is what we are.''  P. S.}
\bigskip
\end{flushright}
In order to describe gauge theories in a rigorous way it is necessary to
introduce the concept of \textit{Fibre Bundle}. A fibre bundle can be seen as a topological space which locally is the product of two manifolds. In the case when the product is also globally defined, the fibre bundle is said to be trivial. 

This mathematical concept is of vital importance in the construction of gauge theories. For this reason, a general however not complete, description will be given through this chapter. For a more detailed analysis, see the books~\cite{Nakahara,Eguchi:1980jx,deAz95}.

\newsection{Fibre bundles}

\begin{definition}
A fibre bundle is composed by the set $\left\{  \mathcal{E},\mathcal{M}%
,\mathcal{F},\pi\right\}  $ where $\mathcal{E},\mathcal{M},\mathcal{F}$ are
topological spaces and $\pi:\mathcal{E}\rightarrow\mathcal{M}$ is a continuous
and surjective projection map. $\mathcal{E}$ is referred as to the total
space, $\mathcal{M}$ is the base space and $\mathcal{F}$ is the fibre.

The local triviality condition $\varphi$ consists in the requirement that for
each $x\in\mathcal{M}$ there exist an open set $U$ such that
$\pi^{-1}\left(  U\right)  $ is homeomorphic to the direct product
$U\times\mathcal{F}$. The homeomorphism $\varphi:\pi^{-1}\left(  U\right)
\rightarrow U\times\mathcal{F}$ is such that the  diagram~(\ref{lotr}) commutes,
\begin{figure}[h!]
\centering
\begin{tikzcd}[column sep=6em,row sep=6em]
\pi^{-1}(U) \arrow[]{r}{\varphi}\arrow[]{d}{\pi}
&U \times F \arrow[]{dl}{\mathrm{proj}_1}\\
U&
\end{tikzcd}
\caption{Local Trivialization} \label{lotr}
\end{figure}
where $\mathrm{proj}_{1}$ denotes the standard projection map
$\mathrm{proj}_{1}:U\times\mathcal{F}\rightarrow U$.
\end{definition}

Let $\left\{  U_{\alpha}\right\}  $ be an open covering of $\mathcal{M}$ in
such a way that $\bigcup\limits_{\alpha}U_{\alpha}=\mathcal{M}$. Each open
$U_{\alpha}$ has a homeomorphism $\varphi_{\alpha}$ associated. The set
$\left\{  U_{\alpha},\varphi_{\alpha}\right\}  $ corresponds to a
\textit{local trivialization} of the fibre bundle. Thus, for each
$x\in\mathcal{M}$ the pre-image $\pi^{-1}\left(  x\right)  $ is homeomorphic to
$\mathcal{F}$ and it will be called \textit{the fibre }in $x$. 

\subsection{Transition functions}

In order to describe the fibre bundle completely in terms of the local
trivializations $\left\{  U_{\alpha},\varphi_{\alpha}\right\}  $ it is
necessary to find juncture conditions in the non-empty overlaps between
different open sets $U$. Given two opens $U_{\alpha}$ and $U_{\beta}$ with
non-empty overlaps $U_{\alpha}\cap U_{\beta}\neq\varnothing$, they respective
local trivializations $\varphi_{\alpha}$ and $\varphi_{\beta}$ will, in general, map $\pi^{-1}\left(  U_{\alpha}\cap U_{\beta
}\right)  $ to $\left(  U_{\alpha}\cap U_{\beta}\right)  \times\mathcal{F}$ in
a different way%
\begin{align}
\varphi_{\alpha} &  :\pi^{-1}\left(  U_{\alpha}\cap U_{\beta}\right)
\rightarrow\left(  U_{\alpha}\cap U_{\beta}\right)  \times\mathcal{F},\\
\varphi_{\beta} &  :\pi^{-1}\left(  U_{\alpha}\cap U_{\beta}\right)
\rightarrow\left(  U_{\alpha}\cap U_{\beta}\right)  \times\mathcal{F},
\end{align}
or more explicitly%
\begin{align}
\varphi_{\alpha}\left(  p\right)   &  =\left(  \pi\left(  p\right)
,y_{\alpha}\left(  p\right)  \right)  =\left(  x,y_{\alpha}\left(  p\right)
\right)  ,\\
\varphi_{\beta}\left(  p\right)   &  =\left(  \pi\left(  p\right)  ,y_{\beta
}\left(  p\right)  \right)  =\left(  x,y_{\beta}\left(  p\right)  \right)  ,
\end{align}
where $x\in U_{\alpha}\cap U_{\beta}$, $y_{\alpha}\in\mathcal{F}$. This
induces the following
\begin{definition}
the composite map of $\varphi_{\alpha}$ and $\varphi_{\beta}^{-1}$ is
\begin{equation}
\varphi_{\alpha}\circ\varphi_{\beta}^{-1}:\left(  U_{\alpha}\cap U_{\beta
}\right)  \times\mathcal{F}\rightarrow\left(  U_{\alpha}\cap U_{\beta}\right)
\times\mathcal{F},
\end{equation}
or in a more explicit way%
\begin{equation}
\varphi_{\alpha}\circ\varphi_{\beta}^{-1}\left(  x,y_{\beta}\right)  =\left(
x,\tau_{\alpha\beta}\left(  x\right)  y_{\beta}\right)  ,
\end{equation}
where $\tau_{\alpha\beta}\left(  x\right)  $ corresponds to a continuous
left-acting operator over $\mathcal{F}$. The mappings $\tau_{\alpha\beta
}\left(  x\right)  :\mathcal{F}\rightarrow\mathcal{F}$ are called the
transitions functions and satisfy the following properties
\begin{enumerate}
\item $\tau_{\alpha\alpha}\left(  x\right)  =\mathrm{I}_{\mathcal{F}},$
\item $\tau_{\alpha\beta}\left(  x\right)  =\tau_{\beta\alpha}^{-1}\left(
x\right)  ,$
\item $\tau_{\alpha\gamma}\left(  x\right)  =\tau_{\alpha\beta}\left(
x\right)  \tau_{\beta\gamma}\left(  x\right)  .$
\end{enumerate}
\end{definition}

The last condition holds in the case $U_{\alpha}\cap U_{\beta}\cap U_{\gamma
}\neq\varnothing$ and is called the co-cycle condition. These three properties
implies that the transition functions $\tau$ form a group; the structure group
$\mathcal{G}$%
\begin{equation}
\tau_{\alpha\beta}:U_{\alpha}\cap U_{\beta}\rightarrow\mathcal{G}.%
\end{equation}

In the case of a \textit{smooth fibre bundle }($\mathcal{E},\mathcal{M}%
,\mathcal{F}$ differentiable manifolds and, $\pi,\varphi$ smooth maps),
$\mathcal{G}$ corresponds to a \textit{Lie group}. Throughout this thesis only
this kind of fibre bundle are considered. Note that the transition functions
were chosen acting by the left, but this is only a matter of convention. In
principle, it is possible to define the \textit{right action} of a group over
the fibre and therefore inducing the right action over the whole bundle.

\begin{definition}
Let $g$ be an element of the structure group $\mathcal{G}$. We denote the right
action of $\mathcal{G}$ over the fibre $\mathcal{F}$ as%
\begin{equation}
y_{\alpha}^{\prime}\left(  p\right)  =y_{\beta}\left(  p\right)  g.
\end{equation}
Let $p,p^{\prime}$ be two point of the fibre bundle. It will be said that%
\begin{equation}
p^{\prime}=pg
\end{equation}
if the following conditions are satisfied%
\begin{align}
\pi\left(  pg\right)   &  =\pi\left(  p\right)  ,\\
y_{\beta}\left(  pg\right)   &  =y_{\beta}\left(  p\right)  g.
\end{align}
\end{definition}

The induced action over the bundle is independent of the homeomorphism
$\varphi$ and therefore of the choosing of $y_{\alpha}$. In fact, since%
\begin{equation}
y_{\alpha}\left(  p\right)  =\tau_{\alpha\beta}\left(  x\right)  y_{\beta
}\left(  p\right)
\end{equation}
it follows,%
\begin{align}
y_{\alpha}\left(  pg\right)   &  =\tau_{\alpha\beta}\left(  x\right)
y_{\beta}\left(  pg\right)  \\
&  =\tau_{\alpha\beta}\left(  x\right)  y_{\beta}\left(  p\right)  g\\
&  =y_{\alpha}\left(  p\right)  g.
\end{align}

The existence of the left action $L_{g}$ and the right action $R_{g}$ of an
element $g\in\mathcal{G}$ acting on the bundle suggest the definition of a
particularly elegant structure: the \textit{principal bundle}

\begin{definition}
A principal bundle is a fibre bundle in which
\begin{enumerate}
\item The fibre $\mathcal{F}$,
\item The set of transitions functions $\left\{  \tau_{\alpha\beta}\right\}  $,
\item The structure group $\mathcal{G}$ acting on the right, 
\end{enumerate}
\end{definition}
\textit{corresponds to the same Lie group} $\mathcal{G}$.

As it will be seen later, this kind of fibre bundle is the fundamental block
in the construction of gauge theories. In general, principal bundles are
denoted by $\mathcal{P}$ instead of $\mathcal{E}$ and of course the fibre
$\mathcal{F}$ is in this case denoted by $\mathcal{G}$.

\subsection{Local sections}

\begin{definition}
given the open covering $\left\{  U_{\alpha}\right\}  $ for $\mathcal{M}$ it
is possible to define a local section $\sigma_{\alpha}$ as the map%
\begin{equation}
\sigma_{\alpha}:U_{\alpha}\longrightarrow\mathcal{P},%
\end{equation}
such that $\forall x\in U_{\alpha}$
\begin{equation}
\pi\circ\sigma_{\alpha}\left(  x\right)  =x.
\end{equation}

\end{definition}

A section $\sigma_{\alpha}$ and the local homeomorphism $\varphi_{\alpha}%
:\pi^{-1}\left(  U_{\alpha}\right)  \rightarrow U_{\alpha}\times\mathcal{G}$
are intimately related. In fact, given a local trivialization $\varphi
_{\alpha}$ induces a particularly local section called \textit{natural section
}%
\begin{equation}
\sigma_{\alpha}\left(  x\right)  =\varphi_{\alpha}^{-1}\left(  x,e\right),
\end{equation}
where $e$ is the identity element of $\mathcal{G}$. The inverse affirmation
it is also true: a section induces a reciprocal local homeomorphism $\varphi$.
Thus, let $y_{\alpha}:\pi^{-1}\left(  U_{\alpha}\right)  \rightarrow
\mathcal{G}$ be the map such that%
\begin{equation}
\sigma_{\alpha}\left(  x\right)  y_{\alpha}\left(  p\right)  =p, \label{ch2.2}%
\end{equation}
where $y_{\alpha}\left(  p\right)  \in\mathcal{G}$ is acting by the right over
the point $\sigma_{\alpha}\left(  x\right)  $. Note that $\sigma_{\alpha
}\left(  x\right)  =p\left[  y_{\alpha}\left(  p\right)  \right]  ^{-1}$ it
does not depend on $p$. In fact, let $p^{\prime}=pg$ then%
\begin{align}
p^{\prime}\left[  y_{\alpha}\left(  p^{\prime}\right)  \right]  ^{-1}  &
=pg\left[  y_{\alpha}\left(  pg\right)  \right]  ^{-1} \nonumber \\
&  =pgg^{-1}\left[  y_{\alpha}\left(  p\right)  \right]  ^{-1}\nonumber \\
&  =p\left[  y_{\alpha}\left(  p\right)  \right]  ^{-1}.
\end{align}
Thus, it is possible to define
\begin{equation}
\varphi_{\alpha}\left(  p\right)  =\left(  x,y_{\alpha}\left(  p\right)
\right),
\end{equation}
where the condition eq.$\left(  \ref{ch2.2}\right)  $ implies that
$\sigma_{\alpha}\left(  x\right)  =\varphi_{\alpha}^{-1}\left(  x,e\right)  $.

Given two local sections $\sigma_{\alpha}$ and $\sigma_{\beta}$ over
$U_{\alpha}\cap U_{\beta}$, it is possible to show that they are related by the transition functions
$\tau_{\alpha\beta}:U_{\alpha}\cap U_{\beta}\rightarrow\mathcal{G}$. Let
$\sigma_{\alpha}$ and $\sigma_{\beta}$ be natural sections. Then, using eq.$\left(
\ref{ch2.2}\right)  $ it follows%
\begin{equation}
\sigma_{\alpha}\left(  x\right)  y_{\alpha}\left(  p\right)  =\sigma_{\beta
}\left(  x\right)  y_{\beta}\left(  p\right)  ,
\end{equation}
but
\begin{equation}
y_{\alpha}(p)=\tau_{\alpha \beta}(x) y_{\beta}(p),
\end{equation}
and therefore,
\begin{align}
\sigma_{\beta}\left(  x\right)   &
=\sigma_{\alpha}\left(  x\right)  y_{\alpha}\left(  p\right)  \left[
y_{\beta}\left(  p\right)  \right]  ^{-1} \nonumber\\  
& =\sigma_{\alpha}\left(  x\right)  \tau_{\alpha\beta}(x)y_{\beta}\left(  p\right)  \left[
y_{\beta}\left(  p\right)  \right]  ^{-1} \nonumber\\
&  =\sigma_{\alpha}\left(  x\right)  \tau_{\alpha\beta}\left(  x\right) .
\end{align}

\subsection{Symmetry group $\mathcal{G}$}

Since in the case of principal bundles the fibre $\mathcal{F}$ coincides with
the structure group $\mathcal{G}$, it is fair at this point to precise some important
features. Throughout this thesis the concepts of \textit{Lie algebra} and
\textit{Lie group} will be used interchangeably. 
Since $\mathcal{G}$ corresponds to a Lie group, it posses a manifold
structure. So, it is possible to define the \textit{tangent} space
$T_{p}\left(  \mathcal{G}\right)  $ as the set of tangent vectors of
$\mathcal{G}$ at the point $p$. In particular, the Lie algebra $\mathfrak{g}$ associated to
$\mathcal{G}$ corresponds to the tangent space of the identity element
$T_{e}\left(  \mathcal{G}\right)  $. The tangent space $T_{e}\left(  \mathcal{G}\right)  $ has vector
space structure so if $\left\{  \mathsf{T}_{A}\right\}  $ is a basis for
$T_{e}\left(  \mathcal{G}\right)  =\mathfrak{g}$; then%
\begin{equation}
\left[  \mathsf{T}_{A},\mathsf{T}_{B}\right]  =C_{AB}^{~~C}\mathsf{T}%
_{C},\label{pr1}%
\end{equation}
where $C_{AB}^{~~C}$ are known as the structure constants which satisfy%
\begin{equation}
C_{AB}^{~~C}=-C_{BA}^{~~C},
\end{equation}
and the \textit{Jacobi identity}%
\begin{equation}
C_{AB}^{~~C}C_{CD}^{~~E}+C_{DA}^{~~C}C_{CB}^{~~E}+C_{BD}^{~~C}C_{CA}^{~~E}=0.
\end{equation}
The commutator eq.$\left(  \ref{pr1}\right)  $ induce the notion of commutator
of differential forms valued in the Lie algebra $\mathfrak{g}$. Let $P$ and $Q$ be a $p$ and a
$q-$form respectively, defined over a manifold $\mathcal{X}$ and taking values in
the Lie algebra $\mathfrak{g}$, i.e.,%
\begin{align}
P &  =P^{A}\mathsf{T}_{A},\\
Q &  =Q^{B}\mathsf{T}_{B},
\end{align}
where $P^{A}\in\Omega^{p}\left(  \mathcal{X}\right)  $ and $Q^{B}\in\Omega
^{q}\left(  \mathcal{X}\right)  $. The commutator $\left[  P,Q\right]  $ is
defined by%
\begin{align}
\left[  P,Q\right]   &  =P^{A}\wedge Q^{B}\left[  \mathsf{T}_{A}%
,\mathsf{T}_{B}\right]   \nonumber \\
&  =P^{A}\wedge Q^{B}C_{AB}^{~~C}\mathsf{T}_{C}.
\end{align}

Since $T_{e}\left(  \mathcal{G}\right)  =\mathfrak{g}$, it is
interesting to define a way to map vectors in $T_{g}\left(  \mathcal{G}\right)$ to vectors in $T_{e}\left(  \mathcal{G}\right)$. In order to do so, it is necessary to introduce some definitions.
Let $\mathcal{X}$ and $\mathcal{W}$ be two manifolds and let $f:\mathcal{X}$
$\longrightarrow\mathcal{W}$ be a map between them. We denote by $f^{\ast}$ to
the reciprocal image or \textit{pull-back} induced by $f$ over a form in
$\mathcal{W}$ to a form in $\mathcal{X}$. Moreover, we denote by $f_{\ast}$ to the direct image
or \textit{pushforward} induced from a vector in $\mathcal{X}$ to a vector in
$\mathcal{W}$. 

Now, given a vector $X\left(  g\right)  \in
T_{g}\left(  \mathcal{G}\right)  $, the corresponding element of the Lie
algebra $\mathsf{X}\in T_{e}\left(  \mathcal{G}\right)  =\mathfrak{g}$
is given by%
\begin{equation}
\mathsf{X}=L_{g^{-1}\ast}\left(  X\left(  g\right)  \right)  .
\end{equation}
This operation induces the definition of the canonical form of a Lie group,
 the\textit{ Maurer--Cartan form} $\mathrm{\theta}\left(  g\right)  \in\Omega
^{1}\left(  \mathcal{G}\right)  \otimes\mathfrak{g}$, as the form which
satisfies%
\begin{equation}
\mathrm{\theta}\left(  g\right)  \left(  X\left(  g\right)  \right)
=\mathsf{X}. \label{MCform}
\end{equation}

Given a matrix representation of $\mathcal{G}$, it is direct to show that the
condition%
\begin{equation}
\mathrm{\theta}\left(  g\right)  (  X\left(  g\right) )
=L_{g^{-1}\ast}\left(  X\left(  g\right)  \right)
\end{equation}
implies%
\begin{equation}
\mathrm{\theta}\left(  g\right)  =\mathrm{Ad}_{g^{-1}}(\dd_{\mathcal{G}})=g^{-1}\dd_{\mathcal{G}}g, \label{pr5}%
\end{equation}
where $\dd_{\mathcal{G}}$ is the exterior derivative over $\mathcal{G}$. Since
$\mathrm{\theta}=\theta^{A}\mathsf{T}_{A}$, the components $\theta^{A}$ form
a dual representation of the Lie algebra and satisfy the Maurer--Cartan
structure equations%
\begin{equation}
\dd_{\mathcal{G}}\theta^{C}+\frac{1}{2}C_{AB}^{~~C}\theta^{A}\wedge\theta
^{B}=0.\label{pr2}%
\end{equation}
The Maurer--Cartan equations are dual to eq.$\left(  \ref{pr1}\right)  $, they
carry the same information. Thus, the dual version of the Jacobi identity is
simply the exterior derivative of eq.$\left(  \ref{pr2}\right)  $,%
\begin{equation}
\frac{1}{2}C_{AB}^{~~C}C_{DC}^{~~E}\theta^{D}\wedge\theta^{A}\wedge\theta
^{B}=0.
\end{equation}

\newsection{Connections over principal bundles}

By definition, any principal fibre bundle $\mathcal{P}$ is locally a structure
of the form $U\times\mathcal{G}$. It seems reasonable though to expect that
the tangent space associated with the fibre bundle can be decomposed in a
direct sum structure. This decomposition can be made in such a way that the
tangent space $T_{p}(\mathcal{P})$ is the direct sum of a \textit{vertical} component $V_{p}(\mathcal{P})$ tangent to
the fibre, and a horizontal component $H_{p}(\mathcal{P})$ which is orthogonal
respect to $V_{p}(\mathcal{P})$. This operation is systematically implemented by
using the so called \textit{Ehresmann Connection} \cite[Chapter 10]{Nakahara}.

Let $T_{p}\left(  \mathcal{P}\right)  $ be the tangent space associated to the principal
bundle $\mathcal{P}$ at the point $p$, and let us decompose it as $T_{p}\left(
\mathcal{P}\right)  =V_{p}\left(  \mathcal{P}\right)  \oplus H_{p}\left(
\mathcal{P}\right)  $. Here, the vertical space $V_{p}\left(  \mathcal{P}\right)  $ corresponds to
the tangent space respect to the fibre $\mathcal{G}$ and the horizontal subspace $H_{p}\left(
\mathcal{P}\right)  $, its orthogonal complement. This motivates the following

\begin{definition}
The vertical subspace $V_{p}\left(  \mathcal{P}\right)  $ of $T_{p}\left(
\mathcal{P}\right)  $ is defined as the kernel of the pushforward of the projection map $\pi_{\ast}$%
\begin{equation}
V_{p}\left(  \mathcal{P}\right)  =\left\{  Y\in T_{p}\left(  \mathcal{P}%
\right)  \text{ such that }\pi_{\ast}\left(  Y\right)  =0\right\}  .
\label{pr17}%
\end{equation}

\end{definition}

In order to define the horizontal subspace $H_{p}\left(  \mathcal{P}\right)  $
in a unique way, it is necessary to define first the notion of connection over
the principal bundle:
\begin{definition}
Let $\mathrm{\omega}\in\Omega^{1}\left(  \mathcal{P}\right)  \otimes
\mathfrak{g}$ be a one form over $\mathcal{P}$ valued in the Lie algebra
$\mathfrak{g}$ satisfying the following conditions
\begin{enumerate}
\item $\mathcal{\omega}$ is continuous and smooth on $\mathcal{P},$
\item for all $Y\in V_{p}\left(  \mathcal{P}\right)  $, it holds%
\begin{equation}
\mathrm{\omega}\left(  Y\right)  =\mathrm{\theta}\left(  y_{\alpha}\left(
p\right)  \right)  \left(  y_{\alpha\ast}Y\right)  =\mathsf{Y}, \label{pr4}%
\end{equation}
\item the right action of the group is given by%
\begin{equation}
R_{g}^{\ast}\left(  \mathrm{\omega}\left(  p\right)  \right) =\mathrm{Ad}_{g^{-1}}(\omega(p)) =g^{-1}%
\mathrm{\omega}\left(  p\right)  g
\end{equation}
where $R_{g}^{\ast}$ is the pull-back induced by the right action $p^{\prime
}=pg$.
\end{enumerate}
\end{definition}
The form $\mathrm{\omega}$ satisfying these properties is called
one--form \textit{connection}. From the first property one sees that
$\mathrm{\omega}$ is globally defined on $\mathcal{P}$. The second condition
implies that $\mathrm{\omega}$ associates to every vector in $V_{p}\left(
\mathcal{P}\right)  $ its corresponding element in the Lie algebra $\mathfrak{g}$ where
$y_{\alpha\ast}$ represents the pushforward $y_{\alpha\ast}:T_{p}\left(
\mathcal{P}\right)  \longrightarrow T_{y_{\alpha}}(\mathcal{G})$ induced by the map $y_{\alpha
}:\mathcal{P}\longrightarrow\mathcal{G}$. Given this definition, it is time to
define the horizontal subspace.

\begin{definition}
The horizontal subspace it is defined as the kernel of $\mathcal{\omega}$%
\begin{equation}
H_{p}\left(  \mathcal{P}\right)  \equiv\left\{  X\in T_{p}\left(
\mathcal{P}\right)  \text{ such that }\mathrm{\omega}\left(  X\right)
=0\right\}  .
\end{equation}
\end{definition}

In this way, given a one-form connection $\mathrm{\omega}$, a unique
definition for $H_{p}\left(  \mathcal{P}\right)  $ is constructed. This
definition fulfils the consistency condition
\begin{equation}
H_{pg}\left(  \mathcal{P}\right)  =R_{g\ast}H_{p}\left(  \mathcal{P}\right),
\label{pr8}%
\end{equation}
so the distribution of $H_{p}\left(  \mathcal{P}\right)  $ it is invariant
under the action of $\mathcal{G}$. In fact, let $X\in H_{p}\left(
\mathcal{P}\right)  ;$ then%
\begin{equation}
\mathrm{\omega}\left(  R_{g\ast}X\right)  =R_{g}^{\ast}\mathrm{\omega}\left(
X\right)  =0,
\end{equation}
and therefore $R_{g\ast}X\in H_{pg}\left(  \mathcal{P}\right)  .$

\subsection{The gauge potential}

The one-form connection $\mathrm{\omega}$ is intimately related with the
concept of gauge potential in the context of gauge theories.

\begin{definition}
Let $\sigma_{\alpha}:U_{\alpha}\longrightarrow P$ be a local section and
$\mathrm{\omega}$ a one-form connection over $\mathcal{P}$. Then the gauge
potential is defined as
\begin{equation}
\mathcal{A}_{\alpha}=\sigma_{\alpha}^{\ast}\left(  \mathrm{\omega}\right)
\label{pr3}%
\end{equation}
where  $\sigma^{\ast}_{\alpha}$ denotes the pull-back map by a local section which projects the connection
$\mathcal{\omega}$ to the open set $U_{\alpha}\subset\mathcal{M}$.
\end{definition}
Now, given two open sets $U_{\alpha}$ and $U_{\beta}$ with non-empty overlap
$U_{\alpha}\cap U_{\beta}\neq\varnothing$ one has two gauge connections
$\mathcal{A}_{\alpha}=\sigma_{\alpha}^{\ast}\left(  \mathrm{\omega}\right)  $
and $\mathcal{A}_{\beta}=\sigma_{\beta}^{\ast}\left(  \mathrm{\omega}\right)
$. In order to find the relation between the two gauge potentials, let us
consider a vector $X\in T_{x}\left(  U_{\alpha}\cap U_{\beta}\right)  $. The
direct image $\sigma_{\beta\ast}:T_{x}\left(  U_{\alpha}\cap U_{\beta}\right)
\longrightarrow T_{\sigma_{\beta}\left(  x\right)  }\left(  \mathcal{P}%
\right)  $ can be expressed as it follows%
\begin{align}
\sigma_{\beta\ast}X &  =\left[  \sigma_{\alpha}\left(  x\right)  \tau
_{\alpha\beta}\left(  x\right)  \right]  _{\ast}\left(  X\right)  \nonumber\\
&  =R_{\tau_{\alpha\beta\ast}}\circ\sigma_{\alpha\ast}\left(  X\right)
+\sigma_{\alpha}\left(  x\right)  _{\ast}\circ\tau_{\alpha\beta\ast}\left(
X\right).
\end{align}
Considering the following direct images%
\begin{align}
\tau_{\alpha\beta\ast} &  :T_{x}\left(  U_{\alpha}\cap U_{\beta}\right)
\longrightarrow T_{\tau_{\alpha\beta}\left(  x\right)  }\left(  \mathcal{G}%
\right)  ,\\
\sigma_{\alpha}\left(  x\right)  _{\ast} &  :T_{\tau_{\alpha\beta}\left(
x\right)  }\left(  \mathcal{G}\right)  \longrightarrow T_{\sigma_{\alpha
}\left(  x\right)  \tau_{\alpha\beta}\left(  x\right)  }\left(  \mathcal{P}%
\right)  ,\\
\sigma_{\alpha\ast} &  :T_{x}\left(  U_{\alpha}\cap U_{\beta}\right)
\longrightarrow T_{\sigma_{\alpha}\left(  x\right)  }\left(  \mathcal{P}%
\right)  ,\\
R_{\tau_{\alpha\beta\ast}} &  :T_{\sigma_{\alpha}\left(  x\right)  }\left(
\mathcal{P}\right)  \longrightarrow T_{\sigma_{\alpha}\left(  x\right)
\tau_{\alpha\beta}\left(  x\right)  }\left(  \mathcal{P}\right),
\end{align}
eq.$\left(  \ref{pr3}\right)  $ reads%
\begin{align}
\mathcal{A}_{\beta}\left(  X\right)   &  =\sigma_{\beta}^{\ast}\mathrm{\omega
}\left(  X\right) \nonumber\\
&  =\mathrm{\omega}\left(  \sigma_{\beta\ast}X\right) \nonumber\\
&  =\mathrm{\omega}\left(  R_{\tau_{\alpha\beta\ast}}\circ\sigma_{\alpha\ast
}\left(  X\right)  \right)  +\mathrm{\omega}\left(  \sigma_{\alpha}\left(
x\right)  _{\ast}\circ\tau_{\alpha\beta\ast}\left(  X\right)  \right).
\end{align}
Using eq.$\left(  \ref{pr4}\right)  $, we find%
\begin{equation}
\mathrm{\omega}\left(  \sigma_{\alpha}\left(  x\right)  _{\ast}\circ
\tau_{\alpha\beta\ast}\left(  X\right)  \right)  =\mathrm{\theta}\left(
\tau_{\alpha\beta}\right)  \left(  \tau_{\alpha\beta\ast}\left(  X\right)
\right)  ,
\end{equation}
and using eq.$\left(  \ref{pr5}\right)  $%
\begin{align}
\mathcal{A}_{\beta}\left(  X\right)   &  =\sigma_{\alpha}^{\ast}\left(
\tau_{\alpha\beta}^{-1}\mathcal{\omega}\tau_{\alpha\beta}\right)  \left(
X\right)  +\left(  \tau_{\alpha\beta}^{-1}\dd_{\mathcal{G}}\tau_{\alpha\beta
}\right)  \left(  \tau_{\alpha\beta\ast}X\right)  \nonumber\\
&  =\tau_{\alpha\beta}^{-1}\sigma_{\alpha}^{\ast}\left(  \mathcal{\omega
}\right)  \tau_{\alpha\beta}\left(  X\right)  +\tau_{\alpha\beta}^{\ast
}\left(  \tau_{\alpha\beta}^{-1}\dd_{\mathcal{G}}\tau_{\alpha\beta}\right)
\left(  X\right)  \nonumber\\
&  =\left(  \tau_{\alpha\beta}^{-1}\mathcal{A}_{\alpha}\tau_{\alpha\beta}%
+\tau_{\alpha\beta}^{-1}\dd_{\mathcal{M}}\tau_{\alpha\beta}\right)  \left(
X\right)  ,
\end{align}
we finally arrive to%
\begin{equation}
\mathcal{A}_{\beta}=\tau_{\alpha\beta}^{-1}\mathcal{A}_{\alpha}\tau
_{\alpha\beta}+\tau_{\alpha\beta}^{-1}\dd_{\mathcal{M}}\tau_{\alpha\beta}.
\label{pr6}%
\end{equation}

The relation between the gauge connections eq.$\left(  \ref{pr6}\right)  $ is
what in physics it is known as \textit{gauge transformations}. If the
principal bundle $\mathcal{P}$ is not trivial, there always exists
$\tau_{\alpha\beta}\left(  x\right)  \in\mathcal{G}$ such that the connection
$\mathcal{A}_{\alpha}$ over $U_{\alpha}$ and the connection $\mathcal{A}%
_{\beta}$ over $U_{\beta}$ are related by eq.$\left(  \ref{pr6}\right)  $ over
$U_{\alpha}\cap U_{\beta}$. In the future we write $g=\tau_{\alpha\beta
}\left(  x\right)  $ and $\dd=\dd_{\mathcal{M}}$ to avoid overload notation. In
this way eq.$\left(  \ref{pr6}\right)  $ reads%
\begin{equation}
\mathcal{A\longrightarrow A}^{\prime}=g^{-1}\mathcal{A}g+g^{-1}\dd g \ . \label{pr7}%
\end{equation}

\subsection{Covariant derivative and curvature}

The definition of connection over $\mathcal{P}$ implies that $\mathrm{\omega}$
projects any vector in $T_{p}\left(  \mathcal{P}\right)  $ only to its vertical
component. It seems natural then to consider a $b$--form $B\in$
$\Omega^{b}\left(  \mathcal{P}\right)  \otimes\mathfrak{g}$ which projects vectors in $T_{p}\left(  \mathcal{P}\right)  $ to the
horizontal component $H_{p}\left(  \mathcal{P}\right)  $.
Let $h:T_{p}\left(  \mathcal{P}\right)  \longrightarrow H_{p}\left(
\mathcal{P}\right)  $ be a map such that to every vector $Z\in T_{p}\left(
\mathcal{P}\right)  $ it associate its projection over the horizontal subspace
$Z^{h}=h\left(  Z\right)  $. With the help of the map $h$, tensor and pseudo
tensorial forms will be defined.
\begin{definition}
a $b$-form $B$ over $\mathcal{P}$ is called pseudotensorial when it
satisfies
\begin{equation}
R_{g}^{\ast} B\left(  pg\right)    =g^{-1}Bg \ ,
\end{equation}
where $R_{g}^{\ast}$ is the reciprocal image induced by the right action
$p^{\prime}=pg$. The form $B$ is called tensorial if
\begin{equation}
B\left(  Z_{1},\ldots,Z_{b}\right)  =B\left(  Z_{1}%
^{h},\ldots,Z_{b}^{h}\right)  ,
\end{equation}
where $Z_{i}^{h}=h\left(  Z_{i}\right)  $, $i=1,\ldots,b$.
\end{definition}

The name \textit{tensorial form} has it origin in the following theorem.

\begin{theorem}
Let $B$ be a tensorial form over $\mathcal{P}$. Then, on the
intersection $U_{\alpha}\cap U_{\beta}\neq\varnothing$ it follows%
\begin{equation}
\mathcal{B}_{\beta}=\tau_{\alpha\beta}^{-1}\mathcal{B}_{\alpha}\tau
_{\alpha\beta},%
\end{equation}
where the $b$--form $\mathcal{B}\in\mathcal{M}$ corresponds to $\mathcal{B}%
_{\alpha}=\sigma_{\alpha}^{\ast}B$.

\begin{proof}
In fact, we have%
\begin{align}
R_{g}^{\ast}B\left(  p\right)  \left(  Z_{1}\left(  p\right)
,\ldots,Z_{b}\left(  p\right)  \right)   &  =R_{g}^{\ast}B\left(
p\right)  \left(  Z_{1}^{h}\left(  p\right)  ,\ldots,Z_{b}^{h}\left(
p\right)  \right)   \nonumber \\
&  =B\left(  pg\right)  \left(  R_{\ast g}Z_{1}^{h}\left(  p\right)
,\ldots,R_{\ast g}Z_{b}^{h}\left(  p\right)  \right)  .
\end{align}
Using eq.$\left(  \ref{pr8}\right)  $, it follows
\begin{align}
R_{g}^{\ast}B\left(  p\right)  \left(  Z_{1}\left(  p\right)
,\ldots,Z_{b}\left(  p\right)  \right)   &  =B\left(  pg\right)
\left(  Z_{1}^{h}\left(  pg\right)  ,\ldots,Z_{b}^{h}\left(  pg\right)
\right) \nonumber \\
&  =B\left(  pg\right)  \left(  Z_{1}\left(  pg\right)  ,\ldots
,Z_{b}\left(  pg\right)  \right) \ .
\end{align}
Thus, when $B$ is tensorial,%
\begin{equation}
R_{g}^{\ast}B\left(  p\right)  \left(  Z_{1}\left(  p\right)
,\ldots,Z_{b}\left(  p\right)  \right)  =B\left(  pg\right)  \left(
Z_{1}\left(  pg\right)  ,\ldots,Z_{b}\left(  pg\right)  \right)  . \label{pr9}%
\end{equation}
Repeating the procedure used to deduce the gauge transformation eq.$\left(
\ref{pr7}\right)  $ one concludes that, in the intersection $U_{\alpha}\cap
U_{\beta},$%
\begin{equation}
\mathcal{B}_{\beta}=\tau_{\alpha\beta}^{-1}\mathcal{B}_{\alpha}\tau
_{\alpha\beta}. \label{pr10}%
\end{equation}

\end{proof}
\end{theorem}

This theorem lead us naturally to another definition; the definition of the
exterior covariant derivative operator $\mathcal{D}$

\begin{definition}
Let $B$ be a pseudotensorial form over $\mathcal{P}$. Then, the
exterior covariant derivative is defined as%
\begin{equation}
\mathcal{D}B=\dd_{\mathcal{P}}B\circ h,
\end{equation}
where $\dd_{\mathcal{P}}$ denotes the exterior derivative in $\mathcal{P}$.
Using the same arguments that were used to show eq.$\left(  \ref{pr9}\right)  $,
it is possible to show that if $B$ is a pseudotensorial form, then
$\mathcal{D}B$ is a $\left(  b+1\right) $-tensorial form.
\end{definition}

Since the connection $\mathrm{\omega}$ is pseudotensorial, it is possible to
define the tensorial two-form $F$ as
\begin{equation}
F=\mathcal{D}\mathrm{\omega},
\end{equation}
which is called the \textit{curvature} of the principal bundle $\mathcal{P}$. Also, it
is possible to show (see~\cite[Theorem 2.13]{deAz95})  directly that%
\begin{equation}
F=\dd_{\mathcal{P}}\mathrm{\omega}+\frac{1}{2}\left[  \mathrm{\omega
},\mathrm{\omega}\right] \ .
\end{equation}
Given the curvature $F$ over the principal bundle $\mathcal{P}$, it
is direct to define the gauge curvature or \textit{field strength} over an open
set $U_{\alpha}$ as $\mathcal{F}_{\alpha}=\sigma_{\alpha}^{\ast
}F$. In terms of eq.$\left(  \ref{pr3}\right)  $, the curvature
$\mathcal{F}_{\alpha}$ corresponds to%
\begin{equation}
\mathcal{F}_{\alpha}=\dd_{\mathcal{M}}\mathcal{A}_{\alpha}+\mathcal{A}_{\alpha
}\wedge\mathcal{A}_{\alpha}.
\end{equation}
Now, $\mathrm{F}$ is tensorial so it satisfies eq.$\left(  \ref{pr10}\right)
$%
\begin{equation}
F_{\beta}=\tau_{\alpha\beta}^{-1}\mathcal{F}_{\alpha}\tau
_{\alpha\beta},%
\end{equation}
and it is direct to verify the consistency with eq.$\left(  \ref{pr6}\right)
$.

Given a tensorial $b$-form $B\in\Omega^{b}\left(  \mathcal{P}\right)
\otimes\mathfrak{g}$, it is possible to show (see~\cite[p.94]{deAz95}) that its covariant
derivative $\mathcal{D}B$ is given by%
\begin{equation}
\mathcal{D}B=\dd_{\mathcal{P}}B+\left[  \mathrm{\omega
},B\right].
\end{equation}
Moreover, the derivative operator $\mathcal{D}$ satisfy the Bianchi identities%
\begin{align}
\mathcal{DD}B  &  =\left[  F,B\right]  , \\
\mathcal{D}F  &  =0. \label{pr19}%
\end{align}

In the base space $\mathcal{M}$, the covariant derivative is obtained by the projection
\begin{align}
D_{\alpha}\mathcal{B}_{\alpha}  &  =\sigma_{\alpha}^{\ast}\left(
\mathcal{D}B\right)  \nonumber\\
&  =\dd_{\mathcal{M}}\mathcal{B}_{\alpha}+\left[  \mathcal{A}_{\alpha
},\mathcal{B}_{\alpha}\right]  , \label{pr12}%
\end{align}
and the Bianchi identities are given by%
\begin{align}
D_{\alpha}D_{\alpha}\mathcal{B}_{\alpha}  &  =\left[  \mathcal{F}_{\alpha
},\mathcal{B}_{\alpha}\right]  ,\\
D_{\alpha}\mathcal{F}  &  =0.
\end{align}
Note that, by definition, $\mathcal{D}B$ is tensorial and therefore
it satisfies
\begin{equation}
D_{\beta}\mathcal{B}_{\beta}=\tau_{\alpha\beta}^{-1}D_{\alpha}\mathcal{B}%
_{\alpha}\tau_{\alpha\beta} \ ,%
\end{equation}
and again, this is consistent with the gauge transformation eq.$\left(
\ref{pr6}\right)  $.

\newsection{Equivariant principal bundles}

Let $\mathcal{P}$ be a principal bundle with base $\mathcal{M}$. If
$\mathcal{M}$ admits the action of a Lie group $G$, we say that $\mathcal{P}$
is $G-$equivariant if the $G-$action lifts to $\mathcal{P}$. This means that
the $G-$action on $\mathcal{P}$ is such that the diagram (\ref{equivfig}) commutes, %
\begin{figure}[h!]
\centering
\begin{tikzcd}[column sep=4em,row sep=4em]
\mathcal{P} \arrow[]{r}{g}\arrow[]{d}{\pi}
& \mathcal{P} \arrow[]{d}{\pi}\\
\mathcal{M} \arrow[]{r}{g}&   \mathcal{M}
\end{tikzcd}
\caption{Equivariant condition} \label{equivfig}
\end{figure}
where $\pi:\mathcal{P\rightarrow M}$ is the projection map and $g$ denotes
both maps $\mathcal{P\rightarrow P}$ and $\mathcal{M\rightarrow M}$.

In what follows we focus on product manifolds of the form $ \mathcal{M}=M\times G/H$, where $M$ is a closed
manifold and $H\subset G$ is a subgroup. For $g\in G$ we denote the coset $gH$
as $\left[  g\right]  $. The action of $G$ on $\mathcal{M}$ is given by%
\begin{equation}
g\left(  x,\left[  g^{\prime}\right]  \right)  =\left(  x,\left[  gg^{\prime
}\right]  \right)\ ,
\end{equation}
for $x\in M$ and $g,g^{\prime}\in G$. So the $G-$action on $\mathcal{M}$ is
extended to be the trivial action on $M$.

Given an $G-$equivariant bundle over $\mathcal{M}$, we can induce an
$H-$equivariant principal bundle over $M$ by restriction%
\begin{equation}
\mathcal{P\mapsto}\left.  \mathcal{P}\right\vert _{\left[e \right]   \times
M}\text{ .}%
\end{equation}
The $H-$action on $M$ is also trivial. However, the action on the fibres of
$\left.  \mathcal{P}\right\vert _{\left[e \right]   \times
M}$ is not
necessary the trivial one. The inverse operation of restriction is called
induction: If $P$ denotes an $H-$equivariant principal bundle over $M$, we can
define
\begin{equation}
\mathcal{P}=G\times_{H}P\label{epb4}%
\end{equation}
where the quotient space $G\times_{H}P$ is the set of equivalent classes
$G\times P$ with respect to the equivalence relation%
\begin{equation}
\left(  g,p\right)  \sim\left(  gh,h^{-1}p\right)  \text{, where }g\in
G\text{, }h\in H\text{ and }p\in P.
\end{equation}
The projection map $\pi:$ $\mathcal{P}\rightarrow\mathcal{M}$ is given by%
\begin{equation}
\pi\left(  \left[  g,p\right]  \right)  =\left(  \left[  g\right]  ,\pi\left(
p\right)  \right)  \text{,}%
\end{equation}
where the projection map in the right hand side is such that $\pi:P\rightarrow
M$. In this way, there is a one-to-one correspondence between $G-$equivariant
principal bundles over $\mathcal{M}$ and $H-$equivariant principal bundles
over $M$. From now we assume $P=\left.  \mathcal{P}\right\vert _{\left[e \right]   \times
M}$ where $\mathcal{P}$ is $G-$equivariant.

The $H-$action on $P$ can be locally described by Lie group homomorphisms
$\rho_{x}:H\rightarrow\mathcal{G}$ as we explain in what follows. Let $\left\{
U_{\alpha}\right\}  _{\alpha\in I}$ be a good open covering of $M$ with
$I $ as a set of indices, and let $\left\{  \sigma_{\alpha}\right\}
_{\alpha\in I}$ be a collection of local sections $\sigma_{\alpha}:U_{\alpha
}\subset M\rightarrow P$. For $\alpha\in I$, the homomorphisms $\rho
_{\alpha,x}:H\rightarrow\mathcal{G}$ are defined by%
\begin{equation}
h\sigma_{\alpha}\left(  x\right)  =\sigma_{\alpha}\left(  x\right)
\rho_{\alpha,x}\left(  h\right)  \text{, with }h\in H\text{ and }x\in
U_{\alpha} \label{epb1}%
\end{equation}

By \cite[Lemma]{HJV} the local sections $\sigma_{\alpha}$ can be chosen
in such a way that the homomorphism $\rho_{\alpha,x}$ does not depend on $x$.
In this way we denote the homomorphism in eq.$\left(  \ref{epb1}\right)  $
simple as $\rho_{\alpha}$. Let $\tau_{\alpha\beta}:U_{\alpha}\cap U_{\beta
}\rightarrow\mathcal{G}$ be the transition functions of $P$ such that
$\sigma_{\alpha}=\tau_{\alpha\beta}\sigma_{\beta}$. Then
\begin{equation}
\tau_{\alpha\beta}\left(  x\right)  \rho_{\beta}\left(  h\right)
=\rho_{\alpha}\left(  h\right)  \tau_{\alpha\beta}\left(  x\right) \ ,
\label{epb2}%
\end{equation}
with $h\in H$ and $x\in U_{\alpha}\cap U_{\beta}\neq\varnothing$
Therefore, on non-empty intersections $U_{\alpha}\cap U_{\beta}$ the
homomorphisms $\rho$ are related by conjugation.

In order to define local trivialisations of $P$ with respect to which the
homomorphisms agree on all patches $U_{\alpha}$, we fix $\alpha_{1}\in I$ and let
$\rho_{\alpha_{1}}:H\rightarrow\mathcal{G}$ be the homomorphism defined by%
\begin{equation}
h\sigma_{\alpha_{1}}\left(  x\right)  =\sigma_{\alpha_{1}}\left(  x\right)
\rho_{\alpha_{1}}\left(  h\right) \ , 
\end{equation}
 with $h\in H$ and $x\in
U_{\alpha}$.%

According to eq.$\left(  \ref{epb2}\right)  $ the homomorphisms $\rho_{\alpha
}$ lie in the same conjugacy class on overlapping $U_{\alpha}$. Therefore, it
is possible to find constant $g_{\alpha}\in\mathcal{G}$ such that%
\begin{equation}
\rho_{\alpha}=g_{\alpha}\rho_{\alpha_{1}}g_{\alpha}^{-1} \ . \label{epb3}%
\end{equation}
These transitions functions $g_{\alpha}$ are not unique. If one defines
sections $\tilde{\sigma}_{\alpha}=\sigma_{\alpha}\left(  x\right)g_{\alpha}  $
for $\alpha\in I$, then using eq.$\left(  \ref{epb1}\right)  $ we get%
\begin{equation}
h\tilde{\sigma}_{\alpha}\left(  x\right)  =\sigma_{\alpha}\left(  x\right)
\rho_{\alpha}\left(  h\right)  g_{\alpha}=\tilde{\sigma}_{\alpha}\left(
x\right)  \rho_{\alpha_{1}}\left(  h\right)  \text{ , }x\in U_{\alpha}%
\end{equation}
Thus, the trivialization defined in terms of the local sections $\left\{
\tilde{\sigma}_{\alpha}\right\}  _{\alpha\in I}$ is such that the action of
$H$ is characterized by only one homomorphism $\rho_{\alpha_{1}}%
:H\rightarrow\mathcal{G}$. Note that different choices of $g_{\alpha}$ in
eq.$\left(  \ref{epb3}\right)  $ leads in general to different trivialization
of $P$. However, the statement that $\rho_{\alpha_{1}}$ is the same across the
open covering $\left\{  U_{\alpha}\right\}  _{\alpha\in I}$ still holds in
this new local trivialization.

Now, let $\tilde{\tau}_{\alpha\beta}:U_{\alpha}\cap U_{\beta}\rightarrow
\mathcal{G}$ denote the transition function with respect to $\tilde{\sigma
}_{\alpha}$. Under change of sections we have
\begin{align}
\tilde{\sigma}_{\beta}\left(  x\right)  \rho_{\alpha_{1}}\left(  h\right)
&=  h\tilde{\sigma}_{\beta}\left(  x\right)  \nonumber \\ 
&=  h\tilde{\sigma}_{\alpha}\left(
x\right)  \tilde{\tau}_{\alpha\beta}\left(  x\right)   \nonumber \\
&= \tilde{\sigma}_{\beta
}\left(  x\right)  \tilde{\tau}_{\alpha\beta}^{-1}\left(  x\right)
\rho_{\alpha_{1}}\left(  h\right)  \tau_{\alpha\beta}\left(  x\right) \ ,
\end{align}
and therefore%
\begin{equation}
\rho_{\alpha_{1}}=\tilde{\tau}_{\alpha\beta}^{-1}\rho_{\alpha_{1}}\tilde\tau
_{\alpha\beta} \ .%
\end{equation}
This means that the structure group of $P$ reduces to $\mathcal{H}%
=\mathcal{Z}_{\mathcal{G}}\left(  \rho_{\alpha_{1}}\left(  H\right)  \right)
$, the centralizer in $\mathcal{G}$ of the image of $H$ by $\rho_{\alpha_{1}}%
$. 

The analysis presented in this section automatically provides a systematic
recipe for constructing bundles: Start with a homomorphism $\rho
:H\rightarrow\mathcal{G}$ and construct a principal bundle $P_{M}$ over $M$
with fibre $\mathcal{H}=\mathcal{Z}_{\mathcal{G}}\left(  \rho\left(  H\right)
\right)  $. Moreover, $P_{M}$ also extends to an $G-$equivariant principle
bundle on $\mathcal{M}=G/H\times M$ by virtue of eq.$\left(  \ref{epb4}%
\right)  $.

\newsection{Invariant connections} \label{sec:invconn}

A connection $\mathsf{\omega}$ over a principal bundle $\mathcal{P}$ over
$\mathcal{M}=M\times G/H$ is called $G-$invariant if
\begin{equation}
g^{\ast}\mathsf{\omega}=\mathsf{\omega}\text{ , \ for all }g\in G.
\end{equation}
Here the map $g^{\ast}$ denotes the pullback of the map
$g:\mathcal{P\rightarrow P}$. The classification of invariant connections in
terms of the structures on the bundles $P=\left.  \mathcal{P}\right\vert
_{\left\{  H\right\}  \times M}$ over $M$ is given by Wang's theorem~\cite{wang1958} in the special case when $M$ is a point. The generalization to
the case of $M$ being contractible is given in~\cite{HJV,JSL}. We follow the
treatment adopted in the later case.

Recall the Maurer--Cartan form defined in (\ref{MCform}). The Maurer--Cartan
form on $G$ is denoted by $\mathsf{\theta}_{G}:T_{g}\left(  G\right)
\rightarrow\mathfrak{g}$, with $g\in G$. To identify $G-$invariant connections
on $\mathcal{P}$ we define the maps%
\begin{align}
\psi_{\alpha}  & :G\times U_{\alpha}\times\mathcal{G\longrightarrow}\left.
\mathcal{P}\right\vert _{G/H\times U_{\alpha}}\text{ } \, \nonumber \\
& \left(  g,x,q\right)  \longmapsto g\sigma_{\alpha}\left(  x\right)  q \label{projconn}
\end{align}
for all $\alpha\in I$. According to \cite[Theorem 2]{HJV}, the pull-back under
$\psi_{\alpha}$ of an invariant connection $\mathsf{\omega}$ defined on
$\mathcal{P}$ is given by%
\begin{equation}
\psi_{\alpha}^{\ast}\mathsf{\omega}_{\left(  g,x,q\right)  }=\mathrm{Ad}(q^{-1})  \left(  \Phi_{\alpha}\left(  x\right)  \circ
\mathsf{\theta}_{G}+\mu_{\alpha}\right)  +\mathsf{\theta}_{\mathcal{G}} \label{pullpsi}
\end{equation}
where
\begin{description}
\item[a)] $\Phi_{\alpha}\left(  x\right)  :\mathfrak{g}\rightarrow
\mathfrak{q}$, being $\mathfrak{q}$ the Lie algebra of $\mathcal{G}$, is a
family of linear maps which depend on $x\in U_{\alpha}$ satisfying%
\begin{align}
\Phi_{\alpha}\left(  x\right)  \circ\mathrm{Ad}\left(  h\right)    &
=\mathrm{Ad}\left(  \rho_{\alpha}\left(  h\right)  \right)  \circ\Phi_{\alpha
}\left(  x\right)  \text{ \ with }h\in H\text{ ,} \label{psicond2}\\
\Phi_{\alpha}\left(  x\right)  \left(  \mathsf{X}_{0}\right)    &
=\rho_{\alpha\ast}\left(  \mathsf{X}_{0}\right)  \text{ for }\mathsf{X}_{0}%
\in\mathfrak{h}\text{ ,} \label{psicond}%
\end{align}
where $\mathfrak{h}$ is the Lie algebra of $H$ and $\rho_{\alpha\ast
}:\mathfrak{h}\rightarrow\mathfrak{q}$ is the Lie algebra homomorphism induced
by $\rho_{\alpha}$.

\item[b)] $\mu_{\alpha}\in\Omega^{1}\left(  U_{\alpha},\mathfrak{q}\right)  $
such that%
\begin{equation}
\mu_{\alpha}=\mathrm{Ad}\left(  \rho_{\alpha}\left(  h\right)  \right)
\mu_{\alpha}\text{ , with }h\in H \label{connacond}
\end{equation}
and therefore $\mu_{\alpha}$ takes values in the Lie algebra of $\mathcal{H}%
=\mathcal{Z}_{\mathcal{G}}\left(  \rho\left(  H\right)  \right)  $.
\end{description}

We now look at the behaviour of the invariant connection $\mathsf{\omega}$ in
terms of the geometric data on $M$. In order to do so, one needs to specify
how $\Phi_{\alpha}$ and $\mu_{\alpha}$ change under transformations in
non-empty overlaps $U_{\alpha}\cap U_{\beta}\neq\varnothing$. To this end,
define the embeddings%
\begin{align}
\iota_{\alpha}  & :U_{\alpha}\hookrightarrow G\times U_{\alpha}\times
\mathcal{G} \nonumber \\
x  & \mapsto\left(  e_{G},x,e_{\mathcal{G}}\right)
\end{align}
for $\alpha\in I$, with $e_{G}$ and $e_{\mathcal{G}}$ the neutral elements in
$G$ and $\mathcal{G}$ respectively. Let $X\in T_{x}\left(  M\right)  $ be a
tangent vector at the point $x\in$ $U_{\alpha}\cap U_{\beta}$. Then $\iota_{\alpha
}^{\ast}\psi_{\alpha}^{\ast}\mathsf{\omega}\left(  X\right)  =\mu_{\alpha}(X)$,
and analogously for $\beta\in I$. Representing the tangent vector $X$ in terms
of a path $\gamma\left(  t\right)  :\left(  -\varepsilon,\varepsilon\right)
\rightarrow U_{\alpha}\cap U_{\beta}$, the pushforward of $X$ under
$\psi_{\beta\ast}\iota_{\beta\ast}$ is given by%
\begin{align}
\psi_{\beta\ast}\iota_{\beta\ast}X  & =\left.  \frac{\dd}{\dd t}\sigma_{\beta
}\left(  \gamma\left(  t\right)  \right)  \right\vert _{t=0} \nonumber\\
& =\left.  \frac{\dd}{\dd t}\sigma_{\alpha}\left(  \gamma\left(  t\right)  \right)
g_{\alpha\beta}\left(  \gamma\left(  t\right)  \right)  \right\vert _{t=0}  \nonumber\\
& =\left(  \psi_{\alpha\ast}\iota_{\alpha\ast}X\right)  g_{\alpha\beta}%
+\sigma_{\beta}g_{\alpha\beta}^{-1}\dd_{{\scriptsize X}}g_{\alpha\beta} \nonumber \\
& =R_{g_{\alpha\beta}\ast}\psi_{\alpha\ast}\iota_{\alpha\ast}X+\widetilde
{\left(  g_{\alpha\beta}^{-1}\dd_{{\scriptsize X}}g_{\alpha\beta}\right)  } \ .%
\end{align}
Hence,%
\begin{equation}
\mu_{\beta}\left(  X\right)  =\omega\left(  \psi_{\beta\ast}\iota_{\beta\ast
}X\right)  =\mathrm{Ad}\left(  g_{\alpha\beta}^{-1}\right)  \mu_{\alpha
}\left(  X\right)  +g_{\alpha\beta}^{-1}\dd_{{\scriptsize X}}g_{\alpha\beta
}%
\end{equation}
which identifies the collection of $\mu_{\alpha}\in\Omega^{1}\left(
U_{\alpha},\mathfrak{q}\right)  $ as connection on $P$. If we specify a set of
sections $\left\{  \tilde{\sigma}_{\alpha}\right\}  _{\alpha\in I}$ such that
there is  only one $\rho_{\alpha_{1}}:H\rightarrow\mathcal{G}$, then for all
$\alpha\in I$ the connection $\mu_{\alpha}$ takes values in the Lie algebra of
$\mathcal{H}=\mathcal{Z}_{\mathcal{G}}\left(  \rho_{\alpha_{1}}\left(  H\right)  \right)
$. This is consistent with the fact that the structure group of $P$ can be
reduced to $\mathcal{H}$.

In the case of $\Phi_{\alpha}\left(  x\right)  $, note that
\begin{equation}
\left.  \Phi_{\alpha}\left(  x\right)  \circ\mathsf{\theta}_{{\scriptsize G}%
}\right\vert _{g}=\left.  \Phi_{\alpha}\left(  x\right)  \left(
\mathsf{T}_{a}\right)  X^{a}\right\vert _{g}\text{ , for }g\in G
\end{equation}
where $\left\{  \mathsf{T}_{a}\right\}  _{a=1}^{\dim\left(  \mathfrak{g}%
\right)  }$ is a basis for $\mathfrak{g}$ and $X^{a}$ denotes a set of
left-invariant forms dual to $\mathsf{T}_{a}$. Let $y\in G/H\,$\ and let $Y$
be a tangent vector at $y$. Moreover, let $N\subset G/H$ be an open region
around $y$ and assume that there exists a section $\eta:N\rightarrow G$. On
$U_{\alpha}\cap U_{\beta}\neq\varnothing$ define the map%
\begin{align}
\tilde{\eta}  & :N\times\left(  U_{\alpha}\cap U_{\beta}\right)  \rightarrow
G\times\left(  U_{\alpha}\cap U_{\beta}\right)  \times\mathcal{G} \nonumber \\
& \left(  \left[  g\right]  x\right)  \mapsto\left(  \eta\left(  \left[
g\right]  \right)  ,x,e_{\mathcal{G}}\right) \ .
\end{align}
Since $\psi_{\beta\ast}\tilde{\eta}_{\ast}=R_{g_{\alpha\beta}\ast}\psi
_{\alpha\ast}\tilde{\eta}_{\ast}$ one find%
\begin{align}
\left.  \Phi_{\beta}\left(  x\right)  \circ\mathsf{\theta}_{{\scriptsize G}%
}\right\vert _{\eta\left(  y\right)  }\left(  \eta_{\ast}Y\right)    &
=\tilde{\eta}^{\ast}\psi_{\beta}^{\ast}\mathsf{\omega}\left(  Y\right)   \nonumber\\
& =\mathsf{\omega}\left(  \psi_{\beta\ast}\tilde{\eta}_{\ast}Y\right)  \nonumber\\
& =\mathrm{Ad}\left(  g_{\alpha\beta}^{-1}\right)  \psi_{\alpha}^{\ast
}\mathsf{\omega}\left(  \tilde{\eta}_{\ast}Y\right)  \nonumber\\
& =\mathrm{Ad}\left(  g_{\alpha\beta}^{-1}\right)  \left(  \left.
\Phi_{\alpha}\left(  x\right)  \circ\mathsf{\theta}_{{\scriptsize G}%
}\right\vert _{\eta\left(  y\right)  }\left(  \eta_{\ast}Y\right)  \right)  .
\end{align}
As a consequence of the last result, one obtains%
\begin{equation}
\Phi_{\beta}\left(  x\right)  \left(  \mathsf{T}_{a}\right)  =\mathrm{Ad}%
\left(  g_{\alpha\beta}^{-1}\right)  \Phi_{\alpha}\left(  x\right)  \left(
\mathsf{T}_{a}\right) \ ,
\end{equation}
and then $\Phi_{\alpha}\left(  x\right)  \left(  \mathsf{T}_{a}\right)  $ define
sections of the associated vector bundle $\mathrm{ad}\left(  P\right)  $ over
$M$. We conclude by summarising the main results of this section and the previous one
in the following table of correspondences between equivariant principal bundles and its corresponding invariants connections.%
\\
\begin{table}[h!]
\centering
\begin{tabular}{ccc}

$G/H\times M$                                                       &  & $M$                                                                               \\ \hline \\
\begin{tabular}[c]{@{}l@{}}$G-$equivariant bundle $\mathcal{P}$\\ with structure group
$\mathcal{G}$\end{tabular}       & $\longrightarrow$ & \begin{tabular}[c]{@{}l@{}}$H-$equivariant bundle $P$\\ with structure
group $\mathcal{H}$\\ \end{tabular} \\[0.7cm] 
\begin{tabular}[c]{@{}l@{}}$G-$invariant connection $\mathsf{\omega}$\\ on $\mathcal{P}$\\ \end{tabular} & $\longrightarrow$  &  \begin{tabular}[c]{@{}l@{}}sections $\Phi_{\alpha
}\left(  \mathsf{T}_{a}\right)  $ on $\mathrm{ad}\left(  P\right)  $\\ and connection $\mu$ on $P$ \\ \end{tabular}             \\ 
\end{tabular} 
\caption{Bundle correspondence.}
\label{one2one}
\end{table}
\bigskip                                               
\newsection{Transgression forms and the Chern class}
Given the notion of principal bundle $\mathcal{P}$ and connection $\mathrm{\omega}$, it is interesting to study the existence of characteristic quantities defined over $\mathcal{P}$. In principle, this quantities are defined in terms of the connection $\mathrm{\omega}$. However, it turns out that they are completely independent of the choice of $\mathrm{\omega}$. Thus, these quantities define \textit{topological invariants} which measure the obstruction of the bundle $\mathcal{P}$ to be trivial. 

In the following, we first introduce some preliminary definitions and then we will use the Chern--Weil theorem to define the Transgression form and the Chern Class.

\subsection{Invariant polynomial}

\begin{definition}
An invariant polynomial of degree $n$, is a $n-$linear map%
\begin{equation}
\left\vert \ldots\right\vert :\underset{n\text{ times}}{\underbrace
{\mathfrak{g\times}\ldots\times\mathfrak{g}}}\longrightarrow%
\mathbb{R}
\end{equation}
which satisfy the condition%
\begin{equation}
\left\vert \left(  g^{-1}Z_{1}g\right)  \wedge\ldots\wedge\left(
g^{-1}Z_{n}g\right)  \right\vert =\left\vert Z_{1}%
\wedge\ldots\wedge Z_{n}\right\vert \label{pr11}%
\end{equation}
where $Z\in\Omega^{z_{i}}\otimes\mathfrak{g}$, $i=1,\ldots,n$ and
$g=exp(\lambda^A \mathsf{T}_A)\in\mathcal{G}$. When the extra condition
\begin{equation}
\left\vert Z_{1}\wedge\ldots\wedge Z_{i}\wedge
 Z_{j}\wedge\ldots\wedge\mathrm{Z}_{n}\right\vert =\left(
-1\right)  ^{z_{i}z_{j}}\left\vert Z_{1}\wedge\ldots\wedge
Z_{j}\wedge Z_{i}\wedge\ldots\wedge Z_{n}\right\vert
\end{equation}
is satisfied for all $Z_{i},Z_{j}$, we say that the
invariant polynomial is symmetric and we denote it by $\left\langle
\ldots\right\rangle $.
\end{definition}

Since $Z=Z^{A}\mathsf{T}_{A}$, it is possible to write
\begin{equation}
\left\langle Z_{1}\wedge\ldots\wedge Z_{n}\right\rangle
=Z_{1}^{A_{1}}\wedge\ldots\wedge Z_{n}^{A_{n}}\left\langle \mathsf{T}_{A_{1}%
}\ldots\mathsf{T}_{A_{n}}\right\rangle,
\end{equation}
where $\left\langle \mathsf{T}_{A_{1}}\ldots\mathsf{T}_{A_{n}}\right\rangle $
is called the symmetric invariant tensor. The invariance condition
eq.$\left(  \ref{pr11}\right)  $ can be written in different ways. In fact, if
we consider an element $g=\exp\left(  \lambda^{A}\mathsf{T}_{A}\right)
\in\mathcal{G}$ infinitesimally close to the identity, the invariant condition
takes the form%
\begin{equation}
\left\langle \left[  \lambda,Z_{1}\right]  \wedge Z_{2}%
\wedge\ldots\wedge Z_{n}\right\rangle +\ldots+\left\langle
Z_{1}\wedge\ldots\wedge Z_{n-1}\wedge\left[  \lambda
,Z_{n}\right]  \right\rangle =0 \ ,
\end{equation}
where $\lambda=\lambda^{A}\mathsf{T}_{A}$. Note that if we replace $\lambda$ with the
one-form connection $\mathcal{A}$ and recalling the definition of covariant
derivative eq.$\left(  \ref{pr12}\right)  $, the invariance condition is
given by%
\begin{equation}
\left\langle D\left(  Z_{1}\wedge\ldots\wedge Z_{n}\right)
\right\rangle =\dd\left\langle Z_{1}\wedge\ldots\wedge Z%
_{n}\right\rangle .\label{pr21}%
\end{equation}
All this expressions for the invariance condition will be used through this
Thesis depending on the context.

\subsection{Projection of differential forms}

So far, we have defined differential forms over the base space $\mathcal{M}$
starting from differential forms defined over the principal bundle
$\mathcal{P}$. In order to do so, we used the reciprocal image induced by the
local sections $\sigma_{\alpha}:U_{\alpha}\subset\mathcal{M\rightarrow P}%
$. A natural question is if one can do the opposite, i.e, defining a form in
$\mathcal{P}$ starting from a form defined on $\mathcal{M}$. This is
possible by using the pull-back of the projection map $\pi
:\mathcal{P\rightarrow M}$.

Given a $p-$form $\mathcal{B}$ over $\mathcal{M}$ it is possible to construct
a $p-$form $B$ over $\mathcal{P}$ as%
\begin{equation}
B=\pi^{\ast}\mathcal{B}.%
\end{equation}
Now, given two points of the same fibre $p$ and $p^{\prime}=pg$ and a $p-$form
$B$, in general it is true that $B\left(  p\right)  $ and
$B\left(  pg\right)  $ will correspond to the pull-back of different
$p-$forms over $\mathcal{M}$,%
\begin{align}
B(p)   &  =\pi^{\ast}\mathcal{B}(x)  \ ,\\
B(pg)   &  =\pi^{\ast}\mathcal{B}^{\prime}(x) \ ,
\end{align}
with $\mathcal{B}(x)  \neq\mathcal{B}^{\prime}(x)  $
in general. When the $p-$form $B$ along the fibre $\pi^{-1}(
x)  $ corresponds to the pull-back of a unique $p-$form $\mathcal{B}%
(x)  $, we say that $B$ is projectable to $\mathcal{B}%
(x)  $. Let us precise the notion of projectable differential
form in terms of the following
\begin{theorem}
Let $B$ be a $p-$form over $\mathcal{P}$ satisfying
\begin{enumerate}
\item $B$ is right invariant under the action of $\mathcal{G}$,%
\begin{equation}
R_{g}^{\ast}B=B \label{pr14}%
\end{equation}
\item $B$ acts on $T_{p}\left(  \mathcal{P}\right)  $ in such a way
that%
\begin{equation}
B\left(  X_{1}\ldots X_{n}\right)  =B\left(  X_{1}^{h}\ldots
X_{n}^{h}\right)  \label{pr13}%
\end{equation}
then, there exist a unique $p-$form $\mathcal{B}(x)  $ defined
over $\mathcal{M}$ such that $B=\pi^{\ast}\mathcal{B}$ and therefore
$B$ is projectable.
\end{enumerate}
\end{theorem}

\begin{proof}
Let $X_{i}\left(  p\right)  \in T_{p}\left(  \mathcal{P}\right)  $ with
$i=1,\ldots,q$. Consider the pushforward induced by the projection map $\pi$%
\begin{align}
\pi_{\ast} &  :T_{p}\left(  \mathcal{P}\right)  \longrightarrow T_{x}\left(
\mathcal{M}\right)   \nonumber \\
&  :X_{i}\left(  p\right)  \longrightarrow Y_{i}\left(  x\right)  =\pi_{\ast
}X_{i}(p) \ .
\end{align}
Let $\mathcal{B}(x)  $ be a $q-$form such that%
\begin{equation}
B\left(  p\right)  =\pi^{\ast}\mathcal{B}\left(  x\right)  .
\end{equation}
Then, we have
\begin{align}
\mathcal{B}\left(  x\right)  \left(  Y_{1}\left(  x\right)  \ldots
,Y_{q}\left(  x\right)  \right)   &  =\mathcal{B}\left(  x\right)  \left(
\pi_{\ast}X_{1}\left(  p\right)  ,\ldots,\pi_{\ast}X_{q}\left(  p\right)
\right)   \nonumber \\
&  =\pi^{\ast}\mathcal{B}\left(  x\right)  \left(  X_{1}\left(  p\right)
,\ldots,X_{q}\left(  p\right)  \right)   \nonumber \\
&  =B\left(  p\right)  \left(  X_{1}\left(  p\right)  ,\ldots
,X_{q}\left(  p\right)  \right) \ ,
\end{align}
and using eq.$\left(  \ref{pr13}\right)  $,%
\begin{equation}
\pi^{\ast}\mathcal{B}\left(  x\right)  \left(  X_{1}\left(  p\right)
,\ldots,X_{q}\left(  p\right)  \right)  =B\left(  p\right)  \left(
X_{1}^{h}\left(  p\right)  ,\ldots,X_{q}^{h}\left(  p\right)  \right)
.\label{pr15}%
\end{equation}
Now, using eq.$\left(  \ref{pr8}\right)  $%
\begin{align}
B\left(  pg\right)  \left(  X_{1}^{h}\left(  pg\right)  ,\ldots
,X_{q}^{h}\left(  pg\right)  \right)   &  =B\left(  pg\right)
\left(  R_{g\ast}X_{1}^{h}\left(  p\right)  ,\ldots,R_{g\ast}X_{q}^{h}\left(
p\right)  \right)   \nonumber \\
&  =R_{g}^{\ast}B\left(  p\right)  \left(  X_{1}^{h}\left(  p\right)
,\ldots,X_{q}^{h}\left(  p\right)  \right)  ,
\end{align}
Finally, using eq.$\left(  \ref{pr14}\right)  $ we have%
\begin{equation}
B\left(  pg\right)  \left(  X_{1}^{h}\left(  pg\right)  ,\ldots
,X_{q}^{h}\left(  pg\right)  \right)  =B\left(  p\right)  \left(
X_{1}^{h}\left(  p\right)  ,\ldots,X_{q}^{h}\left(  p\right)  \right),
\label{pr16}%
\end{equation}
and comparing eq.$\left(  \ref{pr15}\right)  $ with eq.$\left(  \ref{pr16}%
\right)  $ one sees that $B\left(  pg\right)  $ and $B\left(
p\right)  $ corresponds to the pull-back of the same form $\mathcal{B}(x)  $%
\begin{equation}
B\left(  pg\right)  \left(  X_{1}^{h}\left(  pg\right)  ,\ldots
,X_{q}^{h}\left(  pg\right)  \right)  =\pi^{\ast}\mathcal{B}\left(  x\right)
\left(  X_{1}\left(  p\right)  ,\ldots,X_{q}\left(  p\right)  \right).
\end{equation}
Thus, the form $\mathcal{B}\left(  x\right)  $ is the projection of
$B$.
\end{proof}

The important thing about the projection operation is that given a smooth
$p-$form $B$ defined over $\mathcal{P}$, its projection
$\mathcal{B}\left(  x\right)  $ gives a $p-$form \textit{globally} defined
over the base space $\mathcal{M}$. This is in contrast with the construction
of differential forms by using the pull-back induced by local sections
$\sigma_{\alpha}$; given a globally defined $p-$form over
$\mathcal{P}$, in general is only possible to obtain $p-$forms locally
defined in an open set $U_{\alpha}\subset\mathcal{M}$.

However, it is possible to find relations between both procedures. In order to
do so, it is necessary to observe that given an arbitrary local section
$\sigma_{\alpha}$, is always possible to find an element $g_{\alpha}%
\in\mathcal{G}$ such that%
\begin{equation}
\sigma_{\alpha}\circ\pi=g_{\alpha}.
\end{equation}
Now, if $B$ is projectable,
\begin{align}
B &  =R_{g_{\alpha}}^{\ast}B \nonumber\\
&  =R_{\sigma_{\alpha}\circ\pi}^{\ast}B\nonumber\\
&  =\pi^{\ast}\sigma_{\alpha}^{\ast}B \ ,%
\end{align}
and writing $\mathcal{B}_{\alpha}\mathcal{=\sigma}_{\alpha}^{\ast}B$,
we have%
\begin{equation}
B=\pi^{\ast}\mathcal{B}_{\alpha}.\label{pr20}%
\end{equation}
Since $B$ is projectable, its projection $\mathcal{B}$ is such that
$B=\pi^{\ast}\mathcal{B}$ is unique. This means that in the non-empty 
overlap $U_{\alpha}\cap U_{\beta}$, the forms $\mathcal{B}$ are equivalent
\begin{equation}
\mathcal{B}_{\alpha}=\mathcal{B}_{\beta}=\mathcal{B}%
\end{equation}
and globally defined over $\mathcal{M}$.

An important property about the projection operation is its relation with the
covariant derivative operator $\mathcal{D}$ defined over the principal bundle $\mathcal{P}$.

\begin{theorem}
If $B$ is projectable, then%
\begin{equation}
\mathcal{D}B=\dd_{\mathcal{P}}B \ . \label{pr18}%
\end{equation}

\end{theorem}

\begin{proof}
In fact,
\begin{align}
\dd_{\mathcal{P}}B\left(  X_{1},\ldots,X_{p}\right)   &
=\dd_{\mathcal{P}}\pi^{\ast}\mathcal{B}\left(  X_{1},\ldots,X_{p}\right)   \nonumber\\
&  =\pi^{\ast}\left(  \dd_{\mathcal{M}}\mathcal{B}\right)  \left(  X_{1}%
,\ldots,X_{p}\right)   \nonumber \\
&  =\dd_{\mathcal{M}}\mathcal{B}\left(  \pi_{\ast}X_{1},\ldots,\pi_{\ast}%
X_{p}\right)  .
\end{align}
By definition of vertical subspace, we have $\pi_{\ast}X_{i}=\pi_{\ast}X_{i}^{h}$
$\left[  \text{see\thinspace eq.}\left(  \ref{pr17}\right)  \right].  $ Thus,%
\begin{align}
\dd_{\mathcal{P}}B
\left(  X_{1},\ldots,X_{p}\right)   &
=\dd_{\mathcal{M}}\mathcal{B}\left(  \pi_{\ast}X_{1}^{h},\ldots,\pi_{\ast}%
X_{p}^{h}\right)   \nonumber\\
&  =\pi^{\ast}\left(  \dd_{\mathcal{M}}\mathcal{B}\right)  \left(  X_{1}%
^{h},\ldots,X_{p}^{h}\right)  \nonumber\\
&  =\dd_{\mathcal{P}}\pi^{\ast}\mathcal{B}\left(  X_{1}^{h},\ldots,X_{p}%
^{h}\right)  \nonumber\\
&  =\dd_{\mathcal{P}}B\left(  X_{1}^{h},\ldots,X_{p}^{h}\right)  \nonumber\\
&  =\dd_{\mathcal{P}}B\circ h\left(  X_{1},\ldots,X_{p}\right)  \nonumber\\
&  =\mathcal{D}B\left(  X_{1},\ldots,X_{p}\right)  ,
\end{align}
and therefore%
\begin{equation}
\dd_{\mathcal{P}}B=\mathcal{D}B \ .%
\end{equation}

\end{proof}

\subsection{Chern--Weil theorem}

\begin{theorem}
Let $\mathcal{P}\left(  \mathcal{M},\mathcal{G}\right)  $ be a principal
bundle with base space $\mathcal{M}$ endowed with a one--form connection $\mathrm{\omega
}\in\Omega^{1}\left(  \mathcal{P}\right)  \otimes\mathfrak{g}$. Let $F\in\Omega^{2}\left(  \mathcal{P}\right)\otimes\mathfrak{g}$ be its
corresponding curvature, where $\mathfrak{g}$ is the Lie algebra associated to the
structure group $\mathcal{G}$. Let $\left\langle \mathsf{T}_{A_{1}}%
\ldots\mathsf{T}_{A_{n+1}}\right\rangle $ be a symmetric invariant tensor of rank
$n+1$, and let $\left\langle F^{n+1}\right\rangle $ be a $\left(
2n+2\right)-$form
\begin{equation}
\left\langle F^{n+1}\right\rangle =\underset{n+1\text{ times}%
}{\underbrace{\left\langle F\wedge\ldots\wedge F\right\rangle
}} \ .
\end{equation}
Then,
\begin{enumerate}
\item $\left\langle F^{n+1}\right\rangle $ is a closed form,
$\dd_{\mathcal{P}}\left\langle F^{n+1}\right\rangle =0$ and projectable%
\begin{equation}
\left\langle F^{n+1}\right\rangle =\pi^{\ast}\left\langle
\mathcal{F}^{n+1}\right\rangle ,
\end{equation}
with $\left\langle \mathcal{F}^{n+1}\right\rangle $ closed,  $\dd_{\mathcal{M}%
}\left\langle \mathcal{F}^{n+1}\right\rangle =0$.
\item Given two connections $\mathrm{\omega}$ and $\mathrm{\bar{\omega}}$ over
$\mathcal{P}$ and their respective curvatures $F$ and $\bar{F}%
$, the difference $\left\langle F^{n+1}\right\rangle -\left\langle
\bar{F}^{n+1}\right\rangle $ is an exact form and projectable.
\end{enumerate}
\end{theorem}

\begin{proof}[Proof part 1]
$\left\langle F^{n+1}\right\rangle $ is projectable since it is
invariant under the right action $R_{g}^{\ast}$ of the group. In fact, given
that $F$ is tensorial, we have%
\begin{equation}
R_{g}^{\ast}F=g^{-1}Fg,
\end{equation}
and due to $\left\langle F^{n+1}\right\rangle $ is an invariant
polynomial, it follows%
\begin{equation}
R_{g}^{\ast}\left\langle F^{n+1}\right\rangle =\left\langle
F^{n+1}\right\rangle .
\end{equation}
On the other hand, since $F$ is tensorial, $F\left(
X_{1},X_{2}\right)  =F\left(  X_{1}^{h},X_{2}^{h}\right)  $ and
therefore%
\begin{equation}
\left\langle F^{n+1}\right\rangle \left(  X_{1},\ldots,X_{2n+2}%
\right)  =\left\langle F^{n+1}\right\rangle \left(  X_{1}^{h}%
,\ldots,X_{2n+2}^{h}\right)  .
\end{equation}
Thus, since $\left\langle F^{n+1}\right\rangle $ is projectable it
satisfies eq.$\left(  \ref{pr18}\right)  $%
\begin{equation}
\dd_{\mathcal{P}}\left\langle F^{n+1}\right\rangle =\mathcal{D}%
\left\langle F^{n+1}\right\rangle ,
\end{equation}
and using the Bianchi identity eq.$\left(  \ref{pr19}\right)  $ we see that
$\left\langle F^{n+1}\right\rangle $ is a closed form%
\begin{equation}
\dd_{\mathcal{P}}\left\langle F^{n+1}\right\rangle =0.
\end{equation}
Moreover, using eq.$\left(  \ref{pr20}\right)  $ one finds that the polynomial
\begin{equation}
\left\langle \mathcal{F}_{\alpha}^{n+1}\right\rangle =\sigma_{\alpha}^{\ast
}\left\langle F^{n+1}\right\rangle
\end{equation}
corresponds to the projection of $\left\langle F^{n+1}\right\rangle $
over $\mathcal{M}$ and therefore $\left\langle \mathcal{F}_{\alpha}%
^{n+1}\right\rangle $ is globally defined. In order to show that $\left\langle
\mathcal{F}_{\alpha}^{n+1}\right\rangle $ is closed, consider the Bianchi
identity%
\begin{equation}
D_{\alpha}\left\langle \mathcal{F}_{\alpha}^{n+1}\right\rangle =0.
\end{equation}
Making use of the invariance property eq.$\left(  \ref{pr21}\right)  $ it follows that%
\begin{equation}
\dd_{\mathcal{M}}\left\langle \mathcal{F}_{\alpha}^{n+1}\right\rangle =0.
\end{equation}
The quantity $\left\langle \mathcal{F}_{\alpha}^{n+1}\right\rangle ,$ with a
given normalization factor, corresponds to the $\left(  n+1\right)  $
\textit{Chern character}%
\begin{equation}
\mathsf{ch}_{n+1}\left(  \mathcal{F}\right)  =\frac{1}{\left(  n+1\right)
!}\left(  \frac{i}{2\pi}\right)  ^{n+1}\left\langle \mathcal{F}^{n+1}%
\right\rangle .
\end{equation}

\end{proof}

\begin{proof}[Proof part 2]
Consider two one-form connections $\mathrm{\omega}$ and $\mathrm{\bar{\omega
}}$. The difference
\begin{equation}
O=\mathrm{\omega}-\mathrm{\bar{\omega}}%
\end{equation}
is tensorial because given a vertical vector $Y\in V_{p}\left(  \mathcal{P}\right)  $, then%
\begin{equation}
O\left(  Y\right)  =\mathsf{Y}-\mathsf{Y}=0.
\end{equation}
Using these two connections, it is possible to define a third interpolating connection $\mathrm{\omega}_t$ as%
\begin{equation}
\mathrm{\omega}_{t}=\mathrm{\bar{\omega}}+tO\label{pr24}%
\end{equation}
where $t\in\left[  0,1\right]  $, and its corresponding curvature
\begin{align}
F_{t} &  =\mathcal{D}\mathrm{\omega}_{t}\nonumber\\
&  =\dd_{\mathcal{P}}\mathrm{\omega}_{t}+\mathrm{\omega}_{t}\wedge
\mathrm{\omega}_{t}\nonumber\\
&  =\bar{F}+t\mathcal{\bar{D}}O+t^{2}O%
\wedge O,\label{pr25}%
\end{align}
where $\mathcal{\bar{D}}O=\dd_{\mathcal{P}}O+\left[
\mathrm{\bar{\omega}},O\right]  $. Now, the derivative respect to $t$
of $F_{t}$ is given by%
\begin{align}
\frac{\dd F_{t}}{\dd t} &  =\mathcal{\bar{D}}O+t\left[
O,O\right]  \nonumber\\
&  =\mathcal{D}_{t}O,\label{pr22}%
\end{align}
where%
\begin{equation}
\mathcal{D}_{t}O=\dd_{\mathcal{P}}O+\left[  \mathrm{\omega
}_{t},O\right] \  .
\end{equation}
Since $\left.  F_{t}\right\vert _{t=0}=\bar{F}$ and $\left.
F_{t}\right\vert _{t=1}=F$, it is possible to write the
difference $ \left\langle F^{n+1}\right\rangle -\left\langle \bar{F}%
^{n+1}\right\rangle$ as
\begin{equation}
\left\langle F^{n+1}\right\rangle -\left\langle \bar{F}%
^{n+1}\right\rangle =\int_{0}^{1}\dd t\frac{\dd}{\dd t}\left\langle F%
_{t}^{n+1}\right\rangle \ .
\end{equation}
The polynomial $\left\langle F_{t}^{n+1}\right\rangle $ is a $\left(  2n+2\right)-$form and symmetric. This allow us to write%
\begin{equation}
\left\langle F^{n+1}\right\rangle -\left\langle \bar{F}%
^{n+1}\right\rangle =\left(  n+1\right)  \int_{0}^{1}\dd t\left\langle
\frac{\dd F_{t}}{\dd t}\wedge F_{t}^{n}\right\rangle ,
\end{equation}
inserting eq.$\left(  \ref{pr22}\right)  $, we get%
\begin{align}
\left\langle F^{n+1}\right\rangle -\left\langle \bar{F}%
^{n+1}\right\rangle  &  =\left(  n+1\right)  \int_{0}^{1}\dd t\left\langle
\mathcal{D}_{t}O\wedge F_{t}^{n}\right\rangle  \nonumber \\
&  =\left(  n+1\right)  \int_{0}^{1}\dd t\mathcal{D}_{t}\left\langle
O\wedge F_{t}^{n}\right\rangle.
\end{align}
where we have used the Bianchi identity $\mathcal{D}_{t}F_{t}=0$. The
form $O$ is tensorial and since $\left\langle O%
\wedge F_{t}^{n}\right\rangle $ is an invariant polynomial, it is
projectable. This means that $\mathcal{D}_{t}\left\langle O%
\wedge F_{t}^{n}\right\rangle =\dd_{\mathcal{P}}\left\langle
O\wedge F_{t}^{n}\right\rangle $ and therefore we write%
\begin{equation}
\boxed{
\left\langle F^{n+1}\right\rangle -\left\langle \bar{F}%
^{n+1}\right\rangle =\dd_{\mathcal{P}}\mathrm{T}_{\mathrm{\omega}%
\leftarrow\mathrm{\bar{\omega}}}^{\left(  2n+1\right)  }}%
\end{equation}
with%
\begin{equation}
\mathrm{T}_{\mathrm{\omega}\leftarrow\mathrm{\bar{\omega}}}^{\left(
2n+1\right)  }=\left(  n+1\right)  \int_{0}^{1}\dd t\left\langle O%
\wedge F_{t}^{n}\right\rangle \label{pr23}%
\end{equation}

The $\left(  2n+1\right)  -$form $\mathrm{T}_{\mathrm{\omega}\leftarrow
\mathrm{\bar{\omega}}}^{\left(  2n+1\right)  }$ defined over the principal bundle $\mathcal{P}$ is
called \textit{Transgression form}. The transgression form is projectable over
the base space $\mathcal{M}$%
\begin{align}
T_{\mathcal{A}_{\alpha}\leftarrow\mathcal{\bar{A}}_{\alpha}}^{\left(
2n+1\right)  } &  =\sigma_{\alpha}^{\ast}\mathrm{T}_{\mathrm{\omega}%
\leftarrow\mathrm{\bar{\omega}}}^{\left(  2n+1\right)  }\nonumber \\
&  =\left(  n+1\right)  \int_{0}^{1}\dd t\left\langle  \Theta_{\alpha
} \wedge \left[  \mathcal{F}_{t}\right]_{\alpha}^{n}\right\rangle , \label{transform}
\end{align}
with
\begin{align}
\Theta_{\alpha} &  =\sigma_{\alpha}^{\ast}\mathrm{O}\text{,}\\
\left[  \mathcal{F}_{t}\right]  _{\alpha} &  =\sigma_{\alpha}^{\ast}%
F_{t}.
\end{align}
The form $T_{\mathcal{A}_{\alpha}\leftarrow\mathcal{\bar{A}}_{\alpha}%
}^{\left(  2n+1\right)  }$ is called transgression over the base space or
simply transgression if there is no room for confusion. The transgression form
is globally defined so it will be denoted simply by $T_{\mathcal{A}%
\leftarrow\mathcal{\bar{A}}}^{\left(  2n+1\right)  }$. Note that over the base
space $\mathcal{M}$, the relation
\begin{equation}
\left\langle \mathcal{F}^{n+1}\right\rangle -\left\langle \mathcal{\bar{F}%
}^{n+1}\right\rangle =\dd_{\mathcal{M}}T_{\mathcal{A}\leftarrow\mathcal{\bar{A}%
}}^{\left(  2n+1\right)  }\label{pr26}%
\end{equation}
holds.
\end{proof}

\subsection{Chern--Simons forms} \label{chssect}

It is important to emphasize that over the principal bundle $\mathcal{P}$ the
form $\left\langle F^{n+1}\right\rangle $ is not only closed but
exact. However, this is not true for $\left\langle \mathcal{F}^{n+1}%
\right\rangle $ as we will see shortly.

In the definition of the transgression form eq.$\left(  \ref{pr23}\right)  $,
two one--form connections $\mathrm{\omega}$ and $\mathrm{\bar{\omega}}$ were
used to construct a third one $\mathrm{\omega}_{t}$ given in eq.$\left(
\ref{pr24}\right)  $. It is interesting to observe that a connection over
$\mathcal{P}$ cannot be zero; otherwise it would not satisfy eq.$\left(
\ref{pr4}\right)  $. Thus, to impose the condition $\mathrm{\bar{\omega}}=0$
would mean that the expressions for $\mathrm{\omega}_{t}$ and $F_{t}$
$\left[  \text{see eq}.\left(  \ref{pr24},\ref{pr25}\right)  \right]  $
\begin{align}
\mathrm{\omega}_{t} &  =t\mathrm{\omega},\\
F_{t} &  =t \dd_{\mathcal{P}}\mathrm{\omega}+t^{2}\mathrm{\omega}%
\wedge\mathrm{\omega},
\end{align}
no longer correspond to a one--form connection and a two--form curvature over
$\mathcal{P}$. However, the procedure used in the second half of the proof of
the Chern--Weil theorem still holds. In that case, let us to write%
\begin{equation}
\left\langle F^{n+1}\right\rangle =\dd_{\mathcal{P}}\mathrm{Q}^{\left(
2n+1\right)  }\left(  \mathrm{\omega}\right)  ,
\end{equation}
where%
\begin{align}
\mathrm{Q}^{\left(  2n+1\right)  }\left(  \mathrm{\omega}\right)   &
=\mathrm{T}_{\mathrm{\omega}\leftarrow\mathrm{0}}^{\left(  2n+1\right)  } \nonumber \\
&  =\left(  n+1\right)  \int_{0}^{1}\dd t\left\langle \mathrm{\omega}%
\wedge\left(  t\dd_{\mathcal{P}}\mathrm{\omega}+t^{2}\mathrm{\omega}%
\wedge\mathrm{\omega}\right)^{n} \right\rangle ,
\end{align}
is called the \textit{Chern--Simons} form over $\mathcal{P}$. Again, since
$\mathrm{\omega}_{t}$ and $F_{t}$ are not proper connection and
curvature on $\mathcal{P}$ respectively, the Chern--Simons form it is not
invariant and therefore not projectable over the base space $\mathcal{M}$.

In other words, given two local sections $\sigma_{\alpha}:U_{\alpha}%
\subset\mathcal{M}\rightarrow\mathcal{P}$ and $\sigma_{\beta}:U_{\beta
}\subset\mathcal{M}\rightarrow\mathcal{P}$, in general must occur that%
\begin{equation}
\sigma_{\alpha}^{\ast}\mathrm{Q}^{\left(  2n+1\right)  }\neq\sigma_{\beta
}^{\ast}\mathrm{Q}^{\left(  2n+1\right)  }.
\end{equation}
In this way, the definition of the Chern--Simons form over the base space $\mathcal{M}$
\begin{align}
Q_{\alpha}^{\left(  2n+1\right)  }\left(  \mathcal{A}_{\alpha}\right)   &
=\sigma_{\alpha}^{\ast}\mathrm{Q}^{\left(  2n+1\right)  }\left(
\mathrm{\omega}\right)  \nonumber \\
&  =\left(  n+1\right)  \int_{0}^{1} \dd t\left\langle \mathcal{A}_{\alpha}%
\wedge\left(  t \dd_{\mathcal{M}}\mathcal{A}_{\alpha}+t^{2}\mathcal{A}_{\alpha
}\wedge\mathcal{A}_{\alpha}\right)^{n} \right\rangle \ .
\end{align}
it is only \textit{locally} defined over a particular open set $U_{\alpha}%
\in\mathcal{M}$. Consequently, $\left\langle \mathcal{F}^{n+1}\right\rangle $
can be written as the exterior derivative of the Chern--Simons form
$Q_{\alpha}^{\left(  2n+1\right)  }$ only in an open region $U_{\alpha}\in\mathcal{M}$
\begin{equation}
\left.  \left\langle \mathcal{F}^{n+1}\right\rangle \right\vert _{U_{\alpha}%
}=\dd_{\mathcal{M}}Q_{\alpha}^{\left(  2n+1\right)  } \ , \label{csclo}
\end{equation}
but not globally over $\mathcal{M}$. In the future, we omit the index $\alpha$ in
order to avoid overloaded notation. However, it is important to keep in mind
that the Chern--Simons form is only locally defined.

\newsection{Homotopy}

Another interesting property of transgression forms is that the following expression
vanishes identically
\begin{align}
\dd T_{\mathcal{A}\leftarrow\mathcal{\tilde{A}}}^{\left(  2n+1\right)
}+\dd T_{\mathcal{\tilde{A}}\leftarrow\mathcal{\bar{A}}}^{\left(  2n+1\right)
}+\dd T_{\mathcal{\bar{A}}\leftarrow\mathcal{A}}^{\left(  2n+1\right)  } &
=\left\langle \mathcal{F}^{n+1}\right\rangle -\left\langle \mathcal{\tilde{F}%
}^{n+1}\right\rangle +\left\langle \mathcal{\tilde{F}}^{n+1}\right\rangle \nonumber \\
&  -\left\langle \mathcal{\bar{F}}^{n+1}\right\rangle +\left\langle
\mathcal{\bar{F}}^{n+1}\right\rangle -\left\langle \mathcal{F}^{n+1}%
\right\rangle \nonumber \\
&  =0.
\end{align}
for independent gauge connections $\mathcal{A}$, $\mathcal{\bar{A}}$ and
$\mathcal{\tilde{A}}$ . Since $T_{\mathcal{\bar{A}}\leftarrow\mathcal{A}%
}^{\left(  2n+1\right)  }=-T_{\mathcal{A}\leftarrow\mathcal{\bar{A}}}^{\left(
2n+1\right)  }$, this means that%
\begin{equation}
\dd T_{\mathcal{A}\leftarrow\mathcal{\bar{A}}}^{\left(  2n+1\right)
}=\dd
T_{\mathcal{A}\leftarrow\mathcal{\tilde{A}}}^{\left(  2n+1\right)
}+\dd
T_{\mathcal{\tilde{A}}\leftarrow\mathcal{\bar{A}}}^{\left(  2n+1\right)  }.
\end{equation}
Therefore, it is always possible to decompose a transgression as the sum of
two others plus a closed form $\vartheta$%
\begin{equation}
T_{\mathcal{A}\leftarrow\mathcal{\bar{A}}}^{\left(  2n+1\right)
}=T_{\mathcal{A}\leftarrow\mathcal{\tilde{A}}}^{\left(  2n+1\right)
}+T_{\mathcal{\tilde{A}}\leftarrow\mathcal{\bar{A}}}^{\left(  2n+1\right)
}+\vartheta.
\end{equation}
The functional dependence of the closed form $\vartheta$ can be determined by
using a powerful tool, the \textit{Extended Homotopy Cartan formula.}

\subsection{The extended Cartan homotopy formula \textrm{ECHF}}

Consider now a set of $(r+2)$ independent gauge connections $\mathcal{A}%
_{i}\in\Omega^{1}\left(  \mathcal{M}\right)  \otimes\mathfrak{g}$ with
$i=0,...,r+1$. Moreover, let us consider the embedding of the $(r+1)$ simplex $\Delta_{r+1}$ in $\mathbb{R}^{r+2}$ defined by
\begin{equation}
\begin{tabular}{c c c c}
$\{(t^0,t^1, \ldots,t^{r+1}) \in \mathbb{R}^{r+2}$, & 
with& $t^i\geqslant 0, \forall i=1, \ldots, r+1$ and& $\sum\limits_{i=0}^{r+1}t^i=1\}$
\end{tabular}
\end{equation}
 The relation between the simplex and the set of gauge connections is such that the expression
\begin{equation}
\mathcal{A}_{t}=\sum\limits_{i=0}^{r+1}t^{i}\mathcal{A}_{i}%
\end{equation}
transforms as a gauge connection when $(t^0,t^1, \ldots,t^{r+1})\in \Delta_{r+1}$. It is direct to verify that if one performs a gauge transformation to each $\mathcal{A}_{i}$, then $\mathcal{A}_{t}$ transforms as a connection. This
allow us to define a two--form curvature $\mathcal{F}_{t}=\dd\mathcal{A}%
_{t}+\mathcal{A}_{t}\wedge\mathcal{A}_{t}$. 
Thus, for every point of the simplex $\Delta_{r+1}$, we associate a connection $\mathcal{A}_i$. In particular, the $i$-th vertex of $\Delta_{r+1}$ is related to
$i$-th connection $\mathcal{A}_{i}$. We denote the simplex of the associated gauge connections by
\begin{equation}
\Delta_{r+1}=\left(  \mathcal{A}_{0}\mathcal{A}_{1}...\mathcal{A}_{r+1}\right).
\end{equation}

Let $\Upsilon$ be a polynomial in the forms $\left\{  \mathcal{A}%
_{t},\mathcal{F}_{t},\dd_{t}\mathcal{A}_{t},\dd_{t}\mathcal{F}_{t}\right\}  $
which is also a $\left(  q+m\right)  -$form over $\Delta_{r+1}\times\mathcal{M}$. Let $\dd$ and $\dd_t$ be the exterior derivative operator acting on $\mathcal{M}$ and $\Delta_{r+1}$ respectively.
We introduce now the \textit{homotopy derivation
} operator $l_{t}$ which maps differential forms according to%
\begin{equation}
l_{t}:\Omega^{a}\left(  \mathcal{M}\right)  \times\Omega^{b}\left(
\Delta_{r+1}\right)  \longrightarrow\Omega^{a-1}\left(  \mathcal{M}\right)
\times\Omega^{b+1}\left(  \Delta_{r+1}\right)  .
\end{equation}
This operator satisfies the Leibniz rule, and together with the operators $\dd$
and $\dd_{t}$ define the following graded algebra%
\begin{align}
\dd^{2} &  =0,\label{pr27}\\
\dd_{t}^{2} &  =0,\label{pr28}\\
\left[  l_{t},\dd\right]   &  =\dd_{t},\label{pr29}\\
\left\{  \dd,\dd_{t}\right\}   &  =0.\label{pr30}%
\end{align}
A consistent way to define the action of $l_{t}$ over $\mathcal{A}_{t}$ and $\mathcal{F}_{t}$ is the following%
\begin{align}
l_{t}\mathcal{F}_t &  =\dd_{t}\mathcal{A}_{t},\\
l_{t}\mathcal{A}_{t} &  =0.
\end{align}

Using eq.$\left(  \ref{pr27}-\ref{pr30}\right)  $ it is possible to show that%
\begin{equation}
\left[  l_{t}^{p+1},\dd\right]  \Upsilon=\left(  p+1\right)  \dd_{t}l_{t}%
^{p}\Upsilon.
\end{equation}
Now, integrating over the simplex $\Delta_{r+1}$ with $p+q=r$ and $m\geqslant p$, we have%
\begin{align}
\int_{\Delta_{r+1}}l_{t}^{p+1} \dd\Upsilon-\int_{\Delta_{r+1}}\dd l_{t}^{p+1}\Upsilon &
=\left(  p+1\right)  \int_{\Delta_{r+1}}\dd_{t}l_{t}^{p}\Upsilon \nonumber\\
& =\left(  p+1\right)  \int_{\partial \Delta_{r+1}}l_{t}^{p}\Upsilon.\label{pr31}%
\end{align}
Since $l_{t}^{p+1}\Upsilon$ is a $\left(  r+1\right)  -$form over $\Delta_{r+1}$,
it follows%
\begin{equation}
\int_{\Delta_{r+1}}\dd l_{t}^{p+1}\Upsilon=\left(  -1\right)  ^{r+1}\dd \int_{\Delta_{r+1}%
}l_{t}^{p+1}\Upsilon
\end{equation}
Replacing in eq.$\left(  \ref{pr31}\right)  $ and introducing normalization
factors we have%
\begin{equation}
\int_{\partial \Delta_{r+1}}\frac{1}{p!}l_{t}^{p}\Upsilon=\int_{\Delta_{r+1}}\frac
{1}{\left(  p+1\right)  !}l_{t}^{p+1}\dd \Upsilon+\left(  -1\right)  ^{r}%
\dd\int_{\Delta_{r+1}}\frac{1}{\left(  p+1\right)  !}l_{t}^{p+1}\Upsilon.\label{pr32}%
\end{equation}
This important result it is known in literature as the \textit{Extended
Homotopy Cartan Formula }(\textrm{EHCF}) \cite{Manes:2012zz,Izaurieta:2006wv,Izaurieta:2005vp}.

Let us look now at the particular case where the following polynomial is
selected%
\begin{equation}
\Upsilon=\left\langle \mathcal{F}_{t}^{n+1}\right\rangle .
\end{equation}
This choice carries three properties
\begin{enumerate}
\item $\Upsilon$ is $\mathcal{M}-$closed,
\item $\Upsilon$ is a $0-$form on $\Delta_{r+1}$ i.e, $q=0$,
\item $\Upsilon$ is a $\left(  2n+2\right)  -$form on $\mathcal{M}$. Thus, $0\leqslant p \leqslant 2n+2$.
\end{enumerate}

With these considerations
\textrm{EHCF} reduces to%
\begin{equation}
\int_{\partial \Delta_{p+1}}\frac{l_{t}^{p}}{p!}\left\langle \mathcal{F}_{t}%
^{n+1}\right\rangle =\left(  -1\right)  ^{p}\dd\int_{\Delta_{p+1}}\frac{l_{t}^{p+1}%
}{\left(  p+1\right)  !}\left\langle \mathcal{F}_{t}^{n+1}\right\rangle,
\end{equation}
which is known as the restricted or closed version of \textrm{EHCF}.

\subsubsection{$p=0$, the Chern-Weil theorem}

A well known particular case of \textrm{EHCF} is the Chern--Weil theorem.
Setting $p=0$ in the above expression, we get%
\begin{equation}
\int_{\partial \Delta_{1}}\left\langle \mathcal{F}_{t}^{n+1}\right\rangle
=\dd\int_{\Delta_{1}}l_{t}\left\langle \mathcal{F}_{t}^{n+1}\right\rangle
\label{ch2.jojo}%
\end{equation}
where $\mathcal{F}_{t}$ is the curvature for the connection $\mathcal{A}%
_{t}=t_{0}\mathcal{A}_{0}+t_{1}\mathcal{A}_{1}$ where $t_{0}+t_{1}=1$. The
boundary of the simplex $\Delta_{1}=\left(  \mathcal{A}_{0},\mathcal{A}_{1}\right)
$ corresponds to%
\begin{equation}
\partial\left(  \mathcal{A}_{0},\mathcal{A}_{1}\right)  =\left(
\mathcal{A}_{1}\right)  -\left(  \mathcal{A}_{0}\right)  ,
\end{equation}
so the left hand side of eq.$\left(  \ref{ch2.jojo}\right)  $ is then%
\begin{equation}
\int_{\partial \Delta_{1}}\left\langle \mathcal{F}_{t}^{n+1}\right\rangle
=\left\langle \mathcal{F}_{1}^{n+1}\right\rangle -\left\langle \mathcal{F}%
_{0}^{n+1}\right\rangle .
\end{equation}
Since $\left\langle \mathcal{F}_{t}^{n+1}\right\rangle $ is a symmetric
polynomial, we have
\begin{equation}
l_{t}\left\langle \mathcal{F}_{t}^{n+1}\right\rangle =\left(  n+1\right)
\left\langle l_{t}\mathcal{F}_{t}\mathcal{F}_{t}^{n}\right\rangle .
\end{equation}
Now, the homotopic derivative of the curvature $l_{t}\mathcal{F}_{t}$ is given
by
\begin{align}
l_{t}\mathcal{F}_{t} &  =\dd_{t}\mathcal{A}_{t} \  \nonumber\\
&  =\dd t^{0}\mathcal{A}_{0}+\dd t^{1}\mathcal{A}_{1}\nonumber \\
&  =\dd t^{1}\left(  \mathcal{A}_{1}-\mathcal{A}_{0}\right)  .
\end{align}
In this way,%
\begin{equation}
l_{t}\left\langle \mathcal{F}_{t}^{n+1}\right\rangle =\left(  n+1\right)
\dd t^{1}\left\langle \left(  \mathcal{A}_{1}-\mathcal{A}_{0}\right)
\wedge\mathcal{F}_{t}^{n}\right\rangle .\label{pr33}%
\end{equation}
Replacing in eq.$\left(  \ref{ch2.jojo}\right)  $ we obtain the Chern--Weil
theorem%
\begin{align}
\left\langle \mathcal{F}_{1}^{n+1}\right\rangle -\left\langle \mathcal{F}%
_{0}^{n+1}\right\rangle  &  =\left(  n+1\right)  \dd\int_{0}^{1}\dd t\left\langle
\left(  \mathcal{A}_{1}-\mathcal{A}_{0}\right)  \wedge\mathcal{F}_{t}%
^{n}\right\rangle \\
&  =\dd T_{\mathcal{A}_{1}\leftarrow\mathcal{A}_{0}}^{\left(  2n+1\right)  }.
\end{align}

\subsubsection{$p=1$, the triangle equation} \label{tria}

The $p=1$ case corresponds to the so called triangle equation. In fact, for
$p=1$ \textrm{EHCF} reads%
\begin{equation}
\int_{\partial \Delta_{2}}l_{t}\left\langle \mathcal{F}_{t}^{n+1}\right\rangle
=-\frac{1}{2}\dd\int_{\Delta_{2}}l_{t}^{2}\left\langle \mathcal{F}_{t}^{n+1}%
\right\rangle, \label{ch2.2klkl}%
\end{equation}
where $\mathcal{F}_{t}$ is the curvature for the connection $\mathcal{A}%
_{t}=t_{0}\mathcal{A}_{0}+t_{1}\mathcal{A}_{1}+t_{2}\mathcal{A}_{2},$ with
$t_{0}+t_{1}+t_{2}=1$. Again, the boundary of the simplex $\Delta_{2}=\left(
\mathcal{A}_{0},\mathcal{A}_{1},\mathcal{A}_{2}\right)  $ corresponds to%
\begin{equation}
\partial\left(  \mathcal{A}_{0},\mathcal{A}_{1},\mathcal{A}_{2}\right)
=\left(  \mathcal{A}_{1}\mathcal{A}_{2}\right)  -\left(  \mathcal{A}%
_{0}\mathcal{A}_{2}\right)  +\left(  \mathcal{A}_{0}\mathcal{A}_{1}\right)  ,
\end{equation}
and the left hand side in eq.$\left(  \ref{ch2.2klkl}\right)  $ is given by%
\begin{equation}
\int_{\partial \Delta_{2}}l_{t}\left\langle \mathcal{F}_{t}^{n+1}\right\rangle
=\int_{\left(  \mathcal{A}_{1}\mathcal{A}_{2}\right)  }l_{t}\left\langle
\mathcal{F}_{t}^{n+1}\right\rangle -\int_{\left(  \mathcal{A}_{0}%
\mathcal{A}_{2}\right)  }l_{t}\left\langle \mathcal{F}_{t}^{n+1}\right\rangle
+\int_{\left(  \mathcal{A}_{0}\mathcal{A}_{1}\right)  }l_{t}\left\langle
\mathcal{F}_{t}^{n+1}\right\rangle .
\end{equation}
Since $\left[  \text{see eq.}\left(  \ref{pr33}\right)  \right]  $%
\begin{equation}
\int_{\left(  \mathcal{A}_{1}\mathcal{A}_{2}\right)  }l_{t}\left\langle
\mathcal{F}_{t}^{n+1}\right\rangle =T_{\mathcal{A}_{2}\leftarrow
\mathcal{A}_{1}}^{\left(  2n+1\right)  },
\end{equation}
it follows%
\begin{equation}
\int_{\partial \Delta_{2}}l_{t}\left\langle \mathcal{F}_{t}^{n+1}\right\rangle
=T_{\mathcal{A}_{2}\leftarrow\mathcal{A}_{1}}^{\left(  2n+1\right)
}-T_{\mathcal{A}_{2}\leftarrow\mathcal{A}_{0}}^{\left(  2n+1\right)
}+T_{\mathcal{A}_{1}\leftarrow\mathcal{A}_{0}}^{\left(  2n+1\right)
}.\label{ch2.2.ttt}%
\end{equation}
Now, using the symmetry of the polynomial $\left\langle ...\right\rangle $ one
derives%
\begin{equation}
\frac{1}{2}l_{t}^{2}\left\langle \mathcal{F}_{t}^{n+1}\right\rangle =\frac
{1}{2}n\left(  n+1\right)  \left\langle \left(  \dd_{t}\mathcal{A}_{t}\right)
^{2}\wedge\mathcal{F}_{t}^{n-1}\right\rangle ,\label{ch2.2.popo}%
\end{equation}
where%
\begin{equation}
\dd_{t}\mathcal{A}_{t}=\dd t^{0}\mathcal{A}_{0}+\dd t^{1}\mathcal{A}_{1}%
+\dd t^{2}\mathcal{A}_{2,}%
\end{equation}
but considering that
\begin{equation}
\dd t^{0}+\dd t^{1}+\dd t^{2}=0 \ ,
\end{equation}
one finds%
\begin{equation}
\dd_{t}\mathcal{A}_{t}=\dd t^{0}\left(  \mathcal{A}_{0}-\mathcal{A}_{1}\right)
+\dd t^{2}\left(  \mathcal{A}_{2}-\mathcal{A}_{1}\right) \ .
\end{equation}
Replacing in eq.$\left(  \ref{ch2.2.popo}\right)  $ we obtain%
\begin{equation}
\frac{1}{2}l_{t}^{2}\left\langle \mathcal{F}_{t}^{n+1}\right\rangle =-n\left(
n+1\right)  \dd t^{0}\dd t^{2}\left\langle \left(  \mathcal{A}_{2}-\mathcal{A}%
_{1}\right)  \wedge\left(  \mathcal{A}_{1}-\mathcal{A}_{0}\right)
\wedge\mathcal{F}_{t}^{n-1}\right\rangle \ .
\end{equation}
It is convenient to re-define the integration parameters%
\begin{align*}
t &  =1-t^{0},\\
s &  =t^{2},
\end{align*}
and integrate explicitly over $\Delta_{2}$. In this way, we get
\begin{equation}
\frac{1}{2}\dd\int_{ \Delta_{2}}l_{t}^{2}\left\langle \mathcal{F}_{t}%
^{n+1}\right\rangle =Q_{\mathcal{A}_{2}\leftarrow\mathcal{A}_{1}%
\leftarrow\mathcal{A}_{0}}^{\left(  2n\right)  } \ ,%
\end{equation}
where $Q_{\mathcal{A}_{2}\leftarrow\mathcal{A}_{1}\leftarrow\mathcal{A}_{0}%
}^{\left(  2n\right)  }$ is defined as%
\begin{equation}
Q_{\mathcal{A}_{2}\leftarrow\mathcal{A}_{1}\leftarrow\mathcal{A}_{0}}^{\left(
2n\right)  }=n\left(  n+1\right)  \int_{0}^{1} \dd t\int_{0}^{t} \dd s\left\langle
\left(  \mathcal{A}_{2}-\mathcal{A}_{1}\right)  \wedge\left(  \mathcal{A}%
_{1}-\mathcal{A}_{0}\right)  \wedge\mathcal{F}_{st}^{n-1}\right\rangle  \label{chortiz}
\end{equation}
where $\mathcal{F}_{st}$ is the curvature for the connection
\begin{equation}
\mathcal{A}_{s,t}   =s\, \left(  \mathcal{A}_{2}-\mathcal{A}_{1}\right)
+t\, (\mathcal{A}_{1}-\mathcal{A}_{0})  +\mathcal{A}_{0} \ .
\end{equation}
Inserting eq.(\ref{chortiz}) into eq.$\left(  \ref{ch2.2.ttt}\right)  $ the
triangle equation is finally obtained%
\begin{equation}
\ T_{\mathcal{A}_{2}\leftarrow\mathcal{A}_{0}}^{\left(  2n+1\right)
}=T_{\mathcal{A}_{2}\leftarrow\mathcal{A}_{1}}^{\left(  2n+1\right)
}+T_{\mathcal{A}_{1}\leftarrow\mathcal{A}_{0}}^{\left(  2n+1\right)
}+\dd Q_{\mathcal{A}_{2}\leftarrow\mathcal{A}_{1}\leftarrow\mathcal{A}_{0}%
}^{\left(  2n\right)  }.\label{pr34}%
\end{equation}

As we will see in the following chapter,  transgression forms will be used
as Lagrangians for constructing physical theories. For this reason, eq.$\left(  \ref{pr34}\right)  $ is
extremely useful since it allows to split the transgression in different pieces which will correspond
to different interactions present in the resulting theory.

\chapter{Transgression forms as source for gauge theories}
\label{ch:Transgression}
\begin{flushright}
\textit{``...s\'olo cuando transgredo alguna
orden \\
el futuro se vuelve respirable...}''. \\ \textit{Transgresiones, Mario Benedetti.}
\footnote{\scriptsize ``...only when I transgress an order, the future becomes breathable...''. Transgressions, Mario Benedetti.}
\bigskip
\end{flushright}

Given a certain symmetry, the natural way to construct a gauge invariant
theory is by using a Yang-Mills Lagrangian. This is for instance the case for
the interactions of the Standard Model, the most successful and remarkable
model in particle physics. The Yang--Mills action functional is given by%
\begin{equation}
\mathsf{S}_{\mathrm{YM}}=-\frac{1}{4}\int_{\mathcal{M}}\sqrt{g}%
\text{\textbf{Tr}}\left(  \mathcal{F}^{\mu\nu}\mathcal{F}_{\mu\nu}\right)
d^{4}x. \label{tr1}%
\end{equation}
However, is not so hard to face some limitations. As it can be seen from
eq.$\left(  \ref{tr1}\right)  $, the Yang-Mills action requires the inclusion
of a metric structure as background. Thus, in the case of a curved base
space $\mathcal{M}$, the metric tensor becomes dynamical and the action cannot
be considered as describing a pure gauge theory due to the inclusion of a
dynamic field which it is not part of the gauge connection $\mathcal{A}$.

For this reason, it seems natural to think in gauge theories which are
background independent. Transgression forms are natural candidates for
``metric-free" actions in the sense that they can be used to construct gauge
theories without any associated background. The scenario is still better:
Transgressions forms are genuinely gauge invariant objects, they are also
globally defined so they can be integrated over the whole base space
$\mathcal{M}$. The price to pay is the inclusion, in addition to the gauge
field $\mathcal{A}$, a second one--form connection $\mathcal{\bar{A}}$. This
duplicity in the gauge field configurations may be thought as an obstruction from a
traditional point of view. However, it will be shown that the presence of two gauge
connections is a very versatile feature. In fact, 
transgression forms give rise to different types of theories depending on
the conditions imposed on the connections $\mathcal{A}$ and $\mathcal{\bar{A}%
}$.

Throughout this chapter, we review the consequences of choosing any of these conditions. First, we consider the most general transgression field theory. That is, without imposing any condition on the
connections. Secondly, we show that transgressions forms lead to Chern--Simons theories by imposing 
$\mathcal{\bar{A}}=0$, and we will present as an example the construction of Chern--Simons gravitational theories in odd dimensions. Finally, we study the case when both gauge
connections are related by a gauge transformation. In the later case, the resulting theory is the so called
gauged Wess--Zumino--Witten action. For more details about some of the results
presented in this section, see \cite{Mora:2003wy,Borowiec:2003rd,Borowiec:2005ky,Miskovic:2004us}.

\newsection{Transgressions forms as Lagrangians}

\begin{definition}
Let $\mathcal{P}$ be a principal bundle with a $\left(  2n+1\right)
-$dimensional orientable base space $\mathcal{M}$. Let $\mathcal{A}$ and
$\mathcal{\bar{A}}$ be two Lie valued connections, and let
$T_{\mathcal{A}\leftarrow\mathcal{\bar{A}}}^{\left(  2n+1\right)  }$ be the
transgression form over the base space $\mathcal{M}$ given by eq.$(\ref{transform})$.
We define as the Transgression Lagrangian over $\mathcal{M}$, to the $\left(
2n+1\right)  -$form 
\begin{align}
\mathscr{L}_{\mathsf{T}}^{\left(  2n+1\right)  }\left(  \mathcal{A}%
,\mathcal{\bar{A}}\right)   &  =\kappa T_{\mathcal{A}\leftarrow\mathcal{\bar
{A}}}^{\left(  2n+1\right)  }\nonumber\\
&  =\kappa\left(  n+1\right)  \int_{0}^{1}dt\left\langle \Theta\wedge
\mathcal{F}_{t}^{n}\right\rangle \ , \label{ch2transl}%
\end{align}
where $\kappa$ is a constant, $\Theta=\mathcal{A}$ $-\mathcal{\bar{A}}$, and
$\mathcal{F}_{t}=d\mathcal{A}_{t}+\mathcal{A}_{t}\wedge\mathcal{A}_{t}$ is the
curvature of $\mathcal{A}_{t}=\mathcal{\bar{A}}+t\Theta$.
\end{definition}

It is direct to show that the transgression Lagrangian eq.$\left(
\ref{ch2transl}\right)  $ is gauge invariant. In fact, under gauge
transformations~\cite{Izaurieta:2005vp,Mora:2006ka}%
\begin{align}
\mathcal{A}  &  \rightarrow\mathcal{A}^{\prime}=g^{-1}\mathcal{A}%
g+g^{-1}\dd g\label{ch2gin1} \ , \\
\mathcal{\bar{A}}  &  \rightarrow\mathcal{\bar{A}}^{\prime}=g^{-1}%
\mathcal{\bar{A}}g+g^{-1}\dd g \ , \label{ch2gin2}%
\end{align}
with $g\left(  x\right)  =\exp\left\{  \lambda^{A}\left(  x\right)
\mathsf{T}_{A}\right\}  \in\mathcal{G}$ and $\left\{  \mathsf{T}%
_{A},A=1,...,\dim\left(  \mathfrak{g}\right)  \right\}  ,$we have
\begin{align}
\Theta &  \longrightarrow\Theta^{\prime}=g^{-1}\Theta g \ ,\\
\mathcal{F}_{t}  &  \longrightarrow\mathcal{F}_{t}^{\prime}=g^{-1}%
\mathcal{F}_{t}g \ .
\end{align}
Using the invariance property of $\left\langle \ldots\right\rangle $
eq.$(\ref{pr21})$, it follows%
\begin{equation}
\mathscr{L}_{\mathsf{T}}^{\left(  2n+1\right)  }\left(  \mathcal{A}%
,\mathcal{\bar{A}}\right)  =\mathscr{L}_{\mathsf{T}}^{\left(  2n+1\right)
}\left(  \mathcal{A}^{\prime},\mathcal{\bar{A}}^{\prime}\right)  .
\end{equation}

The transgression Lagrangian satisfy two properties
\begin{itemize}
\item Antisymmetry%
\begin{equation}
\mathscr{L}_{\mathsf{T}}^{\left(  2n+1\right)  }\left(  \mathcal{A}%
,\mathcal{\bar{A}}\right)  =-\mathscr{L}_{\mathsf{T}}^{\left(  2n+1\right)
}\left(  \mathcal{\bar{A}},\mathcal{A}\right)  .
\end{equation}

\item Triangle Equation [see Section $(\ref{tria})$]%
\begin{equation}
\mathscr{L}_{\mathcal{A\leftarrow\bar{A}}}^{\left(  2n+1\right)  } =\mathscr{L}_{\mathcal{A\leftarrow\tilde{A}}}^{\left(
2n+1\right)  }-\mathscr{L}_{\mathcal{\bar{A}\leftarrow\tilde{A}}}^{\left(  2n+1\right)
}-\kappa\dd Q_{\mathcal{A\leftarrow\bar{A}\leftarrow\tilde{A}}}^{\left(  2n\right)  } \ .
\label{ch2trian}%
\end{equation}

\end{itemize}

It is interesting to note that imposing $\mathcal{\tilde{A}}=0$ in eq.$\left(
\ref{ch2trian}\right)  $ does not affect the global property of $\mathscr{L}%
_{\mathsf{T}}^{\left(  2n+1\right)  }\left(  \mathcal{A},\mathcal{\bar{A}%
}\right)  $. This is mainly because the definition of $\mathscr{L}_{\mathsf{T}%
}^{\left(  2n+1\right)  }\left(  \mathcal{A},\mathcal{\bar{A}}\right)  $ does
not depend on the election of $\mathcal{\tilde{A}}$.

On the other hand, since $Q^{\left(  2n+1\right)  }\left(  \mathcal{A}\right)
=T_{\mathcal{A}\leftarrow0}^{\left(  2n+1\right)  }$ corresponds to the
Chern--Simons form locally defined over $\mathcal{M}$, it is natural to define
the Chern--Simons Lagrangian as%
\begin{align}
\mathscr{L}_{\mathsf{CS}}^{\left(  2n+1\right)  }\left(  \mathcal{A}\right)
&  =\kappa T_{\mathcal{A}\leftarrow0}^{\left(  2n+1\right)  }\nonumber\\
&  =\kappa\left(  n+1\right)  \int_{0}^{1}dt\left\langle \mathcal{A\wedge
}\left(  td\mathcal{A}+t^{2}\mathcal{A}\wedge\mathcal{A}\right)^n  \right\rangle
. \label{ch2csl}%
\end{align}
Thus, the transgression Lagrangian can be written as the difference
of two Chern--Simons forms plus an exact form 
\begin{equation}
\mathscr{L}_{\mathsf{T}}^{\left(  2n+1\right)  }\left(  \mathcal{A}%
,\mathcal{\bar{A}}\right)  =\mathscr{L}_{\mathsf{CS}}^{\left(  2n+1\right)
}\left(  \mathcal{A}\right)  -\mathscr{L}_{\mathsf{CS}}^{\left(  2n+1\right)
}\left(  \mathcal{\bar{A}}\right)  -\kappa \dd \mathscr{B}^{(2n)}(\mathcal{A},\mathcal{\bar{A}}) \ , \label{ch2tranlag}%
\end{equation}
where $\mathscr{B}=Q_{\mathcal{A\leftarrow\bar
{A}\leftarrow}0}^{\left(  2n\right)  }$. 
It is important to remark that even when the Chern--Simons Lagrangian is only
locally defined, the transgression Lagrangian still is globally defined. This
is due to the presence of the boundary term $ \dd \mathscr{B}$, which plays the role of a
regularizing term, guaranteeing the full invariance of the Lagrangian under
gauge transformations.

\newsection{Transgression gauge field theory}

The most general transgression field theory can be constructed considering
$\mathcal{A}$ and $\mathcal{\bar{A}}$ as independent dynamic fields. In this
case, using eq.$\left(  \ref{ch2tranlag}\right)  $, the transgression action is given by%
\begin{align}
\mathsf{S}_{\mathsf{T}}^{\left(  2n+1\right)  }\left[  \mathcal{A}%
,\mathcal{\bar{A}}\right]   &  =\int_{\mathcal{M}}\mathscr{L}_{\mathsf{T}%
}^{\left(  2n+1\right)  }\left(  \mathcal{A},\mathcal{\bar{A}}\right)
\nonumber\\
&  =\int_{\mathcal{M}}\mathscr{L}_{\mathsf{CS}}^{\left(  2n+1\right)  }\left(
\mathcal{A}\right)  -\int_{\mathcal{M}}\mathscr{L}_{\mathsf{CS}}^{\left(
2n+1\right)  }\left(  \mathcal{\bar{A}}\right)  -\kappa \dd \int_{\mathcal{M}%
}\mathscr{B}^{(2n)}(\mathcal{A},\mathcal{\bar{A}}) \ ,\label{ch2tranact}%
\end{align}
which describes  a gauge field theory composed by two
auto-interacting gauge fields $\mathcal{A}$ and $\mathcal{\bar{A}}$ at the
bulk of $\mathcal{M}$, plus a mutual interaction at the boundary
$\partial\mathcal{M}$.

There are two independent set of symmetries lurking in the transgression
action eq.$\left(  \ref{ch2tranact}\right)  $. The first one is a built-in
symmetry, guaranteed from the outset by our use of differential forms
throughout, namely diffeomorphism invariance%
\begin{align}
\delta_{\text{{\scriptsize diff}}}\mathcal{A}  &  =-\pounds _{\xi}%
\mathcal{A},\label{ch2diffa}\\
\delta_{\text{{\scriptsize diff}}}\mathcal{\bar{A}}  &  =-\pounds _{\xi
}\mathcal{\bar{A}}, 
\end{align}
where $\pounds $ is the Lie derivative operator, and $\xi$ is a vector field
that generates the infinitesimal diffeomorphism.

The second symmetry is gauge symmetry. Under a continuous local gauge
transformation with element $g=\exp\left\{  \lambda^{A}\mathsf{T}_{A}\right\}
$, the gauge connections $\mathcal{A}$ and $\mathcal{\bar{A}}$ change
according to eq.$\left(  \ref{ch2gin1},\ref{ch2gin2}\right)  $. In an
infinitesimal form, these gauge transformations are given by%
\begin{align}
\delta_{\text{{\scriptsize gauge}}}\mathcal{A}  &  =D_{\mathcal{A}}%
\lambda,\label{ch2gaugetr}\\
\delta_{\text{{\scriptsize gauge}}}\mathcal{\bar{A}}  &  =D_{\mathcal{\bar{A}%
}}\lambda,\nonumber
\end{align}
where $\lambda=\lambda^{A}\mathsf{T}_{A}$ is a $\mathfrak{g}-$valued $0-$form
group parameter.

The variation of the transgression action eq.$\left(  \ref{ch2tranact}\right)
$ leads to the following equations of motion%
\begin{align}
\left\langle \mathsf{T}_{A}\mathcal{F}^{n}\right\rangle  &  =0,\\
\left\langle \mathsf{T}_{A}\mathcal{\bar{F}}^{n}\right\rangle  &  =0,
\end{align}
subject to the boundary conditions%
\begin{equation}
\left.  \int_{0}^{1}dt\left\langle \delta\mathcal{A}_{t}\wedge\Theta
\wedge\mathcal{F}_{t}^{n-1}\right\rangle \right\vert _{\partial\mathcal{M}%
}=0.
\end{equation}

For interesting applications of transgression fields theories in the context of supergravity and more recently in the context of hydrodynamics, see~\cite{Mora:2004kb,Izaurieta:2005vp,Haehl:2013hoa}

\newsection{Chern--Simons theory}

Using Chern--Simons forms as Lagrangians seems to be the most direct way to
avoid using two gauge connections. In fact, the Chern--Simons Lagrangian
eq.$\left(  \ref{ch2csl}\right)  $ depends only on a single dynamical field
$\mathcal{A}$. However, the simplification of imposing $\mathcal{\bar{A}}=0$
in the transgression Lagrangian eq.$\left(  \ref{ch2transl}\right)  $ gives
rise to nontrivial problems. As it was mentioned in section \ref{chssect},
Chern--Simons forms are only locally defined over the base space $\mathcal{M}%
$. This means that, if $\left\{  U_{\alpha}\right\}  $ is an open overing of a
$\left(  2n+1\right)  -$dimensional base space $\mathcal{M}$, on the nonempty
intersection $U_{\alpha}\cap U_{\beta}$ it must occur that%
\begin{equation}
Q_{\alpha}^{\left(  2n+1\right)  }\left(  \mathcal{A}_{\alpha}\right)  \neq
Q_{\beta}^{\left(  2n+1\right)  }\left(  \mathcal{A}_{\beta}\right)  .
\end{equation}
This nonlocality property of the Chern--Simons form is a source of ambiguity
in the definition of the action. The reason is basically that integrating
$\left.  \mathscr{L}_{\mathsf{CS}}^{\left(  2n+1\right)  }\right\vert
_{U_{\alpha}}=\kappa Q_{\alpha}^{\left(  2n+1\right)  }$ and $\left.
\mathscr{L}_{\mathsf{CS}}^{\left(  2n+1\right)  }\right\vert _{U_{\beta}%
}=\kappa Q_{\beta}^{\left(  2n+1\right)  }$ on the intersection $U_{\alpha
}\cap U_{\beta}$, leads in general to different answers and hence is
impossible to define the action in a unique manner.

However, in the case of Chern--Simons forms it is possible to partially
circumvent this problem. An easy way to see this is by considering eq.$(\ref{csclo}) $. Since $\left\langle \mathcal{F}^{n+1}\right\rangle $ is
gauge invariant, it follows that under a gauge transformation%
\begin{equation}
\delta_{\text{{\scriptsize gauge}}}\left\langle \mathcal{F}^{n+1}\right\rangle
=\dd \left(  \delta_{\text{{\scriptsize gauge}}}Q^{\left(  2n+1\right)  }\right)
=0.
\end{equation}
Thus, the Chern--Simons form is gauge invariant up to a closed form.
This implies that a Chern--Simons action can be only locally defined  modulo
boundary terms%
\begin{equation}
\mathsf{S}_{\mathsf{CS}}^{\left(  2n+1\right)  }\left[  \mathcal{A}\right]
=\int_{\mathcal{M}}\mathscr{L}_{\mathsf{CS}}^{\left(  2n+1\right)  }\left(
\mathcal{A}\right)  +\int_{\partial\mathcal{M}}X^{\left(  2n\right)  }.
\label{ch2csact}%
\end{equation}

The variation of the Chern--Simons action eq.$\left(  \ref{ch2csact}\right)  $
gives the equations of motion associated to the connection $\mathcal{A}$%
\begin{equation}
\left.  \left\langle \mathsf{T}_{A}\mathcal{F}^{n}\right\rangle \right\vert
_{\mathcal{M}}=0,
\end{equation}
but in principle is not possible to define boundary conditions starting
form the action principle. This has serious implications such as the
non uniqueness of the Noether currents and conserved charges. For this
reason, is more interesting to work with transgression forms. In contrast
to Chern--Simons forms, any action principle constructed using transgression as Lagrangians is uniquely defined
and conduce to finite conserved charges via Noether theorem. See~\cite{Izaurieta:2006ia,Rodriguez:2006xx} for a more detailed discussion.

The Chern--Simons action eq.$(\ref{ch2csact})$ is still diffeomorphism invariant and, as it was mentioned, it remains unchanged under
gauge transformations eq.$\left(  \ref{ch2gaugetr}\right) $ modulo boundary terms.

It is important to mention that the physical properties of a Chern--Simons
theory are strongly related with the choice of the Lie algebra $\mathfrak{g}$. In
principle, the Lie algebra may be composed of different subspaces which will
lead to different interactions in the resulting gauge theory; each of them with a different
physical interpretation. For instance, a particular sector in the Lie algebra
may correspond to gravitational theories and some other to the presence of
fermionic fields. Furthermore, depending on the form of
the invariant tensor, the Chern--Simons Lagrangian splits into pieces defined on the bulk of $\mathcal{M}$,
and some other pieces defined at the boundary $\partial\mathcal{M}$.

It is desirable then to use a systematic procedure in which the Lagrangian naturally splits accordingly with the respective subspaces associated to the Lie algebra, as well as to separate the bulk and the boundary contributions. The advantage in
doing so is that the resulting action principle can always be written in
terms of pieces which reflect the physical features of the theory. This can be
done by applying the \textrm{ECHF} recursively as we will be explained next.

\subsection{Subspace separation method} \label{SSMM}

The key observation behind the subspace separation method comes from the
triangle equation
\begin{equation}
 T_{\mathcal{A}_{2}\leftarrow\mathcal{A}_{1}}^{\left(  2n+1\right)
}=T_{\mathcal{A}_{2}\leftarrow\mathcal{A}_{0}}^{\left(  2n+1\right)
}-T_{\mathcal{A}_{1}\leftarrow\mathcal{A}_{0}}^{\left(  2n+1\right)
}-\dd Q_{\mathcal{A}_{2}\leftarrow\mathcal{A}_{1}\leftarrow\mathcal{A}_{0}%
}^{\left(  2n\right)  }. \label{ch2treq}%
\end{equation}

The transgression form $T_{\mathcal{A}_{2}\leftarrow\mathcal{A}_{1}}^{\left(
2n+1\right)  }$, which interpolates between the connections $\mathcal{A}_{2}$
and $\mathcal{A}_{1}$ can be decomposed as the sum of two transgressions;
$T_{\mathcal{A}_{2}\leftarrow\mathcal{A}_{0}}^{\left(  2n+1\right)  }$ which
interpolates between $\mathcal{A}_{2}$, $\mathcal{A}_{0}$, and $T_{\mathcal{A}%
_{1}\leftarrow\mathcal{A}_{0}}^{\left(  2n+1\right)  }$ interpolating between
$\mathcal{A}_{1}$, $\mathcal{A}_{0}$, plus an exact form which depends on
$\mathcal{A}_{2}$, $\mathcal{A}_{1}$ and $\mathcal{A}_{0}$. The
triangle equation allows to divide a transgression (or Chern--Simons
alternatively) Lagrangian by using an intermediary connection $\mathcal{A}_{1}$.
In this way, the iterative use of the triangle equation enable us to perform the
separation of the Lagrangian in terms of the subspaces early mentioned. For further details about the subspace separation method and its applications see \cite{Izaurieta:2005vp,Izaurieta:2006wv}.

Let $\mathfrak{g}$ be a Lie algebra, $\mathcal{A}$ and $\mathcal{\bar{A}}$ two
Lie valued connections and $\mathscr{L}_{\mathsf{T}}^{\left(  2n+1\right)
}\left(  \mathcal{A},\mathcal{\bar{A}}\right)  =\kappa T_{\mathcal{A}%
\leftarrow\mathcal{\bar{A}}}^{\left(  2n+1\right)  }$. The subspace separation
of the corresponding transgression Lagrangian (or Chern--Simons in the case
$\mathcal{\bar{A}}=0$), is obtained by applying the following procedure

\begin{enumerate}
\item Split the Lie algebra $\mathfrak{g}$ in terms of subspaces  $\mathfrak{g}=V_{0}\oplus\ldots\oplus V_{p}$.

\item Write the connections in pieces valued in each subspace $\mathcal{A}%
=\mathrm{a}_{0}+\ldots+\mathrm{a}_{p}$ and $\mathcal{\bar{A}}=\mathrm{\bar{a}%
}_{0}+\ldots+\mathrm{\bar{a}}_{p}$.

\item Use the triangle equation $\left(  \ref{ch2treq}\right)  $ for writing
$\mathscr{L}_{\mathsf{T}}^{\left(  2n+1\right)  }\left(  \mathcal{A}%
,\mathcal{\bar{A}}\right)  $ (or $\mathscr{L}_{\mathsf{CS}}^{\left(
2n+1\right)  }\left(  \mathcal{A}\right)  $) with%
\begin{align}
\mathcal{A}_{0}  &  =\mathcal{\bar{A}} \ ,\\
\mathcal{A}_{1}  &  =\mathrm{a}_{0}+\ldots+\mathrm{a}_{p-1} \ ,\\
\mathcal{A}_{2}  &  =\mathcal{A} \ .%
\end{align}

\item Iterate step $3$ for the transgression $T_{\mathcal{A}_{1}%
\leftarrow\mathcal{A}_{0}}^{\left(  2n+1\right)  }$ and so forth.
\end{enumerate}

As an example, the derivation of the Chern--Simons
gravity Lagrangian in arbitrary odd-dimension is presented.

\subsection{An example: Chern--Simons gravity}

In this section, we describe general aspects of Chern--Simons
theory of gravity in arbitrary odd dimensions~\cite{Zanelli:2005sa}. In order to do so, we use the tools developed in the previous
sections. Even though, we will come back to Chern--Simons gravity in the next chapter where the classification of topological theories of gravity in arbitrary dimensions is discussed.

Let $\mathcal{M}$ be a $(2n+1)-$dimensional manifold. We consider now the anti--de Sitter algebra in $(2n+1)-$dimensions $\mathfrak{g}=\mathfrak{so}(2n,2)$ with commutation relations 
\begin{align}
\left[  \mathsf{J}_{ab},\mathsf{J}_{cd}\right] &  =\eta_{ac}\, \mathsf{J} \label{adscommrel1}%
_{bd}+\eta_{bd}\, \mathsf{J}_{ca}-\eta_{bc}\, \mathsf{J}_{ad}-\eta_{ad}\,
\mathsf{J}_{bc} \ ,  \\[4pt]
\left[  \mathsf{J}_{ab},\mathsf{P}_{c}\right] &  =\eta_{ac}\, \mathsf{P}%
_{b}-\eta_{bc}\, \mathsf{P}_{a} \ ,  \\[4pt]
\left[  \mathsf{P}_{a},\mathsf{P}_{b}\right] & =\mathsf{J}_{ab} \ . \label{adscommrel2}
\end{align}
Here $\left\{  \mathsf{J}_{ab}\right\}  _{a,b=1}^{2n+1}$ generate
the Lorentz subalgebra $\mathfrak{so}(2n,1)$, $\left\{  \mathsf{P}_{a}\right\}  _{a=1}%
^{2n+1}$ generate the symmetric coset $\mathfrak{so}(2n,2)/\mathfrak{so}(2n,1)$ and 
$(\eta_{ab}) ={\rm diag}\left(  -1,1, \ldots ,1\right)  $ is a $(2n+1)$-dimensional Minkowski metric.

The anti--de Sitter algebra can be naturally decomposed 
into two vector subspaces $\mathfrak{so}\left(  2n,2\right)  =\mathsf{V}_{0}\oplus \mathsf{V}_{1}$ where
$\mathsf{V}_{0}={\mathrm{ Span}}_{\mathbb{C}}\left\{  \mathsf{J}_{ab}\right\}
$ and $\mathsf{V}_{1}={\mathrm{ Span}}_{\mathbb{C}} \left\{  \mathsf{P}_{a}\right\}  $. Since we are interested in pure Chern--Simons theory, the Lie valued one--form connections can be defined as follows
\begin{align}
\mathcal{\bar{A}} &=0 \ , \\
\mathcal{A}&= \omega + e \ ,
\end{align} 
with
\begin{align}
\omega&=\dfrac{1}{2} \omega^{ab}\mathsf{J}_{ab}, \\
e&= \frac{1}{l}e^a \mathsf{P}_{a} \ .
\end{align} 
The field $\omega$ is called the \textit{spin connection} since it transforms as a connection under gauge transformations valued in the Lorentz subgroup $SO(2n,1)$. The gauge field $e$ is the \textit{vielbeine} and transforms as a vector under the Lorentz gauge transformations. The constant $l$ has length units and is called the radius of curvature of anti--de Sitter space.

The invariant tensor associated to $\mathfrak{so}\left(  2n,2\right)$ has only one nonvanishing component given by
\begin{equation}
\left\langle \mathsf{J}_{a_{1}a_{2}}\cdots \mathsf{J}_{a_{2n-1}a_{2n}} \,
\mathsf{P}_{a_{2n+1}}\right\rangle =\frac{2^{n}}{n+1} \,
\epsilon_{a_{1}\cdots a_{2n+1}} \ . \label{ch2invten}
\end{equation}

In order to construct the Chern--Simons Lagrangian for the $SO(2n,2)$ group, we apply the subspace separation method. First, we use the vector subspace decomposition for the gauge algebra $\mathsf{V}_{0}={\mathrm{ Span}}_{\mathbb{C}}\left\{  \mathsf{J}_{ab}\right\}
$ and $\mathsf{V}_{1}={\mathrm{ Span}}_{\mathbb{C}} \left\{  \mathsf{P}_{a}\right\}  $. Next, we split the  gauge potential $\mathcal{A}$ into pieces valued in each subspace of the gauge algebra
\begin{align}
\mathcal{A}_{0} & =0 ,\\  
\mathcal{A}_{1}  &=\omega, \\
\mathcal{A}_{2}  &=\omega+e.
\end{align}
Finally, we write the triangle equation $(\ref{pr34})$
\begin{equation}
\mathscr{L}^{(2n+1)}_{\mathsf{CS}}(\omega,e)=
\kappa T^{(2n+1)}_{(\omega+e)\leftarrow \omega} + 
\kappa T^{(2n+1)}_{\omega\leftarrow 0} +
\kappa \dd Q^{(2n)}_{(\omega+e)\leftarrow \omega\leftarrow 0}.
\end{equation}

Given the invariant tensor eq.$(\ref{ch2invten})$, it is direct to show that $T^{(2n+1)}_{\omega\leftarrow 0}=0$. Moreover, the explicit form of $\kappa T^{(2n+1)}_{(\omega+e)\leftarrow \omega} $ is given by
\begin{equation}
\kappa T^{(2n+1)}_{(\omega+e)\leftarrow \omega}
=(n+1)\kappa\int_{0}^{1} \dd t \langle e \wedge \mathcal{F}^{n}_{t} \rangle \label{adsaction}
\end{equation}
with 
\begin{equation}
\mathcal{F}_{t}=\mathcal{R}+t^2e \wedge e+t\mathcal{T} \ ,
\end{equation}
and
\begin{align}
\mathcal{R} &= \frac{1}{2}R^{ab} \mathsf{J}_{ab} \ , \\
\mathcal{T} &= \frac{1}{l}T^a\mathsf{P}_{a} \ , \\
e \wedge e  &=\frac{1}{l^2}e^a \wedge e^b \mathsf{J}_{ab} \ .
\end{align}
The quantities $R^{ab}=\dd\omega^{ab}+\omega^{a}_{~c}\wedge \omega^{cb}$ and $T^a=\dd e^a+\omega^{a}_{~b}\wedge e^b$ correspond to the two--forms Lorentz curvature and Torsion respectively.
 
 For the expression $Q^{(2n)}_{(\omega+e)\leftarrow \omega\leftarrow 0}$, one has 
\begin{equation}
Q^{(2n)}_{(\omega+e)\leftarrow \omega\leftarrow 0}
=-n(n+1) \int_{0}^{1} \dd t \int_{0}^{t} \dd s \langle \mathcal{F}^{n-1}_{st} \wedge \omega \wedge e \rangle \ ,
\end{equation}
with 
\begin{equation}
\mathcal{F}_{st}=s\mathcal{T}_{t}+t\mathcal{R}_t+s^2e\wedge e \ .
\end{equation}
Here, we have defined
\begin{align}
\mathcal{T}_{t}  & =\frac{1}{l}T_{t}^{a}\mathsf{P}_{a}=\frac{1}{l}\left(
\dd e^{a}+t\omega_{~b}^{a}\wedge e^{b}\right)  \mathsf{P}_{a} \ ,\\
\mathcal{R}_{t}  & =\frac{1}{2}R_{t}^{ab}\mathsf{J}_{ab}=\frac{1}{2}\left(
\dd\omega^{ab}+t\omega_{~c}^{a}\wedge\omega^{cb}\right)  \mathsf{J}_{ab} \ .%
\end{align}

The Chern--Simons action for the $SO(2n,2)$ group is then
\begin{align}
&\mathsf{S}_{\mathsf{CS}}^{\left(  2n+1\right)  }   =\frac{\kappa}{l}%
\int_{\mathcal{M}}\int_{0}^{1}\dd t \, \epsilon_{a_{1}\ldots a_{2n+1}}\check{R}%
_{t}^{a_{1}a_{2}}\wedge\ldots\wedge\check{R}_{t}^{a_{2n-1}a_{2n}}\wedge
e^{a_{2n+1}}\\
& +\frac{n\kappa}{l}\dd\int_{\mathcal{M}}\int_{0}^{1}\dd t\int_{0}^{t}%
\dd s\,t^{n-1}\epsilon_{a_{1}\ldots a_{2n+1}}\check{R}_{st}^{a_{1}a_{2}}\wedge
\ldots\wedge\check{R}_{st}^{a_{2n-3}a_{2n-2}}\wedge\omega^{a_{2n-1}a_{2n}}\wedge
e^{a_{2n+1}} . \label{adscs} 
\end{align}
where we have used the abbreviations
\begin{equation}
\begin{tabular}
[c]{ll}%
$\check{R}_{t}^{ab}=\left(  R^{ab}+\frac{t^{2}}{l^{2}}e^{a}\wedge
e^{b}\right)  ;$ & $\check{R}_{st}^{ab}=\left(  R_{t}^{ab}+\frac{s^{2}}{l^{2}%
}e^{a}\wedge e^{b}\right)  $%
\end{tabular}  \label{chekr}
\end{equation}

Note that because of the form of the invariant tensor, the torsion does not appear explicitly in eq.$(\ref{adscs})$. Since the Chern--Simons Lagrangian it is not globally defined, its corresponding action is gauge invariant only modulo boundary terms. In fact, the Chern--Simons action defined on the bulk of $\mathcal{M}$ remains unchanged up to boundary contributions under the following transformations
\begin{align}
\delta e^{a}  & =-\lambda_{~b}^{a}e^{b}+D\lambda^{a},\\
\delta\omega^{ab}  & =D\lambda^{ab}+e^{a}\lambda^{b}-e^{b}\lambda^{a}.%
\end{align} 

Finally, The variation of the action eq.$(\ref{adscs})$, up to boundary terms, leads to the following equations of motion
\begin{align}
\mathcal{R}_{abc}T^{c}  & =0 \ ,\\
\mathcal{R}_{abc}\left(  R^{ab}+\frac{1}{l^{2}}e^{a}\wedge e^{b}\right)    & =0 \ ,
\end{align}
where%
\begin{equation}
\mathcal{R}_{abc}\equiv\epsilon_{abca_{1}...a_{2n-2}}\left(  R^{a_{1}a_{2}}+\frac
{1}{l^{2}}e^{a_{1}}\wedge e^{a_{2}}\right)  \wedge\ldots\wedge\left(
R^{a_{2n-3}a_{2n-2}}+\frac{1}{l^{2}}e^{a_{2n-3}}\wedge e^{a_{2n-2}}\right)  .
\end{equation}

\newsection{Gauged Wess--Zumino--Witten term}

We have learned that the most general transgression field theory eq.$\left(
\ref{ch2tranact}\right) $ is globally defined. The price to pay was the
inclusion of two connections $\mathcal{A}$ and $\mathcal{\bar{A}}$ as
independent dynamic fields. This certainly opens the question about the
physical meaning of the field $\mathcal{\bar{A}}$. We also have shown that
further simplifications, like imposing $\mathcal{\bar{A}}=0$ in the
transgression action, conduces to Chern--Simons theories eq.$\left(
\ref{ch2csact}\right)  $ which are not globally defined, and therefore they are gauge
invariant modulo boundary terms.

There is a third and no less interesting possibility: that is to regard
$\mathcal{A}$ and $\mathcal{\bar{A}}$ as two connections defining the same
nontrivial principal bundle $\mathcal{P}$. This means that in non empty overlaps $U_{\alpha}\cap U_{\beta}\neq\varnothing$, they are related by
a gauge transformation.

If a nontrivial bundle is considered, the transgression action eq.$\left(
\ref{ch2tranact}\right)  $ is not globally defined. However, the action can
be treated formally provided more than one chart is used. This justifies the
introduction of two gauge connections defined in different charts, such that in
the overlap of two charts $U_{\alpha}\cap U_{\beta}$, the connections are
related by
\begin{equation}
\mathcal{\bar{A}}=g^{-1}\mathcal{A}g+g^{-1}dg\equiv\mathcal{A}^{g} \ ,%
\end{equation}
where $g=\tau_{\alpha\beta}\left(  x\right) \in \mathcal{G} $ is a transition function which
determines the nontriviality of the principal bundle $\mathcal{P}$.

To illustrate how the transgression action eq.$\left(  \ref{ch2tranact}\right)  $
changes under $\mathcal{\bar{A}}\longrightarrow\mathcal{A}^{g}$ we first
consider
\begin{align}
\mathscr{L}_{\mathsf{CS}}^{\left(  2n+1\right)  }\left(  \mathcal{\bar{A}%
}\right)   &  =\mathscr{L}_{\mathsf{CS}}^{\left(  2n+1\right)  }\left(
\mathcal{A}^{g}\right)   \nonumber \\
&  =\kappa T_{\mathcal{A}^{g}\leftarrow0}^{\left(  2n+1\right)  }.
\end{align}
We can use the subspace separation method to split the transgression
$T_{\mathcal{A}^{g}\leftarrow0}^{\left(  2n+1\right)  }$ as follows%
\begin{equation}
T_{\mathcal{A}^{g}\leftarrow0}^{\left(  2n+1\right)  }=T_{\mathcal{A}%
^{g}\leftarrow g^{-1}\dd g}^{\left(  2n+1\right)  }+T_{g^{-1}\dd g\leftarrow
0}^{\left(  2n+1\right)  }+\dd Q_{\mathcal{A}^{g}\leftarrow g^{-1}\dd g\leftarrow
0}^{\left(  2n\right)  } \ . \label{ch2wzw}%
\end{equation}
The right hand side of eq.$\left(  \ref{ch2wzw}\right)  $ can be evaluated
term by term. In fact, the first term corresponds to the Chern--Simons form
$T_{\mathcal{A}\leftarrow0}^{\left(  2n+1\right)  }$, where%
\begin{equation}
T_{\mathcal{A}^{g}\leftarrow g^{-1}\dd g}^{\left(  2n+1\right)  }=\left(
n+1\right)  \int_{0}^{1}\dd t\left\langle g^{-1}\mathcal{A}g\wedge\mathcal{F} %
_{t}^{n}\right\rangle \ ,
\end{equation}
and $\mathcal{F}_{t}$ is the curvature for the 
\begin{equation}
\mathcal{A}_{t}=g^{-1}\dd g+tg\mathcal{A}g^{-1} \ .
\end{equation}
Using eq.$\left(  \ref{pr2}\right)  $ it is direct to show that%
\begin{equation}
\mathcal{F}_{t}=g^{-1}\left(  t\dd\mathcal{A}+t^{2}\mathcal{A}\wedge
\mathcal{A}\right)  g \ ,
\end{equation}
and then%
\begin{equation}
T_{\mathcal{A}^{g}\leftarrow g^{-1}\dd g}^{\left(  2n+1\right)  }=T_{\mathcal{A}%
\leftarrow0}^{\left(  2n+1\right)  } \ .%
\end{equation}

Let us now to evaluate the second term in eq.$\left(  \ref{ch2wzw}\right)  $.
Again, using eq.$\left( \ref{pr2}\right )  $ we obtain
\begin{align}
T_{g^{-1}\dd g\leftarrow0}^{\left(  2n+1\right)  }  &  =\left(  n+1\right)
\int_{0}^{1}\dd t\left\langle \left(  g^{-1}\dd g\right)  \wedge\left(  t\dd\left(
g^{-1}\dd g\right)  +t^{2}\left(  g^{-1}\dd g\right)  ^{2}\right)  ^{n}\right\rangle
\nonumber\\
&  =\left(  n+1\right)  \left(  -1\right)  ^{n}\int_{0}^{1}\dd tt^{n}\left(
1-t\right)  ^{n}\left\langle \left(  g^{-1}\dd g\right)  \wedge\left(
g^{-1}\dd g\right)  ^{2n}\right\rangle  \nonumber\\
&  =\left(  -1\right)  ^{n}\frac{n!\left(  n+1\right)  !}{\left(  2n+1\right)
!}\left\langle \left(  g^{-1}\dd g\right)  \wedge\left(  g^{-1}\dd g\right)
^{2n}\right\rangle \ ,
\end{align}
where we have used%
\begin{equation}
\int_{0}^{1}\dd tt^{n}\left(  1-t\right)  ^{n}=\frac{\left(  n!\right)  ^{2}%
}{\left(  2n+1\right)  !} \ .
\end{equation}
This term is the so called Wess--Zumino term and corresponds to a closed,
however not exact, $\left(  2n+1\right) -$form $\dd T_{g^{-1}\dd g\leftarrow
0}^{\left(  2n+1\right)  }=0$. Since the Wess--Zumino term represents a
winding number, it will be a total derivative unless $\mathcal{G}$ has
nontrivial homotopy group $\pi(\mathcal{G})$ and a large gauge transformation
is performed.

The remaining term in eq.$\left(  \ref{ch2wzw}\right)  $ corresponds to%
\begin{equation}
Q_{\mathcal{A}^{g}\leftarrow g^{-1}\dd g\leftarrow0}^{\left(  2n\right)
}=-n\left(  n+1\right)  \int_{0}^{1}\dd t\int_{0}^{t}\dd s\left\langle
g^{-1}\dd g\wedge\mathcal{A}\wedge\mathcal{F}_{st}^{n-1}\right\rangle ,
\end{equation}
where $\mathcal{F}_{st}=\dd\mathcal{A}_{st}+\mathcal{A}_{st}\wedge
\mathcal{A}_{st}$, with
\begin{equation}
\mathcal{A}_{st}=tg^{-1}\dd g+sg^{-1}\mathcal{A}g \ .
\end{equation}

In this case, it is direct to show that
\begin{equation}
\mathcal{F}_{st}=g^{-1}\left[  s\left(  \dd\mathcal{A}+s\mathcal{A}%
\wedge\mathcal{A}\right)  +t\left(  1-t\right)  \left(  \dd gg^{-1}\right)
^{2}\right]  g,
\end{equation}
and therefore,%
\begin{align}
Q_{\mathcal{A}^{g}\leftarrow g^{-1}\dd g\leftarrow0}^{\left(  2n\right)  }  &
=-n\left(  n+1\right)  \int_{0}^{1}\dd t\int_{0}^{t}\dd s\times\\
& \times\left\langle g^{-1}\dd g\wedge\mathcal{A}\wedge\left[  s\left(
\dd\mathcal{A}+s\mathcal{A}\wedge\mathcal{A}\right)  +t\left(  1-t\right)
\left(  g^{-1}\dd g\right)  ^{2}\right]  ^{n-1}\right\rangle .
\end{align}

Finally, the Chern--Simons Lagrangian $\mathscr{L}_{\mathsf{CS}}^{\left(
2n+1\right)  }\left(  \mathcal{A}^{g}\right)  $, is given by
\begin{align}
\mathscr{L}_{\mathsf{CS}}^{\left(  2n+1\right)  }\left(  \mathcal{A}%
^{g}\right)   &  =\mathscr{L}_{\mathsf{CS}}^{\left(  2n+1\right)  }\left(
\mathcal{A}\right)  +\left(  -1\right)  ^{n}\frac{n!\left(  n+1\right)
!}{\left(  2n+1\right)  !}\left\langle \left(  g^{-1}\dd g\right)  \wedge\left(
g^{-1}\dd g\right)  ^{2n}\right\rangle \nonumber \\
&  -n\left(  n+1\right)  \dd \int_{0}^{1}\dd t\int_{0}^{t}\dd s\times \nonumber \\
&  \times\left\langle g^{-1}\dd g\wedge\mathcal{A}\wedge\left[  s\left(
\dd\mathcal{A}+s\mathcal{A}\wedge\mathcal{A}\right)  +t\left(  1-t\right)
\left(  g^{-1}\dd g\right)  ^{2}\right]  ^{n-1}\right\rangle .\label{ch2csgauge}
\end{align}
Thus, the Chern--Simons Lagrangian changes by a closed form under gauge transformations.

We turn now to the transgression Lagrangian eq.$\left(  \ref{ch2tranact}%
\right)  $. In fact inserting eq.$\left(  \ref{ch2csgauge}\right)  $ the
transgression becomes into a gauged Wess--Zumino--Witten Lagrangian%
\begin{align}
\mathscr{L}_{\mathsf{T}}^{\left(  2n+1\right)  }\left(  \mathcal{A}%
,\mathcal{A}^{g}\right)    & =\mathscr{L}_{\mathsf{gWZW}}^{\left(  2n+1\right)
}\left(  \mathcal{A},\mathcal{A}^{g}\right)   \\
& =\left(  -1\right)  ^{n+1}\frac{n!\left(  n+1\right)  !}{\left(
2n+1\right)  !}\left\langle \left(  g^{-1}\dd g\right)  \wedge\left(
g^{-1}\dd g\right)  ^{2n}\right\rangle +\dd\left(  \mathcal{C}^{\left(  2n\right)
}-\mathcal{B}^{\left(  2n\right)  }\right), \nonumber
\end{align}
where 
\begin{align}
\mathcal{C}^{\left(  2n\right)  }& \equiv \kappa Q_{\mathcal{A}^{g}\leftarrow g^{-1}\dd g\leftarrow
0}^{\left(  2n\right)  } \ , \\
\mathcal{B}^{\left(  2n\right)  }&\equiv\kappa Q_{\mathcal{A}\leftarrow
\mathcal{A}^{g}\leftarrow0}^{\left(  2n\right)  } \ .
\end{align}
This Lagrangian has interesting properties. For instance, the $n=1$ case
conduces to a gauged Wess--Zumino--Witten action given by%
\begin{equation}
\mathsf{S}_{\mathsf{gWZW}}\left[  \mathcal{A},g\right]  =\kappa\int
_{\mathcal{M}}\frac{1}{3}\left\langle \left(  g^{-1}\dd g\right)  ^{3}%
\right\rangle +\dd\left\langle 2\left(  g^{-1}\dd g\right)  \wedge\mathcal{A+A}%
g^{-1}\wedge\mathcal{A}g\right\rangle .
\end{equation}

As it will be shown later, if we consider the action valued in the Poincare group $ISO(2,1)$, this model
collapses to a boundary term. In fact, we will show that the
resulting gauged Wess--Zumino--Witten term can be used as a Lagrangian for an action principle in two dimensions. In that case, the resulting model is the simplest
version for topological actions for gravity in even dimensions classified in~\cite{Cha90}. 
In the next chapter, we will construct the early mentioned gauged Wess--Zumino--Witten theory in two dimensions and its generalization to arbitrary even dimensions. We will show that the supersymmetric extension leads to topological supergravity
in two dimensions starting from a transgression field theory which is invariant under the
supersymmetric extension of the Poincare group in three dimensions. We also apply this
construction to a three-dimensional Chern--Simons theory of gravity genuinely invariant under
the Maxwell algebra and obtain the corresponding gauged Wess--Zumino--Witten model. For similar approaches in three dimensions see \cite{Ogura198961,Carlip:1991zm}.
\chapter{Transgression field theory and topological gravity actions}
\label{ch:top_grav}
\begin{flushright}
\textit{``...presiento que por lo emp\'irico se ha enloquecido la br\'ujula...''.  \\ C\'andidos, Inti-Illimani.}
\footnote{\scriptsize ``...I sense that empiricism has made ​​foolish the compass''. C\'andidos, Inti-Illimani.}
\bigskip
\end{flushright}

The problem of unifying gravity to the other fundamental interactions still remains as an open problem. A particularly interesting direction to look at this problem is from the field theory point of view. This approach aims to find a unified theory with gauge symmetries incorporating gravity which also should be well behave at the quantum level. The nearest in spirit to this approach are supergravity models~\cite{Nilles:1983ge}, especially the ones which are extensions of the standard Weinberg--Salam model. Those, however are not renormalizable.

In the past, many attempts were made to construct gravity as a gauge theory of the Lorentz or Poincar\'e group in four dimensions~\cite{Utiyama:1956sy}. It later became clear that if both the vierbein and the spin connection are to be viewed as gauge fields, the Einstein--Hilbert action is then only invariant under a constrained symmetry in which the torsion is set to zero~\cite{Chamseddine:1976bf,MacDowell:1977jt}. In that case the spin connection transformation law must be modified nonlinearly to satisfy the constraints and the gauge system cannot be considered as a standard gauge theory. Therefore, this understanding was only useful in constructing geometric actions but cannot not be used to unify gravity with the other known interactions within the framework of a renormalizable gauge theory.
Later on, at the end of the eighties, it was shown that three-dimensional gravity can be written not only as renormalizable but finite gauge theory~\cite{Witten:1988hc,Witten:1989sx}. This relies on the curious fact that three-dimensional Einstein gravity corresponds to a Chern--Simons action for the gauge group $ISO(2,1)$. This gave new momentum to explore new features about topological gauge theories of gravity. 

The classification of topological gauge theories for gravity and its supersymmetric extensions were introduced by A. H. Chamseddine at the beginning of the decade of the nineties~\cite{Chamseddine:1989yz,Cha89,Cha90}. The natural gauge groups $\mathcal{G}$ considered are the anti-de~Sitter
group $SO(d-1,2)$, the de~Sitter group $SO(d,1)$, and the Poincar\'e
group $ISO(d-1,1)$ in $d$ spacetime dimensions depending on the sign of the
cosmological constant: $-1,+1,0$ respectively. In odd dimensions
$d=2n+1$, the gravitational theories are constructed in terms of Chern--Simons forms. As we have learned from previous chapters, Chern--Simons forms are useful objects because they
lead to gauge invariant theories modulo boundary terms. They also
have a rich mathematical structure similar to those of the characteristic classes that arise in Yang--Mills theories; they are constructed in terms of a gauge potential which descends from a connection on a principal $\mathcal{G}$-bundle.
 In even dimensions, there is no
natural candidate such as the Chern--Simons forms; hence in order to construct an invariant $2n$-form, the product of $n$ field strengths is not
sufficient and requires the insertion of a scalar multiplet $\phi^{a}$ in the fundamental
representation of the gauge group $\mathcal{G}$. This requirement
ensures gauge invariance but it threatens the topological origin of
the theory. 

In this chapter we pursue of the main results of this Thesis. In section \ref{firstresult}, we show that even-dimensional topological gravity can be
formulated as a transgression field theory genuinely invariant under the Poincar\'e group~\cite{Merino:2010zz,Salgado:2013pva,PhysRevD.89.084077}. The gauge connections are
considered taking values in the Lie algebras associated to linear
and nonlinear realizations of the gauge group. The resulting theory corresponds to
a gauged Wess--Zumino--Witten model
\cite{Wess:1967jq,Witten:1983tw} where the scalar field $\phi$ is now identified with the coset parameter of the nonlinear realization of the Poincar\'e group
$ISO(d-1,1)$. By similar arguments we also
compute the transgression action for the $\mathcal{N}=1$ Poincar\'e
supergroup in three dimensions, and show that the resulting action is the one proposed in~\cite{Chamseddine:1989yz}.

\newsection{Topological gauge theories of gravity}

Topological gauge theories of gravity were classified in~\cite{Chamseddine:1989yz,Cha89,Cha90}. The natural gauge groups
$\mathcal{G}$ involved in the classification are given by (\ref{gaugegrav})
\begin{table}[h!]
\centering
$\mathcal{G}$ \quad : \quad
\begin{tabular}
[c]{||l|l|l||}\hline\hline
\textrm{AdS} & $SO(d-1,2)$ & $\Lambda<0$\\\hline
\textrm{dS} & $SO(d,1)$ & $\Lambda>0$\\\hline
\textrm{Poincar\'e} & $ISO(d-1,1)$ & $\Lambda=0$\\\hline\hline
\end{tabular}
\caption{Gauge groups}
\label{gaugegrav}
\end{table}
depending on the spacetime dimension $d$ and the sign of the cosmological constant $\Lambda$. These
gauge groups are the smallest nontrivial choices which contain the
Lorentz symmetry $SO(d-1,1)$, as well as symmetries analogous to local
translations.

In odd dimensions $d=2n+1$, the action for topological gravity is written in terms of a
Chern--Simons form defined by
\begin{align}
\mathsf{S}^{\left(  2n+1\right)  }_{\mathsf{CS}}\left[  \mathcal{A}\right]  & =\kappa \,
\int_{\mathcal{M}}\, \mathscr{L}_{\mathsf{CS}}^{\left(  2n+1\right)  }\left(
\mathcal{A}\right) \nonumber   \\ 
&=\kappa\, \left(  n+1\right) \, \int_{\mathcal{M}}
\ \int_{0}%
^{1} \, dt \ \left\langle \mathcal{A} \wedge\left(  t\,
    \dd\mathcal{A}+t^{2} \, 
\mathcal{A\wedge A}\right)  ^{n}\right\rangle \ . \label{cspoinc}%
\end{align}

In even dimensions there is no topological candidate such as the Chern--Simons form. In fact, the
exterior product of $n$ field strengths makes the required $2n$-form in a
$2n$-dimensional spacetime, but in order to obtain a gauge invariant
differential $2n$-form, a scalar multiplet $\phi^{a}$ with $a=1,\ldots
,2n+1$ transforming in the fundamental representation of the gauge group must be
added in such a way that the action can be written as
\begin{equation}
\mathrm{S}^{\left(  2n\right)  }\left[  \mathcal{A},\phi\right]  =\kappa\,
(n+1)\int_{\mathcal{M}}\, \big\langle \mathcal{F}^n \wedge
\phi\big\rangle \ . \label{maineq}%
\end{equation}
Here $\mathcal{F}=\dd\mathcal{A+A\wedge A}$ is the curvature two-form
associated to the gauge potential $\mathcal{A}$. Note that here the
Lagrangian $\big\langle \mathcal{F}^n \wedge
\phi\big\rangle$ is a global differential form on $\mathcal{M}$. This topological action has interesting
applications; for instance, it describes the Liouville theory 
of gravity from a local Lagrangian in two dimensions \cite{D'Hoker:1983ef,D'Hoker:1983is}.
 
 In the following, we focus on the Chern--Simons action for the anti--de Sitter gauge group and indicate how to recover the Poincar\'e action by In\"onu-Wigner contraction or equivalently by taking the limit $l\rightarrow \infty$ of the anti-de Sitter curvature radius.
  Recalling the invariant tensor associated to the $\mathfrak{so}(2n,2)$ algebra eq.$(\ref{ch2invten})$ and the commutation relations eq.$(\ref{adscommrel1}-\ref{adscommrel2})$, it is direct to show that the odd-dimensional topological gravity action for the AdS group takes the form (\ref{adscs}). On the other hand, the even-dimensional action takes the form
 \begin{equation}
 \mathrm{S}^{\left(  2n\right)  }\left[ e, \omega,\phi\right]  =\kappa\,(n+1) \int
 _{\mathcal{M}_{2n}}\epsilon_{a_{1} \cdots a_{2n+1}}\, \check{R}^{a_{1}a_{2}}\wedge \cdots \wedge \check{R}^{a_{2n-1}%
 a_{2n}}\, \phi^{a_{2n+1}} \label{topgads}
 \end{equation}
with $\check{R}^{ab}$ given in the left hand side of (\ref{chekr}) with $t=1$.
It is interesting to mention that the action (\ref{topgads}) can be obtained starting form its odd dimensional counterpart by dimensional reduction where the symmetry breaking to $ISO(2n-1,1)$, $SO(2n,1)$ or $SO(2n-1,2)$, is subject to suitable field truncation \cite{Cha90}.

\subsection{Lanczos--Lovelock gravity}

The most general Lagrangian in $d$ dimensions which is compatible with the
Einstein--Hilbert action for gravity is a polynomial of degree $\left[  d/2\right]  $ in
the curvatures known as the Lanczos--Lovelock Lagrangian
\cite{Lan38,Lov71,Cno05,Des05,Mar91}. Lanczos--Lovelock theories share
the same fields, symmetries and local degrees of freedom of General
Relativity. In fact the Lanczos--Lovelock Lagrangian constructed by considering a generalization of the Einstein tensor which is second order in derivatives, symmetric and divergence-free. In the Cartan formalism, the Lagrangian is built from the
vielbein $e^{a}$ and the spin connection $\omega^{ab}$ via the Riemann
curvature two-form $R^{ab}=\dd\omega^{ab}+\omega
_{~c}^{a}\wedge\omega^{cb}$, leading to the action
\begin{equation}
\mathrm{S}^{(d)}_{\mathsf{LL}} =\int_{\mathcal{M}} \ \sum\limits_{p=0}^{\left[
d/2\right]  }\, \alpha_{p}\, \epsilon_{a_{1}\cdots a_{d}}\,
R^{a_{1}a_{2}}\wedge \cdots \wedge
R^{a_{2p-1}a_{2p}}\wedge e^{a_{2p+1}}\wedge \cdots \wedge e^{a_d} \ .
\end{equation}
Here $\alpha_{p}$ are arbitrary parameters that cannot be fixed from first
principles. However, in~\cite{Tro99} it is shown that by requiring the
equations of motion to uniquely determine the dynamics for as many
components of the independent fields as possible, one can fix $\alpha_{p}$ (in any
dimension) in terms of the gravitational and cosmological constants.

In $d=2n$ dimensions the parameters $\alpha_{p}$ are given by%
\begin{equation}
\alpha_{p}=\alpha_{0}\, \left(  2\gamma\right)  ^{p}\, \binom{n}{p} \label{coeffeven}
\end{equation}
and the Lagrangian takes a Born--Infeld form. The Lanczos--Lovelock action
constructed in this dimension is only invariant under the Lorentz symmetry
$SO(2n-1,1)$.
In odd dimensions $d=2n+1$ the coefficients are given by
\begin{equation}
\alpha_{p}=\alpha_{0}\, \frac{(2n-1)\, (2\gamma)  ^{p}%
}{2n-2p-1}\, \binom{n-1}{p} \ . \label{coeffodd}
\end{equation}
Here
\begin{equation}
\alpha_{0}=\frac{\kappa}{d\, l^{d-1}} \qquad \mbox{and} \qquad \gamma=-\mathrm{sgn}(
\mathrm{\Lambda}) \, \frac{l^{2}}{2}
\end{equation}
with 
\begin{equation}
\kappa^{-1}=2(d-2)!\Omega_{d-2}G
\end{equation}
where $G$ is the gravitational constant~\cite{Crisostomo:2000bb}, and $l$ is a length parameter related to the
cosmological constant by
\begin{equation}
\mathrm{\Lambda}=\pm\, \frac{(d-1)\, (d-2)}{2l^{2}} \ .
\end{equation}

With this choice of coefficients, the
Lanczos--Lovelock Lagrangian for $d=2n+1$
coincides exactly with a Chern--Simons form for the $\mathrm{AdS}$
group $SO(2n,2)$. This means that the exterior derivative of the Lanczos--Lovelock Lagrangian corresponds to a $2n+2$-dimensional Euler density. This is the reason
why there is no analogous construction in even dimensions: There are no known
topological invariants in odd dimensions which can be constructed
in terms of exterior products of curvatures alone. However, in ref.~\cite{Sal05} a Lanczos--Lovelock theory
genuinely invariant under the $\mathrm{AdS}$ group, in any dimension, is
proposed. The construction is based on the Stelle--West mechanism \cite{Ste80,Grignani:1991nj},
which is an application of the theory of nonlinear realizations of Lie groups
to gravity. For a detailed treatment of nonlinear realization theory and it applications, See Appendix \ref{ch:app2}.

\subsection{SWGN formalism}

The Stelle--West--Grignani--Nardelli (SWGN) formalism \cite{Ste80,Grignani:1991nj}
is an application of the theory of nonlinear realizations of Lie groups~[Appendix \ref{ch:app2}] to
gravity. In particular, it allows the construction of the
Lanczos--Lovelock theory of gravity which is genuinely
invariant under the anti-de~Sitter group $G=SO(d-1,2)$. This model is discussed by the action~\cite{Sal03b}
\begin{equation}
\mathrm{S}_{\mathsf{SW}}^{\left(  d\right)  }=\int_{\mathcal{M}} \
\sum\limits_{p=0}^{\left[  d/2\right]  } \,
\alpha_{p}\, \epsilon_{a_{1} \cdots a_{d}}\, \bar{R}{}^{a_{1}a_{2}}\wedge
\cdots \wedge\bar{R}{}^{a_{2p-1}a_{2p}}\wedge\bar{e}{}^{a_{2p+1}}%
\wedge \cdots \wedge\bar{e}{}^{a_{d}} \ .\label{llact}%
\end{equation}
Here $\bar{R}{}^{ab}=\dd\bar{\omega}{}^{ab}+\bar{\omega}{}_{~c}^{a}\wedge
\bar{\omega}{}^{cb} 
$ and $\bar{e}{}^{a}$ are nonlinear gauge
fields and the coefficients $\alpha_{p}$ are given by either eq.~$(\ref{coeffeven})$ or eq.~$(\ref{coeffodd})$ depending on the dimension of the spacetime. The relation between linear and nonlinear gauge fields is obtained using eq.~$(\ref{nonl13})$. For the present case one finds that 
\begin{align}
\bar{e}{}^{a} =&\Omega_{~b}^{a}( \cosh z) \, e^{b}-\Omega
_{~b}^{a}\Big(\, \frac{\sinh z}{z}\, \Big) \, D_{\omega}\phi^{b} \ , \label{xor1}\\[4pt]
\bar{\omega}{}^{ab} =&\omega^{ab}+\frac{\sigma}{l^{2}}\, \Big(\,
\frac{\sinh z}{z}\, \big(  \phi^{a}\, e^{b}-\phi^{b}\, e^{a}\big)
-\frac{\cosh z-1}{z^{2}} \, \big( \phi^{a}\,
D_{\omega}\phi^{b}-\phi^{b}\, D_{\omega}\phi
^{a}\big) \, \Big) \ , \label{xor3}%
\end{align}
where $e^a$ and $\omega^{ab}$ are the usual vielbein and spin connection, respectively. Here
we have defined
\begin{align}
D_{\omega}\phi^a & := \dd\phi^a + \omega_{~b}^{a} \, \phi^b \ ,
\nonumber \\[4pt]
z  & := \frac{\phi}{l}=\frac{\sqrt{\phi^{a}\, \phi_{a}}}{l} \ ,
\nonumber \\[4pt]
\Omega_{~b}^{a}(  u)    & := u\, \delta_{~b}^{a}+\left(
1-u\right) \, \frac{\phi^{a}\, \phi_{b}}{\phi^{2}} \ ,
\end{align}
where $l$ is the radius of curvature of AdS and $\phi^{a}$ are the \textrm{AdS} coordinates
 which parametrize the coset space $
\frac{SO(d-1,2)}{SO(d-1,1)} $. In this scheme, this coordinate carries no
dynamics as any value that we pick for it is equivalent to a gauge
fixing condition
which breaks the symmetry from \textrm{AdS} to the Lorentz subgroup. This is
best seen using the equations of motion; they are the same as those
for the ordinary Lanczos--Lovelock theory where the vielbein $e^{a}$ and the spin
connection $\omega^{ab}$ are replaced by their nonlinear versions
$\bar{e}{}^{a}$ and $\bar{\omega}{}^{ab}$ given in eqs.~$\left(  \ref{xor1}%
,\ref{xor3}\right)  $. 

In odd
dimensions $d=2n+1$, the Chern--Simons action written in terms of the linear
gauge fields $e^a$ and $\omega^{ab}$ with values in the Lie algebra of
$SO(2n,2)$ differs only by a boundary term from that written using
the nonlinear gauge fields $\bar{e}{}^a$
and $\bar{\omega}{}^{ab}$. This is by virtue of eq.~$(\ref{nonl13})$ which has the form of a gauge transformation
\begin{equation}
\mathcal{A} \ \longmapsto \ \mathcal{\bar{A}}=g^{-1}\, \left(  \dd+\mathcal{A}%
\right)  \, g \label{gtc}
\end{equation}
with $g=\e^{-\phi^{a}\, \mathsf{P}_{a}} \in  \frac{SO(2n,2)}%
{SO(2n,1)}  $. Alternatively, since $\mathcal{\bar{F}}=g^{-1}\,
\mathcal{F}\, g$ we have
\begin{equation}
\dd \mathscr{L}_{\mathsf{CS}}^{\left(  2n+1\right)  }\left(
  \mathcal{\bar{A}
} \, \right)  =\left\langle \mathcal{\bar{F}}^{n+1}\right\rangle =\left\langle
\mathcal{F}^{n+1}\right\rangle =\dd \mathscr{L}_{\mathsf{CS}}^{\left(
2n+1\right)  }\left(  \mathcal{A}\right) \label{laglnl}
\end{equation}
and hence both Lagrangians may locally differ only by a total
derivative.

\subsection{Chern--Simons gravity invariant under the Poincar\'e group}

Poincar\'e gravity in $2n+1$ dimensions can be formulated as a Chern--Simons theory 
for the gauge group 
$ISO(2n,1)$. This group can
be obtained by performing an In\"{o}n\"{u}--Wigner contraction of
the \textrm{AdS} group in odd dimensions $SO(2n,2)$.

The fundamental field is the one-form connection 
\begin{equation}
\mathcal{A}=e^a \, \mathsf{P}_a+\mbox{$\frac{1}{2}$}\, \omega^{ab}\, \mathsf{J}_{ab}
\end{equation}
with values in the Lie algebra $\mathfrak{iso}(2n,1)$ whose
commutation relations are given by
\begin{align}
\left[  \mathsf{J}_{ab},\mathsf{J}_{cd}\right] &  =-\eta_{ac}\, \mathsf{J}%
_{bd}-\eta_{bd}\, \mathsf{J}_{ca}+\eta_{bc}\, \mathsf{J}_{ad}+\eta_{ad}\,
\mathsf{J}_{bc} \ , \nonumber \\[4pt]
\left[  \mathsf{J}_{ab},\mathsf{P}_{c}\right] &  =\eta_{bc}\, \mathsf{P}_{a}-\eta_{ac}\, \mathsf{P}%
_{b} \ , \nonumber \\[4pt]
\left[  \mathsf{P}_{a},\mathsf{P}_{b}\right] & =0 \ .
\end{align}
Here $\left\{  \mathsf{J}_{ab}\right\}  _{a,b=1}^{2n+1}$ generate
the Lorentz subalgebra $\mathfrak{so}(2n,1)$, $\left\{  \mathsf{P}_{a}\right\}  _{a=1}%
^{2n+1}$ generate local Poincar\'e translations and
$(\eta_{ab}) ={\rm diag}\left(  -1,1, \ldots ,1\right)  $ is a $(2n+1)$-dimensional Minkowski metric.

The explicit form of the action can be obtained in the limit $l \rightarrow \infty$ of the Chern--Simons gravity action for the AdS group $SO(2n,2)$ or, alternatively, by using the
subspace separation method introduced in (\ref{SSMM}). Following this option, we first we decompose the gauge algebra
into vector subspaces $\mathfrak{iso}\left(  2n,1\right)  =\mathsf{V}_{1}\oplus \mathsf{V}_{2}$ where
$\mathsf{V}_{1}={\rm Span}_{\complex}\left\{  \mathsf{J}_{ab}\right\}
$ and $\mathsf{V}_{2}={\rm Span}_{\complex}\left\{  \mathsf{P}_{a}\right\}  $. Next
we split the  gauge potential into pieces valued in each subspace of the gauge algebra
\begin{align}
\mathcal{A}_{0} & =0 \ , \\
\mathcal{A}_{1}  &=\omega \ , \\
\mathcal{A}_{2} & =\omega+e \ .
\end{align}
where $\omega=\frac{1}{2}\, \omega^{ab}\, \mathsf{J}_{ab}$ and
$e=e^{a}\, \mathsf{P}_{a}$. Computing each component of the triangle
equation of eq.~$(\ref{pr34})$ one finds
\begin{align}
T_{\mathcal{A}_{2}\leftarrow\mathcal{A}_{1}}^{\left(  2n+1\right)  }  
& =\epsilon_{a_{1}\cdots a_{2n+1}}\, R^{a_{1}a_{2}}\wedge \cdots \wedge R^{a_{2n-1}a_{2n}
}\wedge e^{a_{2n+1}} \ , \label{t1x}\\[4pt]
T_{\mathcal{A}_{1}\leftarrow \mathcal{A}_0}^{\left(  2n+1\right)  }   & =0 \
, \label{t2x}\\[4pt] 
Q_{\mathcal{A}_{2}\leftarrow\mathcal{A}_{1}\leftarrow\mathcal{A}_0}^{\left(  2n\right)  }
& =-n\, \int_{0}^{1}\, \dd t \ t^{n-1}\, \epsilon_{a_{1}\cdots
  a_{2n+1}}\, R_{t}^{a_{1}a_{2}}
\wedge \cdots \wedge R_{t}^{a_{2n-3}a_{2n-2}}\wedge\omega^{a_{2n-1}a_{2n}}\wedge
e^{a_{2n+1}}  \ . \label{t3x}
\end{align}
Here we have used the fact that the only nonvanishing components of the
invariant tensor for the Poincar\'e algebra are given by
\begin{equation}
\left\langle \mathsf{J}_{a_{1}a_{2}}\cdots \mathsf{J}_{a_{2n-1}a_{2n}} \,
\mathsf{P}_{a_{2n+1}}\right\rangle =\frac{2^{n}}{n+1} \,
\epsilon_{a_{1}\cdots a_{2n+1}} \ .
\end{equation}
Since $\mathscr{L}_{\mathsf{CS}}^{\left(
    2n+1\right)  }\left(  \mathcal{A}\right)
=T_{\mathcal{A}_{2}\leftarrow\mathcal{A}_{0}}^{\left(  2n+1\right)  }$,
we obtain
\begin{align}
\mathscr{L}_{\mathsf{CS}}&^{\left(  2n+1\right)  }\left(  \mathcal{A}\right)
=\ \epsilon_{a_{1}\cdots a_{2n+1}}\, R^{a_{1}a_{2}}\wedge \cdots \wedge R^{a_{2n-1}a_{2n}}\wedge
e^{a_{2n+1}}\nonumber\\
& -n\, \dd\int_{0}^{1}\, \dd t \ t^{n-1}\, \epsilon_{a_{1}
  \cdots a_{2n+1}}\, R_{t}^{a_{1}a_{2}}%
\wedge \cdots \wedge R_{t}^{a_{2n-3}a_{2n-2}}\wedge\omega^{a_{2n-1}a_{2n}}\wedge
e^{a_{2n+1}}\label{cslin}%
\end{align}
where $R_{t}^{ab}=\dd\omega^{ab}+t\, \omega_{~c}^{a}\wedge\omega^{cb}$.

 Note that the piece of the Lagrangian which corresponds to the volume (bulk) term of $\mathcal{M}$ in the action, can still be recovered in the limit $l\to\infty$ from the Lanczos--Lovelock
series in the case $d=2n+1$ and $p=n$. However, there is an extra
boundary term which arises once the computation of the relevant Chern--Simons action by using transgressions. This boundary contribution is very important in the construction of the even dimensional topological gravity theory as we shall see later.

Under infinitesimal local gauge transformations with parameter
$\lambda=\frac{1}{2}\, \kappa^{ab}\, \mathsf{J}_{ab}+\rho^a \, \mathsf{P}_a$, the gauge fields transform as 
\begin{align}
\delta e^{a}    &=-D_{\omega}\rho^{a}+\kappa_{~b}^{a}\, e^{b} \ , \\
\delta\omega^{ab}   &=-D_{\omega}\kappa^{ab} \ .
\end{align}
and these transformations leave eq.~($\ref{cslin}$) invariant.

We now write the expression for the Chern--Simons Lagrangian where the gauge fields are written in terms of the nonlinear realization of the Poincar\'e group $ISO(2n+1,1)$. This can be done using eqs.~$(\ref{xor1},
\ref{xor3})$ in the limit $l\rightarrow\infty$. In that case one finds that the gauge linear and non linear gauge fields are related by
\begin{align}
\bar{e}^a &=e^a-D_{\omega}\phi^a \ , \\
\bar{\omega}^{ab} &=\omega^{ab} \ , 
\end{align}
and therefore the nonlinear gauge connection can be expressed as
\begin{equation}
\mathcal{\bar{A}}=\left( e^a-D_{\omega}\phi^a \right)\mathsf{P}_a+\frac{1}{2}\omega^{ab}\mathsf{J}_{ab} \ .
\end{equation}
and substituting into eq.~$(\ref{cslin})$ we obtain
\begin{align}
\mathscr{L}_{\mathsf{CS}}^{\left(  2n+1\right)  }\left(  \mathcal{\bar{A}%
}\, \right) = & \ \epsilon_{a_{1} \cdots a_{2n+1}}\, R^{a_{1}a_{2}}\wedge \cdots \wedge R^{a_{2n-1}%
a_{2n}}\wedge\left(
e^{a_{2n+1}}-D_{\omega}\phi^{a_{2n+1}}\right) \nonumber \\
& -n\, \dd\int_{0}^{1}\, \dd t \ t^{n-1}\,
\epsilon_{a_{1}\cdots a_{2n+1}}\, R_{t}^{a_{1}a_{2}}%
\wedge \cdots \wedge R_{t}^{a_{2n-3}a_{2n-2}}\wedge\omega^{a_{2n-1}a_{2n}}
\nonumber \\ & \hspace{5cm} \wedge\left(
e^{a_{2n+1}}-D_{\omega}\phi^{a_{2n+1}}\right) \ . \label{csnonlin} 
\end{align}
The gauge transformations for the coset field $\phi$ can be obtained
from eq.~$(\ref{nonl3})$ using $g_{0}-1=-\phi^a \, \mathsf{P}_a$. In
this case one shows that under local Poincar\'e translations the coset
field $\phi$ transforms as $\delta \phi^a=\rho^a$. One can
directly check, as in the case of the linear Lagrangian, that
eq.~$(\ref{csnonlin})$ remains unchanged under gauge transformations.

\newsection{Topological gravity actions} \label{firstresult}

Let now $\mathcal{M}$ be a manifold of dimension $d=2n+1$ with
boundary $\partial\mathcal{M}$. Let $\mathcal{A}$ and $\mathcal{\bar{A}}$ be the linear and nonlinear one-form gauge potentials both taking values in
the Lie algebra $\mathfrak{g}= \mathfrak{iso}(2n,1)$. In the following we assume
that both gauge potentials can be obtained as the pull-back by a local
section $\sigma$ of a one-form connection $\mathsf{\omega}$ defined on a
nontrivial principal $G$-bundle $\mathcal{P}$ over $\mathcal{M}$. 

Let us recall again the transgression action eq.~$\left(  \ref{ch2tranact}\right)  $. In the case for a manifold $\mathcal{M}$ with boundary $\partial\mathcal{M}$ we have
\begin{equation}
\mathsf{S}_{\mathsf{T}}^{\left(  2n+1\right)  }\left[  \mathcal{A}%
,\mathcal{\bar{A}}\right]  
 =\int_{\mathcal{M}}\mathscr{L}_{\mathsf{CS}}^{\left(  2n+1\right)  }\left(
\mathcal{A}\right)  -\int_{\mathcal{M}}\mathscr{L}_{\mathsf{CS}}^{\left(
2n+1\right)  }\left(  \mathcal{\bar{A}}\right)  -\kappa \int_{\partial\mathcal{M}%
}\mathscr{B}^{(2n)}(\mathcal{A},\mathcal{\bar{A}})
. \label{tact}%
\end{equation}
Here 
\begin{align}
\mathscr{B}^{(2n)}& := 
Q_{\mathcal{A}\leftarrow\mathcal{\bar{A}}\leftarrow 0} ^{\left(
2n\right)  } \nonumber \\ 
& =  n\, \left(  n+1\right) \, \int_{0}^{1}\, \dd t \ \int_{0}^{t}\,
\dd s\ \left\langle \left(  \mathcal{A}-\mathcal{\bar{A}}\right)  \wedge
\mathcal{\bar{A}}  \wedge\mathcal{F}_{s,t}^{n-1}%
\right\rangle \ ,
\label{Q2n}
\end{align}
with
\begin{equation}
\mathcal{F}_{s,t}  =\dd\mathcal{A}_{s,t}+\mathcal{A}_{s,t}\wedge\mathcal{A}%
_{s,t} \ ,
\end{equation}
and
\begin{equation}
\mathcal{A}_{s,t}   =s\, \left(  \mathcal{A}-\mathcal{\bar{A}}\right)
+t\, \mathcal{\bar{A}} \ .
\end{equation}
 
If the $G$-bundle $\mathcal{P}$ is nontrivial, then eq.~$(\ref{tact})$ can
be written more precisely by covering $\mathcal{M}$ with local charts. This explains the introduction of the second gauge potential $\mathcal{\bar{A}}$ such that in the overlap of two charts
the connections are related by a gauge transformation which we take to
be given by $g=\e^{-\phi^a \,\mathsf{P}_a} \in G/H$, where
$H=SO(2n,1)$ is the Lorentz subgroup. In this setting the coset element $g \in G/H$ is
interpreted as a transition function determining the nontriviality of
$\mathcal{P}$ \cite{Anabalon:2006fj}.

Now we construct transgression actions for the Poincar\'e group using
the Lagrangian of eq.~$(\ref{cslin})$ and its nonlinear representation
in eq.~$(\ref{csnonlin})$. 
In this case the boundary term $Q_{\mathcal{A\leftarrow\bar{A}\leftarrow}0}^{\left(  2n\right)
} $ defined by eq.~$\left(  \ref{Q2n}%
\right)  $ reads%
\begin{equation}
Q_{\mathcal{A\leftarrow\bar{A}\leftarrow}0}^{\left(  2n\right)
}
= n\, \int_{0}^{1}\, \dd t \ t^{n-1}\, \epsilon_{a_{1} \cdots
  a_{2n+1}}\, R_{t}^{a_{1}a_{2}}
\wedge \cdots \wedge R_{t}^{a_{2n-3}a_{2n-2}}\wedge\omega^{a_{2n-1}a_{2n}}\wedge
D_{\omega}\phi^{a_{2n+1}} \ . \label{borde}%
\end{equation}
Inserting eqs.~$\left(  \ref{cslin},\ref{csnonlin}\right)  $ and eq.~$\left(
\ref{borde}\right)  $ into eq.~$\left(  \ref{tact}\right)  $ we get%
\begin{equation}
\mathsf{S}_{\mathrm{T}}^{\left(  2n+1\right)  }\left[  \mathcal{A}%
,\mathcal{\bar{A}}\, \right]    =\kappa\, \int_{\mathcal{M}}\, \epsilon
_{a_{1} \cdots a_{2n+1}}\, R^{a_{1}a_{2}}\wedge \cdots \wedge R^{a_{2n-1}a_{2n}}\wedge
D_{\omega}\phi^{a_{2n+1}} \ ,
\end{equation}
which is a boundary term because of the Bianchi identity
$D_\omega R^{ab}=0$ and Stokes' theorem. This motivates the writing
\begin{equation}
\mathrm{S}^{\left(  2n\right)  }\left[  \omega,\phi\right]  =\kappa\, \int
_{\partial\mathcal{M}}\, \epsilon_{a_{1} \cdots a_{2n+1}}\, R^{a_{1}a_{2}}\wedge \cdots \wedge R^{a_{2n-1}%
a_{2n}}\, \phi^{a_{2n+1}} \label{topg}
\end{equation}
as an action
principle in one less dimension which corresponds to
$2n$-dimensional topological Poincar\'e gravity. Our derivation
can be regarded as a holographic principle in the
sense that the transgression action in eq.~$(\ref{tact})$ collapses to
its boundary contribution once we consider gauge connections taking
values in the Lie algebras associated to the linear and nonlinear
realizations of the Poincar\'e group. The topological action of
eq.~$(\ref{topg})$ is the action of a gauged WZW
model~\cite{Anabalon:2007dr}; this is because the transformation law for
the nonlinear gauge fields has the same form as a gauge transformation
from eq.~$(\ref{gtc})$ with gauge element $g=\e^{-\phi^{a}\,
  \mathsf{P}_a} \in \frac{ISO(2n,1)}{SO(2n,1)}$
\cite{Salgado:2013pva}. 

Recall that the nonlinear realization prescribes a transformation law
for the field $\phi$ under local translations given by $\delta
\phi^a=\rho^a$. This transformation breaks the symmetry of
eq.~$(\ref{topg})$ from $ISO(2n,1) $ to $ SO(2n,1)$; this is due to
the fact that the transformation law of the coset field $\phi$ under
local translations is not a proper adjoint transformation (see
eq.~$(\ref{nonl3})$). 
 
The variation of the action in eq.~$(\ref{topg})$ leads to the field
equations
\begin{align}
\epsilon_{abc a_{1} \cdots a_{2n-2}}\,D_\omega\phi^c \wedge R^{a_{1}a_{2}} \wedge \cdots  \wedge R^{a_{2n-3}%
a_{2n-2}} & =0 \ , \label{eomtop1} \\[4pt]
\epsilon_{c a_{1} \cdots a_{2n}}\, R^{a_{1}a_{2}} \wedge \cdots \wedge
R^{a_{2n-1}a_{2n}}  & =0 \ . \label{eomtop2}
\end{align}

Note that one can always use a gauge transformation to rotate to a
frame in which $\phi^1= \cdots = \phi^{a_{2n}}=0$ and $\phi^{a_{2n+1}}
:= \phi$. This choice breaks the gauge symmetry to the residual gauge
symmetry preserving the frame, which is a subgroup $SO(2n-1,1)
\hookrightarrow SO(2n,1)$; this is just the usual Lorentz symmetry
in $2n$ dimensions. If in addition one imposes the condition
$\omega^{a,{2n+1}}=0$ for $a=1, \ldots ,2n$, then gauge invariance of
eq.~(\ref{topg}) is also preserved. 

\subsection{Three-dimensional supergravity}
\label{3dsugra}
Supergravity in three dimensions \cite{Des83,Ach86} can be formulated
as a Chern--Simons theory for the Poincar\'e supergroup
\cite{Banh96a}. The action
is invariant (up to boundary terms) under Lorentz rotations, Poincar\'e
translations and $\mathcal{N}=1$ supersymmetry transformations. The gauge fields $e^a$, $\omega^{ab}
$ and the Majorana spinor $\bar{\psi}$ transform as components of a gauge connection
valued in the $\mathcal{N}=1$ supersymmetric extension of Poincar\'e
algebra in three dimensions given by
\begin{equation}
\mathcal{A}=\ii e^{a}\,  \mathsf{P}_{a}+\mbox{$\frac{\ii}{2}$}\,
\omega^{ab}\, \mathsf{J}_{ab}%
+\bar{\psi}\, \mathsf{Q} \ .
\end{equation}
This algebra contains, in addition to the bosonic commutation relations, the
supersymmetry algebra structure given by
\begin{align}
[  \mathsf{Q}_{\alpha},\mathsf{J}_{ab}] &=-
\mbox{$\frac{\ii}{2}$}\,  \left(\Gamma_{ab}
\right)_{\alpha}^{~\beta}\, \mathsf{Q}_{\beta} \ , \\ 
 \left\{  \mathsf{Q}_{\alpha},\mathsf{Q}_{\beta}\right\} & =\left(
\Gamma^{a}\right)  _{\alpha \beta}\, \mathsf{P}_{a} \ .
\end{align}
Here $\Gamma_{ab}=\left[  \Gamma_{a},\Gamma_{b}\right]  $, and the set of
gamma-matrices $\Gamma_{a}$ with $a=1,2,3$ defines a representation of the
Clifford algebra in $2+1$ dimensions. Our spinor conventions are summarized in Appendix~\ref{ch:app1}. In the case of $\mathcal{N}=1$ supersymmetry, the model is
described by the action%
\begin{align}
\mathrm{S}^{\left(  3\right)  }\left(  \mathcal{A}\right)
& =\kappa\, \int_{\mathcal{M}}\, \mathscr{L}_{\mathsf{CS}}^{\left(  3\right)  }\left(
\mathcal{A}\right) \nonumber \\
& = \kappa\, \int_{\mathcal{M}}\, \big(
\epsilon_{abc}\, R^{ab}\wedge
e^{c}-\ii\bar{\psi}\wedge D_\omega\psi\big) -\frac{\kappa}{2}\, \int_{\partial\mathcal{M}}\, \epsilon_{abc}\, \omega^{ab}\wedge
e^{c} \ , \label{sulinact}%
\end{align}
where ${\psi}$ is a two component Majorana spinor
one-form 
and
\begin{equation}
D_\omega\psi:=\dd\psi+\frac14\, \omega^{ab}\wedge \Gamma_{ab}\psi \ ,
\end{equation}
is the Lorentz covariant derivative in the spinor representation. Under an infinitesimal
gauge transformation with parameter $\lambda=\ii \rho^{a}\, \mathsf{P}_{a}+\frac
{\ii}{2}\, \kappa^{ab}\, \mathsf{J}_{ab}+\bar{\varepsilon}\, \mathsf{Q}$, the gauge fields transform
as
\begin{align}
\delta e^{a} &  =-D_{\omega}\rho^{a}+\kappa_{~b}^{a}\,
e^{b} \ ,
\nonumber \\[4pt]
\delta\omega^{ab} &  =-D_{\omega}\kappa^{ab} \ ,  \nonumber \\[4pt]
\delta\bar{\psi} &  =-D_\omega\bar{\varepsilon}-\mbox{$\frac{1}{4}$}\,
\kappa^{ab}\, \bar{\psi
}\, \Gamma_{ab} \ . 
\end{align}
These transformations leave the action of eq.~$\left(  \ref{sulinact}\right)  $ invariant
modulo boundary terms.

\subsection{Supersymmetric SWGN formalism}

The supersymmetric Stelle--West--Grignani--Nardelli formalism is treated in~\cite{Sal03a} where
the nonlinear realization of the supersymmetric \textrm{AdS} group in
three dimensions is considered. Here we consider the
nonlinear realization of the three-dimensional $\mathcal{N}=1$ Poincar\'e supergroup \cite{Sal01}.

Let $G$ denote the Poincar\'e supergroup generated by $\left\{  \mathsf{J}%
_{ab},\mathsf{P}_{a},\mathsf{Q} \right\}  $. It is convenient
to decompose $G$ into two subgroups: The Lorentz subgroup $L=SO(2,1)$ generated by
$\left\{  \mathsf{J}_{ab}\right\}  $ as the stability subgroup, and the
Poincar\'e subgroup $H=ISO(2,1)$ generated by $\left\{  \mathsf{J}_{ab},\mathsf{P}%
_{a}\right\}  $. We introduce a coset field associated to each generator in
the coset space $G/L$ through
$\bar{\chi}\,\mathsf{Q}$ and $\phi^{a}\, \mathsf{P}_{a}$. 
Let us write eq.~$\left(  \ref{nonl2}\right)  $ in the form%
\begin{equation}
g_{0}\, \e^{-\bar{\chi}\,\mathsf{Q}}\,  \e^{-\phi\text{\textperiodcentered}\mathsf{P}%
}= \e^{-\bar{\chi}^{\prime}\,\mathsf{Q}}\, 
\e^{-\phi^{\prime}\text{\textperiodcentered}\mathsf{P}}\, l_{1}
\end{equation}
with $l_1\in L$.
Multiplying on the right by $\e^{\phi\text{\textperiodcentered}\mathsf{P}}$ we
get%
\begin{equation}
g_{0}\, \e^{-\bar{\chi}\,\mathsf{Q}}\,  =\e^{-\bar{\chi}^{\prime}%
\,\mathsf{Q}}\,  h_{1}\qquad \mbox{and}
\qquad 
h_{1}\, \e^{-\phi\text{\textperiodcentered}\mathsf{P}} =\e^{-\phi^{\prime
}\text{\textperiodcentered}\mathsf{P}}\, l_{1}%
\end{equation}
with $h_{1}=\e^{-\phi^{\prime}\text{\textperiodcentered}\mathsf{P}}\,
l_{1}\, \e^{\phi
\text{\textperiodcentered}\mathsf{P}}\in H$. To obtain the transformation law of
the coset fields, we write these expressions in infinitesimal form%
\begin{align}
\e^{\bar{\chi}\,\mathsf{Q}} \left(  g_{0}-1\right)  \e^{-\bar{\chi}\, \mathsf{Q}}-\e^{\bar{\chi}\,\mathsf{Q}}\, \delta\big(
\e^{-\bar{\chi}\, \mathsf{Q}}\big) &  =h_{1}-1 \ , \label{suinf1}\\[4pt]
\e^{\phi\text{\textperiodcentered}\mathsf{P}}\, \left(  h_{1}-1\right)\,
\e^{-\phi\text{\textperiodcentered}\mathsf{P}}-\e^{\phi\text{\textperiodcentered
}\mathsf{P}}\, \delta \e^{-\phi\text{\textperiodcentered}\mathsf{P}} &
=l_{1}-1 \ , \label{suinf2}%
\end{align}
where $h_{1}=h_{1}\left(  \bar{\chi},\bar{\varepsilon}
,\rho,\kappa\right)  $ and $l_{1}=l_{1}\left(  \bar{\chi},\phi
,\bar{\varepsilon},\rho,\kappa\right)  $. Inserting $g_{0}%
-1=-\ii\rho^{a}\, \mathsf{P}_{a}-\frac{\ii}{2}\,
\kappa^{ab}\,\mathsf{J}_{ab}-\bar{\varepsilon}\, \mathsf{Q} $,
$h_{1}-1=-\ii \rho^{a}\, \mathsf{P}_{a}-\frac{\ii}{2}\kappa^{ab}\,\mathsf{J}_{ab}$ and $l_{1}-1=-\frac{\ii}%
{2}\, \kappa^{ab}\, \mathsf{J}_{ab}$ into eqs.~$\left(  \ref{suinf1}\text{,}%
\ref{suinf2}\right)  $, we find the symmetry transformations for the
coset fields%
\begin{align}
\delta\phi^{a} &  =\rho^{a}+\mbox{$\frac{\ii}{2}$}\,
\bar{\varepsilon}\, \Gamma^{a}\chi -\kappa_{~c}^a\, \phi^c \ , \label{nltf1}\\[4pt]
\delta\bar{\chi} &  = \mbox{$\frac{1}{4}$}\, \bar{\chi}\,
\kappa^{ab}\, \Gamma_{ab}+\bar{\varepsilon} \ . \label{nltf3}%
\end{align}
The relations between the linear and nonlinear gauge fields can be obtained
from eq.~$\left(  \ref{gtc}\right)  $. With $g=\e^{-\bar{\chi}%
\,\mathsf{Q}}\, \e^{-\phi\text{\textperiodcentered}\mathsf{P}}$
we get%
\begin{align}
V^{a} &  =e^{a}-D_\omega\phi^{a}-\mbox{$\frac{\ii}{2}$}\, D_\omega\bar{\chi} \,\Gamma^{a} \chi+\ii\bar
{\chi}\,   \Gamma^{a} \psi
 \ , \label{nlf1}\\[4pt]
W^{ab} &  =\omega^{ab} \ , \label{nlf2}\\[4pt]
\bar{\Psi} &
=\bar{\psi}-D_\omega\bar{\chi} \ . \label{nlf3} 
\end{align}

Note that the action for supergravity in three dimensions
written in terms of nonlinear fields reads%
\begin{align}
\mathrm{S}^{\left(  3\right)  }\left(  \mathcal{\bar{A}%
}\, \right) &=\kappa\, \int_{\mathcal{M}}\, \mathscr{L}_{\mathsf{CS}}^{\left(  3\right) 
}\left(  \mathcal{\bar{A}}\, \right) \nonumber \\  
 & = \kappa\,
\int_{\mathcal{M}}\, \big( \epsilon
_{abc}\, R^{ab}\wedge V^{c}-\ii 
\bar{\Psi}\wedge D_\omega\Psi \big) -\frac{\kappa}{2}\, \int_{\partial\mathcal{M}}\, \epsilon_{abc}\,
\omega^{ab} \wedge
V^{c} \ . \label{sunolinact}%
\end{align}
where 
\begin{equation}
\bar{\mathcal{A}}=V^{a}\mathsf{P}_{a}+\dfrac{1}{2}W^{ab}\mathsf{J}_{ab}+\bar{\Psi}\mathsf{Q} \ .
\end{equation}
\subsection{Topological supergravity in two-dimensions}

In complete analogy with the bosonic case, we now construct a transgression
action for the Poincar\'e supergroup in three dimensions.
Inserting eq.~$\left(  \ref{sulinact}\right)  $ and eq.~$\left(
\ref{sunolinact}\right)  $ in eq.~$\left(  \ref{tact}\right)  $ with
\begin{equation}
\mathscr{B}^{\left(  2\right)  }\left(
  \mathcal{A},\mathcal{\bar{A}}\, \right)
= \mbox{$-\frac{1}{2}$}\, \epsilon_{abc}\, \omega^{ab}\wedge \big(  D_\omega\phi^{c}+\ii D_\omega\bar{\chi
}\, \Gamma^{c} \chi-\ii\bar{\chi}\,  \Gamma^{c}
\psi \big)  
 +\ii\bar{\psi}\wedge D_\omega\chi
\end{equation}
we obtain%
\begin{equation}
\mathrm{S}^{\left(  2\right)  }\left[\omega,\phi;\bar{\psi},\chi\right]
=\kappa\, \int_{\partial\mathcal{M}}\, \big( \epsilon_{abc}\, R^{ab}\,
 \phi^{c} -2\ii\bar{\psi}\wedge D_\omega\chi
 \big) \ . \label{supact}%
\end{equation}
This action corresponds to the supersymmetric extension of topological
gravity in two dimensions proposed by~\cite{Chamseddine:1989yz}. As in the
purely bosonic case, supersymmetry is broken to the Lorentz symmetry
$SO(2,1)$ because of the nonlinear transformation laws in
eqs.~(\ref{nltf1},\ref{nltf3}); however, the action is invariant
under the full supersymmetry if one prescribes the correct
transformation laws for the coset fields $\bar{\chi}, \phi$
instead of considering the symmetries dictated by the nonlinear
realization. The variation of the action in eq.~$(\ref{supact})$ leads
to the field equations
\begin{align}
 \epsilon_{abc} \, \left(
   D_\omega\phi^{c}-\ii\bar{\psi}\,\Gamma^{c}\chi \right) &= 0 \ , \\ 
  \epsilon_{abc}\, R^{ab}&=0 \ , \\
  D_{\omega} \chi & =0 \ , \\
  D_{\omega}{\bar{\psi}} &=0  \ .
\end{align}

In this way, it has been shown that even-dimensional topological (super)gravity can be obtained by using a transgression field theory where the gauge connections take values in the Lie algebra associated to the linear and nonlinear realization of the (super)Poincar\'e group. The topological actions corresponds to a gauged Wess--Zumino--Witten term where the gauge transformation relating both gauge connections lives in the coset $ISO(2n,1)/SO(2n,1)$. 
It would be interesting to explore the relationship between the space of solutions of the Chern--Simons theory and its corresponding gauged Wess--Zumino--Witten models since there is evidence that Wess--Zumino--Witten actions corresponds to the holographic dual of the gravitational theory in one more dimension~\cite{Barnich:2013yka}. This problem is out of the scope of this Thesis but it constitutes a very interesting possibility to explore.

In principle, given a Lie algebra and its invariant tensors, one should be able to construct the associated gauged Wess--Zumino--Witten model. An interesting possibility is to consider non trivial extensions of the classical gravitational groups studied in\cite{Chamseddine:1990gk}. This is the case for instance of the Maxwell algebra. In the next Chapter, the three-dimensional Chern--Simons action for the Maxwell algebra is constructed, as well as its corresponding gauged Wess--Zumino--Witten model.

\chapter{Gauged WZW model for the Maxwell algebra}
\label{ch:Max_alg}
\begin{flushright}
\textit{``...\textquestiondown valdr\'a la pena jugarse 
la vida por una idea 
que puede resultar falsa?... \\
...claro que vale la pena}''. \\ \textit{Preguntas y respuestas, Nicanor Parra.}
\footnote{\scriptsize ``...will it be worth gamble
life for an idea
that may be false ?...Clearly worth''. Nicanor Parra, Questions and answers.}
\bigskip
\end{flushright}
The Maxwell algebra was introduced in the early seventies \cite{Bacry:1970ye,Schrader:1972zd}. In this context, the Maxwell algebra encodes the symmetries of a particle moving in an
electromagnetic background. Recently,
it has attracted more attention at some extent due to its supersymmetric extension~\cite{Bonanos:2009wy}.
In the context of gravitational theories, in~\cite{deAzcarraga:2010sw,deAzcarraga:2012qj} it is argued that gauging the Maxwell algebra leads to new contributions to the
cosmological term in Einstein gravity. Recently, in~\cite{deAzcarraga:2012qj} it has been shown that $D=4$, $N=1$ supergravity can be obtained geometrically as a quadratic expression in the curvatures of the Maxwell superalgebra. Thus, the Maxwell (super)algebra carries an interesting set of symmetries beyond the standard (super)Poincar\'e ones.

In this chapter we explore the
implications of the gauged Maxwell algebra in the context of Chern--Simons
gravity \cite{Salgado:2014qqa}. In particular, we consider the three-dimensional case and the construction of the
corresponding gauged
WZW model in two dimensions. In order to do so, we first need to specify the nonzero component of the invariant tensor associated to the Maxwell algebra. As we will see, the invariant tensors can be obtained by performing an $S-$expansion procedure~\cite{Izaurieta:2006ia} starting form the anti-de Sitter algebra $SO(d-1,2)$ with a suitable semigroup $S$.

\newsection{Maxwell algebra and $S-$expansion procedure}

The Maxwell algebra is a noncentral extension of the Poincar\'e algebra by a rank two tensor $\mathsf{Z}_{ab}=-\mathsf{Z}_{ba}$ such that
 \begin{align} 
 \left[  \mathsf{J}_{ab},\mathsf{P}_{c}\right]      &=\eta_{bc}\mathsf{P}%
 _{a}-\eta_{ac}\mathsf{P}_{b} \ , \label{maxalg1}\\
 \left[  \mathsf{Z}_{ab},\mathsf{P}_{c}\right]    & =0 \ ,\\
 \left[ \mathsf{J}_{ab},\mathsf{J}_{cd}\right] & =\eta _{bc}\, \mathsf{J}%
 _{ad}+\eta _{ad}\, \mathsf{J}_{bc}-\eta _{ac}\, \mathsf{J}_{bd}-\eta
 _{bd}\, \mathsf{J%
 }_{ac} \ , \\                                
 \left[ \mathsf{P}_{a},\mathsf{P}_{b}\right] & =\mathsf{Z}_{ab} \ , \\
\left[ \mathsf{J}_{ab},\mathsf{Z}_{cd}\right]  &=\eta _{bc}\, \mathsf{Z}%
_{ad}+\eta _{ad}\, \mathsf{Z}_{bc}-\eta _{ac}\, \mathsf{Z}_{bd}-\eta
_{bd}\, \mathsf{Z%
}_{ac} \ . \label{maxalg2}
\end{align}

We show now that the Maxwell algebra can be obtained as an $S$-expansion starting from the \textrm{AdS}
algebra. $S$-expansions consist of systematic Lie
algebra enhancements which enlarge symmetries. They have the nice property
that they provide the right invariant tensor of the expanded algebra
\cite{Iza06b}, which is a key ingredient in the evaluation of Chern--Simons forms. For related approaches regarding $S-$expansions in the context of gravitational Lie algebras see~\cite{Diaz:2012zza,Fierro:2014lka,Caroca:2011zz} and  Appendix~\ref{ch:app3} for generalities about the procedure itself.

\subsection{$S$-expansion of the AdS algebra}

We now show that the Maxwell algebra can be obtained by an $S$-expansion
of the AdS algebra. Let $\mathrm{S}_{\mathrm{E}}^{\left(  2\right)  }$ be the semigroup~\cite{Salgado:2014qqa}
\begin{equation}
\mathrm{S}_{\mathrm{E}}^{\left(  2\right)  }=\left\{  \lambda_{0},\lambda
_{1},\lambda_{2},\lambda_{3}\right\} \ ,
\end{equation}
with composition law%
\begin{equation}
\lambda_{\alpha}\text{\textperiodcentered}\lambda_{\beta}:= \left\{
\genfrac{.}{.}{0pt}{}{\lambda_{\alpha+\beta} \qquad \text{ if } \ \alpha+\beta
\leq3 \ ,}{\lambda_{3} \qquad \text{ if } \ \alpha+\beta>3 \ .}%
\right.
\label{SE2law}\end{equation}
Recall that the \textrm{AdS} algebra $\frg=\mathfrak{so}(d-1,2)$ in $d$ dimensions is given by%
\begin{align}
\left[\,  \bar{\mathsf{J}}_{ab},\bar{\mathsf{J}}_{cd} \, \right]    & =\eta
_{bc}\, \bar{\mathsf{J}}_{ad}+\eta_{ad}\, \bar{\mathsf{J}}_{bc}-\eta_{ac} \,
\bar{\mathsf{J}}_{bd}-\eta_{bd}\, \bar{\mathsf{J}}_{ac} \ , \label{ads1}\\[4pt]
\left[ \, \bar{\mathsf{J}}_{ab},\bar{\mathsf{P}}_{c} \, \right]    & =\eta
_{bc}\, \bar{\mathsf{P}}_{a}-\eta_{ac}\, \bar{\mathsf{P}}_{b}  \ , \label{ads2}  \\[4pt]
\left[ \, \bar{\mathsf{P}}_{a},\bar{\mathsf{P}}_{b} \, \right]    & =\bar{\mathsf{J}}_{ab} \ . \label{ads3}
\end{align}
This algebra can be decomposed into two subspaces
$
\mathfrak{g}=\mathsf{V}_{0}\oplus\mathsf{V}_{1}%
$
where $\mathsf{V}_{0}={\rm Span}_{\complex}\left\{  \bar{\mathsf{J}}_{ab}\right\}$ and $\mathsf{V}_{1}={\rm Span}_{\complex}\left\{  \bar{\mathsf{P}}_{a}\right\}$. In terms of these subspaces, the $\mathrm{AdS}$ algebra has
the structure%
\begin{equation}
\left[  \mathsf{V}_{0},\mathsf{V}_{0}\right]  \subset\mathsf{V}_{0} \
, \qquad \left[  \mathsf{V}_{0},\mathsf{V}_{1}\right]
\subset\mathsf{V}_{1} \qquad \mbox{and} \qquad \left[
  \mathsf{V}_{1},\mathsf{V}_{1}\right]  \subset\mathsf{V}_{0} \ . \label{subspa}
\end{equation}
If we now choose the partition for the semigroup [See. \ref{resexp}]
$\mathrm{S}_{\mathrm{E}}^{\left(  2\right)  }$ given by
\begin{equation}
\mathrm{S}_{0}  =\left\{  \lambda_{0},\lambda_{2}\right\}  \cup\left\{
\lambda_{3}\right\} \qquad \mbox{and} \qquad 
\mathrm{S}_{1} =\left\{  \lambda_{1}\right\}  \cup\left\{  \lambda
_{3}\right\} \ ,
\end{equation}
then this partition is resonant with respect to the structure of the
AdS algebra; under the
semigroup multiplication law we have%
\begin{equation}
\mathrm{S}_{0}\, \text{\textperiodcentered}\, \mathrm{S}_{0}\subset\mathrm{S}_{0}
\ , \qquad \mathrm{S}_{0}\,\text{\textperiodcentered}\, \mathrm{S}_{1}\subset\mathrm{S}%
_{1} \qquad \mbox{and} \qquad \mathrm{S}_{1}\,
\text{\textperiodcentered}\, \mathrm{S}_{1}%
\subset\mathrm{S}_{0}
\end{equation}
which agrees with the decomposition in eq.~(\ref{subspa}). The resonance condition allows us to construct a resonant subalgebra
$\mathfrak{g}_{\rm R}$ defined by%
\begin{equation}
\mathfrak{g}_{\rm R}=\mathsf{W}_{0}\oplus\mathsf{W}_{1}:= \left(
\mathrm{S}_{0}\times\mathsf{V}_{0}\right)  \oplus\left(  \mathrm{S}%
_{1}\times\mathsf{V}_{1}\right) \ .
\end{equation}
Explicitly one has
\begin{align}
\mathsf{W}_{0}  & =\left\{  \lambda_{0},\lambda_{2},\lambda_{3}\right\}
\times{\rm Span}_{\complex}\left\{  \bar{\mathsf{J}}_{ab}\right\}
=:{\rm Span}_{\complex} \left\{  \mathsf{J}%
_{ab,0},\mathsf{J}_{ab,2},\mathsf{J}_{ab,3}\right\} \ , \nonumber \\[4pt]
\mathsf{W}_{1}  & =\left\{  \lambda_{1},\lambda_{3}\right\}
\times{\rm Span}_{\complex} \left\{
\bar{\mathsf{P}}_{a}\right\}  =: {\rm Span}_{\complex} \left\{  \mathsf{P}_{a,1},\mathsf{P}%
_{a,3}\right\} \ .
\end{align}
Since $\lambda_{3}$ is a zero element
in the semigroup [See. \ref{0reduc}], one can extract another subalgebra by setting $\mathsf{J}_{ab,3}=\mathsf{P}%
_{a,3}=0$; this choice still preserves the Lie algebra
structure of the residual algebra. This algebra is called a $0_{S}$-forced
resonant algebra and is composed by the subspaces
\begin{align}
{\tilde{\sf W}}_{0}  ={\rm Span}_{\complex} \left\{  \mathsf{J}_{ab,0},\mathsf{J}%
_{ab,2}\right\}  \qquad \mbox{and} \qquad
{\tilde{\sf W}}_{1}  ={\rm Span}_{\complex} \left\{
  \mathsf{P}_{a,1}\right\} \ .
\end{align}
In order to obtain a presentation for the $0_{S}$-forced resonant algebra we
use eqs. (\ref{ads1}--\ref{ads3}) together with eq. (\ref{SE2law}) to compute the commutation relations 
\begin{align}
\left[  \mathsf{J}_{ab,0},\mathsf{J}_{cd,0}\right]    & =\lambda_{0}%
\lambda_{0}\left[  \mathsf{\bar{J}}_{ab},\mathsf{\bar{J}}_{cd}\right]
\sim\mathsf{J}_{ab,0} \ ,\\
\left[  \mathsf{J}_{ab,0},\mathsf{J}_{cd,2}\right]    & =\lambda_{0}%
\lambda_{1}\left[  \mathsf{\bar{J}}_{ab},\mathsf{\bar{J}}_{cd}\right]
\sim\mathsf{J}_{ab,2} \ , \\
\left[  \mathsf{J}_{ab,2},\mathsf{J}_{cd,2}\right]    & =\lambda_{2}%
\lambda_{2}\left[  \mathsf{\bar{J}}_{ab},\mathsf{\bar{J}}_{cd}\right]  \sim0 \ , \\
\left[  \mathsf{J}_{ab,0},\mathsf{P}_{a,1}\right]    & =\lambda_{0}\lambda
_{1}\left[  \mathsf{\bar{J}}_{ab},\mathsf{\bar{P}}_{c}\right]  \sim
\mathsf{P}_{a,1} \ , \\
\left[  \mathsf{J}_{ab,2},\mathsf{P}_{a,1}\right]    & =\lambda_{2}\lambda
_{1}\left[  \mathsf{\bar{J}}_{ab},\mathsf{\bar{P}}_{c}\right]  \sim0 \ , \\
\left[  \mathsf{P}_{a,1},\mathsf{P}_{b,1}\right]    & =\lambda_{1}\lambda
_{1}\left[  \mathsf{\bar{P}}_{a},\mathsf{\bar{P}}_{b}\right]  \sim
\mathsf{J}_{ab,2} \ .%
\end{align}
Here we have used the symbol $\sim$ just for denoting the subspace structure of the expanded algebra. Now, identifying
\begin{align}
\mathsf{J}_{ab}:=\mathsf{J}_{ab,0}  \ , \qquad 
\mathsf{Z}_{ab}:= \mathsf{J}_{ab,2}  \qquad \mbox{and} \qquad
\mathsf{P}_{a}:= \mathsf{P}_{a,1} 
\end{align}
we obtain the Maxwell algebra in $d$ dimensions (\ref{maxalg1}-\ref{maxalg2}).

\subsection{Invariant tensors}

The $S$-expansion procedure also provides the invariant tensors
associated to the expanded algebra [See. \ref{invite}]; here we study the particular
case of $d=3$ dimensions. The invariant tensors of the AdS algebra $\mathfrak{so}(2,2)$ are
given by \cite{Hayashi:1991mu}%
\begin{align}
\left\langle \,\bar{\mathsf{J}}_{ab}\, \bar{\mathsf{J}}_{cd}\, \right\rangle  &
=\mu_{0}\, \left(  \eta_{ad}\, \eta_{bc}-\eta_{ac}\, \eta_{bd}\right)
\ , \nonumber \\[4pt]
\left\langle \, \bar{\mathsf{J}}_{ab}\, \bar{\mathsf{P}}_{c}\, \right\rangle  &
=\mu_{1}\, \epsilon_{abc} \ , \nonumber \\[4pt]
\left\langle \, \bar{\mathsf{P}}_{a}\, \bar{\mathsf{P}}_{b}\, \right\rangle  & =\mu
_{0}\, \eta_{ab} \ ,
\end{align}
where $\mu_i$, $i=0,1$ are arbitrary constants. By \cite[Theorem~7.2]{Iza06b}, the $S$-expanded tensors
are given by the formula%
\begin{equation}
\left\langle \mathsf{T}_{A,\alpha}\, \mathsf{T}_{B,\beta}\right\rangle
=\tilde{\alpha}_{\gamma}\, K_{\alpha\beta}^{~~\gamma}\, \left\langle \mathsf{T}%
_{A}\, \mathsf{T}_{B}\right\rangle \label{sinvten}
\end{equation}
where $\tilde{\alpha}_\gamma$ are also arbitrary constants, and $K_{\alpha\beta}^{~~\gamma}$ is
called a two-selector (\ref{twosel}). The application of the formula eq.~(\ref{sinvten}) for the $S$-expanded generators 
$\mathsf{J}_{ab,0}$, $\mathsf{J}_{ab,2}$ and $\mathsf{P}_{a,1}$ gives
the following invariant tensors for the Maxwell algebra 
\begin{align}
 \left\langle \mathsf{J}_{ab}\, \mathsf{J}_{cd}\right\rangle  &=\alpha
 _{0}\, \left( \eta _{ad}\, \eta _{bc}-\eta _{ac}\, \eta _{bd}\right)
 \ ,  \label{tens1} \\[4pt]
 \left\langle \mathsf{J}_{ab}\, \mathsf{P}_{c}\right\rangle  &=\alpha
 _{1}\, \epsilon _{abc} \ , \label{tens2} \\[4pt]
 \left\langle \mathsf{J}_{ab}\, \mathsf{Z}_{cd}\right\rangle  &=\alpha
 _{2}\, \left( \eta _{ad}\, \eta _{bc}-\eta _{ac}\, \eta _{bd}\right)
 \ , \label{tens3}  \\[4pt]
 \left\langle \mathsf{P}_{a}\, \mathsf{P}_{b}\right\rangle  &=\alpha
 _{2}\, \eta_{ab} \ , \label{tens4}
 \end{align}%
 with the redefined arbitrary constants
\begin{equation}
\alpha_{0}:= \tilde{\alpha}_{0}\, \mu_{0} \ , \qquad 
\alpha_{1}:= \tilde{\alpha}_{1}\, \mu_{1} \qquad \mbox{and} \qquad
\alpha_{2}:= \tilde{\alpha}_{2}\, \mu_{0} \ .
\end{equation}

\newsection{Maxwell algebra and Chern--Simons gravity}

In order to construct the Chern--Simons gravitational Lagrangian, we need to \textit{gauge} the Maxwell algebra. Let us to consider a connection one-form
$\mathcal{A}$ taking values in the Maxwell algebra, which can be
expanded as%
\begin{equation}
\mathcal{A}=e^{a}\, \mathsf{P}_{a}+\mbox{$\frac{1}{2}$}\, \omega
^{ab}\, \mathsf{J}_{ab}+\mbox{$\frac{%
1}{2}$}\, \sigma ^{ab}\, \mathsf{Z}_{ab}
\end{equation}
Here, $e^{a}$ and $\omega ^{ab}$ are identified with the standard vielbein and spin
connection, and we introduce an additional rank two
antisymmetric one-form $\sigma^{ab}=-\sigma^{ba}$ as the gauge field
corresponding to the generator $\mathsf{Z}_{ab}$. 
 
 With this data, we can apply the subspace separation method once again to obtain the associated Chern--Simons action. In fact, using 
 \begin{align}
 \mathcal{A}_{0}  & =0,\\
 \mathcal{A}_{1}  & =\omega,\\
 \mathcal{A}_{2}  & =e+\omega,\\
 \mathcal{A}_{3}  & =\sigma+e+\omega,
 \end{align}
 where $e=e^{a}\mathsf{P}_{a}$, $\omega=\frac{1}{2}\omega^{ab}\mathsf{J}_{ab}$
 and $\sigma=\frac{1}{2}\sigma^{ab}\mathsf{Z}_{ab}$ recursively in the triangle equation (\ref{pr34}), and the nonvanishing components of the invariant tensor eq.(\ref{tens1}-\ref{tens4}), one finds the Chern--Simons gravity action for the
 Maxwell algebra 
 \begin{align}    
 \mathrm{S}_{\mathsf{CS}}^{(3)}(\mathcal{A})=&\ \kappa \,
 \int_{\mathcal{M}} \, \bigg(
 \frac{\alpha_{0}}{2}\, \omega_{~b}^a \wedge \Big( \dd\omega _{~c}^{b}+\frac{2%
 }{3}\, \omega _{~d}^{b}\wedge \omega _{~c}^{d}\Big) +\alpha _{1}\, \epsilon
 _{abc}\, R^{ab}\wedge e^{c}  \label{maxcs}\\
 & \qquad\qquad +\alpha _{2}\, \left( T^{a}\wedge e_{a}+R_{~c}^{a}\wedge \sigma
 _{~a}^{c}\right) -\dd\Big( \, \frac{\alpha _{1}}{2}\, \epsilon _{abc}\, \omega
 ^{ab}\wedge e^{c}+\frac{\alpha _{2}}{2}\, \omega _{~b}^{a}\wedge \sigma
 _{~a}^{b}\, \Big) \bigg) \ , \nonumber
 \end{align}%
where $T^a=D_\omega e^a$ is the torsion two-form. The resulting theory
contains three sectors governed by the different values of the
coupling constants $\alpha_i$. The first term is the gravitational
Chern--Simons Lagrangian~\cite{Witten:1988hc} while the second term is the
usual Einstein--Hilbert Lagrangian. The sector proportional to
$\alpha_2$ contains the torsional term plus a new coupling between the
gauge field $\sigma^{ab}$ and the Lorentz curvature. Up to boundary
terms, the action of eq.~$(\ref{maxcs})$ is invariant under the local
gauge transformations%
\begin{align}
\delta e^{a} &=-D_{\omega }\rho ^{a}+\kappa _{~b}^{a}\, e^{b} \ ,
\nonumber \\[4pt]
\delta \omega ^{ab} &=-D_{\omega }\kappa ^{ab} \ , \nonumber \\[4pt]
\delta \sigma ^{ab} &= -D_{\omega }\tau ^{ab}-2e^{a}\, \rho ^{b}-2\omega
_{~c}^{a} \tau ^{cb}+2\kappa _{~c}^{a}\, \sigma ^{cb} \ .
\end{align}

The variation of eq.(\ref{maxcs}) leads to the following equations of motion
\begin{align}
\alpha_{0}R_{ab}+\alpha
_{1}\epsilon_{abc}T^{c}-\alpha_{2}\left(  e_{a} \wedge e_{b}+\frac{1}{2}D_{\omega}\sigma_{ab}\right)  =0 ,  \label{eom1}\\
\alpha_{1}\epsilon_{abc}%
R^{ab}+2\alpha_{2}T_{c}=0 \label{eom2},\\
\alpha_{2}R_{ab}=0 \label{eom3}.
\end{align}
Substituting eq.$(\ref{eom3})$, with $\alpha_2 \neq 0$, into eq.$(\ref{eom2})$ we get $T^a=0$. Substituting again into eq.$(\ref{eom1})$ we finally get
\begin{align}
R_{ab} &  =0,\label{flatr}\\
T_{c} &  =0,\label{flatt}\\
D_{\omega}\sigma_{ab} +2e_{a} \wedge e_{b}&  =0. \label{flats}%
\end{align}
Thus, according to eq.$(\ref{flatr}, \ref{flatt})$, the three dimensional Chern--Simons action for the Maxwell algebra describes a flat geometry. The new feature of this theory comes from eq.$(\ref{flats})$ which can be interpreted as the coupling of a matter field $\sigma$ to the flat three dimensional space.
\newsection{Maxwell gauged WZW model}                      
          
The Maxwell group $G$ contains the Lorentz subgroup $H$ generated by $\{%
\mathsf{J}_{ab}\}$ and the coset $G/H$ generated by $\left\{ \mathsf{P}_{a},%
\mathsf{Z}_{ab}\right\} $. Under gauge transformations, the gauge field
transforms according to eq.~$(\ref{gtc})$.                                                     
Let us now perform a gauge transformation with gauge element $g \in G/H $ given by
\begin{equation}
g=\e^{-\frac{1}{2}\, h^{ab}\, \mathsf{Z}_{ab}}\, \e^{-\phi ^{a}\,
  \mathsf{P}_{a}} \ .
\end{equation}           
According to eq.~$(\ref{gtc})$ with
\begin{equation}      
\bar{\mathcal{A}}=V^{a}\, \mathsf{P}_{a}+\mbox{$\frac{1}{2}$}\, W^{ab}\,
\mathsf{J}_{ab}+ \mbox{$\frac{1}{2}$}\, \Sigma^{ab} \,
\mathsf{Z}_{ab}
\end{equation}                           
it is straightforward to show, using the commutation relations (\ref{maxalg1}-\ref{maxalg2}), that  
\begin{align}                 
V^{a} &=e^a -D_{\omega}\phi^a \ , \nonumber \\[4pt]
W^{ab} &=\omega ^{ab} \ , \nonumber \\[4pt]
\Sigma ^{ab} &=2\phi ^{a}\, e^{b}+\sigma ^{ab}-\phi ^{a}\,
D_\omega\phi ^{b}-D_\omega h^{ab} \ .
\end{align}%
On the other hand, the Chern--Simons action written in terms of the nonlinear connection $\mathcal{\bar{A}}$ is given by
\begin{align}    
 &\mathrm{S}_{\mathsf{CS}}^{(3)}(\mathcal{\bar{A}})=\ \kappa \,
 \int_{\mathcal{M}} \, \bigg(
 \frac{\alpha_{0}}{2}\, W_{~b}^a \wedge \Big( \dd W _{~c}^{b}+\frac{2%
 }{3}\, W_{~d}^{b}\wedge W_{~c}^{d}\Big) +\alpha _{1}\, \epsilon
 _{abc}\, \bar{R}^{ab}\wedge V^{c}  \label{nlmaxcs}\\
 & \qquad\qquad +\alpha _{2}\, \left( \bar{T}^{a}\wedge V_{a}+\bar{R}_{~c}^{a}\wedge \Sigma
 _{~a}^{c}\right) -\dd\Big( \, \frac{\alpha _{1}}{2}\, \epsilon _{abc}\, W
 ^{ab}\wedge V^{c}+\frac{\alpha _{2}}{2}\, W_{~b}^{a}\wedge \Sigma
 _{~a}^{b}\, \Big) \bigg), \nonumber
 \end{align}%
where
\begin{align}    
\bar{R}^{ab} = &\dd W^{ab}+W^{a}_{~c}\wedge W^{cb}
   \\
 \bar{T}^{a}=& \dd V^{a}+W^{a}_{~b}\wedge V^{b}.
 \end{align}%

The final step is to compute the transgression action for the Maxwell algebra. In order to do so note that the boundary contribution eq.(\ref{Q2n}) reads
\begin{align}
Q_{\mathcal{A\leftarrow\bar{A}\leftarrow}0}^{\left(  2\right)}=& \ -\alpha
_{2}\,e^{a}\wedge D_\omega \phi _{a} -\mbox{$\frac{%
\alpha _{1}}{2}$}\, \epsilon _{abc}\, \omega ^{ab}\wedge D_\omega\phi
^{c} \nonumber \\ & + \mbox{$\frac{\alpha _{2}}{2}$}\, \omega
_{~c}^{a}\wedge \left( 2\phi ^{c}\, e_{a}-\phi ^{c}\, D_\omega\phi
_{a}-D_\omega h_{~a}^{c}\right) \ . \label{maxQ2n}
\end{align}
Now, inserting eq.~(\ref{maxcs}, \ref{nlmaxcs}) and eq.~({\ref{maxQ2n}}) into eq.~$(\ref{ch2tranact})$,
The resulting action is a boundary term which corresponds to the gauged
WZW action associated to the Maxwell algebra. As previously, we propose it as a Lagrangian in one less dimension%
\begin{equation}
\mathrm{S}^{(2)}\left[\omega, \phi, h, e \right]=\kappa\,
\int_{\partial \mathcal{M}}\, \big(\alpha _{1}\, \epsilon _{abc}\, R^{ab}\, \phi ^{c}+\alpha
_{2}\,  R_{~c}^{a}\wedge h_{~a}^{c}
\big) \ . \label{wzwmax}
\end{equation}
This action generalizes the topological action for gravity from eq.~$(\ref{topg})$. However, it is interesting to note that both actions are classically equivalent, on-shell.  

The variation of eq.$(\ref{wzwmax})$ gives the following equations of motion
\begin{align}
\alpha_{2}D_{\omega}h_{ab}-\alpha_{1}\epsilon_{abc}D_{\omega}\phi^{c}   &=0 , \label{wzwmax1}\\
R_{ab} &=0 . \label{wzwmax2}
\end{align}
Now, making the following redefinition $\bar{\phi}^c=\alpha_2 \epsilon^{cjk}h_{jk}-2\alpha_1 \phi^c$, it is direct to show that eq.$(\ref{wzwmax1})$ satisfy $\epsilon_{abc}D_{\omega}\bar{\phi}^c=0$, which corresponds, together with eq.$(\ref{wzwmax2})$, to the field equations for the topological gravity theory in the $n=1$ case eq.$(\ref{eomtop1}, \ref{eomtop2})$. Note that this equivalence is only classical. It would be interesting to investigate what are the implications of this model at the quantum level.
\chapter{Covariant Quiver Gauge Theories}
\label{ch:covquiv}
\begin{flushright}
\textit{``I am coming more and more to the conviction that the necessity of our geometry cannot be demonstrated, at least neither by, nor for, the human intellect.'' \\ Carl Friedrich Gauss.}
\bigskip
\end{flushright}
In
this Chapter we pursue the equivariant dimensional
reduction of topological gauge theories. Starting form studying the related problems of generalizing equivariant dimensional reduction to arbitrary
gauge groups $\mathcal{G}$ and extending these techniques to Chern--Simons gauge theories.  Finally the case of non-compact gauge supergroups is explored. In particular, we perform the dimensional reduction of five-dimensional
Chern--Simons supergravity over $\mathbb{C}P^1$. 

\newsection{$SU(2)-$Equivariant principal bundles}

In this section we study gauge theories on the product
space $\mathcal{M}=M\times S^{2}$. Here $M$ is a closed
$d$-dimensional manifold with local coordinates $(x^\mu)_{\mu=1}^d$. On the sphere $S^{2}\simeq%
\mathbb{C}
P^{1}$ we use complex coordinates $(y,\bar y)$ defined by stereographic parametrization. 
We identify $S^{2}$ with the coset space $SU(2)/U(1)$. This induces a transitive action of $SU(2)$ on $S^{2}$ which we extend to the trivial action on $M$.
In order to obtain dimensionally reduced gauge invariant field theories starting from arbitrary gauge groups $\mathcal{G}$, in this section we study $SU(2)$-equivariant principal bundles on
$\mathcal{M}$ and their corresponding $SU(2)$-invariant connections. We follow for a large part the treatment of \cite{Manton:2010mj}.

Every $SU(2)$-equivariant principal bundle
over $S^{2}$ with structure group $\mathcal{G}$ is isomorphic to a quotient
space \cite{JSL}%
\begin{equation}
\mathcal{P}_{\rho}=SU(2)  \times_{\rho}\mathcal{G}\label{eq9}%
\end{equation}
where 
$\rho:U(1)  \rightarrow\mathcal{G}$ is a homomorphism and the elements of $SU(2)  \times_{\rho}\mathcal{G}$ are
equivalence classes $\left[s,g\right]$ on $SU(2)  \times\mathcal{G}$ with respect to the equivalence relation%
\begin{equation}
\left(  s,g\right)  \equiv\big(  s\, s_{0}\,,\,\rho( s_{0})
^{-1}\, g\big) \label{eq10}%
\end{equation}
for all elements $s_{0}\in U(1)\subset SU(2) $. The bundle projection
$\pi:\mathcal{P}_{\rho}\rightarrow S^{2}$ is given by%
\begin{equation}
\pi\left(  \left[  s,g\right]  \right)  =\left[  s\right] \label{eq11}%
\end{equation}
where $\left[  s\right]  $ denotes the left coset $
s\cdot U(1)   $ in $SU(2)  $. Bundles
$\mathcal{P}_{\rho},\mathcal{P}_{\rho'}$ are isomorphic if and only if
the homomorphisms $\rho,\rho'
:U(1)  \rightarrow\mathcal{G}$ take values in the same conjugacy class of $\mathcal{G}$.

Let $P$ be an $SU(2)$-equivariant principal $\CG$-bundle
over $\mathcal{M}=M\times S^{2}$ and select a good open covering $\left\{
U_{i}\right\}  _{i\in I}$ of $M$, i.e. all $U_{i}$ are contractible. Then
the restrictions $\left.  P\right\vert _{U_{i}\times S^{2}}$ are $SU(2)$-equivariant bundles which are trivial on each $U_{i}$, so
that %
\begin{equation}
\left.  P\right\vert _{U_{i}\times S^{2}}\simeq U_{i}\times \mathcal{P}_{\rho_{i}%
}\label{eq12}%
\end{equation}
where the homomorphisms $\rho_{i}:U(1)  \rightarrow
\mathcal{G}$ may be different for every open set $U_{i}\subset M$. However,
on the non-empty intersections $U_{ij}=U_{i}\cap U_{j}$ in $M$, the
restrictions $\left.  P\right\vert _{U_{ij}\times S^{2}}$ are isomorphic
to%
\begin{equation}
U_{ij}\times \mathcal{P}_{\rho_{j}}\simeq\left.  P\right\vert _{U_{ij}\times S^{2}%
}\simeq U_{ij}\times \mathcal{P}_{\rho_{i}} \ .\label{eq13}%
\end{equation}
This means that $\mathcal{P}_{\rho_{j}}\simeq \mathcal{P}_{\rho_{i}}$ and hence $\rho
_{i},\rho_{j}$ take values in the same conjugacy class of $\mathcal{G}$. If $M$ is connected, a
representative homomorphism $\rho$ can be chosen such that%
\begin{equation}
\left.  P\right\vert _{U_{i}\times S^{2}}\simeq U_{i}\times \mathcal{P}_{\rho
}\label{eq14}%
\end{equation}
for all $ i\in I$, and which satisfies%
\begin{equation}
\rho=h_{ij}^{-1}\, \rho \, h_{ij}\label{eq15}%
\end{equation}
for all transition functions $h_{ij}:U_{ij}\rightarrow\mathcal{G}$.
This implies that $h_{ij}$ take values in the centralizer of the image 
$\rho\left(  U(1)  \right)  $ in $\mathcal{G}$, which we denote by
\begin{equation}
\mathcal{H}=Z_{\mathcal{G}}\big(  \rho(  U(1)
)  \big) \ . \label{eq16}%
\end{equation}
Thus the collection of transition functions $\{h_{ij}\}$ for $ i,j\in I$ defines a principal bundle $P_{M}$ over $M$ with structure
group $\mathcal{H}$ which is the residual gauge group after dimensional reduction.

The homomorphism $\rho$ is determined by specifying a unique element $\Lambda
\in\mathfrak{g}$, where $\mathfrak{g}$ is the Lie algebra of $\mathcal{G}$. For this, introduce the Pauli spin matrices%
\begin{equation}%
\sigma_{1}=%
\begin{pmatrix}
0 & 1\\
1 & 0
\end{pmatrix}
\ , \qquad \sigma_{2}=%
\begin{pmatrix}
0 & -\ii\\
\ii & 0
\end{pmatrix}
\ , \qquad \sigma_{3}=%
\begin{pmatrix}
1 & 0\\
0 & -1
\end{pmatrix}
\label{eq17}%
\end{equation}
so that $\mathsf{T}_{a}=-\frac{\ii}{2}\, \sigma_{a}$ for $a=1,2,3$ generate the defining representation of the Lie algebra $\mathfrak{su}(2)  $, where the $U(1)  $ subgroup of
$SU(2)  $ is generated by $T_{3}$. Any element of $U(1)  $ can be written as $\exp(  t\, \mathsf{T}_{3})  $, where $t\in \mathbb{R}$, and the
image of this element under the homomorphism $\rho$ is 
\begin{equation}
\rho\big(  \exp( t\, \mathsf{T}_{3})  \big)  =\exp( t\, \Lambda) \label{eq18}%
\end{equation}
where $\exp( t\, \Lambda )  \in\mathcal{G}$. 
Note that the identity element of $U(1)\subset SU(2)$ corresponds to $t=4\pi$, so that
\begin{equation}
\exp(  4\pi \, \mathsf{T}_{3})  =\unit_{SU(2)} \ ,\label{eq19}%
\end{equation}
and since $\rho$ is a homomorphism 
it follows that $\Lambda$ must satisfy
\begin{equation}
\exp(  4\pi\, \Lambda)  =\unit_{\mathcal{G}} \ . \label{eq20}%
\end{equation}
This leads generally to an algebraic quantization condition on $\rho: U(1)\rightarrow \mathcal{G}$ which we describe explicitly in what follows.

The operations of restriction and induction \cite{AlGar12}
work for principal bundles in the same way as for
vector bundles. Given an $SU(2)$-equivariant principal bundle
$P\rightarrow M\times S^2 $, its restriction $P \vert _{M\times [\unit_{SU(2)}]}$ defines a $U(1)$-equivariant principal bundle on $M$ which is isomorphic to $P_{M}$. The $U(1)$-action on the fibre is defined by the homomorphism $\rho : U(1) \rightarrow \mathcal{G}$ and it extends trivially on the base space $M$. The inverse operation gives $P=SU(2)\times_{\rho}\left.P\right\vert _{M\times[\unit_{SU(2)}]  }$. 

\newsection{$SU(2)-$Invariant connections} 

Let us turn now to the derivation of the $SU(2)-$invariant connection
$\mathcal{A}$ defined on $\mathcal{M}=SU(2)/U(1)\times M$. To this end, we
apply the results obtained in section \ref{sec:invconn} with the special
identifications $G=SU(2)$ and $H=U(1)$. The key point of this construction is that dimensional reduction of gauge theories on $\mathcal{M}$ with $SU(2)-$invariant connection conduce naturally to a gauge theory on $M$.

Recall that the most general element $g\in SU(2)$ can be written in term of
two complex parameters $z$, $w\in%
\mathbb{C}
$%
\begin{equation}
g=%
\begin{pmatrix}
z & -\bar{w}\\
w & \bar{z}%
\end{pmatrix}
\,\ \text{, \ such that }\left\vert z\right\vert ^{2}+\left\vert w\right\vert
^{2}=1 \ .%
\end{equation}
The complex numbers $z$, $w\in%
\mathbb{C}
$ can be written in terms of polar coordinates
\begin{equation}%
\begin{tabular}
[c]{ll}%
$z=\cos\frac{\vartheta}{2}\mathrm{e}^{-\frac{i}{2}\left(  \chi+\varphi\right)
}$, & $w=\sin\frac{\vartheta}{2}\mathrm{e}^{-\frac{i}{2}\left(  \chi
-\varphi\right)  }$ ,
\end{tabular}
\ \label{su2conn1}%
\end{equation}
with $\vartheta\in\left[  0,\pi\right]  $, $\varphi\in\left[  0,2\pi\right]  $
and $\chi\in\left[  0,4\pi\right]  $. Using the generators of the Lie algebra
of $\mathfrak{su}(2)$, and the respective Maurer--Cartan form $\theta
_{{\scriptsize SU(2)}}$%
\begin{align}
\theta_{SU(2)}  &  =g^{-1}\dd g=\mathsf{T}_{a}X^{a} \nonumber \\
&  =-\frac{i}{2}%
\begin{pmatrix}
X^{3} & X^{1}-iX^{2}\\
X^{1}+iX^{2} & -X^{3}%
\end{pmatrix}
\ ,
\end{align}
we obtain the left invariant forms \cite[Section 11.7]{rohtua}
\begin{align}
X^{1}  &  =\sin\chi d\vartheta-\sin\vartheta\cos\chi d\varphi\text{
},\label{su2conn2}\\
X^{2}  &  =\cos\chi d\vartheta+\sin\vartheta\sin\chi d\varphi~,\\
X^{3}  &  =d\chi+\cos\vartheta \dd \varphi~. \label{su2conn3}%
\end{align}
Since the subgroup $U(1)\subset SU(2)$ is generated by $\mathsf{T}_{3}$,
using eq.$\left(  \ref{su2conn1}\right)  $ we can factorize any element $g\in SU(2)$
as follows%
\begin{equation}
g=%
\begin{pmatrix}
\cos\frac{\vartheta}{2}\mathrm{e}^{-\frac{i}{2}\varphi} & -\sin\frac
{\vartheta}{2}\mathrm{e}^{-\frac{i}{2}\varphi}\\
\sin\frac{\vartheta}{2}\mathrm{e}^{\frac{i}{2}\varphi} & \cos\frac{\vartheta
}{2}\mathrm{e}^{\frac{i}{2}\varphi}%
\end{pmatrix}
\mathrm{e}^{\mathsf{T}_{3}\chi}~.
\end{equation}
Furthermore, it is possible to choose representatives of the classes in $SU(2)/U(1)$ in
which $\chi=\varphi$. In that case eq.(\ref{su2conn1}) reads
\begin{equation}
\begin{tabular}
[c]{ll}%
$z=\cos\frac{\vartheta}{2}\mathrm{e}^{-\frac{i}{2}\varphi}~,$ & $w=\sin
\frac{\vartheta}{2}$%
\end{tabular}
\ .
\end{equation}
Introducing real coordinates $x^{1}$, $x^{2}$ and $x^{3}$ we have%
\begin{align}
x^{1}-ix^{2}  &  =2zw \nonumber \\
&  =\sin\vartheta\left(  \cos\varphi-i\sin\varphi\right)  ~,
\end{align}
and%
\begin{equation}
x^{3}=\left\vert z\right\vert ^{2}-w^{2}=\cos\vartheta~.
\end{equation}
This allows us to identify $SU(2)/U(1)$ with $S^{2}$. Note that the point
$\vartheta=\pi$ is the south pole of $S^{2}$. 

Let $N\subset S^{2}$ be an open
region around the south pole of $S^{2}$. In this open set $N$, we can define a
complex coordinate $y$ by stereographic projection from the north pole%
\begin{equation}
y=\frac{x^{1}+ix^{2}}{1-x^{3}}=\frac{\sin\vartheta~\mathrm{e}^{i\varphi}%
}{1-\cos\vartheta}=y\left(  \vartheta,\varphi\right)  ~.
\end{equation}
The stereographic projection parametrizes $N=S^{2}\backslash\left\{
\vartheta=0\right\}  $. Introducing a section $\eta:N\rightarrow SU(2)$ given
by%
\begin{equation}
\eta\left(  y\left(  \vartheta,\varphi\right)  \right)  =%
\begin{pmatrix}
\cos\frac{\vartheta}{2}\mathrm{e}^{-i\varphi} & -\sin\frac{\vartheta}{2}\\
\sin\frac{\vartheta}{2} & \cos\frac{\vartheta}{2}\mathrm{e}^{i\varphi}%
\end{pmatrix} \ ,
\end{equation}
we can pullback under $\eta$ the left-invariants forms on $SU(2)$ eq.$\left(
\ref{su2conn2}-\ref{su2conn3}\right)  $%
\begin{align}
\eta^{\ast}X^{1}  &  =\sin\varphi d\vartheta-\sin\vartheta\cos\varphi
d\varphi \ , \label{etapull1}\\
\eta^{\ast}X^{2}  &  =\cos\varphi d\vartheta+\sin\vartheta\sin\varphi
d\varphi~,\\
\eta^{\ast}X^{3}  &  =\left(  1+\cos\vartheta\right)  d\varphi~. \label{etapull2}
\end{align}
With this information we can write down the explicit expressions for the
$SU(2)-$invariant connection over $\mathcal{M}$. 

Let $P_{M}=\left.
\mathcal{P}\right\vert _{\left\{  y=0\right\}  \times M}$ be an
$U(1)-$equivariant principal bundle over $M$, $\left\{  \sigma
_{i}\right\}  _{i\in I}$ be a set of local sections $\sigma_{i}:U_{i}%
\rightarrow P_{M}$ and assume the homomorphisms $\rho_{i}:U(1)\rightarrow
\mathcal{G}$ agree for any $i\in I$. In that case, we can simply drop the
sub index $i$ and thus $\rho_{i}=\rho$. Consider a local section of
$\mathcal{P}$ defined by%
\begin{align}
\varepsilon_{i}  &  :N\times U_{i}\longrightarrow\mathcal{P} \nonumber \\
&  \left(  y,x\right)  \mapsto\left(  \eta\left(  y\right)  \sigma_{i}\left(
x\right)  \right)  ~. \label{epsi}
\end{align}
In order to find an expression for $\mathcal{A}_{i,N}=\varepsilon_{i}^{\ast
}\mathsf{\omega}$, we use eq.(\ref{pullpsi}). To this end introduce the
map%
\begin{align}
\tilde{\eta}  &  :N\times U_{i}\longrightarrow G\times U_{i}\times
\mathcal{G} \nonumber \\
&  \left(  y,x\right)  \mapsto\left(  \eta\left(  y\right)  ,x,e_{\mathcal{G}%
}\right)  ~, \label{etagorro}
\end{align}
and note that using eq.(\ref{projconn}) we obtain $\varepsilon_{i}=\psi_{i}\circ\tilde{\eta}$. Then,
\begin{align}
\mathcal{A}_{i,N}  &  =\varepsilon_{i}^{\ast}\mathsf{\omega~,} \nonumber \\
&  =\psi_{i}^{\ast}\tilde{\eta}^{\ast}\mathsf{\omega}_{\left(  y,x\right)
}~,\nonumber \\
&  =\left.  \Phi_{i}\left(  x\right)  \left(  \mathsf{T}_{a}\right)
\eta^{\ast}X^{a}\right\vert _{y}+\mu_{i}\left(  x\right)  ~.
\end{align}
Note that the pulled-back connection $\mathcal{A}_{i,N}$ over $N\times U_{i}$ implies that the last term $\theta_{\mathcal{G}}$ in eq.(\ref{pullpsi}) vanishes. 

As it was mentioned in section \ref{sec:invconn}, the $\mu_{i}$ are in fact one-form
connections on $P$. Now, writing $\Phi_{i,a}$ with $a=1,2,3$ for $\Phi
_{i}\left(  \mathsf{T}_{a}\right)  $ and using eq.$\left(  {\ref{psicond}}\right)  $ we
arrive to the following expression%
\begin{equation}
\Phi_{i,3}  =\rho_{3}
  =\rho_{\ast}\left(  \mathsf{T}_{3}\right)  ~.
\end{equation}
This last equation can be regarded as the definition of $\rho_{3}$. Note that
there is no contradiction in having a constant section $\rho_{3}$ in the
associated vector bundle $\mathrm{ad}\left(  P\right)  $ due to the elements
of $\mathcal{H=Z}_{\mathcal{G}}\left(  \rho\left(  U(1)\right)  \right)  $
commute with $\rho_{3}$.

Using eq.$\left( \ref{etapull1}-\ref{etapull2}\right)  $ it is possible to show that
\begin{align}
\mathcal{A}_{i,N}\left(  y,x\right)   &  =\left(  \Phi_{i,1}\left(  x\right)
\sin\varphi+\Phi_{i,2}\left(  x\right)  \cos\varphi\right)  d\vartheta
\nonumber\\
& + \left(-  \Phi_{i,1}\left(  x\right)  \cos\varphi+\Phi_{i,2}\left(
x\right)  \sin\varphi\right)  \sin\vartheta d\varphi\nonumber\\
&  +\rho_{3}\left(  1+\cos\vartheta\right)  d\varphi+\mu_{i}\left(  x\right)
~.\label{su2conn5}%
\end{align}
Moreover, using eq.$\left(  {\ref{psicond2}}\right)  $ one obtains%
\begin{equation}
\Phi_{i}\left(  \mathrm{e}^{-\mathsf{T}_{3}t}\mathsf{T}_{a}\mathrm{e}%
^{\mathsf{T}_{3}t}\right)  =\mathrm{e}^{-\rho_{3}t}\Phi_{i}\left(
\mathsf{T}_{a}\right)  \mathrm{e}^{\rho_{3}t} \label{psiexp}
\end{equation}
where $t\in%
\mathbb{R}
$, $a=1,2,3$. In an infinitesimal form, eq.(\ref{psiexp}) becomes%
\begin{equation}
\Phi_{i}\left(  \left[  \mathsf{T}_{3},\mathsf{T}_{a}\right]  \right)
=\left[  \rho_{3},\Phi_{i}\left(  \mathsf{T}_{a}\right)  \right]
~.\label{su2conn4}%
\end{equation}
From the commutations relations of $\mathfrak{su}(2)  $%
\begin{equation}
\begin{tabular}
[c]{ll}%
$\left[  \mathsf{T}_{3},\mathsf{T}_{1}\right]  =\mathsf{T}_{2}~,$ & $\left[
\mathsf{T}_{3},\mathsf{T}_{2}\right]  =-\mathsf{T}_{1}~,$%
\end{tabular}
\end{equation}
it is direct to show that eq.$\left(  \ref{su2conn4}\right)  $ leads to%
\begin{equation}%
\begin{tabular}
[c]{ll}%
$\Phi_{i,2}=\left[  \rho_{3},\Phi_{i,1}\right]  ~,$ & $-\Phi_{i,1}=\left[
\rho_{3},\Phi_{i,1}\right]  ~.$%
\end{tabular}
\label{su2conn6}%
\end{equation}

It is convenient to write $\mathcal{A}_{i,N}\left(  y,x\right)  $ in a more
compact way. In order to do so we use complex coordinates $y$ along $S^{2}$
and the section $\Phi=-i\Phi_{1}+\Phi_{2}$ of the complexified vector bundle
$\mathrm{ad}\left(  P\right)  ^{%
\mathbb{C}
}=P\times\mathfrak{g}^{%
\mathbb{C}
}/\sim$. If one also assumes that $\mathfrak{g\subset u}(n)$ for suitable
$n\in%
\mathbb{N}
$, the field $\Phi$ becomes $\Phi^{\dagger}=-i\Phi_{1}-\Phi_{2}$. Thus,
the gauge potential eq.$\left(  \ref{su2conn5}\right)  $ can be written as \cite{Popov:2005ik}

\begin{equation}
\mathcal{A}_{i,N}\left(  y,x\right)  =\mathcal{A}_{i,y}\dd y+\mathcal{A}%
_{i,\bar{y}}\dd\bar{y}+\mathcal{A}_{i,\mu}\dd x^{\mu}%
\end{equation}
with
\begin{align}
\mathcal{A}_{i,\mu} &  =A_{i,\mu} \ , \label{eq21}\\[4pt]
\mathcal{A}_{i,y} &  =\frac{-1}{1+y\, \bar{y}}\,\big(  \ii\bar{y}\, \Lambda+\Phi
_{i}\big) \ , \label{eq23}\\[4pt]
\mathcal{A}_{i,\bar{y}} &  =\frac{1}{1+y\,\bar{y}}\, \big(  \ii y\, \Lambda+\Phi
_{i}^{\dagger}\big) \ , \label{eq24}%
\end{align}
Here we have renamed $\mu_{i,\mu}=A_{i,\mu}$ where $\mu=1,...,d$, being $d$ the dimension of $M$, and
$\rho_{3}=\Lambda$. With these identifications, the commutation relations
eq.$\left(  \ref{su2conn6}\right)  $ become \cite{Manton:2010mj}
\begin{align}
\left[  \Lambda,\Phi\right]   &  =-\ii\Phi \ , \label{eq26}\\[4pt]
\big[  \Lambda,\Phi^{\dagger}\big]   &  =\ii\Phi^{\dagger
} \ . \label{eq27}%
\end{align}
where we have omitted the $i$ index since these relations are globally
meaningful: $\Phi$ transforms from $U_{i}$ to $U_{j}$ in the bifundamental
representation of the structure group $\mathcal{H}$ on $\mathfrak{g}^{%
\mathbb{C}
}$, and $\rho_{3}$ is invariant under the adjoint action. One also con write the infinitesimal expression of
eq.$\left(  {\ref{connacond}}\right)  $%
\begin{equation}
\left[  \Lambda,A_{\mu}\right]   =0 \ , \label{eq25}
\end{equation}
which is again globally meaningful since the inhomogeneous part in the
transformation law for $A$ lies in the Lie algebra of $\mathcal{H}$.

Thus, on non-empty
overlaps $U_{ij}\subset M$ these fields obey the relations%
\begin{align}
{A}_{j} &  =h_{ij}^{-1}\, {A}_{i}\, h_{ij}+h_{ij}^{-1}%
\, \dd h_{ij} \ , \label{eq28}\\[4pt]
\Phi_{j} &  =h_{ij}^{-1}\, \Phi_{i}\, h_{ij} \ , \label{eq29}%
\end{align}
where $h_{ij}:U_{ij}\rightarrow\mathcal{H}$ are the transition functions of
$P_{M}$, and ${A}_{i}={A}_{i,\mu}\, \dd x^{\mu}$. The collection of
 local gauge potentials $A_{i}$ defines a connection on $P_{M}$, and
the constraints eq.(\ref{eq25}) imply that $A_{i}$ take values in the Lie algebra $\mathfrak{h}$ of $\mathcal{H}$
which is consistent with $P_{M}$ having $\mathcal{H}$ as structure group. The
collection of local adjoint scalar fields $\Phi_{i}$ define a section of the complexified vector bundle ${\rm ad}(P_M)^{\mathbb{C}}:= P_{M}\times_{\text{ad}}\mathfrak{g}^{\mathbb{C}}$ associated 
to $P_{M}$ by the adjoint representation of $\mathcal{H}$ on $\mathfrak{g}$. In the following we write $A$, $\Phi$ with $\left.  A\right\vert _{U_{i}}=A_{i}$ and $\left. \Phi\right\vert _{U_{i}}%
=\Phi_{i}.$

\subsection{Dimensional reduction of Yang--Mills theory}

We consider as an illustrative example the dimensional reduction of Yang--Mills theories ~\cite{Popov:2005ik,Lechtenfeld:2007st,Lechtenfeld:2008nh,Dolan:2009ie,Dolan:2010ur,Popov2012,Szabo:2014zua}.
On $\mathcal{M}=M\times S^{2}$ the metric
is taken to be the direct
product of a chosen metric $g_{\mu\nu}$ on $M$ and the round metric of the two-sphere, so that
\begin{equation}
\dd s^{2}
=G_{\mu^{\prime}\nu^{\prime}}\,\dd x^{\mu^{\prime}}\otimes \dd x^{\nu^{\prime}}
 =g_{\mu\nu}\, \dd x^{\mu} \otimes \dd x^{\nu}+\frac{4R^{2}}{\left(  1+y\, \bar{y}\right)
^{2}}\, \dd y \otimes \dd\bar{y}\label{ym5}%
\end{equation}
where the indices $\mu^{\prime},\nu^{\prime}$ run over $1, \ldots,  d+2$ and $R$ is the radius of $S^2$. For a principal $\mathcal{G}$-bundle $P \rightarrow \mathcal{M}$ with gauge potential $\alg$, the 
 Yang--Mills Lagrangian is given by
\begin{equation}
\mathscr{L}_{\mathsf{YM}}=-\frac{1}{4g_{\rm YM}^2}\, \sqrt{G}~ \text{\textsf{Tr}}\big( \mathcal{F}_{\mu^{\prime}%
\nu^{\prime}}\, \mathcal{F}^{\mu^{\prime}\nu^{\prime}} \big) \label{nym1}%
\end{equation}
where $\mathcal{F}$ is the curvature two-form
\begin{equation}
\mathcal{F} 
=\dd\mathcal{A+A\wedge A}=\mbox{$\frac{1}{2}$}\, \mathcal{F}_{\mu^{\prime}\nu^{\prime}}\, \dd x^{\mu^{\prime}}\wedge \dd x^{\nu^{\prime}}\label{nym2}
\end{equation}
and $G= \det (G_{\mu^{\prime}\nu^{\prime}} )$. Here $g_{\rm YM}$ is the Yang--Mills coupling constant and $\textsf{Tr}$ denotes a non-degenerate invariant quadratic form on the Lie algebra $\frg$ of the gauge group $\mathcal{G}$, which for $\mathcal{G}$ semisimple is proportional to the Killing--Cartan form.

Expanding eq.(\ref{nym1}) into components along $M$ and $\mathbb{C}P^{1}$ we get%
\begin{equation}
\mathscr{L}_{\mathsf{YM}}=-\frac{1}{4g_{\rm YM}^2}\, \sqrt{G}~\text{\textsf{Tr}}\Big( \mathcal{F}_{\mu\nu
}\, \mathcal{F}^{\mu\nu}+\frac{\left(  1+y\, \bar{y}\right)  ^{2}}{2R^{2}}\, g^{\mu\nu
}\, \big(  \mathcal{F}_{\mu y}\, \mathcal{F}_{\nu\bar{y}}+\mathcal{F}_{\mu\bar{y}%
}\, \mathcal{F}_{\nu y}\big)  +\frac{\left(  1+y\, \bar{y}\right)  ^{4}}{8R^{4}%
}\, \mathcal{F}_{y\bar{y}}\, \mathcal{F}_{\bar{y}y}\Big)
\end{equation}
where from eq.($\ref{eq21}-\ref{eq24}$) we have
\begin{align}
\mathcal{F}_{\mu\nu} &  =F_{\mu\nu} \ , \label{ym22}\\[4pt]
\mathcal{F}_{\mu y} &  =-\frac{1}{1+y\, \bar{y}}\, \nabla_{\mu}\Phi \ , \label{ym23}\\[4pt]
\mathcal{F}_{\mu\bar{y}} &  =\frac{1}{1+y\, \bar{y}}\, \nabla_{\mu}\Phi^{\dagger
} \ , \label{ym24}\\[4pt]
\mathcal{F}_{y\bar{y}} &  =\frac{1}{\left(  1+y\, \bar{y}\right)  ^{2}}\, \big(
2\ii\Lambda-\big[  \Phi,\Phi^{\dagger}\big] \big) \ , \label{ym27}%
\end{align} 
with
\begin{align}
F  & =\dd A+A\wedge A=\mbox{$\frac{1}{2}$}\, F_{\mu\nu}\, \dd x^{\mu}\wedge \dd x^{\nu} \ , \label{curv}\\[4pt]
\nabla\Phi & =\dd\Phi+\left[  A,\Phi\right] = \nabla_\mu\Phi\, \dd x^\mu \label{covd} \ .
\end{align}
Integrating the corresponding Yang--Mills action 
\begin{equation}
S_{\mathsf{YM}}=\int_{\mathcal{M}}\, \dd^{d+2}x\ \sqrt{G}\ L_{\mathrm{YM}}
\end{equation}
over $S^2 \simeq \mathbb{C}P^{1}$ using
\begin{equation}
\int_{\mathbb{C}P^{1}}\, \frac{R^{2}}{\left(  1+y\, \bar{y}\right)
  ^{2}}\, \dd y \wedge \dd\bar{y}=4\pi \, R^{2} \ ,
\end{equation}
we get the action
\begin{align}
S_{\mathsf{YMH}}=\frac{\pi\, R^{2}}{g_{\rm YM}^{2}}\, \int_{M}\, \dd^{d}%
x~\sqrt{g}\ \text{\textsf{Tr}}\Big( F_{\mu\nu}\, \left(  F^{\mu\nu}\right)  ^{\dagger
}& +\frac{1}{2R^{2}}\, \big(  \nabla_{\mu}\Phi \, \nabla^{\mu}\Phi^{\dagger}+\nabla_{\mu}%
\Phi^{\dagger}\, \nabla^{\mu}\Phi\big) \nonumber \\ 
& + \frac{1}{8R^{4}}\, \big(  2\ii\Lambda-\big[  \Phi,\Phi^{\dagger}\big]
\big)  ^{2}\Big) \label{ymredact}
\end{align}
which describes a Yang--Mills--Higgs theory on $M$ with gauge group $\mathcal{H}$ \cite{ForgaCS:1979zs,Dolan:2009ie,Manton:2010mj}.

\section{Principal quiver bundles} \label{sector3}

In order to solve the constraint equations eq.($\ref{eq26}-\ref{eq25}$) explicitly, it is necessary to fix the element
$\Lambda\in\frg$ and therefore the gauge group $\mathcal{G}$. In this section we consider the case where $\CG$ is one of the classical Lie groups $U(n)$, $SO(2n)$,
$SO(2n+1)$, or $Sp(2n)$. In this case equivariant dimensional reduction gives principal $\mathcal{H}$-bundles $P_{M} \rightarrow M$ which can be characterized in terms of quivers, and eq.(\ref{ymredact}) becomes an action for a quiver gauge theory on $M$.

In
the Cartan--Weyl basis, the generators of the gauge group $\mathcal{G}$ satisfy the
commutation relations%
\begin{align}
\left[  H_{i},H_{j}\right]   &  =0 \ , \label{cwb1}\\[4pt]
\left[  H_{i},X_{\alpha}\right]   &  =\alpha_{i}\, X_{\alpha} \ , \label{cwb2}\\[4pt]
\left[  X_{\alpha},X_{\beta}\right]   &  =\left\{
\begin{array}
[c]{l}%
N_{\alpha,\beta}\, X_{\alpha+\beta} \quad \text{ ~~if }\alpha+\beta\text{ is a root} \ , \\
0 \quad \text{ ~~~~~~~~~~~~~~~otherwise} \ ,
\end{array}
\right.  \text{ }\label{cwb3}\\[4pt]
\left[  X_{\alpha},X_{-\alpha}\right]   &  =\frac{2}{\left\vert \alpha
\right\vert ^{2}}\, \sum\limits_{i=1}^{n}\,\alpha_{i}\, H_{i}\label{cwb4} \ ,
\end{align}
where $n$ is the rank of $\mathcal{G}$, the subset $\left\{  {H}_{i}\right\}_{i=1}^n  $ generates the Cartan subalgebra
$\mathfrak{t\subset g}$, the vectors $\alpha$ are the roots of the Lie algebra $\mathfrak{g}$ of $\mathcal{G}$, and $\left\{
X_{\alpha}\right\}  $ are the root vectors with normalization constants $N_{\alpha,\beta}$. By gauge invariance, the element $\Lambda
\in\mathfrak{g}$ can be conjugated into the Cartan subalgebra
generated by $\left\{H_{i}\right\}  $. Then there is still a residual gauge symmetry under the discrete Weyl subgroup $\Wcal\subset\mathcal{G}$ which acts by permuting the eigenvalues $\lambda_{i}$, $i=1,\dots, n$ of $\Lambda$. We can use this symmetry to group $\lambda_{i}$ into $m+1$ degenerate blocks, $0\leq m\leq n-1$, of dimensions $k_{\ell}$ such that $\lambda_{k_0+ k_{1}+\cdots+k_{\ell-1}+1}=\cdots=\lambda_{k_0+ k_{1}+\cdots+k_{\ell-1}+k_{\ell}}=:\alpha_{\ell}$ for $\ell=0, 1,\ldots,m$, where $k_{-1}:=0$ and
\begin{equation}
\sum\limits_{\ell=0}^{m}\, k_{\ell}=n \ . \label{cb1}%
\end{equation}
Then the element $\Lambda$ can be expanded as
\begin{equation}
\Lambda=\ii\sum\limits_{\ell=0}^{m}\, \alpha_{\ell} \ \sum\limits_{i=1}^{k_{\ell}} \,
H_{k_{1}+\cdots+k_{\ell-1}+i} \ . \label{cwb5}%
\end{equation}
Similarly, the Higgs fields $\Phi$ and the gauge field $A$ can both be
expanded in the Cartan--Weyl basis as
\begin{align}
\Phi&=\sum\limits_{i=1}^{n}\, \phi_{i}\, H_{i}+\sum\limits_{\alpha>0}\, \big(
\phi_{\alpha}\, X_{\alpha}+\phi_{-\alpha}\, X_{-\alpha}\, \big) \ ,
\label{cw6} \\[4pt]
A&=\sum\limits_{i=1}^{n}\, A_{i}\, H_{i}+\sum\limits_{\alpha>0}\, \big(
A_{\alpha}\, X_{\alpha}+A_{-\alpha}\, X_{-\alpha}\big) \ .
\label{cb7}%
\end{align}

Let us first consider the unitary gauge group $\mathcal{G}=U(n)
$. Since $\Lambda\in\mathfrak{u}(n)$, it may be represented by a
Hermitian $n\times n$ matrix which can always be taken to be diagonal
by conjugation with a suitable element $g \in U(n)  $. The roots and
the forms of the generators in the Cartan--Weyl basis are summarized in appendix \ref{ch:app4}.

Using
\begin{equation}
\left[  H_{i},X_{e_{j}-e_{k}}\right]  =\left(
\delta_{ji}-\delta_{ki}\right) \, X_{e_{j}-e_{k}}\label{u3}
\end{equation}
the invariance constraints eq.(\ref{eq26}) and eq.(\ref{eq27}) yield
\begin{equation}
\phi_{i}=0 \ , \qquad \phi_{jk}\, \left(
  \lambda_{j}-\lambda_{k}+1\right)    =0 = \phi_{kj}\, \left(
  \lambda_{k}-\lambda_{j}+1\right) \ . \label{u7} 
\end{equation}
To allow for non-trivial solutions, it is necessary to require
$\lambda_{k}-\lambda_{j}=\pm\, 1$. Using Weyl symmetry to restrict attention to
$
\lambda_{j}-\lambda_{k}=-1\label{er44}%
$
with $\lambda_{j}\neq\lambda_{k}\neq0$, we find $\phi_{kj}=0$ while $\phi_{jk} $ can be non-vanishing.
However, not all of the fields $\phi_{jk}$ are non-zero.
The only non-vanishing components arise when $j$ and $k
$ belong to neighbouring blocks of indices. 
If $j,k$ belong to the same block $K_{\left(\ell\right)  }:=\{{k_0+ k_{1}+\cdots +k_{\ell-1}+i }\}_{i=1}^{k_{\ell}}$, then $\lambda_{j}=\lambda_{k}=\alpha_{\ell}$ and so $\phi_{jk}=0$ by eq.(\ref{u7}).
On the other hand, if  $j\in K_{\left(  \ell\right)  }$ and  $k\in K_{\left(  \ell+1\right)  }$, then $\lambda_{j}=\alpha_{\ell}$ and $\lambda_{k}=\alpha_{\ell+1}$, and by eq.(\ref{u7}) if $\phi_{jk}\neq0$ then $\alpha_{\ell}-\alpha_{\ell+1}=-1$, so we have 
$
\alpha_{\ell}  =\alpha+\ell
$
for $\ell=0, 1, \ldots,m$ and $\alpha :=\alpha_{0}$. Therefore the Higgs field eq.(\ref{cw6}) has the form
\begin{equation}
\Phi=\sum\limits_{\ell=0}^{m}\, \phi_{(\ell+1)}\label{u19}%
\end{equation}
where%
\begin{equation}
\phi_{(\ell+1)}=\sum\limits_{\stackrel{\scriptstyle j\in K_{\left(  \ell\right)  }\,,\, k\in K_{\left(
\ell+1\right)  }}{\scriptstyle j<k}} \, \phi_{jk} \, X_{e_{j}-e_{k}}
\end{equation}
with $\phi_{(m+1)}:=0$.

The constraint equation eq.(\ref{eq25}) gives
\begin{equation}
A_{jk}\, \left(  \lambda_{j}-\lambda_{k}\right)  =0=
A_{kj}\, \left(  \lambda_{k}-\lambda_{j}\right) \ . \label{u22}%
\end{equation}
Here non-trivial solutions occur when $\lambda_{k}=\lambda_{j}$. This
happens when $j,k$ belong to the same block $K_{\left(  \ell\right)  }$ and thus%
\begin{equation}
A=\sum\limits_{\ell=0}^{m}\, A_{(\ell)}\label{u23}%
\end{equation}
where%
\begin{equation}
A_{(\ell)}=\sum\limits_{i\in K_{\left(  \ell\right)  }}\, A_{i}\, H_{i}%
+\sum\limits_{\stackrel{\scriptstyle j,k\in K_{\left(  \ell\right)  }}{\scriptstyle j<k}}\, \left(
A_{jk}\, X_{e_{j}-e_{k}}+A_{kj}\, X_{e_{k}-e_{j}}\right) \ .
\end{equation}
This calculation also shows that the breaking of the original $U(n)$ gauge symmetry to the centralizer subgroup eq.(\ref{eq16}) is given by
\begin{equation}
\mathcal{H}=\prod\limits_{\ell=0}^{m}\, U(k_{\ell}) \ . \label{u24}%
\end{equation}

The $\mathfrak{u}(n)$-valued gauge potential $\mathcal{A}$ on $\mathcal{M}$ is by construction $SU(2)$-invariant and 
decomposes into $k_{\ell}\times k_{\ell'}$ blocks $\mathcal{A}^{\ell,\ell'}$
with $\ell,\ell'=0, 1, \ldots ,m$ and
\begin{align}
\mathcal{A}^{\ell,\ell} &  =A_{(\ell)}-\a_{(\ell)} \ , \label{ansa1}\\[4pt]
\mathcal{A}^{\ell, \ell+1} &  =-\phi_{(\ell+1)}\, \beta \ , \label{ansa2}\\[4pt]
\mathcal{A}^{\ell+1, \ell} &  =-\big(  \mathcal{A}^{\ell, \ell+1}\big)  ^{\dagger}%
=\phi_{(\ell+1)}^{\dagger}\, \bar{\beta} \ , \\[4pt]
\mathcal{A}^{\ell+i, \ell} &  =0=\mathcal{A}^{\ell, \ell+i} \qquad \text{for} \quad i\geq2 \ . \label{ansa4}
\end{align}
Here the local one-forms $\a_{(\ell)}$ on $\C P^1$ are given by
\begin{equation}
\a_{(\ell)}=-\frac{\alpha_{\ell}\, \left(  \bar{y}\, \dd y-y\, \dd\bar{y}\right) }{ 1+y\, \bar
{y}} \ , \label{magm}
\end{equation} 
and 
\begin{equation}
\beta=\frac{\dd y}{1+y\, \bar{y}} \ , \qquad \bar{\beta}=\frac
{\dd\bar{y}}{1+y\, \bar{y}} \ , \label{cdf}
\end{equation}
are the unique covariantly constant $SU(2)$-invariant $(1,0)$- and $(0,1)$-forms on $\mathbb{C}P^1$ respectively. From eq.($\ref{ansa1}-\ref{ansa4}$) it follows that the curvature two-form splits into $k_{\ell}\times k_{\ell'}$ blocks 
\begin{equation}
\mathcal{F}^{\ell,\ell'}=\dd\mathcal{A}^{\ell,\ell'}+\sum\limits_{\ell''=0}^{m}\, \mathcal{A}%
^{\ell,\ell''} \wedge \mathcal{A}^{\ell'',\ell'}\label{lolazo}
\end{equation} 
and its only non-vanishing components are
\begin{align}
\mathcal{F}^{\ell,\ell} &  =F_{(\ell)}-\sff_{(\ell)}+\big(  \phi_{(\ell)}^{\dagger}\, \phi_{(\ell)}-\phi
_{(\ell+1)}\, \phi_{(\ell+1)}^{\dagger}\big) \, \beta \wedge \bar{\beta} \ , \nonumber \\[4pt]
\mathcal{F}^{\ell,\ell+1} &  =-\nabla\phi_{(\ell+1)}\wedge \beta \ , \nonumber \\[4pt]
\mathcal{F}^{\ell+1,\ell} &  =\nabla\phi_{(\ell+1)}^{\dagger}\wedge
\bar{\beta} \ ,
\end{align}
where
\begin{align}
\sff_{(\ell)} &  =2\alpha_\ell\, \beta \wedge \bar{\beta} \ , \nonumber \\[4pt]
F_{(\ell)}& = \dd A_{(\ell)}+A_{(\ell)}\wedge A_{(\ell)} \ , \nonumber \\[4pt]
 \nabla\phi_{(\ell+1)} &  =\dd\phi_{(\ell+1)}+A_{(\ell)}\, \phi_{(\ell+1)}-\phi_{(\ell+1)}\, A_{(\ell+1)} \ , \nonumber \\[4pt]
\nabla\phi_{(\ell+1)}^{\dagger} &  =\dd\phi_{(\ell+1)}^{\dagger}+A_{(\ell+1)}\, \phi_{(\ell+1)}^{\dagger
}-\phi_{(\ell+1)}^{\dagger}\, A_{(\ell)}
\end{align}
with $\phi_{(0)}:=0=:\phi_{(m+1)}$.

The eigenvalues of the matrix $\Lambda$ from eq.(\ref{cwb5}) are constrained by eq.(\ref{eq20}) to quantized values $\alpha_{\ell} \in \frac12\, \mathbb{Z}$ given by
\begin{equation}
\alpha_{\ell}=\frac{p+2\ell}{2}
\end{equation}
for arbitrary $p\in\mathbb{Z}$. 
It follows that the matrix $\Lambda$ geometrically parameterizes the quantized magnetic charges of the unique $SU(2)$-invariant family of monopole connections $\a_{(\ell)}$ on $\C P^1$. With $p=-m$ the Yang--Mills--Higgs model eq.(\ref{ymredact}) reproduces the quiver gauge theories from~\cite{Popov:2005ik} which are based on the linear $A_{m}$ quivers
\begin{equation}
\xymatrix@R=0.28pc{
\bullet \ar[r] & \bullet \ar[r]&\bullet \ \cdots \ \bullet \ar[r]& \bullet 
} \label{quiver}
\end{equation}
containing $m+1$ nodes corresponding to the gauge groups $U(k_\ell)$ and gauge fields $A_{(\ell)}$, and $m$ arrows corresponding to the $U(k_{\ell+1})\times U(k_\ell)$ bifundamental Higgs fields $\phi_{(\ell+1)}$. The quiver (\ref{quiver}) characterizes how $SU(2)$-invariance is incorporated into the gauge theory on $\CM=M\times S^2$.

Note that this correspondence with quivers is somewhat symbolic, as an
$SU(2)$-equivariant principal $\CG$-bundle does not belong to a
representation category for the quiver (\ref{quiver}). The association
is possible because in the present case the gauge group $\CG$ is a
matrix Lie group: One may regard $U(k_\ell)$ as the group of unitary
automorphisms of a complex inner product space $V_{k_\ell}\simeq
\C^{k_\ell}$ and the Higgs fields $\phi_{(\ell+1)}$ fibrewise as maps
in $\Hom(V_{k_{\ell+1}}, V_{k_\ell})$. To associate a quiver bundle to
our construction we need a suitable representation of the quiver
(\ref{quiver}) in the category of vector bundles on $M$. For this, we
can take the complex vector bundle $E=P\times_\varrho V$ on $\CM$
associated to the fundamental representation $\varrho: \CG\to U(V)$ of
$\CG=U(n)$ on $V\simeq \C^n$. Then the restriction $E_M:= E|_{M\times[\unit_{SU(2)}]}=P_M\times_\varrho V|_\CH$ is a $U(1)$-equivariant vector bundle on $M$ with fibre the restriction $V|_\CH=\bigoplus_{\ell=0}^m\, V_{k_\ell}$ of the linear representation $(\varrho,V)$ to $\CH$. The $U(1)$-action on the fibre is given by $\exp(t\, \Lambda)|_{V_{k_\ell}}=\e^{\ii t\, (\frac p2+\ell)}\, \unit_{V_{k_\ell}}$ and the Higgs fields are morphisms $\Phi|_{E_{k_{\ell+1}}}:E_{k_{\ell+1}} \to E_{k_\ell}$ of the vector bundles $E_{k_\ell}:=P_M\times_\varrho V_{k_\ell}$ for each $\ell=0,1,\dots,m$.

Our detailed treatment here of the standard case with $\CG=U(n)$ has
the virtue that the exact same analysis can be performed for the
remaining classical gauge groups $\mathcal{G}=SO(2n)$,  $SO(2n+1)$,
and $Sp(2n)$; the requisite group theory data for their decompositions
in the Cartan--Weyl basis are summarised in appendix~\ref{ch:app4}. In
every case one shows that, for generic eigenvalues $\alpha_\ell$ of
the matrix $\Lambda$, the residual gauge symmetry group is again given
by eq.(\ref{u24}) (as a subgroup of $\CG$) and the structure of the
dimensionally reduced gauge theory can again be encoded in the $A_m$
quiver (\ref{quiver}), with only trivial redefinitions of the coupling
constants in eq.(\ref{ymredact}) distinguishing the different cases. Such
redefinitions may have implications in matching the quiver gauge
theories with more realistic models as in~\cite{Dolan:2009ie}.

\newsection{Topological Chern--Simons--Higgs models} \label{secondresult}

We will now perform the $SU(2)$-equivariant dimensional
reduction of the Chern--Simons gauge theory on $\mathcal{M}=M \times
S^2$, where $M$ is an oriented manifold of dimension $d=2n-1$. Throughout we assume that the manifold $M$ is closed, as no novel boundary effects arise in the models we derive.
The gauge field defined by eq.($\ref{eq21}-\ref{eq24}$) can be written in the form
\begin{equation}
\mathcal{A}=A-\a-\Phi\otimes \beta+\Phi^{\dagger
}\otimes \bar{\beta} \ ,
\end{equation}
where
\begin{equation}
\a:=\Lambda\otimes\frac{\ii\left(  \bar{y}\, \dd y-y\, \dd\bar{y}\right) }{ 1+y\, \bar
{y}} 
\end{equation}
and we have used eq.(\ref{cdf}). 
In general, the computation of the reduced Chern--Simons action directly from its definition is somewhat involved; to simplify the calculations considerably we use the subspace separation method~\cite{Izaurieta:2006wv} introduced in section (\ref{SSMM})

For the present case we decompose $\frg=\frg_0\oplus\frg_1$ with $\frg_0=\frh$ and $\frg_1= \frg\ominus\frh$, and expand the gauge potential as
\begin{align}
\mathcal{A}_{0} &  =0 \ ,\label{ssmd1}\\[4pt]
\mathcal{A}_{1} &  =-\a \ ,\label{ssmd2}\\[4pt]
\mathcal{A}_{2} &  =A-\a \ ,\label{ssmd3}\\[4pt]
\mathcal{A}_{3} &  =\Phi^{\dagger}\otimes \bar\beta-\Phi\otimes
\beta+A-\a \ . \label{ssmd4}%
\end{align}
 By applying the triangle equation (\ref{ch2treq}) with $\mathcal{\bar{A}}=0$, we obtain the 
expression for the reduced Chern--Simons action: The reduced Lagrangian splits into the sum of three terms%
\begin{align}
\mathscr{L}_{\Phi} &  =\kappa T_{\mathcal{A}_{3}\leftarrow\mathcal{A}_{2}}^{\left(
2n+1\right)  }=2\kappa \left(  n+1\right)  \int_{0}^{1}\dd t\
\left\langle t \big(  \Phi
 \nabla\Phi^{\dagger}-\Phi^{\dagger} \nabla\Phi\big) \wedge \beta \wedge\bar{\beta
} \wedge F^{n-1}\right\rangle \ , \nonumber \\[4pt]
\mathscr{L}_{A} &  =\kappa  T_{\mathcal{A}_{2}\leftarrow\mathcal{A}_{1}}^{\left(
2n+1\right)  }=2\kappa \left(  n+1\right)  \int_{0}^{1} \dd t\
\left\langle 2\ii \Lambda \beta
\wedge \bar{\beta} \wedge
A\wedge\big(  t\dd A+t^{2} A\wedge A\big)  ^{n-1}\right\rangle \ , \nonumber \\[4pt]
\mathscr{L}_{\Lambda} &  =\kappa T_{\mathcal{A}_{1}\leftarrow\mathcal{A}_{0}}^{\left(
2n+1\right)  }=0 \ .
\end{align}
By integrating over $S^2$, the original $(2n+1)$-dimensional Chern--Simons gauge theory reduces to a Chern--Simons--Higgs type model in $d=2n-1$ dimensions with action
\begin{equation}
{S}_{\mathsf{CSH}}^{\left(  2n-1\right)  }=  \kappa^{\prime} 
\int_{M} \int_{0}^{1} \dd t\ \left\langle t \big(  \Phi \nabla\Phi^{\dagger} 
-\Phi^{\dagger} \nabla\Phi\big)  \wedge F^{n-1} +2\ii\Lambda A\wedge \big(
t\dd A+t^{2} A\wedge A\big)  ^{n-1}\right\rangle \label{CSredact}%
\end{equation}
subject to the constraints eq.($\ref{eq26}-\ref{eq25} $).
Here we have defined $\kappa^{\prime}=8\pi R^{2}\left(  n+1\right) \kappa$ and the fields $F$, $\nabla\Phi$ are given by
eq.($\ref{curv}-\ref{covd}$) respectively. 

This action is ``topological"
in the sense that it is diffeomorphism invariant; this point is
actually somewhat subtle and we return to it below. The first term of eq.(\ref{CSredact}) is also manifestly invariant under the gauge transformations
\eq
A^h = h^{-1}\, A\, h+h^{-1}\, \dd h \ , \qquad \Phi^h= h^{-1}\, \Phi\, h
\label{APhigaugetransf}\eqend
for $h\in\Omega^0(M,\CH)$, but the second Chern--Simons type term is
generically not: Using~\cite[eq.~(3.5)]{Chamseddine:1990gk} one finds
that the full action transforms as
\eq 
{S}_{\mathsf{CSH}}^{\left(  2n-1\right)  }\big[A^h,\Phi^h \big] = {S}_{\mathsf{CSH}}^{\left(  2n-1\right)  }[A,\Phi] -2\ii (-1)^n \, \frac{(n-1)!\, n!}{(2n-1)!}\, \kappa' \, \int_M\, \Big\langle \Lambda\, \big(h^{-1}\, \dd h\big)^{2n-1} \Big\rangle \ .
\label{CStransf}\eqend
Due
to the constraint eq.(\ref{eq20}), the closed $(2n-1)$-form $\big\langle \Lambda\, (h^{-1}\, \dd h)^{2n-1}\big\rangle$ gives a de~Rham representative for a class in the cohomology
group $H^{2n-1}(M,\pi_{2n-1}(\CH))$. Hence the deficit term in
eq.(\ref{CStransf}) generically vanishes if and only if the free part of
the homotopy group
$\pi_{2n-1}(\CH)$ is trivial. Otherwise, the path integral
for the quantum field theory
is well-defined provided that the functional $\exp\big(\ii {S}_{\mathsf{CSH}}^{\left(
    2n-1\right)  }\big)$ is invariant under gauge transformations; this requirement generically imposes a
further topological quantization condition on the effective coupling constant
$\kappa'$ after dimensional reduction if the group
$\pi_{2n-1}(\CH)/\, {\rm Tor}(\pi_{2n-1}(\CH))$ is non-trivial.
Then up to a gauge transformation with parameter $\lambda=\xi\, \lrcorner\, A$, the infinitesimal action of diffeomorphisms of $M$ can be represented as contractions
\eq 
\delta_\xi A = \xi \, \lrcorner\, F \ , \qquad \delta_\xi \Phi=\xi \, \lrcorner\, \nabla\Phi
\label{diffgauge}\eqend
along vector fields $\xi \in \Omega^0(M,T(M))$.

The field equations can be obtained by varying the reduced action eq.(\ref{CSredact}) or equivalently by dimensional reduction over the general
condition%
\begin{equation}
\delta {S_{\mathsf{CS}}^{(2n+1)}}=\kappa\, \int_{\mathcal{M}}\,\big\langle \mathcal{F}^{n}\wedge
\delta\mathcal{A}\big\rangle =0\label{CS59}%
\end{equation}
on $\mathcal{M}= M\times S^2$. One finds that the equations of motion reduce to 
\begin{align}
\Big\langle \Big(  F^{n-1}\, \big(  2\ii\Lambda-\big[  \Phi,\Phi^{\dagger
}\big]  \big)  +\left(  n-1\right) \, F^{n-2}\wedge \nabla\Phi^{\dagger}\wedge \nabla\Phi\Big)\wedge
\delta A\Big\rangle  &  =0 \ , \nonumber \\[4pt]
\big\langle F^{n-1} \wedge \nabla\Phi^{\dagger} \ \delta\Phi\big\rangle  &  =0 \ , \nonumber \\[4pt]
\big\langle F^{n-1} \wedge \nabla\Phi \ \delta\Phi^{\dagger}\big\rangle  &  =0 \label{varphidag} \ ,
\end{align}
subject to the linear constraints eq.($\ref{eq26}-\ref{eq25}$).
In the following we will study various aspects of the moduli space
$\CCM_n$ of
solutions to these equations modulo gauge transformations and
diffeomorphisms. As a special class of topological solutions, note
that the Higgs fields $\Phi$ are (locally) parallel
sections of the adjoint bundle ${\rm ad}(P_M)$ if and only if
the curvature two-form $F$ of $P_M$ vanishes, in which case the field
equations are immediately satisfied when $n>1$. Since in this case the diffeomorphisms
eq.(\ref{diffgauge}) vanish on-shell, this subspace of the solution space
is the finite-dimensional moduli space of flat $\CH$-connections on $M$ modulo gauge transformations, or equivalently the moduli space of representations of the fundamental group $\pi_1(M)$ in $\CH$ modulo conjugation.

\subsection{Moduli spaces of solutions}

For some explicit examples, let us look at the case where $\CG$ is one of the classical gauge groups from section \ref{sector3}, focusing without loss of generality on $\CG=U(n)$. 
The dynamics of the reduced topological quiver gauge theory is then controlled by the invariant tensor associated to the residual gauge group eq.(\ref{u24}).
In general, if $\{\mathsf{t}_a\}^{\dim \mathfrak{h}}_{a=1}$ denotes the generators of the Lie algebra $\mathfrak{h}$ of $\mathcal{H}$, then the invariant tensor $g_{a_{1} \cdots a_{n+1}}$ is a symmetric tensor of rank $n+1$ that is invariant under the adjoint action of $\mathcal{H}$ which we take to be the symmetrized trace~\cite{deAzcarraga:1997ya}
\begin{equation}
g_{a_{1} \cdots a_{n+1}}=\left\langle \mathsf{t}_{a_{1}} \cdots
  \mathsf{t}_{a_{n+1}} \right\rangle = \frac{1}{(n+1)!} \
\sum\limits_{\sigma\in S_{n+1}}\, \mathsf{Tr}\big(  \mathsf{t}_{a_{\sigma\left(
1\right)}  }\cdots\mathsf{t}_{a_{\sigma\left(  n+1\right) } }\big)
\label{symmtrace}\end{equation}
where $S_{n+1}$ is the symmetric group of degree $n+1$. In the
Cartan--Weyl basis the reduced gauge group $\mathcal{H}$ of
eq.(\ref{u24}) is generated by $\{H_{i},
X_{e_{j}-e_{k}}\}_{i,j,k=1}^n$. Let us now examine in detail some cases in lower dimensionalities.

\subsubsection*{$\mbf{d=1}$}

The non-zero components of the invariant tensor for $d=1$ coincide
with the Killing--Cartan form
\begin{align}
\left\langle X_{e_{j}-e_{k}}\, X_{e_{l}-e_{m}}\right\rangle  & =%
\delta_{jm}\, \delta_{kl} \ , \nonumber \\[4pt]
\left\langle H_{i}\, X_{e_{j}-e_{k}}\right\rangle   &=\delta
_{ik}\, \delta_{ij} \ , \nonumber \\[4pt]
\left\langle H_{i}\, H_{j}\right\rangle   &=\delta_{ij} \ ,
\end{align}
and the resulting action functional is that of a topological matrix quantum mechanics given by 
\begin{equation}
S^{(1)}_{\mathsf{CSH}}=8\pi\, R^{2}\, \kappa\, \int\, \dd\tau \ \sum\limits_{\ell=0}^{m}\,
\mathsf{Tr}\big(  \phi_{(\ell+1)}\, \nabla_\tau\phi_{(\ell+1)}^{\dagger}-\phi
_{(\ell)}^{\dagger}\, \nabla_\tau\phi_{(\ell)}-2\alpha_\ell\, A_{(\ell)}\big)  \label{trivact}
\end{equation}
where $\nabla_\tau\phi_{(\ell)}=\dot{\phi}{}_{(\ell)} +A_{(\ell-1)}\,
\phi_{(\ell)}-\phi_{(\ell)}\, A_{(\ell)}$.
In this case the gauge potentials $A_{(\ell)}(\tau)\in\frh$ are Lagrange
multipliers and integrating them out of the action eq.(\ref{trivact}) yields the constraints
\eq 
\mu_{V}^{(\ell)}(\Phi):= \phi_{(\ell+1)}\, \phi_{(\ell+1)}^{\dagger}-\phi
_{(\ell)}^{\dagger}\, \phi_{(\ell)} = 2\alpha_\ell \, \unit_{k_\ell} \ ,
\label{momentmap}\eqend
while the remaining equations of motion for the Higgs fields read $\dot{\phi}{}_{(\ell)}=0= \dot{\phi}{}_{(\ell)}^\dag$ for $\ell=0,1,\dots,m$. 

Thus in this case moduli space $\CCM_1$ of classical solutions is finite-dimensional and can be described as the subvariety cut out by the quadric eq.(\ref{momentmap}) in the quotient of the affine variety $\prod_{\ell=0}^{m}\, \Hom(\C^{k_{\ell+1}}, \C^{k_\ell})$ by the natural action of the gauge group eq.(\ref{u24}) given by $\phi_{(\ell+1)}\mapsto g_{\ell+1}\, \phi_{(\ell+1)}\, g_{\ell}^\dag$ with $g_\ell\in U(k_\ell)$. The moduli space $\CCM_1$ also has a representation theoretic description as an affine quiver variety in the following way. The vector space of linear representations of the $A_m$ quiver (\ref{quiver}) with fixed $V|_\CH=\bigoplus_{\ell=0}^m\, V_{k_\ell}$ is
\beq
\CCR_m(V) = \bigoplus_{\ell=0}^m\, \Hom(V_{k_{\ell+1}}, V_{k_\ell}) \ .
\eeq
The corresponding representation space for the opposite quiver, obtained by reversing the directions of all arrows, is the dual vector space $\CCR_m(V)^*$ and the cotangent bundle on $\CCR_m(V)$ is
\beq
T^*\CCR_m(V)=\CCR_m(V)\oplus \CCR_m(V)^* \ .
\eeq
It carries a canonical $\CH$-invariant symplectic structure such that
the linear $\CH$-action on $T^*\CCR_m(V)$ is
Hamiltonian~\cite{Popov:2010rf} and the corresponding moment map is
given by $\mu_V=\big(\mu_V^{(\ell)}\big)_{\ell=0}^m:T^*\CCR_m(V)\to
\frh^*$. The moduli space is then the symplectic quotient
\eq
\CCM_1= \mu_V^{-1}(2\alpha_0,2\alpha_1,\ldots,2\alpha_m)\,
\big/\!\!\big/\, \CH \ .
\label{quivervariety}\eqend
This moduli space parameterizes isomorphism classes of semisimple
representations of a certain preprojective
algebra deformed by the eigenvalues $\alpha_\ell$~\cite{Popov:2010rf}.

The topological nature of the quiver gauge theory in this instance is
not surprising as the original pure three-dimensional
Chern--Simons theory with Lagrangian
\eq
\mathscr{L}_{\sf CS}^{(3)} = \big\langle
\alg\wedge\dd\alg+\mbox{$\frac13$}\, \alg\wedge \alg\wedge\alg 
\big\rangle
\eqend
is a topological gauge theory, and hence so is its dimensional
reduction. In this setting the affine quiver variety
(\ref{quivervariety}) is described geometrically as the
finite-dimensional moduli space of flat $SU(2)$-invariant
$\CG$-connections on the three-manifold $\CM$, which can be regarded as the symplectic quotient of the space of all $SU(2)$-invariant $\CG$-connections on $\CM$ by the action of the group of gauge transformations $\Omega^0(\CM,\CH)$.

\subsubsection*{$\mbf{d=3}$} \label{d3feq}

The Chern--Simons--Higgs like system in the case $d=3$ is the three-dimensional diffeomorphism-invariant gauge theory reduced from pure $U(n)$ Chern--Simons theory in five dimensions which has Lagrangian
\eq 
\mathscr{L}_{\sf CS}^{(5)} = \big\langle\alg\wedge(\dd\alg)^2+\mbox{$\frac32$}\, \alg^3\wedge \dd\alg+\mbox{$\frac35$}\, \alg^5\big\rangle \ .
\eqend
As a consequence, the components of the invariant tensor are inherited from the five-dimensional theory and read as
\begin{align}
\big\langle X_{e_{j}-e_{k}}\, X_{e_{j^{\prime}}-e_{k^{\prime}}}\, X_{e_{j^{\prime
\prime}}-e_{k^{\prime\prime}}}\big\rangle  & =\delta_{kj^{\prime}}\,
\delta_{jk^{\prime\prime}}\, \delta_{k^{\prime}j^{\prime\prime}}+\delta
_{kj^{\prime\prime}}\, \delta_{jk^{\prime}}\, \delta_{k^{\prime\prime}j^{\prime}} \ , \nonumber \\[4pt]
\big\langle H_{j}\, X_{e_{j^{\prime}}-e_{k^{\prime}}}\, X_{e_{j^{\prime\prime}%
}-e_{k^{\prime\prime}}}\big\rangle  & =\delta_{jj^{\prime}}\, \delta
_{jk^{\prime\prime}}\, \delta_{k^{\prime}j^{\prime\prime}}+\delta_{jj^{\prime
\prime}}\, \delta_{jk^{\prime}}\, \delta_{k^{\prime\prime}j^{\prime}} \ , \nonumber \\[4pt]
\big\langle H_{j}\, H_{j^{\prime}}\, X_{e_{j^{\prime\prime}}-e_{k^{\prime\prime}}%
}\big\rangle  & =\delta_{jj^{\prime}}\, \big(  \delta_{jk^{\prime\prime}%
}\, \delta_{j^{\prime}j^{\prime\prime}}+\delta_{jj^{\prime\prime}}\, \delta
_{k^{\prime\prime}j^{\prime}}\big) \ , \nonumber \\[4pt]
\big\langle H_{j}\, H_{j^{\prime}}\, H_{j^{\prime\prime}}\big\rangle  &
=2\, \delta_{jj^{\prime}}\, \delta_{jj^{\prime\prime}}\, \delta_{j^{\prime}%
j^{\prime\prime}} \ .
\end{align}
With this data, the reduced action becomes
\begin{align}
S^{(3)}_{\mathsf{CSH}}=12\pi \, R^{2}\, \kappa\, \int_M \ \sum_{\ell=0}^{m}\, \mathsf{Tr}\Big( \big( & \, \phi_{(\ell+1)}\, \nabla\phi_{(\ell+1)}^{\dagger}-\phi_{(\ell)}^{\dagger
}\, \nabla\phi_{(\ell)}\big)  \wedge F_{(\ell)} \nonumber \\ & \, - 2\alpha_\ell \, A_{(\ell)}\wedge \big(
\dd A_{(\ell)}+\mbox{$\frac{2}{3}$}\, A_{(\ell)}\wedge A_{(\ell)}\big)
\Big)
\end{align}
with the field equations
\begin{align}
F_{(\ell)}\, \big(  4\alpha_{\ell}+\phi_{(\ell+1)}\, \phi_{(\ell+1)}^{\dagger}-\phi_{(\ell)}^{\dagger}\,
\phi_{(\ell)}\big) -\nabla\phi_{(\ell)}^{\dagger}\wedge \nabla\phi_{(\ell)} & =0 \ , \nonumber \\[4pt]
F_{(\ell)}\wedge \nabla\phi_{(\ell)}^{\dagger}  & =0 \ , \nonumber \\[4pt]
F_{(\ell)}\wedge \nabla\phi_{(\ell+1)}  & =0 \ .
\label{d=3equations}\end{align}

Note that the pure gauge sector of this field theory is governed by the three-dimensional Chern--Simons action with gauge group $\CH$, whose classical solution space is the moduli space of flat $\CH$-connections on $M$ modulo gauge transformations. As an explicit example, consider the case $m=1$, so that the gauge group $\CG=U(2)$ is broken to $\CH=U(1)\times U(1)$, and consider $A_1$ quiver gauge field configurations with $A_{(0)}=-A_{(1)}$ which further breaks the gauge symmetry to the diagonal $U(1)$ subgroup of $\CH$. It is then easy to reduce the field equations to the flatness conditions $F_{(0)}=-F_{(1)}=0$, and as a consequence there exists a local basis of parallel sections of the adjoint bundle ${\rm ad}(P_M)$. Hence in this case the solution space is again the finite-dimensional moduli space of flat $\CH$-connections on $M$. Owing to the topological nature of the system in this dimensionality, it may be possible that this is the generic moduli space of solutions in this dimension, but a rigorous proof of this fact is needed.

\subsubsection*{$\mbf{d\geq5}$}

Although for $d=3$ the moduli space of solutions is completely classified by the topology of the manifold $M$ and hence has no local degrees of freedom, in dimensions $d\geq5$ one can argue following~\cite{Banados:1995mq,Banados:1996yj,Miskovic:2005di} that the space of solutions of the diffeomorphism invariant Chern--Simons--Higgs model cannot be uniquely associated to the topology of $M$ as it generically contains local propagating degrees of freedom, depending on the algebraic properties of the invariant tensor. Our model presents an example of an \emph{irregular} Hamiltonian system~\cite{Saavedra:2000wk,Miskovic:2003ex} whose phase space is stratified into branches with different numbers of degrees of freedom and gauge symmetries, due to the dependence of the symplectic form on the fields. When certain \textit{generic} conditions are fulfilled, the symplectic form is of maximal rank and it is shown by~\cite{Banados:1996yj} using the standard Hamiltonian formalism that the number of local degrees of freedom in the pure gauge sector is given by
\begin{equation}
\mathcal{N}=\mbox{$\frac{1}{2}$} \, \big( 2(d-1)\, h-2(h+d-1)-(d-1)\, (h-1) \big) = \mbox{$\frac{1}{2}$}\, (d-1)\, (h-1) -h \ , \label{dofn}
\end{equation}
where $h>1$ is the dimension of the residual gauge group $\CH$; the first term in eq.(\ref{dofn}) is the number of canonical variables, the second term is twice the number $h$ of first class constraints associated with the gauge symmetry plus $d-1$ first class constraints associated to spatial diffeomorphism invariance, and the third term corresponds to the second class constraints. Note that this number is zero only for $d=5$ and $h=2$, i.e. the $A_1$ quiver gauge theory in five dimensions with gauge group $\CH=U(1)\times U(1)$. 

There are also \emph{degenerate} sectors where the rank of the symplectic form is smaller, additional local symmetries emerge, and fewer degrees of freedom propagate; on these branches the constraints are functionally dependent and the standard Dirac analysis is not applicable. Thus the dynamical structure of the theory changes throughout the phase space, from purely topological sectors to sectors with the maximal number eq.(\ref{dofn}) of local degrees of freedom. Moreover, the sector with maximal rank is stable under perturbations of the initial conditions, and on open neighbourhoods of the maximal rank solutions one can ignore the field-dependent nature of the constraints; on the contrary, degenerate sectors form measure zero subspaces of the phase space and around such degenerate backgrounds local degrees of freedom can propagate.

\newsection{Quiver gauge theory of AdS gravity} \label{sect4}

\subsection{$SU(2,2|1)$ Chern--Simons supergravity}

The most general action for gravity in arbitrary dimensionality is given by the
dimensional continuation of the Einstein--Hilbert action, called the
Lovelock series \cite{Lov71,Lan38,Zegers:2005vx}. In this expansion there are free
parameters which cannot be fixed from first principles. However, in
$d=2n+1$ dimensions a
special choice for the coefficients can be made in such a way that the Lovelock
Lagrangian becomes a Chern--Simons form~\cite{Cha89,Zanelli:2005sa,Allemandi:2003rs,Tro99}. The importance of
this feature lies in the fact that the gravity theory then
possesses a gauge symmetry once the spin
connection $\omega$ and the vielbein $e$ are arranged into a
connection $\mathcal{A}$ valued in the Lie
algebra of one of the Lie groups $SO(
d-1,2)$, $SO(d,1)$ or $ISO(d-1,1)$ corresponding respectively to the local
isometry groups of
spacetimes with negative, positive or vanishing cosmological constant. Another important reason for
considering Chern--Simons gravity theories is that
they admit natural supersymmetric extensions
\cite{Chamseddine:1990gk,Troncoso:1998ng,Banh96a}. In this
section we study as an example the $SU(2)$-equivariant dimensional reduction of 
five-dimensional Chern--Simons supergravity on $\CM= M\times S^2$, where
$M$ is a three-manifold.

Five-dimensional supergravity can be constructed as a Chern--Simons
gauge theory which is invariant
under the supergroup $SU(2,2|N)$~\cite{Troncoso:1997me}. The superalgebra $\mathfrak{su}%
(2,2|N)$ is the minimal supersymmetric extension of
$\mathfrak{su}(2,2)$, which is isomorphic to the
anti-de~Sitter (AdS) algebra $\mathfrak{so}(4,2)$.
A crucial observation is that in any dimension $d$ an explicit representation
of the AdS algebra can be given in terms of gamma-matrices $\Gamma_a$ which
satisfy the Clifford algebra relations (see appendix~\ref{explrep})
\begin{equation}
\left\{  \Gamma_{a},\Gamma_{b}\right\}  =2\eta
_{ab}\label{cliffal}%
\end{equation}
where $\eta={\rm diag}\left(-1,1,\ldots,1\right)  $ is the metric of $d$-dimensional
Minkowski space. By defining%
\begin{equation}
\Gamma_{ab}=\mbox{$\frac{1}{2}$}\, \left[  \Gamma_{a},\Gamma
_{b}\right]
\label{Gammaab}\end{equation}
it is easy to show that%
\begin{align}
\left[  \Gamma_{a},\Gamma_{b}\right]  & =2\Gamma
_{ab}  \ , \label{gam1}\\[4pt]
\left[  \Gamma_{ab},\Gamma_{cd}\right]  & =2\left(
\eta_{cb}\, \Gamma_{ab}-\eta_{ca}\, \Gamma_{bd}+\eta
_{db}\, \Gamma_{ca}-\eta_{da}\, \Gamma_{cb}\right) \ ,  \\[4pt]
\left[  \Gamma_{ab},\Gamma_{c}\right]  & =2\left(
\eta_{cb}\, \Gamma_{a}-\eta_{ca}\, \Gamma_{b}\right) \ . \label{gam2}
\end{align}
In this way, by choosing a set of $4 \times 4$ matrices satisfying
eq.($\ref{gam1}-\ref{gam2}$) it is possible to represent the Lie algebra
$\mathfrak{su}(2,2)  $ as a matrix algebra
by defining
\begin{equation}
\mathsf{J}_{ab}  =\mbox{$\frac{1}{2}$}\, \Gamma_{ab} \ , \qquad \mathsf{P}_{a}
=\mbox{$\frac{1}{2}$}\, \Gamma_{a} \ .
\end{equation}

Let us now turn to the supersymmetric extension
$\mathfrak{su}(2,2|N)$. For definiteness, we consider the case $N=1$
which accommodates
the minimum number $\mathcal{N}=2$ of supersymmetries. A representation of
$\mathfrak{su}(2,2|1)  $ can be obtained by extending the
bosonic generators $\left\{  \mathsf{P}_a,\mathsf{J}_{ab}\right\}  $ as
\begin{equation}
\mathsf{P}_{a}   =
\begin{pmatrix}
\frac{1}{2}\, \left(  \Gamma_{a}\right)  _{~\beta}^{\alpha} & 0\\
0 & 0
\end{pmatrix} \ , \qquad \mathsf{J}_{ab}   =
\begin{pmatrix}
\frac{1}{2}\, \left(  \Gamma_{ab}\right)  _{~\beta}^{\alpha} & 0\\
0 & 0
\end{pmatrix}
\end{equation}
and inserting the fermionic generators
\begin{equation}
\mathsf{Q}^{\gamma}  =
\begin{pmatrix}
0 & 0\\
-2\delta_{\beta}^{\gamma} & 0
\end{pmatrix} 
\ , \qquad \bar{\mathsf{Q}}_{\gamma}   =%
\begin{pmatrix}
0 & -2\delta_{\gamma}^{\alpha}\\
0 & 0
\end{pmatrix} \  .
\end{equation}
The supersymmetry algebra further requires the inclusion of a $U(1) $ generator
\begin{equation}
\mathsf{K}=%
\begin{pmatrix}
\frac{\ii}{4}\, \delta_{~\beta}^{\alpha} & 0\\
0 & \ii
\end{pmatrix}
\end{equation}
so that gauge invariance is preserved~\cite{Banados:1999kv}.

\subsection{Dimensional reduction}

In order to perform the $SU(2)$-equivariant dimensional reduction of $SU(2,2|1)  $
Chern--Simons supergravity, we choose the element
$\Lambda$ to take values in the Lorentz subalgebra $\mathfrak{so}%
(1,4)$ generated by $\left\{  \mathsf{J}_{ab}\right\}  $ and expand it as
\begin{equation}
\Lambda=\mbox{$\frac{\ii}{2}$}\, \lambda^{ab}\, \mathsf{J}_{ab} \ .
\end{equation}
This choice is not arbitrary, in the sense that it is the only one that
leads to an Einstein--Hilbert term after dimensional
reduction. Furthermore, non-trivial solutions of the constraint
equations eq.($\ref{eq26}-\ref{eq27}$) are possible only if the Higgs
fields $\Phi$ take values in the fermionic sector of $\mathfrak{su}(
2,2|1)  $; we expand them as
\begin{equation}
\Phi  =\bar{\mathsf{Q}}_{\beta}\, \chi^{\beta} \ , \qquad \bar{\Phi}
=\bar{\chi}_{\beta}\, \mathsf{Q}^{\beta}
\end{equation}
where $\chi$ and $\bar{\chi}$ are four-component Dirac spinor zero-forms with
$\beta$ running over $1,2,3,4$. In this way the constraints
eq.($\ref{eq26}-\ref{eq27}$) read as%
\begin{equation}
\big(  \mbox{$\frac{1}{4}$}\, \lambda^{ab}\, \left(  \Gamma_{ab}\right)  _{~\beta}^{\alpha
}+\delta_{~\beta}^{\alpha}\big) \, \chi^{\beta}=0 \ , \qquad
\bar{\chi}_{\alpha}\, \big(  \mbox{$\frac{1}{4}$}\, \lambda^{ab}\, \left(  \Gamma_{ab}\right)
_{~\beta}^{\alpha}+\delta_{~\beta}^{\alpha}\big)  =0 \ . \label{hfer}
\end{equation}
Gauging the Lie superalgebra $\mathfrak{su}(  2,2|1)  $ means that the
gauge potential decomposes as
\begin{equation}
A=\mbox{$\frac{1}{2}$}\, \omega^{ab}\, \mathsf{J}_{ab}+e^{a}\, \mathsf{P}%
_{a}+b\, \mathsf{K}+\bar{\psi}_{\alpha}\, \mathsf{Q}^{\alpha}-\bar{\mathsf{Q}%
}_{\beta}\, \psi^{\beta}%
\end{equation}
where $e, \omega$ are the standard vielbein and spin connection, $b$
is a $U(1)$ gauge field and $\psi, \bar{\psi}$ are four-component spin
$\frac32$ gravitino fields.
The constraint equation (\ref{eq25}) reads
\begin{align}
\lambda_{~b}^{a}\, \omega^{bd}  & =0 \ , \label{acons1}\\[4pt]
\lambda_{~b}^{a}\, e^{b}  & =0 \ , \\[4pt]
\bar{\psi}_{\alpha}\, \lambda^{ab}\, \left(  \Gamma_{ab}\right)  _{~\beta}^{\alpha}
& =0 \ , \qquad \\[4pt]
\lambda^{ab}\, \left(  \Gamma_{ab}\right)  _{~\beta}^{\alpha}\,
\psi^{\beta}  & =0 \ . \label{acons4}
\end{align}
These equations are still generic and will characterize the symmetry
breaking pattern once the non-zero components of $\lambda^{ab}$ are
specified. For this, we choose a particular representation of
$\mathfrak{su}(  2,2|1)  $. Using the Pauli matrices eq.(\ref{eq17}), a representation of the Clifford algebra in five dimensions is given by 
\begin{align}
\Gamma_{0} &  =\ii\sigma_{1}\otimes\mathbbm{1}_2 \ , \label{repexpl}\\[4pt]
\Gamma_{1} &  =\sigma_{2}\otimes\mathbbm{1}_2 \ , \\[4pt]
\Gamma_{2} &  =\sigma_{3}\otimes\sigma_{1} \ , \\[4pt]
\Gamma_{3} &  =\sigma_{3}\otimes\sigma_{2} \ , \\[4pt]
\Gamma_{4} &  =\sigma_{3}\otimes\sigma_{3} \ .\label{repexpl1}
\end{align}
The explicit construction is detailed in appendix~\ref{explrep}.
We now restrict $\lambda^{ab}\, \mathsf{J}_{ab}$ to be $\lambda
^{01}\, \mathsf{J}_{01}$; other restrictions are possible and they all
lead to the same qualitative results below. With this choice the algebraic quantization condition
eq.(\ref{eq20}) is satisfied and the constraint equation (\ref{hfer}) has non-trivial solutions if $\lambda^{01}=4$. In that case, one finds
\begin{equation}
\chi^{2}   =\chi^{4}=0=\bar{\chi}^{2}   =\bar{\chi}^{4} \ .
\end{equation}
Similarly, non-trivial solutions of eq.($\ref{acons1}-\ref{acons4}$) are
given by taking
\begin{equation}
\omega^{1a}   =0=\omega^{0a} \ , \qquad e^{1}  =0 =e^{0} \ , \qquad
\bar{\psi}_{\alpha}  =0 =\psi^{\alpha}
\end{equation}
for $a=0,1,2,3,4$ and $\alpha=1,2,3,4$.

The reduced field content can therefore be summarised as
\begin{align}
&  e^{a},\omega^{ab} \qquad \text{for} \quad a,b=2,3,4 \ , \nonumber \\[4pt]
&  \chi_{\alpha},\bar{\chi}^{\alpha} \qquad \text{for} \quad
\alpha=1,3 \ , \nonumber \\[4pt]
&  b \qquad \text{as $U(1)$ gauge field} \ .
\end{align}
Since the reduced gauge potential becomes%
\begin{equation}
A=\mbox{$\frac{1}{2}$}\, \omega^{ab}\, \mathsf{J}_{ab}+ e^{a}\, \mathsf{P}%
_{a}+b\, \mathsf{K} \ \in \ \mathfrak{so}(  2,2)  \oplus
\mathfrak{u}(  1) \ ,
\end{equation}
the gauge symmetry $\CG=SU(2,2|1)$ is broken by this construction to
\begin{equation}
\CH=  SO(
2,2)  \times U(  1) \ .
\end{equation}
The quiver gauge
theory is thus based on the $A_1$ quiver
\begin{equation}
\xymatrix@R=0.3pc{
& \bullet \ar[r] & \bullet  }
\end{equation}
with the left node containing the $SO(2,2)$ gravitational content $e,\omega$, the
right node containing the $U(1)$ gauge field $b$, and the arrow
corresponding to the Higgs fermions $\chi$ and $\bar{\chi}$ which transform in the
bifundamental representation of $SO(2,2)\times U(1)$. Since $\pi_3(U(1))=0=\pi_3(SO(2,2))$, there is no topological quantization
condition required of the gravitational constant $\kappa'$ after
dimensional reduction.

In order to evaluate the reduced Chern--Simons--Higgs action, note
that the curvature two-form associated to the group $SO(2,2)\times U(1)$ is  
\begin{equation}
{F}  =\mbox{$\frac{1}{2}$} \, \big(
R^{ab}+\mbox{$\frac{1}{l^{2}}$} \,
e^{a}\wedge e^{b}\big) \,
\mathsf{J}_{ab}+\mbox{$\frac{1}{l}$} \, T^{a}\, \mathsf{P}_{a}+\dd b\, \mathsf{K}
\end{equation}
where $l$ is the AdS radius, $R^{ab}=\dd\omega^{ab}+\omega^a_{~c}\wedge
\omega^{cb}$ is the Lorentz curvature two-form, and $T^a=\dd
e^a+\omega^a_{~b}\wedge e^b$ is the torsion two-form.
The non-vanishing components of the $\mathfrak{su}%
(2,2|1)$-invariant tensor of rank three are given in appendix \ref{invtens}.
With this, one finds that the Chern--Simons--Higgs gravitational action is given by
\begin{equation}
{S}_{\mathsf{CSH}}^{(3)} =\frac{\kappa^\prime}{l}\, \int_{M}\, \bigg( \epsilon_{abc}\, \Big(
R^{ab}+\frac{1}{3l^{2}} \,
e^{a}\wedge e^{b}\Big)  \wedge e^{c} -\ii \nabla\bar{\chi}_{\alpha} \wedge \mathcal{Z}_{~\beta}^{\alpha
}\, \chi^{\beta}+\ii\bar{\chi}_{\alpha}\, \mathcal{Z}_{~\beta}^{\alpha} \wedge \nabla\chi^{\beta} \, \bigg) \label{CShferm}
\end{equation}
where $\kappa^\prime =8\pi \,R^{2}\, \kappa$ and
\begin{align}
\mathcal{Z}_{~\beta}^{\alpha}  & =\mbox{$\frac{1}{2}$} \, \big(
R^{ab}+\mbox{$\frac{1}{l^{2}}$} \,
e^{a}\wedge e^{b}\big) \, \left(  \Gamma_{ab}\right)
_{~\beta}^{\alpha}-\mbox{$\frac{1}{l}$}\, T^{a}\, \left(  \Gamma_{a}\right)  _{~\beta}^{\alpha}%
+\mbox{$\frac{5\ii}{2}$}\, \delta_{~\beta}^{\alpha}\, \dd b \ , \nonumber \\[4pt]
\nabla\bar{\chi}_{\alpha}  & =\dd\bar{\chi}_{\alpha}- \mbox{$\frac{1}{4}$}\,
\bar{\chi}_{\beta}\, \omega^{ab}\, \left( \Gamma_{ab}\right)  _{~\alpha}^{\beta
}-\mbox{$\frac{1}{2}$}\, \bar{\chi}_{\beta}\, e^{a}\, \left(  \Gamma_{a}\right)  _{~\alpha
}^{\beta}+\mbox{$\frac{3\ii}{4}$}\, b\, \bar{\chi}_{\beta}\, \delta_{~\alpha
}^{\beta} \ , \nonumber \\[4pt]
\nabla\chi^{\beta}  & =\dd\chi^{\beta}+\mbox{$\frac{1}{4}$}\, \omega^{ab}\, \left(  \Gamma
_{ab}\right)  _{~\alpha}^{\beta}\, \chi^{\alpha}+\mbox{$\frac{1}{2}$}\, e^{a}\, \left(
\Gamma_{a}\right)  _{~\alpha}^{\beta}\, \chi^{\alpha}-\mbox{$\frac{3\ii}{4}$}\, b\, 
\delta_{~\alpha}^{\beta} \, \chi^{\alpha} \ .
\label{eme}\end{align}
Note that the reduced field content restricts the gamma-matrices of the five-dimensional representation according to 
\begin{equation}
\Gamma_{0}=%
\begin{pmatrix}
0 & \ii\\
\ii & 0
\end{pmatrix}
\ , \qquad \Gamma_{1}=%
\begin{pmatrix}
0 & -\ii\\
\ii & 0
\end{pmatrix}
\ , \qquad \Gamma_{2}=%
\begin{pmatrix}
1 & 0\\
0 & -1
\end{pmatrix}
\ , \nonumber 
\end{equation}
\begin{equation}
\Gamma_{01}=%
\begin{pmatrix}
-1 & 0\\
0 & 1
\end{pmatrix}
\ , \qquad \Gamma_{02}=%
\begin{pmatrix}
0 & -\ii\\
\ii & 0
\end{pmatrix}
\ , \qquad \Gamma_{12}=%
\begin{pmatrix}
0 & \ii\\
\ii & 0
\end{pmatrix}
\ ,
\end{equation}
which gives a representation of the Clifford algebra in $d=2+1$
dimensions.

The infinitesimal gauge transformations corresponding to
eq.(\ref{APhigaugetransf}) yield local symmetry transformations for
the gauge fields and Higgs fermions given by
 \begin{align}
 \delta_{\lambda,\rho}\omega^{ab}  &
 =\dd\lambda^{ab}+\omega_{~c}^{a}\, \lambda^{cb}+ \omega
 _{~c}^{b}\, \lambda^{ac}+\mbox{$\frac{1}{l^{2}}$} \, \big(
 e^{a}\wedge\rho^{b}-\rho^{a} \wedge
 e^{b}\big) \ , \\[4pt]
 \delta_{\lambda,\rho} e^{a}  & =\dd\rho^{a}+\omega_{~b}^{a}\,
 \rho^{b}-\lambda_{~b}^{a}\, e^{b} \ , \\[4pt]
 \delta_\beta b  & =\dd\beta \ , \\[4pt]
 \delta_{\rho,\kappa,\beta}\chi & =\mbox{$\frac{1}{2l}$} \, \rho^{a}\,
 \Gamma_{a}\chi- \mbox{$\frac{1}{2}$} \, \epsilon_{abc}\,
 \kappa^{ab}\, \Gamma^{c}\chi-\mbox{$\frac{3\ii}{4}$} \, \beta\, \chi \ , \label{hferm1}\\[4pt]
 \delta_{\rho,\kappa,\beta}\bar{\chi}& =-\mbox{$\frac{1}{2l}$} \, \bar{\chi}\,
 \rho^{a}\, \Gamma_{a}+\mbox{$\frac{1}{2}$} \,
 \epsilon_{abc}\, \bar{\chi}\, \kappa^{ab}\, \Gamma^{c}+
 \mbox{$\frac{3\ii}{4}$} \,
 \bar{\chi}\, \beta  \ . \label{hferm2}
 \end{align}
The action eq.(\ref{CShferm}) describes a theory of Einstein--Hilbert
gravity with cosmological constant in three dimensions, plus a
non-minimal coupling between Higgs fermions and the fields associated
to the curvature of the residual gauge symmetry $SO(2,2)\times
U(1)$. This model is not supersymmetric as one sees from the gauge
transformations eq.($\ref{hferm1}-\ref{hferm2}$). The equivariant
dimensional reduction scheme thus provides a novel and systematic way to couple
scalar fermions to gravitational theories, which
is normally cumbersome to do.

The variation of the Chern--Simons--Higgs action eq.(\ref{CShferm}) leads to the field equations \begin{align}
2\epsilon_{abc}\check{R}^{ab}+\mbox{$\frac{\ii}{l}$}\, T_{c}\, \bar{\chi}\, \chi-\mbox{$\frac{1}{2}$}\, 
\dd b\, \bar{\chi}\, \Gamma_{c}\, \chi -\ii\nabla\bar\chi\, \wedge \Gamma_c\nabla\chi & =0 \ , \nonumber \\[4pt]
\ii\check{R}^{ab}\, \bar{\chi}\,\chi+\mbox{$\frac{1}{l}$}\, \epsilon^{abc}\, T_{c}+ \mbox{$\frac{1}{4}$}\, 
\dd b\, \bar{\chi}\, \Gamma^{ab}\chi -\ii  \nabla\bar\chi\,\wedge \Gamma^{ab}\nabla\chi & =0 \ , \nonumber \\[4pt]
\check{R}^{ab}\, \bar{\chi}\, \Gamma_{ab}\chi-\mbox{$\frac1l$}\, T^{a}\, \bar{\chi}\, \Gamma_{a}\chi
+\mbox{$\frac{15\ii}{2}$}\, \dd b\, \bar{\chi}\, \chi & =0 \ , \nonumber \\[4pt]
\mathcal{Z}\wedge\nabla\chi & =0 \ , \nonumber \\[4pt]
\nabla\bar{\chi}\wedge\mathcal{Z}  & =0 \ ,
\label{AdSmattereqs}\end{align}
where we have used the abbreviation 
\begin{equation}
\check{R}^{ab}:=\mbox{$\frac{1}{2}$}\, \big(  R^{ab}+\mbox{$\frac{1}{l^{2}}$}\, e^{a}\wedge
e^{b}\big) \ .
\end{equation}
These equations demonstrate an interesting coupling between curvature and the matter currents; note that at least one of the torsion field $T^a$ or the $U(1)$ field strength $\dd b$ must be non-zero to get a non-trivial matter coupling; otherwise, when $T^a=0=\dd b$ the matter fields freely decouple from gravity and the field equations reduce to those of pure AdS gravity in three dimensions.

\chapter{Conclusions}
\label{ch:conclusions}
\begin{flushright}
\textit{``...Caminante, no hay camino, 
se hace camino al andar.}''. \\ \textit{Caminante no hay camino, Antonio Machado.}
\footnote{\scriptsize ``...Wayfarer, there is no way,
make your way by going farther.''. Wayfarer, there is no path , Antonio Machado.}
\bigskip
\end{flushright}
In the present thesis, we have had the opportunity to investigate the construction of different types of topological gauge theories by means of transgression forms. In Chapter \ref{ch:top_grav} we have made the connection between even dimensional topological gravity and transgression field theories for the special case of Poincar\'e symmetry. By similar arguments, in Chapter \ref{ch:Max_alg} a gauged Wess--Zumino--Witten model for the Maxwell algebra in two dimensions is constructed. Finally, in Chapter \ref{ch:covquiv} we use transgressions to obtain an action principle for what we called a Chern--Simons--Higgs model as dimensional reduction of a pure Chern--Simons term in higher dimensions. As physical application, we studied the Chern--Simons--Higgs Lagrangian in the context of five dimensional supergravity. In each of these chapters some answers have been provided but also some interesting questions have arisen which could extend this Thesis to further research directions.

\begin{itemize}
\item \textbf{Transgression forms:} Any field theory constructed using transgression forms as Lagrangians have very good qualities. A striking property is that they are built in terms of topological invariants and consequently any transgression action turns out to be background independent (metric free). This is why in literature they are usually called ``topological''. 

In the most general case, when the two connections are treated as independent fields, the transgression field theory is fully gauge invariant (since the transgression defined on the fibre bundle is projectable) \cite{Mora:2004kb,Izaurieta:2005vp,Mora:2006ka}. Despite of the counter intuitive idea of carrying two connections as dynamical fields, the fact turns out to be very versatile. For instance, turning off one of the connections conduces to the definition of a Chern--Simons form. This restriction is however not free in the sense that Chern--Simons forms are not globally defined and therefore any action principle constructed with Chern--Simons forms as Lagrangians is only gauge invariant modulo boundary terms. 

Another interesting possibility is to relate both connections by a gauge transformation. In this case the transgression field theory can be treated as more than one chart on the base $\mathcal{M}$ is provided. This means that the gauge fields are independent up to the intersection region in which they relate by the transitions functions $g_{ij}: U_i \cap U_j \rightarrow \mathcal{G}$. The resulting Lagrangian is a gauged Wess--Zumino--Witten term.

A more technical but no less interesting application of transgressions is given in terms of the triangle equation and the Subspace Separation Method \ref{SSMM}. With this method, the transgression and subsequently the Chern--Simons Lagrangian can be explicitly written in pieces corresponding to different interactions present in the theory, as well as to split the volume and the boundary contributions in the Lagrangian.

It should be emphasized that even though there are strong indicators that transgression or Chern--Simons theories are renormalizable \cite{Witten:1988hf,Witten:1988hc,Cha89,Cha90,Zanelli:1994ti,Mora:2003wy}, the quantum behaviour of these theories it is not well understood in dimensions higher than three. The main reason is that the kinetic and the potential terms are very complicated and therefore the interactions are highly nonlinear. This has as a consequence that the dynamic is strongly constrained; much more than the usual field theories \cite{Saavedra:2000wk,Miskovic:2003ex,Banados:1995mq,Banados:1996yj}. Thus, the quantum mechanics of these type of theories still remains as an open question and even when the problem seems to be soluble, it does not looks like the solution can be achieved by conventional quantization and renormalization rules. This could mean that new quantization method are needed or, more dramatically, the requirement that he axioms which support the notion of quantum mechanics that we know today should be reformulated.

\item \textbf{Topological theories of gravity:} 

In the classification of topological theories of (super)gravity \cite{Cha89,Cha90} the gauge groups anti-de Sitter, de-Sitter and Poincar\'e were considered, depending on the sign of the cosmological constant $\Lambda$. In odd dimensions the gravitational actions are obtained by using Chern--Simons forms once the gauge potential is arranged in terms of the vierbein $e$ and the spin connection $\omega$. A very interesting link between Chern--Simons gravity and General Relativity was first realized in three dimensions where it was shown that both theories are classically equivalent \cite{Ach86,Witten:1988hc,Deser:1981wh}. Moreover, in any odd dimension, Chern--Simons gravity theory  invariant under the anti-de~Sitter group turns out to be equivalent to a Lanczos--Lovelock Lagrangian. The identification is realized by requiring that the equations of motion determine uniquely the dynamics for as many components of the independent fields as possible. In this way one can fix the free parameters in the Lovelock theorem in
terms of the gravitational and cosmological constant \cite{Tro99}. This is quite interesting, Lanczos--Lovelock theory is the most general Lagrangian compatible with General Relativity principles in higher dimensions and corresponds to a Chern--Simons Lagrangian for a specific choice of the coupling constants.

In even dimensions, there is no geometrical candidate as Chern--Simons forms. In this case a $2n-$dimensional invariant form can be obtained by wedging with $n$ two-form curvature but, in addition, a scalar field transforming in the fundamental representation of the gauge group must be inserted. This inclusion may seem unnatural but it is actually motivated by dimensional reduction of a Chern--Simons gravity theory in one higher dimensions \cite{Cha90}. A non-trivial observation pointed out in this thesis is that even-dimensional topological gravity action, which is genuinely invariant under the Poincar\'e group, can be written as a transgression field theory where the gauge connections are related by a gauge transformation with group element taking values in the coset $g=\mathrm{e}^{-\phi^a \mathsf{P}_a} \in ISO(2n,1)/SO(2n,1)$. The geometrical interpretation regarding the inclusion a second gauge field $\mathcal{\bar{A}}$ is that more than one chart $U_\alpha$ on the base space $\mathcal{M}$ is provided. In this way, on non-empty overlaps $U_{ij}$ the gauge connections are related by transitions functions $g=\mathrm{e}^{-\phi^a \mathsf{P}_a}$ determining the non-triviality of the principal bundle $\mathcal{P}$ \cite{Anabalon:2006fj}.
The resulting even dimensional topological action is a gauged Wess--Zumino--Witten term \cite{Anabalon:2007dr}. In fact, in the present case the pure $(2n+1)$-dimensional Wess--Zumino term vanish due to the form of the invariant tensor. This has as a consequence that the full theory collapses to its $2n$-dimensional boundary which we identify with the even-dimensional topological gravity model.

There is also another interesting identification made at this point. In the context of nonlinear realization theory of the Poincar\'e group, the transformation law for the nonlinear counterpart of the gauge connection has the same form as a gauge transformation law for the connection with element $g=\mathrm{e}^{-\phi^a \mathsf{P}_a} \in ISO(2n,1)/SO(2n,1)$. This allowed us to obtain the even-dimensional topological gravity model as the difference of the Chern--Simons forms constructed in terms of the linear $\mathcal{A}$ and the nonlinear connection $\mathcal{\bar{A}}$ both valued in the Lie algebra of $\mathfrak{iso}(2n,1)$ plus a boundary contribution. From this point of view the topological gravity action remains invariant under the Lorentz subgroup $SO(2n,1)$. However the local translation symmetry is broken since the nonlinear realization theory does not prescribe of the appropriate adjoint transformation law for the coset field $\phi$.

Since transgressions and subsequently Chern--Simons theories admit very natural supersymmetric extensions, we have extended the construction of  gauged Wess--Zumino--Witten models to the case of the super Poincar\'e algebra in three dimensions. In complete analogy to the pure bosonic case, the Wess--Zumino--Witten action collapses to its boundary $\partial \mathcal{M}$ providing in this way of a two dimensional field theory containing in addition to the coset fields $\phi$, spin 3/2 gravitinos $\chi$. It would be interesting to extend this construction to higher dimensions as well as to different gauge groups. In particular to study eleven-dimensional supergravity and its associated ten-dimensional Wess--Zumino--Witten models. This could provide of an interesting relation with supergravity theories which arise as the low energy limit of superstrings in ten-dimensions. 

\item {\textbf{WZW model and the Maxwell algebra:}}

The Maxwell algebra was introduced as the non-central extension of the Poincar\'e algebra $\mathfrak{iso}(d-1,1)$ by a rank-two abelian generator $\mathsf{Z}_{ab}$ \cite{Bacry:1970ye,Schrader:1972zd}. The initial motivation for considering the Maxwell algebra was the description of the symmetries of particles moving in a constant electromagnetic background. More recently, it has been shown that gauging the Maxwell algebra in four dimensions leads to generalizations of standard General Relativity where the new abelian gauge fields play the role of vector inflatons which contribute to a generalization of the cosmological term \cite{deAzcarraga:2012qj}. The construction is based on considering all the possible four dimensional non-metric Lagrangian densities constructed in terms of the components field strength and the Levi-Civita tensor. Motivated by these models, the three-dimensional Chern--Simons gravity Lagrangian invariant under the Maxwell algebra was constructed. The Chern--Simons theory contains in addition to the Einstein--Hilbert term, the so called exotic term for the Lorentz connection plus the torsional term $T^a \wedge e_a$, and a coupling between an additional gauge field $\sigma$ with the Lorentz curvature $R^{a}_{~b} \wedge \sigma^{b}_{~a} $. The geometrical interpretation of this theory can be understood by looking at the field equation which only support flat solutions. In this way, three dimensional gravity for the Maxwell algebra describes flat geometries coupled to the gauge fields $\sigma$.

In order to obtain the Chern--Simons Lagrangian one need to specify the nonvanishing components of the invariant tensors. This is done by considering the Maxwell algebra as an $S-$expansion mechanism of the anti-de Sitter algebra $SO(2,2)$ in three dimensions \cite{Salgado:2014qqa}. The problem of extending Lie algebras allows to obtain new Lie algebras in terms of, for instance an already knew one. Depending on which type of extension is performed, the symmetries can be sometimes enhanced or reduced. The $S-$expansion mechanism consists basically in the construction of a Lie algebra as the direct product of a Lie algebra $\mathfrak{g}$ and an abelian semigroup $S$. Furthermore, smaller Lie algebras can be extracted from any $S-$expanded algebra by performing systematic reductions like $0_{S}-$force or resonant conditions \cite{Iza06b,Izaurieta:2006ia}.  More interestingly, $S-$expansion provides of the invariant tensors associated to the new Lie algebra once the invariant tensors of $\mathfrak{g}$ are specified. In this way, starting form the anti-de Sitter algebra in three dimensions and for a given semigroup $S$, the Maxwell algebra is recovered and therefore the invariant tensors are obtained.

A new insight considered in this Thesis is also the construction of the gauged Wess--Zumino--Witten model associated to the Maxwell algebra in two dimensions. In this case one shows that the resulting Lagrangian generalizes the $(1+1)-$topological gravity action for the Poincar\'e case. Thus, the model contains a new coset field $h^{ab}$ associated to the generator $\mathsf{Z}_{ab}$ minimally coupled to the Lorentz curvature. It would be interesting to study this model in a deeper way. In fact, it has been shown recently that the two dimensional holographic dual of three-dimensional asymptotically flat (or anti-de Sitter) gravity theory corresponds to a Wess--Zumino-Witten model which connects with Liouville theory of gravity \cite{Coussaert:1995zp,Campoleoni:2011hg,Barnich:2013yka}. From this point of view, an obvious generalization to this idea could be carried out by considering three-dimensional Maxwell gravity and the associated Wess--Zumino--Witten model.

\item{\textbf{Quivers and Chern--Simons--Higgs theory:}}

The geometric structures arising from reductions of $G$-invariant Yang--Mills theory have been thoroughly studied in a multitude of
different contexts \cite{Popov:2008gw,Popov:2005ik,Dolan:2009ie}, while coset space dimensional reduction of five-dimensional Chern--Simons theory with gauge group $\CG=SU(2)$ is considered in~\cite{TempleRaston:1994sk}. In the context of $SU(2)-$equivariant dimensional reduction we have shown that
the symmetry breaking patterns induced by $SU(2)$-invariant connections for the case of the classical gauge groups $U(n)$, $SO(2n)$,
$SO(2n+1)$, or $Sp(2n)$ is generically the same (without any conditions on the
background Dirac monopole charges) in all cases. As a
consequence, the induced quiver gauge theories are the same for any classical gauge group (up to redefinitions of the coupling
constants).

$SU(2)-$equivariant dimensional
reduction of Chern--Simons theories over $\mathbb{C}P^{1}$ leads to a novel diffeomorphism-invariant
Chern--Simons--Higgs model, which can have local degrees of freedom whose dynamics and canonical structure are rather delicate to disentangle; from this point of view the generally covariant models are therefore generically \emph{not} topological field theories. However, the definite answer about the topological origin of Chern--Simons--Higgs models, requires Hamiltonian analysis for the case of degenerate systems. Similar treatments have been considered in \cite{Banados:1995mq,Banados:1996yj,Saavedra:2000wk,Miskovic:2003ex}.

In $d=1$ dimensions, Chern--Simons--Higgs model describes a trivial system of covariantly constant Higgs fields. It is shown that the moduli spaces of classical solutions is finite dimensional and the resulting theory is that of a topological matrix quantum mechanics. This result is somehow expected in the sense that the original pure three-dimensional
Chern--Simons theory
is a topological gauge theory, and hence so is its dimensional
reduction.

In $d=3$ dimensions
the reduced field equations similar to those of the $m=1$ case \ref{d3feq} were obtained in~\cite{TempleRaston:1994sk}. It is interesting to note that one can consider regions of $M$ with monopole type Higgs field configurations having $\nabla\Phi=0$ but $F\neq0$; in this case the monopole charge is non-zero only through two-cycles of $M$ which enclose regions where $\nabla\Phi\neq0$. According to the field equations (\ref{d=3equations}), in such regions the Higgs fields must in addition satisfy $\big[  \Phi,\Phi^{\dagger
}\big]= 2\ii\Lambda$, which is the minimum of the Higgs potential in (\ref{ymredact}). Thus monopole configurations are allowed in the Higgs vacuum and are triggered by spontaneous symmetry breaking. It would be interesting to examine the dynamics after symmetry breaking of the coupled Yang--Mills--Chern--Simons--Higgs models defined by the sum of the action functionals (\ref{ymredact}) and (\ref{CSredact}), along the lines of~\cite{Dolan:2009ie}; in this model the gauge sector also contains massive spin one degrees of freedom~\cite{Deser:1981wh}.

In $d\geqslant5$ dimensions the equations of motion become rather complicated. However, 
it may be possible that the main features of pure Chern--Simons dynamic will not be spoilt by the coupling to the Higgs fields, so the essential features should remain: The equations of motion do not constrain the connection to be flat. As our choice of invariant tensor (\ref{symmtrace}) for $\CG$ is primitive~\cite{deAzcarraga:1997ya}, we expect the generic condition to hold; note that this choice is the one that leads to the appropriate Higgs branching structure of the quiver gauge theory from section~\ref{sector3}. In fact, the phase $F=0$ is degenerate because small perturbations around it are trivial. It would be interesting to see how the degree of freedom count (\ref{dofn}) is modified by performing the analogous canonical analysis for the full Chern--Simons--Higgs model, but this seems far more complicated than the analysis of the pure Chern--Simons gauge theory. Moreover, even in the pure gauge sector, no explicit propagating solutions have been found thus far. If we choose to discard solutions with $F^n=0$, $n>1$ as degenerate backgrounds, then one can find a phase with $F$ of maximal rank which carries the maximum number of degrees of freedom (\ref{dofn}). Such a propagating phase contains ``Higgs monopole" type configurations analogous to those discussed above for the case $d=3$.

In the context of $SU(2)-$equivariant dimensional reduction to five-dimensional
Chern--Simons supergravity over $\mathbb{C}P^1$. It is found that if the
Higgs fields are bifundamental fields in the fermionic sector of the gauge algebra, then the reduced action contains the standard
Einstein--Hilbert term plus a non-minimal coupling of the Higgs fermions to the curvature.
This reduction scheme constitutes a novel and systematic way to couple scalar fermions to
gravitational Lagrangians. Note that
by restricting to the pure bosonic sector by setting all fermions to zero,
our reduction reduces from ordinary five-dimensional to
three-dimensional anti-de Sitter gravity without any matter fields; hence our
reduction scheme further provides a means for lifting purely
gravitational configurations on $M$ to solutions on $M\times
S^2$, and it would be interesting to examine this lifting in more
detail on some explicit solutions. 

\end{itemize}

\appendix
\renewcommand{\chaptermark}[1]{\markboth{Appendix \thechapter.\ #1}{}}
\chapter{Nonlinear realizations of Lie groups}
\label{ch:app2}

Most of the applications of Lie groups theory to physics is by using linear
representations. In this picture, for each element $g$ of a Lie group $G$
there is a linear operator acting on a vector space (the space of the
representation) in such a way that the composition law defined by the group
axioms is preserved under the product of the associated linear operators which
define the representation.

It is also possible to define nonlinear realization of Lie groups which
corresponds to maps from a manifold $M$ to itself characterized by an element
$g_{0}\in G$. Let $\left\{  x\right\}  $ be the set of coordinates labelling
the points of $M$. The action of $g_{0}$ on $M$ is characterized by%
\begin{equation}
x^{\prime}=f\left(  g_{0};x\right) \ ,
\end{equation}
where $f:G\times M\longrightarrow M$ satisfies the following properties%
\begin{align}
x  &  =f\left(  I_{G};x\right) \ , \\
f\left(  g_{2};\left[  f\left(  g_{1};x\right)  \right]  \right)   &
=f\left(  g_{1}g_{2};x\right) \ .
\end{align}
In general, the map $f$ is nonlinear. The standard notion of linear
representation is recovered when $f$ is linear%
\begin{equation}
f\left(  g,\alpha x+\beta y\right)  =\alpha f\left(  g,x\right)  +\beta
f\left(  g,y\right) \ ,
\end{equation}
and $M$ is a vector space.

\newsection{Standard form of a nonlinear representation}

Let $G$ be a Lie group of dimension $n$ and $\mathfrak{g}$ its Lie algebra.
Let $H$ be the stability subgroup of $G$ of dimension $n-d$ whose Lie algebra
$\mathfrak{h}$ is generated by $\mathrm{Span}_{%
\mathbb{C}
}\left\{  \mathsf{V}_{i}\right\}  _{i=1}^{n-d}$. Let us denote by
$\mathfrak{p}$ to the vector subspace generated by the remaining generators of
$\mathfrak{g}$, $\mathrm{Span}_{%
\mathbb{C}
}\left\{  \mathsf{P}_{l}\right\}  _{l=1}^{d}$. In this way, as vector spaces,
we can write $\mathfrak{g}=\mathfrak{h}\oplus\mathfrak{p}$. In general, the
Lie bracket between two any elements in $\mathfrak{p}$ is $\left[
\mathfrak{p,p}\right]  \subset\mathfrak{p}\oplus\mathfrak{h}$. Since
$\mathfrak{h}$ is a subalgebra, then $\left[  \mathfrak{h,h}\right]
\subset\mathfrak{h}$. Moreover, we will assume that $\mathfrak{p}$ can be
chosen in such a way that it defines a representation of $H$. This means that
$\left[  \mathfrak{h},\mathfrak{p}\right]  \subset\mathfrak{p}$.

With this decomposition, any element $g_{0}\in G$ can always be written as \cite{Coleman:1969sm,Callan:1969sn}%
\begin{equation}
g_{0}=e^{\zeta\text{\textperiodcentered}\mathsf{P}}h \ .\label{nonl1}%
\end{equation}
where $h\in H$ and $e^{\zeta\text{\textperiodcentered}\mathsf{P}}=e^{\zeta
^{l}\mathsf{P}_{l}}\in G/H$ for $l=1,...,d$. The coordinates $\zeta$ 
parametrize the coset space $G/H$~\cite{Ste80}. By virtue of eq.$\left(  \ref{nonl1}%
\right)  $, the action of $g_{0}$ on the coset space $G/H$ is given by%
\begin{equation}
g_{0}e^{\zeta\text{\textperiodcentered}\mathsf{P}}=e^{\zeta^{\prime
}\text{\textperiodcentered}\mathsf{P}}h_{1} \ . \label{nonl2}%
\end{equation}
This expression allows to obtain $\zeta^{^{\prime}}$ and $h_{1}$ as nonlinear
functions of $g_{0}$ and $\zeta$ and thus%
\begin{align}
\zeta^{^{\prime}}  & =\zeta^{^{\prime}}\left(  g_{0};\zeta\right)
,\label{nonl12}\\
h_{1}  & =h_{1}\left(  g_{0};\zeta\right)  .\nonumber
\end{align}
The expression eq.$\left(  \ref{nonl12}\right)  $ defines by itself a
nonlineal realization of $G$ over the manifold with coordinates $\zeta^{i}$.

To obtain the transformation law of the coset parameters under the action of
$G$, it is useful to consider $g_{0}$ infinitesimal. In that case eq.$\left(
\ref{nonl2}\right)  $ reads%
\begin{equation}
e^{-\zeta\text{\textperiodcentered}\mathsf{P}}\left(  g_{0}-1\right)
e^{\zeta\text{\textperiodcentered}\mathsf{P}}-e^{-\zeta
\text{\textperiodcentered}\mathsf{P}}\delta e^{\zeta\text{\textperiodcentered
}\mathsf{P}}=h_{1}-1\label{nonl3}%
\end{equation}
and this allows us to obtain $\delta\zeta$ as the transformation law of the
coset parameter $\zeta^{i}$ under the infinitesimal action of $G$.

In order to characterize the standard form of a nonlinear realization of a Lie
group, let $\varphi$ be a field transforming in a linear representation of $G$%
\begin{equation}
\varphi^{\prime}=D\left(  g_{0}\right)  \varphi
\end{equation}
here $D\left(  g_{0}\right)  $ denotes the linear operator $D$ associated to
the element $g_{0}\in G$. Let us define the nonlinear field $\bar{\varphi}$ as
the action of an element of the coset space on $\varphi$ by%
\begin{equation}
\bar{\varphi}=D\left(  e^{-\zeta\text{\textperiodcentered}\mathsf{P}}\right)
\varphi
\end{equation}
Using this relation we see how the nonlinear fields transform under the action
of $G$. In fact, since%
\begin{align}
\bar{\varphi}^{\prime} &  =D\left(  e^{-\zeta^{\prime}\text{\textperiodcentered
}\mathsf{P}}\right)  \varphi^{\prime} \ ,\\
&  =D\left(  e^{-\zeta^{\prime}\text{\textperiodcentered}\mathsf{P}}\right)  D\left(
g_{0}\right)  \varphi \ ,\\
&  =D\left(  e^{-\zeta^{\prime}\text{\textperiodcentered}\mathsf{P}}\right)  D\left(
g_{0}\right)  D^{-1}\left(  e^{-\zeta\text{\textperiodcentered}\mathsf{P}%
}\right)  \bar{\varphi} \ ,\\
&  =D\left(  e^{-\zeta^{\prime}\text{\textperiodcentered}\mathsf{P}}g_{0}%
e^{\zeta\text{\textperiodcentered}\mathsf{P}}\right)  \bar{\varphi} \ .%
\end{align}
and since that $g_{0}e^{\zeta\text{\textperiodcentered}\mathsf{P}}=e^{-\zeta^{\prime}\text{\textperiodcentered}\mathsf{P}}h_{1}$, we
find
\begin{equation}
\bar{\varphi}^{\prime}=D\left(  h_{1}\right)  \bar{\varphi} \ ,%
\end{equation}
where $h_{1}=h_{1}\left(  g_{0},\zeta\right)  $ and $D\left(  h\right)  $ as a
linear representation of the subgroup $H$. The field $\bar{\varphi}$ has the
special property that under the action of $g_{0}\in G$ it transforms as under
an element $h_{1}\in H$ where $h_{1}=h_{1}\left(  g_{0},\zeta\right)  $ is
nonlinear. Thus, the complete set of relations which define a nonlinear
realization is given by \cite{Callan:1969sn}
\begin{align}
\zeta^{^{\prime}} &  =\zeta^{^{\prime}}\left(  g_{0},\zeta\right)
,\label{nonl8}\\
\bar{\varphi}^{\prime} &  =D\left[  h_{1}\left(  g_{0},\zeta\right)  \right]
\bar{\varphi},\label{nonl9}\\
g_{0}e^{\zeta\text{\textperiodcentered}\mathsf{P}} &  =e^{\zeta^{\prime
}\text{\textperiodcentered}\mathsf{P}}h_{1}.\label{nonl10}%
\end{align}

The nonlinear realization of a Lie group $G$ can be understood as a set of
maps acting on a manifold $M$ with coordinates $\left(  \zeta,\bar{\varphi
}\right)  $. Note that in the case that we restrict $G$ to the subgroup $H$
the nonlinear representation becomes linear. If $g_{0}=h_{0}\in H$,
eq.$\left(  \ref{nonl2}\right)  $ takes the form%
\begin{align}
e^{\zeta^{\prime}\text{\textperiodcentered}\mathsf{P}}h_{1}  & =h_{0}%
e^{\zeta\text{\textperiodcentered}\mathsf{P}} \ ,\\
& =\left(  h_{0}e^{\zeta\text{\textperiodcentered}\mathsf{P}}h_{0}%
^{-1}\right)  h_{0} \ .%
\end{align}
and since $\left[  \mathfrak{h},\mathfrak{p}\right]  \subset\mathfrak{p}$, the
term $h_{0}e^{\zeta\text{\textperiodcentered}\mathsf{P}}h_{0}^{-1}$ it is
proportional to the generators of $\mathfrak{p}$ and we can split the last
expression as%
\begin{align}
h_{1} &  =h_{0} \ ,\label{nonl6}\\
e^{\zeta^{\prime}\text{\textperiodcentered}\mathsf{P}} &  =h_{0}%
e^{\zeta\text{\textperiodcentered}\mathsf{P}}h_{0}^{-1} \ . \label{nonl7}%
\end{align}
From eq.$\left(  \ref{nonl7}\right)  $ on sees that it is always possible to
find a linear transformation for the coset parameters $\zeta^{\prime}%
=\tilde{D}\left(  h_{0}\right)  \zeta$. On the other hand, eq.$\left(
\ref{nonl6}\right)  $ says that $h_{1}$ is no longer $\zeta$ dependent and
therefore we can write eq.$\left(  \ref{nonl9}\right)  $ as
\begin{equation}
\bar{\varphi}^{\prime}=D\left(  h_{0}\right)  \bar{\varphi} \ ,%
\end{equation}
which corresponds to a linear transformation for the nonlinear field under
$h_{0}\in H$. Therefore, the restriction to the subgroup $H$ implies that the
nonlinear realization eq.$\left(  \ref{nonl8}-\ref{nonl10}\right)  $ becomes linear.

\newsection{Nonlinear gauge fields}

In the case that the group elements $g_{0}\in G$ are local $g_{0}=g_{0}\left(
x\right)  $, one needs to introduce, as in the case of linear representations,
a nonlinear gauge connection $\mathcal{\bar{A}}$ in order to guarantee that
the derivatives of the fields $\zeta,\bar{\varphi}$ transforms covariantly
with respect to the standard form of nonlinear realizations eq.$\left(
\ref{nonl8}-\ref{nonl10}\right)  $. The linear gauge potential $\mathcal{A}$,
can be naturally divided in terms of gauge fields associated to $H$ and $G/H$%
\begin{equation}
\mathcal{A}=v^{i}\mathsf{V}_{i}+p^{l}\mathsf{P}_{l}.
\end{equation}
Now, under gauge transformations, the linear connection changes as%
\begin{equation}
\mathcal{A}^{\prime}=g_{0}\mathcal{A}g_{0}^{-1}+\dd g_{0}^{-1}g_{0 \ .} %
\end{equation}
Introducing the non linear gauge potential $\mathcal{\bar{A}}$, we can write
no the non-linear gauge fields
\begin{equation}
\mathcal{\bar{A}}=\bar{v}^{i}\mathsf{V}_{i}+\bar{p}^{l}\mathsf{P}_{l}.
\end{equation}

It can be shown that the relation between the gauge field
associated to the linear and nonlinear gauge potential is given by \cite{Callan:1969sn}%
\begin{equation}
\bar{v}^{i}\mathsf{V}_{i}+\bar{p}^{l}\mathsf{P}_{l}=e^{-\zeta
\text{\textperiodcentered}P}\left[  \dd+v^{i}\mathsf{V}_{i}+p^{l}\mathsf{P}%
_{l}\right]  e^{\zeta\text{\textperiodcentered}P}.\label{nonl13}%
\end{equation}
This relation is very interesting because it has exactly the form of a gauge
transformation with parameter $z=e^{-\zeta\text{\textperiodcentered}%
\mathsf{P}}\in G/H$%
\begin{equation}
\mathcal{A}\longrightarrow\mathcal{\bar{A}}=z\mathcal{A}z^{-1}+z\dd z^{-1}%
\end{equation}
However, strictly speaking, this is not true since te nonlinear gauge
connection possesses its own transformation law as we see in the following:
From eq.$\left(  \ref{nonl2}\right)  $ we see that%
\begin{align}
z^{\prime}g  & =h_{1}z,\\
g^{-1}z^{\prime-1}  & =z^{-1}h_{1}^{-1},
\end{align}
taking the exterior derivative, we get
\begin{align}
\dd z^{\prime}g+z^{\prime}\dd g  & =\dd h_{1}z+h_{1}\dd z  \\
\dd g^{-1}z^{\prime-1}+g^{-1}\dd z^{\prime-1}  & =\dd z^{-1}h_{1}^{-1}+z^{-1}%
\dd h_{1}^{-1}%
\end{align}
Thus,
\begin{align}
\mathcal{\bar{A}}^{\prime}  & =z^{\prime}\mathcal{A}^{\prime}z^{\prime
-1}+z^{\prime}\dd z^{\prime-1}\\
& =h_{1}\left(  z\mathcal{A}z^{-1}+\dd z^{-1}\dd z\right)  h_{1}^{-1}+\dd h_{1}%
h_{1}^{-1}\\
& =h_{1}\mathcal{\bar{A}}h_{1}^{-1}+\dd h_{1}h_{1}^{-1}%
\end{align}
Then, under gauge transformations with elements $g\in G$, the nonlinear
connection changes as%
\begin{equation}
\mathcal{\bar{A}}=h_{1}\mathcal{\bar{A}}h_{1}^{-1}+\dd h_{1}h_{1}^{-1}%
\label{nonl11}%
\end{equation}
with $h\in H$. Thus, under the whole group $G$, the gauge potential transforms
as a one-form connection under $h_{1}\in H$ where $h_{1}=h_{1}\left(
g_{0},\zeta\right)  $ is non linear. From eq.$\left(  \ref{nonl11}\right)  $
one sees that the nonlinear gauge fields transform in the following way%
\begin{align}
\bar{v}^{\prime} &  =h_{1}\bar{\nu}h_{1}^{-1}\\
\bar{p}^{\prime} &  =h_{1}\bar{p}h_{1}^{-1}+\dd h_{1}^{-1}h_{1}%
\end{align}
One important observation is that if one writes an action principle in terms
of a gauge potential and its derivatives which is invariant under $H$,
\begin{equation}
S=S\left[  \mathcal{A},\dd\mathcal{A}\right] \ ,
\end{equation}
the replacement of the gauge connection by its nonlinear version does not
change the form of the action and moreover, it guarantees the invariance of
the action not only by the subgroup $H$ but the whole group $G$%
\begin{equation}
S=S\left[  \mathcal{\bar{A}},\dd\mathcal{\bar{A}}\right] \ ,
\end{equation}
enhancing the symmetry form $H$ to $G$. The field strength associated to
$\mathcal{\bar{A}}$ is defined as usual
\begin{equation}
\mathcal{\bar{F}}=\dd\mathcal{\bar{A}}+\mathcal{\bar{A}}\wedge\mathcal{\bar{A}}%
\end{equation}
and under gauge transformations with $g\in G$, it changes as%
\begin{equation}
\mathcal{\bar{F}}^{\prime}=h_{1}\mathcal{\bar{F}}h_{1}^{-1} \ ,%
\end{equation}
where $h_{1}\in H$ is nonlinear.

\newsection{Nonlinear realization of the Poincar\'{e} group}

Let $G=ISO(d-1,1)$ generated by $\left\{  \mathsf{J}_{ab},\mathsf{P}%
_{a}\right\}  $. It is possible to decompose the Poincar\'{e} algebra in term
of two subspaces~\cite{Ste80}

\begin{itemize}
\item The Lorentz subalgebra $\mathfrak{so}(d-1,1)$ generated by $\left\{
\mathsf{J}_{ab}\right\}  $

\item The symmetric coset $\mathfrak{iso}(d-1,1)/\mathfrak{so}(d-1,1)$
generated by $\left\{  \mathsf{P}_{a}\right\}  $
\end{itemize}

This decomposition satisfies%
\begin{align}
\left[  \mathsf{J},\mathsf{J}\right]    & \sim\mathsf{J,}\\
\left[  \mathsf{J},\mathsf{P}\right]    & \sim\mathsf{P,}\\
\left[  \mathsf{J},\mathsf{J}\right]    & \sim\mathsf{J.}%
\end{align}
This means that the commutator of any element in the stability subgroup
$H=SO(d-1,1)$ with an element of the coset $G/H$ will remain in $G/H$. This is
a key ingredient for obtaining nonlinear realizations of the Poincar\'{e} group.

We introduce now a coset coordinate associated to the generators of $G/H$%
\begin{equation}
\phi^{a}\rightarrow\mathsf{P}_{a}.
\end{equation}
To obtain the transformation law of the coset parameter we use eq.$\left(
\ref{nonl3}\right)  $. In fact under a group element of the form $g=1-\rho
^{a}\mathsf{P}_{a}$, the coset coordinate changes as%
\begin{equation}
\delta\phi^{a}=\rho^{a}.
\end{equation}

Now, using eq.$\left(  \ref{nonl13}\right)  $ we have%
\begin{equation}
\bar{e}^{a}\mathsf{P}_{a}+\frac{1}{2}\bar{\omega}^{ab}\mathsf{J}_{ab}%
=\exp\left( \phi^{c}\mathsf{P}_{c}\right)  \left[ \dd+ e^{a}\mathsf{P}_{a}%
+\frac{1}{2}\omega^{ab}\mathsf{J}_{ab}\right]  \exp\left( - \phi^{c}%
\mathsf{P}_{c}\right)  ,
\end{equation}
and then%
\begin{align}
\bar{e}^{a}  & =e^{a}-D_{\omega}\phi^{a},\\
\bar{\omega}^{ab}  & =\omega^{ab}.
\end{align}

\chapter{Chern--Simons supergravity}
\label{ch:app1}
\newsection{$d=3$ Majorana spinors}

The minimal irreducible spinor in three dimensions is a two real component
Majorana spinor. Every Majorana spinor satisfies a reality condition which can
be established by demanding that the Majorana conjugate equals the Dirac
conjugate
\begin{equation}
\bar{\psi}:= \psi^{\top}\mathcal{C}=-\ii\psi^{\top}\Gamma_{1} \ .\label{maj1}%
\end{equation}
Spinors carry indices $\psi_{\alpha}$ and gamma-matrices act on them in such a
way that $\Gamma_{a}\psi:= \left(  \Gamma_{a}\right)  _{~\beta}^{\alpha
}\, \psi_{\alpha}$. In order to raise and lower indices, we introduce matrices
$(\mathcal{C}^{\alpha\beta})$, $(\mathcal{C}_{\alpha\beta})$ related to the
charge conjugation matrix, and we use the convention of raising and lowering
indices according to the NorthWest--SouthEast convention $\left(
\searrow\right)  $. This means that the position of the indices should appear
in that relative position as
\begin{equation}
\psi^{\alpha}=\mathcal{C}^{\alpha\beta}\, \psi_{\beta} \qquad
\mbox{and} \qquad \psi_{\alpha}%
=\psi^{\beta}\, \mathcal{C}_{\alpha\beta} \ , \label{maj2}
\end{equation}
which implies that
\beq
\mathcal{C}^{\alpha\beta}\, \mathcal{C}_{\gamma\beta}=\delta_{\gamma}^{\alpha
} \qquad \mbox{and} \qquad \text{ }\mathcal{C}_{\beta\alpha}\, \mathcal{C}^{\beta\gamma}=\delta_{\alpha
}^{\gamma} \ .
\eeq
We choose the identifications in such a way that the Majorana conjugate
$\bar{\psi}$ is written as $\psi^{\alpha}$. Comparing eq.~(\ref{maj1})
with eq.~(\ref{maj2}), one then finds
$(\mathcal{C}^{\alpha\beta})=\mathcal{C}^{\top}$ and $(\mathcal{C}_{\alpha\beta
})=\mathcal{C}^{-1}$.

Note that in section (\ref{3dsugra}) we have used the following presentation for the super Poincar\'e algebra in three-dimensions
\begin{align}
\left[  \mathsf{J}_{ab},\mathsf{J}_{cd}\right]    & =-\ii\left(  \eta
_{bc}\mathsf{J}_{ad}+\eta_{ad}\mathsf{J}_{bc}-\eta_{bd}\mathsf{J}_{ac}%
-\eta_{ac}\mathsf{J}_{bd}\right)  , \nonumber\\
\left[  \mathsf{J}_{ab},\mathsf{P}_{c}\right]    & =-\ii\left(  \eta
_{bc}\mathsf{P}_{a}-\eta_{ac}\mathsf{P}_{b}\right)  ,\nonumber\\
\left[  \mathsf{P}_{a},\mathsf{P}_{b}\right]    & =0,\nonumber\\
\left[  \mathsf{Q}_{\alpha},\mathsf{J}_{ab}\right]    & =-\frac{\ii}{2}\left(
\Gamma_{ab}\right)  _{\alpha}^{~\beta}\mathsf{Q}_{\beta},\nonumber\\
\left\{  \mathsf{Q}_{\alpha},\mathsf{Q}_{\beta}\right\}    & =\left(
\Gamma^{a}\right)  _{\alpha\beta}\mathsf{P}_{a},\nonumber\\
\left[  \mathsf{Q}_{\alpha},\mathsf{P}_{a}\right]    & =0.
\end{align}

The invariant tensor associated to this gauge algebra are given in~\cite{Banh96a} 
\begin{align}
\left\langle \mathsf{J}_{ab}\mathsf{P}_{c}\right\rangle  & =\epsilon_{abc} \ ,\\
\left\langle \mathsf{Q}_{\alpha}\mathsf{Q}_{\beta}\right\rangle  &
=-i\mathcal{C}_{\alpha\beta} \ .%
\end{align}

\newsection{Five dimensional supergravity Lagrangian}\label{SUGRALag}

The supersymmetric extension of the AdS algebra in five dimensions is the Lie superalgebra $\mathfrak{su}(  2,2|N)  $~\cite{Chamseddine:1990gk}. The
associated gauge field decomposes into generators as
\begin{equation}
\mathcal{A}=e^{a}\, \mathsf{P}_{a}+\mbox{$\frac{1}{2}$}\, \omega^{ab}\, \mathsf{J}_{ab}%
+a_{~n}^{m}\, \mathsf{M}_{~m}^{n}+b\, \mathsf{K}+\bar{\psi}_{\alpha}^{k}%
\, \mathsf{Q}_{k}^{\alpha}-\bar{\mathsf{Q}}_{\beta}^{k}\, \psi_{k}^{\beta} \ .
\end{equation}
Here the generators $\left\{ \mathsf{P}_a, \mathsf{J}_{ab} \right\}$ span
an $\mathfrak{so}(4,2)  $ subalgebra, $\mathsf{M}_{~m}^{n}$ are
$N^{2}-1$ generators of $SU(  N)  $, $\mathsf{K}$
generates a $U(1)  $ subgroup, and $\mathsf{Q}_{k}^{\alpha
},\bar{\mathsf{Q}}_{\beta}^{k}$ are the supersymmetry generators. The
Chern--Simons Lagrangian associated to this superalgebra is given by \cite{Chamseddine:1990gk,Troncoso:1998ng,Izaurieta:2006wv}
\begin{equation}
\mathscr{L}_{\text{CS}}^{\left(  5\right)  }=\mathscr{L}_{\psi}
+\mathscr{L}_{a}+\mathscr{L}_{b}+\mathscr{L}_{e}
\end{equation}
where%
\begin{align}
\mathscr{L}_{\psi}&  =\mbox{$\frac{3}{2\ii}$}\, \left(  \bar{\psi}^n%
\wedge \mathcal{R}\wedge \nabla\psi_n +\bar{\psi}^{n}\wedge \mathcal{F}_{~n}^{m}\wedge \nabla\psi_{m}%
-\nabla\bar{\psi}^n\wedge \mathcal{R}\wedge \psi_n -\nabla\bar{\psi}^{n}\wedge \mathcal{F}_{~n}^{m}%
\wedge \psi_{m}\right) \ , \nonumber \\[4pt]
\mathscr{L}_{a} &  =\mbox{$\frac{3}{N}$}\, \dd b \wedge
\mathsf{Tr}\big(  a\wedge \dd a+\mbox{$\frac{2}{3}$}\, a^{3}\big)  -\ii \mathsf{Tr}\big( a\wedge \left(
\dd a\right)  ^{2}+\mbox{$\frac{3}{2}$}\, a^{3}\wedge \dd a+\mbox{$\frac{3}{5}$}\, a^{5}\big) \ , \nonumber \\[4pt]
\mathscr{L}_{b} &  = \big( \mbox{$ \frac{1}{16}-\frac{1}
{N^{2}}$} \big) \, b\wedge \left(  \dd b\right)  ^{2}-\mbox{$\frac{3}{4l^{2}}$}\, b\wedge \big(
T^{a}\wedge T_{a}-R_{ab}\wedge e^{a}\wedge e^{b}-\mbox{$\frac{l^{2}}{2}$}\, R^{ab}\wedge R_{ab}\big) \ , \nonumber \\[4pt]
\mathscr{L}_{e} &  =\mbox{$\frac{3}{8l}$}\, \epsilon_{abcdh}\, \big(
R^{ab}\wedge R^{cd}+\mbox{$\frac{2}{3}$}\, R^{ab}\wedge e^{c}\wedge e^{d}+\mbox{$\frac{1}{5}$}\, e^{a}\wedge e^{b}\wedge e^{c}\wedge
e^{d}\big)\wedge e^{h} \ ,
\end{align}
and%
\begin{align}
\mathcal{R} &  =\ii \big( \mbox{$ \frac{1}{4}+\frac{1}{N}$} \big)\,
\big( \dd b+\mbox{$\frac{\ii}{2l}$}\, \bar{\psi}^n\wedge \psi_n \big) + \mbox{$\frac{1}{2}$}\, \big(  T^{a}%
-\mbox{$\frac{1}{4}$}\, \bar{\psi}^n\wedge \Gamma^{a}\psi_n \big)\, \Gamma_{a} \nonumber \\ & \qquad +\, \mbox{$\frac{1}{4}$}\, \big(
R^{ab}+\mbox{$\frac{1}{l}$}\, e^{a}\wedge
e^{b}+\mbox{$\frac{1}{4l}$}\, \bar{\psi}^n\wedge \Gamma^{ab}\psi_n \big)\,
\Gamma_{ab} \ , \nonumber \\[4pt]
\mathcal{F}_{~n}^{m} & =f_{~n}^{m}-\mbox{$\frac{1}{2l}$}\, \bar{\psi}^{m}\wedge \psi_{n} \ .
\end{align}
Here the spinor covariant derivatives are defined by
\begin{align}
\nabla\psi_k& =\dd \psi_k+\mbox{$\frac1{2l}$}\, e^a\wedge \Gamma_a\psi_k+\mbox{$\frac14$} \, \omega^{ab}\wedge\Gamma_{ab}\psi_k -a^n_{~k}\wedge \psi_n+\ii\big(\mbox{$\frac14-\frac1N$} \big)\, b\wedge\psi_k \ , \nonumber \\[4pt]
\nabla\bar{\psi}^k& =\dd \bar{\psi}^k-\mbox{$\frac1{2l}$}\, e^a\wedge \bar{\psi}^k\Gamma_a-\mbox{$\frac14$} \, \omega^{ab}\wedge\bar{\psi}^k\Gamma_{ab}+ a^k_{~n}\wedge \bar{\psi}^n-\ii\big(\mbox{$\frac14-\frac1N$} \big)\, b\wedge\bar{\psi}^k \ ,
\end{align}
while $f=\dd a+a\wedge a$ is the curvature of the $SU(N)$ gauge field $a$. 

\newsection{Representation of $\mathfrak{su}(  2,2|1)  $} \label{explrep}

For simplicity we consider now the particular instance
$N=1$. This case furnishes the minimum number $\mathcal{N} =2$ of supersymmetries, and the commutation relations are given by%
\begin{align}
\left[  \mathsf{K},\mathsf{Q}^{\rho}\right]   &  =\mbox{$\frac{3\ii}{4}$}\, \mathsf{Q}%
^{\rho} \ , \nonumber \\[4pt]
\left[  \mathsf{K},\bar{\mathsf{Q}}_{\rho}\right]   &  =-\mbox{$\frac{3\ii}%
{4}$}\, \bar{\mathsf{Q}}_{\rho} \ , \nonumber \\[4pt]
\left[  \mathsf{P}_{a},\mathsf{P}_{b}\right]   &  = \mbox{$\frac{1}{l^{2}}$}\,
\mathsf{J}_{ab} \ , \nonumber \\[4pt]
\left[  \mathsf{P}_{a},\mathsf{J}_{bc}\right] &  =\eta_{ba}\, \mathsf{P}%
_{c}-\eta_{ac}\, \mathsf{P}_{b} \ , \nonumber \\[4pt]
\left[  \mathsf{P}_{a},\mathsf{Q}^{\rho}\right]   &  =-\mbox{$\frac{1}{2l}$}\, \left(
\Gamma_{a}\right)  _{~\gamma}^{\rho}\, \mathsf{Q}^{\gamma} \ , \nonumber \\[4pt]
\left[  \mathsf{P}_{a},\bar{\mathsf{Q}}_{\rho}\right]   &  =\mbox{$\frac{1}%
{2l}$}\, \bar{\mathsf{Q}}_{\gamma}\, \left(  \Gamma_{a}\right)  _{~\rho}^{\gamma} \ , \nonumber \\[4pt]
\left[\mathsf{J}_{ab},\mathsf{J}_{cd}\right] &  =\eta_{cb}\, \mathsf{J}%
_{ad}-\eta_{ac}\, \mathsf{J}_{bd}+\eta_{db}\, \mathsf{J}_{ca}-\eta_{ad}%
\, \mathsf{J}_{cb} \ , \nonumber \\[4pt]
\left[  \mathsf{J}_{ab},\mathsf{Q}^{\rho}\right]   &  =- \mbox{$\frac{1}{2}$}\, \left(
\Gamma_{ab}\right)  _{~\gamma}^{\rho}\, \mathsf{Q}^{\gamma} \ , \nonumber \\[4pt]
\left[  \mathsf{J}_{ab},\bar{\mathsf{Q}}_{\rho}\right]   &  =\mbox{$\frac{1}%
{2}$}\, \bar{\mathsf{Q}}_{\gamma}\left(  \Gamma_{ab}\right)  _{~\rho}^{\gamma} \ , \nonumber \\[4pt]
\left\{  \mathsf{Q}^{\rho},\bar{\mathsf{Q}}_{\sigma}\right\}   &
=-4\ii\delta_{~\sigma}^{\rho}\,\mathsf{K}+2\left(  \Gamma^{a}\right)  _{~\sigma
}^{\rho}\, \mathsf{P}_{a}-\left(  \Gamma_{ab}\right)  _{~\sigma}^{\rho}%
\, \mathsf{J}_{ab} \ .
\end{align}
According to (\ref{repexpl})--(\ref{repexpl1}) the matrix
generators explicitly read as
\begin{equation}
\Gamma_{0}=%
\begin{pmatrix}
0 & \ii & 0 & 0\\
\ii & 0 & 0 & 0\\
0 & 0 & 0 & -\ii\\
0 & 0 & -\ii & 0
\end{pmatrix} \ , \qquad \Gamma_{1}=%
\begin{pmatrix}
0 & -\ii & 0 & 0\\
\ii & 0 & 0 & 0\\
0 & 0 & 0 & \ii\\
0 & 0 & -\ii & 0
\end{pmatrix} \ , \qquad \Gamma_{2}=%
\begin{pmatrix}
1 & 0 & 0 & 0\\
0 & -1 & 0 & 0\\
0 & 0 & -1 & 0\\
0 & 0 & 0 & 1
\end{pmatrix} \ ,
\nonumber \end{equation}
\begin{equation}
\Gamma_{3}=%
\begin{pmatrix}
0 & 0 & 1 & 0\\
0 & 0 & 0 & 1\\
1 & 0 & 0 & 0\\
0 & 1 & 0 & 0
\end{pmatrix}
\ , \qquad \Gamma_{4}=%
\begin{pmatrix}
0 & 0 & -\ii & 0\\
0 & 0 & 0 & -\ii\\
\ii & 0 & 0 & 0\\
0 & \ii & 0 & 0
\end{pmatrix} \ ,
\end{equation}
and using (\ref{Gammaab}) we find
\begin{equation}
\Gamma_{01}=%
\begin{pmatrix}
-1 & 0 & 0 & 0\\
0 & 1 & 0 & 0\\
0 & 0 & -1 & 0\\
0 & 0 & 0 & 1
\end{pmatrix}
\ , \qquad \Gamma_{02}=%
\begin{pmatrix}
0 & -\ii & 0 & 0\\
\ii & 0 & 0 & 0\\
0 & 0 & 0 & -\ii\\
0 & 0 & \ii & 0
\end{pmatrix}
\ , \qquad \Gamma_{03}=%
\begin{pmatrix}
0 & 0 & 0 & \ii\\
0 & 0 & \ii & 0\\
0 & -\ii & 0 & 0\\
-\ii & 0 & 0 & 0
\end{pmatrix} \ ,
\nonumber \end{equation}
\begin{equation}
\Gamma_{04}=%
\begin{pmatrix}
0 & 0 & 0 & 1\\
0 & 0 & 1 & 0\\
0 & 1 & 0 & 0\\
1 & 0 & 0 & 0
\end{pmatrix}
\ , \qquad \Gamma_{12}=%
\begin{pmatrix}
0 & \ii & 0 & 0\\
\ii & 0 & 0 & 0\\
0 & 0 & 0 & \ii\\
0 & 0 & \ii & 0
\end{pmatrix}
\ , \qquad \Gamma_{13}=%
\begin{pmatrix}
0 & 0 & 0 & -\ii\\
0 & 0 & \ii & 0\\
0 & \ii & 0 & 0\\
-\ii & 0 & 0 & 0
\end{pmatrix} \ ,
\nonumber \end{equation}
\begin{equation}
\Gamma_{14}=%
\begin{pmatrix}
0 & 0 & 0 & -1\\
0 & 0 & 1 & 0\\
0 & -1 & 0 & 0\\
1 & 0 & 0 & 0
\end{pmatrix}
\ , \qquad \Gamma_{23}=%
\begin{pmatrix}
0 & 0 & 1 & 0\\
0 & 0 & 0 & -1\\
-1 & 0 & 0 & 0\\
0 & 1 & 0 & 0
\end{pmatrix}
\ , \qquad \Gamma_{24}=%
\begin{pmatrix}
0 & 0 & -\ii & 0\\
0 & 0 & 0 & \ii\\
-\ii & 0 & 0 & 0\\
0 & \ii & 0 & 0
\end{pmatrix} \ ,
\nonumber \end{equation}
\begin{equation}
\Gamma_{34}=%
\begin{pmatrix}
\ii & 0 & 0 & 0\\
0 & \ii & 0 & 0\\
0 & 0 & -\ii & 0\\
0 & 0 & 0 & -\ii
\end{pmatrix} \ .
\end{equation}
It is then easy to show that this particular choice of basis for the Lie
algebra $\mathfrak{su}(2,2)$ has traceless generators all satisfying the Clifford algebra relations (\ref{cliffal}).

The $\mathfrak{su}(
2,2| 1)  $-invariant tensor of rank three can be computed from this
representation as the supersymmetrized supertraces of products of
triples of supermatrices. The non-vanishing components are given by~\cite{Izaurieta:2006wv}
\begin{align}
\left\langle \mathsf{J}_{ab}\, \mathsf{J}_{cd}\, \mathsf{P}_{e}\right\rangle  &
=-\mbox{$\frac{\gamma}{2l}$}\, \epsilon_{abcde} \ , \nonumber \\[4pt]
\left\langle \mathsf{K\, K\, K}\right\rangle  &  =-\mbox{$\frac{15}{16}$} \ , \nonumber \\[4pt]
\left\langle \mathsf{K\, P}_{a}\, \mathsf{P}_{b}\right\rangle  &
=-\mbox{$\frac{1}{4l^{2}}$}\, \delta_{ab} \ , \nonumber \\[4pt]
\left\langle \mathsf{J}_{ab}\, \mathsf{K\, J}_{cd}\right\rangle  &
=-\mbox{$\frac{1}{4}$}\, (\delta_{ad}\, \delta_{bc}+\delta_{ac}\, \delta_{bd}) \ , \nonumber \\[4pt]
\left\langle \mathsf{Q}^{\alpha}%
\, \mathsf{K\, \bar{Q}}_{\beta}\right\rangle  &  =\mbox{$\frac{5}{2l}$}\, \delta^\alpha_{~\beta} \ , \nonumber \\[4pt]
\left\langle \mathsf{Q}^{\alpha
}\, \mathsf{P}_{a}\, \bar{\mathsf{Q}}_{\beta}\right\rangle  &  =-\mbox{$\frac{\ii}{l}$}\, 
\left(  \Gamma_{a}\right)  _{~\beta}^{\alpha}%
\ , \nonumber \\[4pt]
\left\langle \mathsf{Q}^{\alpha
}\, \mathsf{J}_{ab}\, \bar{\mathsf{Q}}_{\beta}\right\rangle  &  =-\mbox{$\frac{\ii}{l}$}\,
\left(  \Gamma_{ab}\right)  _{~\beta}^{\alpha}%
 \ ,
\end{align} \label{invtens}
where $\gamma$ is an arbitrary constant.

\chapter{The $S-$expansion procedure}
\label{ch:app3}

\newsection{$S-$Expansion method for Lie algebras}

In this section we describe general aspects of the $S-$expansion mechanism \cite{Iza06b,Izaurieta:2006ia,Rodriguez:2006xx}.

\begin{definition}
Let $S=\left\{  \lambda_{\alpha},\alpha=1,...,N\right\}  $ be an abelian
finite semigroup. We define the two-selector $K_{\alpha\beta}^{~~\gamma}$ as follows
\begin{equation}
K_{\alpha\beta}^{~~\gamma}=\left\{
\begin{tabular}
[c]{l}%
$1$, if $\ \lambda_{\alpha}\lambda_{\beta}=\lambda_{\gamma}$\\
$0$, otherwise
\end{tabular}
\right. \label{twosel}
\end{equation}

\end{definition}

This definition induces the multiplication law of the semigroup, i.e.,
\begin{equation}
\lambda_{\alpha}\lambda_{\beta}=K_{\alpha\beta}^{~~\gamma}\lambda_{\gamma
}.\label{sexp1}%
\end{equation}
Since $S$ is abelian, the two-selectors are symmetric in the lower indices.
Moreover, the two-selectors allows us to define matrix representations of the
elements of the semigroup%
\begin{equation}
\left[  \lambda_{\alpha}\right]  _{\rho}^{~\sigma}=K_{\alpha\rho}^{~\ \sigma}
\end{equation}
and it is direct to show, using the semigroup axioms, that these matrices
satisfy eq.$\left(  \ref{sexp1}\right)  $.

Under the same considerations, it is possible to extend the definition of
two-selectors to $n$--selectors. This can be by taking the product of $n$
element of the semigroup
\begin{equation}
\lambda_{\alpha_{1}}\ldots\lambda_{\alpha_{n}}=K_{\alpha_{1}\ldots\alpha_{n}%
}^{~~~~~~~~~\gamma}\lambda_{\gamma}.
\end{equation}

In what follows it will be implicitly assumed that whenever we refer to an
element $\lambda_{\alpha}\in S$, we mean to the matrix representation of the
elements of the semigroup over a vector space. For instance, the representation
defined by the two-selectors.

Let $\mathfrak{g}$ be a Lie algebra with structure constants $C_{AB}^{~~C}$.
Then, according to ref.~\cite[Theorem~3.1]{Iza06b}, the product $\mathfrak{S}%
=S\times\mathfrak{g}$ is also a Lie algebra and it is given by%
\begin{equation}
\left[  \mathsf{T}_{\left(  A,\alpha\right)  },\mathsf{T}_{\left(
B,\beta\right)  }\right]  =K_{\alpha\beta}^{~~\gamma}C_{AB}^{~~C}%
\mathsf{T}_{\left(  C,\gamma\right)  }\label{sexp2}%
\end{equation}
this can be clearly seen by taking the generators of the expanded algebra
$\mathfrak{S}$ as%
\begin{equation}
\mathsf{T}_{\left(  A,\alpha\right)  }=\lambda_{\alpha}\mathsf{T}_{A},
\end{equation}
so the commutator eq.$\left(  \ref{sexp2}\right)  $ reads%
\begin{align}
\left[  \mathsf{T}_{\left(  A,\alpha\right)  },\mathsf{T}_{\left(
B,\beta\right)  }\right]   & =\lambda_{\alpha}\lambda_{\beta}\left[
\mathsf{T}_{A},\mathsf{T}_{B}\right]  ,\\
& =K_{\alpha\beta}^{~~\gamma}C_{AB}^{~~C}\lambda_{\gamma}\mathsf{T}_{C},\\
& =K_{\alpha\beta}^{~~\gamma}C_{AB}^{~~C}\mathsf{T}_{\left(  C,\gamma\right)
}.
\end{align}
\subsection{Resonant subalgebra} \label{resexp}

There are cases in which it is possible to systematically
extract Lie subalgebras from $\mathrm{S}\times\mathfrak{g}$. For
instance,
let us suppose that the Lie algebra $\mathfrak{g}$ in
$\mathfrak{S}=S\times\mathfrak{g}$ has, as a vector space, decompositions in
terms of subspaces $\mathsf{V}_{p}$%
\begin{equation}
\mathfrak{g}=\bigoplus\limits_{p\in I}\mathsf{V}_{p},\label{sexp3}%
\end{equation}
where $I$ denotes a set of indices. Let us suppose in addition that the
commutation relations of the Lie algebra $\mathfrak{g}$ have the following
structure%
\begin{equation}
\left[  \mathsf{V}_{p},\mathsf{V}_{q}\right]  \subset\bigoplus\limits_{r\in
i_{\left(  p,q\right)  }}\mathsf{V}_{r}
\end{equation}
where $i_{\left(  p,q\right)  }\subset I$ encodes the information of the
subspaces structure of $\mathfrak{g}$.

Now, if the semigroup $S$ admits a subset decomposition
\begin{equation}
S=\bigcup\limits_{p\in I}S_{p}
\end{equation}
such that the subsets $S_{p}\times S_{q}$ satisfy%
\begin{equation}
S_{p}\times S_{q}=\bigcap\limits_{r\in i_{\left(  p,q\right)  }}S_{r}
\end{equation}
then, we say that the semigroup admits a decomposition which is \textit{resonant}
respect to the algebra $\left(  \ref{sexp3}\right)  $.

The structure
\begin{equation}
\mathfrak{S}_{R}=\bigoplus\limits_{p\in I}S_{p}\times\mathsf{V}_{p}
\end{equation}
defines a sublagebra $\mathfrak{S}_{R}\subset\mathfrak{S}$ called the resonant
subalgebra of the $S-$expanded algebra $\mathfrak{S}$~\cite[Theorem~4.2]{Iza06b}.

\subsection{$0_{S}-$Reduced algebra} \label{0reduc}

Let us consider a semigroup $S=\left\{  \lambda_{i},0_{S}\right\}  $ where $i=1,\ldots,N$ with a zero element $0_{ S}$. This means that $0_{ S}\cdot \lambda_{\alpha}=0_{ S}$ for all $\lambda_{\alpha
}\in \mathrm{S}$. Denoting the zero
element as $\lambda_{N+1}=0_{S}$, the $S-$expanded algebra $\mathfrak{S}%
=S\times\mathfrak{g}$ takes the form%
\begin{align*}
\left[  \mathsf{T}_{\left(  A,i\right)  },\mathsf{T}_{\left(  B,j\right)
}\right]   & =K_{ij}^{~~k}C_{AB}^{~~C}\mathsf{T}_{\left(  C,k\right)  }%
+K_{ij}^{~~N+1}C_{AB}^{~~C}\mathsf{T}_{\left(  C,N+1\right)  },\\
\left[  \mathsf{T}_{\left(  A,N+1\right)  },\mathsf{T}_{\left(  B,j\right)
}\right]   & =C_{AB}^{~~C}\mathsf{T}_{\left(  C,N+1\right)  },\\
\left[  \mathsf{T}_{\left(  A,N+1\right)  },\mathsf{T}_{\left(  B,N+1\right)
}\right]   & =C_{AB}^{~~C}\mathsf{T}_{\left(  C,N+1\right)  }.
\end{align*}
The $0_{S}-$reduction process consist in removing from the expanded algebra
all the generators of the form $\mathsf{T}_{\left(  C,N+1\right)  }%
=0_{S}\mathsf{T}_{C}$ . In other words, the whole sector
$0_{ S}\times\mathfrak{g}$ can be removed from the algebra
by imposing $0_{S}\times\mathfrak{g}=0$. In this way, the remaining sector 
\begin{equation}
\left[  \mathsf{T}_{\left(  A,i\right)  },\mathsf{T}_{\left(  B,j\right)
}\right]  =K_{ij}^{~~k}C_{AB}^{~~C}\mathsf{T}_{\left(  C,k\right)  }.
\end{equation}
is referred to as the $0_{ S}$-reduced algebra, which is still a Lie
algebra~\cite[Theorem~6.1]{Iza06b}.

\subsection{Invariant tensors} \label{invite}

The problem of finding invariant tensors associated to Lie algebras in highly
nontrivial. Usually they are constructed in terms of symmetrized traces but
this is not the only possibility.

The $S-$expansion method of Lie algebras provides a mechanism to obtain
invariant tensors for the expanded algebra $\mathfrak{S}$ starting from an
invariant tensor of the original algebra $\mathfrak{g}$. The
systematic process is contained in the following theorem

\begin{theorem}
Let $S$ be an abelian semigroup, $\mathfrak{g}$ a Lie algebra with base
element $\mathsf{T}_{A}$ and let $\left\langle \mathsf{T}_{A_{1}}%
\ldots\mathsf{T}_{A_{n}}\right\rangle $ an invariant tensor for $\mathfrak{g}$.
Then, the expression%
\begin{equation}
\left\langle \mathsf{T}_{\left(  A_{1},\alpha_{1}\right)  }\ldots
\mathsf{T}_{\left(  A_{n},\alpha_{n}\right)  }\right\rangle =\alpha_{\gamma
}K_{\alpha_{1}\ldots\alpha_{n}}^{~~~~\gamma}\left\langle \mathsf{T}_{A_{1}%
}\ldots\mathsf{T}_{A_{n}}\right\rangle ,
\end{equation}
corresponds to an invariant tensor for the $S-$expanded algebra $\mathfrak{S}%
=S\times\mathfrak{g}$, where $\alpha_{\gamma}$ are arbitrary constants and
$K_{\alpha_{1}\ldots\alpha_{n}}^{~~~~\gamma}$ is the $n-$selector.

\begin{proof}
See  \cite[Theorem 7.1, 7.2]{Iza06b}
\end{proof}
\end{theorem}

\chapter{Classical gauge groups}
\label{ch:app4}

In this appendix we summarize the group theory data which are used in
section~\ref{sector3} in the case when the gauge symmetry belongs to
one of the four infinite families $A_n,B_n,C_n,D_n$ of  classical Lie
groups in the Cartan classification; we consider each family in turn. Below $\{E_{i,j}\}_{i,j=1}^n$
denotes the orthonormal basis of $n\times n$ matrix units with elements
$\left( E_{i,j}\right) _{kl}=\delta_{ik}\, \delta_{jl}$, and $\left\{  e_{i}\right\}_{i=1}^n  $ is
the canonical orthonormal basis of~$%
\mathbb{R}
^{n}$.
\newsection{Cartan--Weyl decomposition}
\subsubsection*{$\mbf{\mathcal{G}=U(n)}$}

\begin{equation}
\centering%
\begin{tabular}
[c]{l|l|l}
Positive roots $\alpha>0$& $e_{i}-e_{j}$ & $1\leq i<j\leq n$\\\hline
Cartan generators & ${H}_{i}={E}_{i,i}$ & $1\leq i\leq
n$\\\hline
Root vectors & $X_{e_{i}-e_{j}}=E_{i,j}$ & $i\neq j$,
$i,j=1,\ldots,n$\\\hline
Weyl symmetry $\Wcal$ & $S_{n}$ & 
\end{tabular}
\ \ \ \ \
\label{u1}%
\end{equation}

\subsubsection*{$\mbf{\mathcal{G}=SO(2n+1)}  $}

\begin{equation}
\begin{tabular}
[c]{l|l|l}
Positive roots $\alpha>0$ & $e_{i}\pm e_{j}$ & $1\leq i<j\leq n$\\
& $e_{i}$ & $1\leq i\leq n$\\\hline
Cartan generators & $H_{i}=E_{i,i}-E_{i+n,i+n}$ &
$1\leq i\leq n$\\\hline
Root vectors & $X_{e_{i}-e_{j}}=E_{j+1,i+1}-E_{i+n+1,j+n+1}$ & $i\neq j$\\
& $X_{e_{i}+e_{j}}=E_{i+n+1,j+1}-E_{j+n+1,i+1}$ &
$i<j$\\
& $X_{e_{i}}=E_{1,i+1}-E_{i+n+1,1}$ & $1\leq i\leq
n$\\\hline
Weyl symmetry $\Wcal$ & $S_{n}\ltimes\left(
\mathbb{Z}
_{2}\right)  ^{n}$ &
\end{tabular}
\end{equation}

\subsubsection*{$\mbf{\mathcal{G}=Sp(2n) } $}

\begin{equation}
\begin{tabular}
[c]{l|l|l}
Positive roots $\alpha>0$ & $e_{i}\pm e_{j}$ & $1\leq i<j\leq n$\\
& $e_{2i}$ & $1\leq i\leq n$\\\hline
Cartan generators & $H_{i}=E_{i,i}-E_{i+n,i+n}$ &
$1\leq i\leq n$\\\hline
Root vectors & $X_{e_{i}-e_{j}}=E_{j,i}-E_{i+n,j+n}$ & $i\neq j$\\
& $X_{e_{i}+e_{j}}=E_{i+n,j}-E_{j+n,i}$ &
$i<j$\\
& $X_{2e_{i}}=E_{i+n,i}$ & $1\leq i\leq n$\\\hline
Weyl symmetry $\Wcal$ & $S_{n}\ltimes\left(
\mathbb{Z}
_{2}\right)  ^{n}$ &
\end{tabular}
\end{equation}

\subsubsection*{$\mbf{\mathcal{G}=SO(2n)  }$}

\begin{equation}
\begin{tabular}
[c]{l|l|l}
Positive roots $\alpha>0$ & $e_{i}\pm e_{j}$ & $1\leq i<j\leq n$\\\hline
Cartan generators & $H_{i}=E_{i,i}-E_{i+n,i+n}$ &
$1\leq i\leq n$\\\hline
Root vectors & $X_{e_{i}-e_{j}}=E_{j,i}-E_{i+n,j+n}$ & $i\neq j$\\
& $X_{e_{i}+e_{j}}=E_{i+n,j}-E_{j+n,i}$ &
$i<j$\\
& $X_{-e_{i}-e_{j}}=E_{j,i+n}-E_{i,j+n}$ &
$i<j$\\\hline
Weyl symmetry $\Wcal$ & $S_{n}\ltimes\left(
\mathbb{Z}
_{2}\right)  ^{n-1}$ &
\end{tabular}
\end{equation}

\bibliographystyle{utphys}
\bibliography{Bibliography}

\end{document}